\documentclass[11pt]{article}
\usepackage[colorlinks=false]{hyperref}
\usepackage{url}
\usepackage[utf8]{inputenc}
\usepackage{microtype}
\usepackage{color}
\usepackage{xcolor}
\usepackage{xspace}
\usepackage{graphicx}

\usepackage{algorithmic}
\usepackage{algorithm}
\usepackage{fullpage}

\usepackage{enumitem}
\usepackage[margin=1in]{geometry}
\usepackage{subcaption}
\let\proof\relax\let\endproof\relax
\usepackage{amsthm}
\usepackage{amsmath}
\usepackage{amsthm}
\usepackage{amssymb}
\usepackage{esint}
\usepackage[capitalise]{cleveref}
\usepackage{thmtools}
\usepackage{thm-restate}

\newenvironment{innerproof}
 {\proof}
 {\endproof}

\newtheorem{thm}{Theorem}[section]
\newtheorem{problem}{Open problem}[section]
\newtheorem{lemma}[thm]{Lemma}
\newtheorem{observation}[thm]{Observation}
\newtheorem{proposition}[thm]{Proposition}

\newtheorem{definition}[thm]{Definition}

\newtheorem{claim}[thm]{Claim}

\newcommand{\tree}{\mathrm{tree}}

\newcommand{\eps}{\varepsilon}

\newcommand{\Oh}{\mathcal{O}}

\newcommand{\bb}{\mathcal{B}}

\newcommand{\ww}{\mathcal{W}}
\newcommand{\cc}{\mathcal{C}}

\newcommand{\tw}{\mathrm{tw}}
\newcommand{\mvp}{\mathrm{mvp}}
\newcommand{\vp}{\textsc{vp}}
\newcommand{\opt}{\mathsf{OPT}}

\newcommand{\unbreak}{$c$-inseparable\xspace}
\newcommand{\brb}[1]{\langle #1 \rangle}
\newcommand{\nice}{single-faced\xspace}
\newcommand{\circum}{circumscribed\xspace}
\newcommand{\planardel}{\textsc{Vertex planarization}\xspace}

\ifdefined\DEBUG{}
\newcommand{\mic}[1]{{\color{blue}{#1}}}

\def\rem#1{{\marginpar{\raggedright\scriptsize #1}}}
\newcommand{\micr}[1]{\rem{\textcolor{blue}{\(\bullet \) #1}}}

\newcommand{\bmp}[1]{{\color{purple}{#1}}}
\newcommand{\bmpr}[1]{\rem{\textcolor{purple}{\(\bullet \) #1}}}

\else
\newcommand{\mic}[1]{#1}
\newcommand{\bmp}[1]{#1}
\newcommand{\micr}[1]{ }
\newcommand{\bmpr}[1]{ }
\fi

\title{Lossy Planarization: A Constant-Factor Approximate Kernelization for Planar Vertex Deletion\footnote{This project has received funding from the European Research Council (ERC) under the European Union's Horizon 2020 research and innovation programme (grant agreement No 803421, ReduceSearch).}}

\author{Bart M. P. Jansen\footnote{Address: \texttt{b.m.p.jansen@tue.nl}}  \\ Eindhoven University of~Technology
\and Micha{\l} W{\l}odarczyk\footnote{Address: \texttt{m.wlodarczyk@tue.nl}} \\ Eindhoven University of~Technology}

\author{Bart M. P. Jansen\footnote{Address: \texttt{b.m.p.jansen@tue.nl}}  \\ Eindhoven University of~Technology
\and Micha{\l} W{\l}odarczyk\footnote{Address: \texttt{m.wlodarczyk@tue.nl}} \\ Eindhoven University of~Technology}

\date{}

\begin{document}
\maketitle{}

\thispagestyle{empty} 
 
\begin{abstract}
In the \textsc{$\mathcal{F}$-minor-free deletion} problem we are given an undirected graph $G$ and the~goal is to find \bmp{a minimum vertex set that intersects} all minor models \bmp{of graphs} from the family $\mathcal{F}$.
This captures numerous important problems including \textsc{Vertex cover}, \textsc{Feedback vertex set}, \textsc{Treewidth-$\eta$ modulator}, and \textsc{Vertex planarization}.
In the latter one, we ask for \bmp{a~minimum vertex set} whose removal makes the graph planar.
This is a special case of  \textsc{$\mathcal{F}$-minor-free deletion} for the family $\mathcal{F} = \{K_5, K_{3,3}\}$.

Whenever the family $\mathcal{F}$ contains at least one planar graph, then
\textsc{$\mathcal{F}$-minor-free deletion} is known to admit a constant-factor approximation algorithm and a~polynomial kernelization
[Fomin, Lokshtanov, Misra, and Saurabh, FOCS'12].
A~polynomial kernelization is a~polynomial-time algorithm that, given a graph $G$ and integer $k$, outputs a graph $G'$ on $poly(k)$ vertices and integer $k'$, so that $ \mathsf{OPT}(G) \le k$ if and only if $ \mathsf{OPT}(G') \le k'$.
The \textsc{Vertex planarization} problem is arguably the simplest setting for which $\mathcal{F}$ does not contain a~planar graph and
the existence of~a~constant-factor approximation or a~polynomial kernelization remains a major open problem.

In this work we show that \textsc{Vertex planarization} admits an algorithm which is a combination of both approaches.
Namely, we present a polynomial $\alpha$-approximate kernelization, for some constant $\alpha > 1$, \bmp{based} on the framework of lossy kernelization [Lokshtanov, Panolan, Ramanujan, and Saurabh, STOC'17].
Simply speaking, when given a graph $G$ and integer $k$, we show how to compute a graph $G'$ on $poly(k)$ vertices so that any $\beta$-approximate solution to $G'$ can be lifted to an $(\alpha\cdot \beta)$-approximate solution to $G$, as long as $(\alpha\cdot \beta)\cdot \mathsf{OPT}(G) \le k$.
In order to achieve this, we develop a framework for sparsification of planar graphs which approximately preserves all separators and near-separators between subsets of the given terminal set.

Our result yields an improvement over the state-of-art approximation algorithms for \textsc{Vertex planarization}.
The problem admits a polynomial-time $\Oh(n^\eps)$-approximation algorithm, for any $\eps > 0$, and a~quasi-polynomial-time $(\log n)^{\Oh(1)}$-approximation algorithm, where $n$ is the input size, both randomized [Kawarabayashi and Sidiropoulos, FOCS'17].
By pipelining these algorithms with our approximate kernelization, we improve the approximation factors to respectively $\Oh(\mathsf{OPT}^\eps)$ and $(\log \mathsf{OPT})^{\Oh(1)}$. 
\end{abstract}

\includegraphics[scale=0.15]{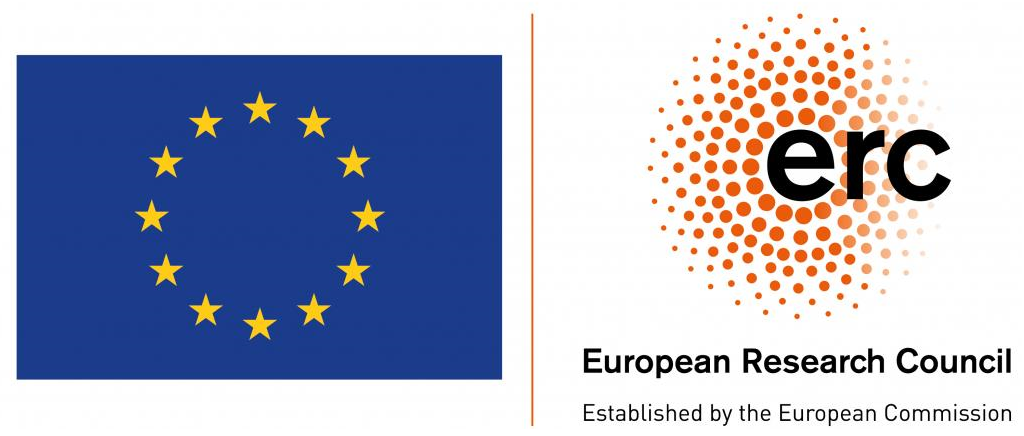}

\clearpage

\tableofcontents

 \thispagestyle{empty} 
 
\clearpage

\pagenumbering{arabic} 

\section{Introduction}
\subparagraph*{Background and motivation}
A graph is planar if it can be drawn in the plane without crossings. Planar graphs play an important role in (algorithmic) graph theory, forming the subject of well-known characterizations of planar graphs by Kuratowski~\cite{Kuratowski30}, Wagner~\cite{Wagner37}, Mac Lane~\cite{MacLane37}, Whitney~\cite{Whitney32}, and Tutte~\cite{Tutte59}, along with a variety of efficient algorithms to test whether a graph is planar~\cite{BoyerM04,HopcroftT74,MehlhornM96,MeiG70} and to solve optimization problems on planar graphs~\cite{Baker94,DemaineH08}. For a non-planar graph~$G$, the \textsc{Vertex planarization} problem asks to find a minimum vertex set whose removal makes~$G$ planar. Aside from obvious applications in graph drawing and visualization, it is important because several other NP-complete problems can be solved efficiently on graphs which are close to planar, if a deletion set is known~\cite{DemaineGKLLSVP19}. \textsc{Vertex planarization} is NP-complete~\cite{FariaFGNS06,LewisY80}, but its decision version is \emph{fixed-parameter tractable} (FPT) when parameterized by the solution size: Jansen, Lokshtanov, and Saurabh~\cite{JansenLS14} gave an algorithm that, given an $n$-vertex graph~$G$ and integer~$k$, determines whether~$G$ can be made planar by removing~$k$ vertices in the near-optimal time bound~$2^{\Oh(k \log k)} \cdot n$~\cite{JansenLS14}, following earlier algorithms for the problem by Marx and Schlotter~\cite{MarxS12} and Kawarabayashi~\cite{Kawarabayashi09}. The existence of a single-exponential algorithm with running time~$2^{\Oh(k)} \cdot n^{\Oh(1)}$ remains open. From the perspective of two other algorithmic paradigms, approximation algorithms and kernelization algorithms, the problem remains far from understood.

Despite significant interest in the problem, no constant-factor approximation algorithm for \textsc{Vertex planarization} is known. Kawarabayashi and Sidiropoulos~\cite{KawarabayashiS19} gave a polynomial-time algorithm that, given an $n$-vertex graph of maximum degree~$\Delta$, outputs a solution of size~$\Delta^{\Oh(1)} \cdot (\log n)^{3.5} \cdot \mathsf{OPT}$, improving on an earlier approximation algorithm for bounded-degree graphs by Chekuri and Sidiropoulos~\cite{ChekuriS18}. Kawarabayashi and Sidiropoulos~\cite{kawarabayashi2017polylogarithmic} also developed a randomized algorithm that outputs a solution of size $(\log n)^{\Oh(1)} \cdot \mathsf{OPT}$ in quasi-polynomial time~$n^{\Oh(\log n / \log \log n)}$, as well as a randomized algorithm to find a solution of size~$\mic{\Oh(n^{\varepsilon})} \cdot \mathsf{OPT}$ in time~$n^{\Oh(1/\varepsilon)}$ for any fixed~$\varepsilon > 0$. Whether or not a constant-factor approximation algorithm exists remains wide open. 

The second paradigm in which \textsc{Vertex planarization} remains elusive is that of kernelization~\cite{cygan2015parameterized,fomin2019kernelization,Kratsch14}, a formalization of efficient preprocessing originating in parameterized complexity theory. A parameterized problem is a decision problem where each instance~$x$ is associated with a positive integer~$k$ called the parameter, which forms an additional measurement of its complexity aside from the total input length. A \emph{kernelization} for a parameterized problem is a polynomial-time algorithm that reduces any parameterized instance~$(x,k)$ to an instance~$(x',k')$ with the same \textsc{yes/no} answer, such that~$|x'|+k' \leq f(k)$ for some function~$f$ that is called the \emph{size} of the kernelization. While all fixed-parameter tractable problems have a kernelization~\cite[Lemma 2.2]{cygan2015parameterized}, a large body of research~\cite{bodlaender2016meta,BodlaenderDFH09,Drucker15,fomin2012planar,FortnowS11,KratschW20} is devoted to investigating which problems have kernelizations of \emph{polynomial size}, which are especially interesting from a practical perspective. Whether \textsc{Vertex planarization} has a kernelization of polynomial size is one of the main open problems in kernelization~\cite[\S A.1]{fomin2019kernelization},~\cite{JansenPvL19},~\cite{fomin2012planar},~\cite{Misra16a} as advocated in the 2019 Workshop on Kernelization~\cite{OpenProblems}.

\textsc{Vertex planarization} is the prototypical example of a large class of vertex-deletion problems for which the existence of three desirable types of algorithms remains unknown: constant-factor approximation, polynomial-size kernelization, and single-exponential FPT algorithm. This family of problems is defined in terms of hitting forbidden minors. For a finite family of graphs~$\mathcal{F}$, the \textsc{$\mathcal{F}$-minor-free deletion} problem asks to find a minimum vertex set whose removal ensures the resulting graph does not contain any graph~$H \in \mathcal{F}$ as a minor. By Wagner's theorem~\cite{Wagner37}, \textsc{Vertex planarization} corresponds to~$\mathcal{F} = \{K_{3,3}, K_5\}$. Fomin et al.~\cite{fomin2012planar} developed a~framework based on \emph{protrusion replacement} to show that for each finite family~$\mathcal{F}$ that contains at least one planar graph, the corresponding deletion problem admits a polynomial kernelization, a randomized constant-factor approximation algorithm, and a single-exponential FPT algorithm. The relevance of this restriction on~$\mathcal{F}$ stems from the fact that the family of $\mathcal{F}$-minor-free graphs has bounded treewidth if and only if~$\mathcal{F}$ contains a planar graph~\cite{RobertsonS86}. \bmp{For \textsc{Vertex planarization}, the machinery of Fomin et al.~\cite{fomin2012planar} cannot be applied since planar graphs can have arbitrary large treewidth, as witnessed by grids. Hence the existing tools for kernelization break completely for the family of forbidden minors~$\mathcal{F} = \{K_{3,3}, K_5\}$.} As Fomin et al.~note in their conclusion~\cite[\S 9]{fomin2020planar-arxiv}, not a single family~$\mathcal{F}$ is known that does not contain a planar graph but for which one of the three desirable types of algorithms described above exists, making it tempting to conjecture that the presence of a planar graph in~$\mathcal{F}$ marks the boundary of this type of tractability. Understanding the complexity of \textsc{$\mathcal{F}$-minor-free deletion}, which was identified as an important open problem in the appendix of the textbook by Downey and Fellows~\cite[\S 33.2]{DowneyF13}, forms another motivation for the study of \textsc{Vertex planarization}.

In this work, we develop a new algorithm for \textsc{Vertex planarization} via the framework of \emph{lossy kernelization} which was introduced by Lokshtanov, Panolan, Ramanujan, and Saurabh~\cite{lossy}. While the formal details are deferred to the preliminaries, the main idea behind this framework is to relax the correctness requirements of a kernelization: a lossy kernelization of size~$f(k)$ reduces an instance~$(x,k)$ to an instance~$(x',k')$ with~$|x'| \mic{+} k' \leq f(k)$, such that if the optimum of~$x$ is at most~$k$, then a good approximation to~$x'$ can be efficiently lifted to a good approximation on the original instance~$x$. Hence to approximate the problem, it suffices to work with the reduced instance~$(x',k')$. The framework is of interest when the approximation guarantee of the lossy kernelization is better than the best-known polynomial-time approximation factor. Lossy kernelizations are known for several problems that do not admit polynomial kernelizations, including
\mic{connected variants of minor/subgraph hitting~\cite{EibenHR17, lossy, Ramanujan19, Ramanujan21},}
graph contraction problems~\cite{KrithikaM0T16}, and \textsc{Connected dominating set} on sparse graphs~\cite{EibenKMPS19}. 

\subparagraph*{Our contribution}
We present the first polynomial-size lossy kernelization for \textsc{Vertex planarization}. To state it, we use the term \emph{planar modulator} in a graph~$G$ to refer to a vertex set~$S$ whose removal makes~$G$ planar, and use~$\mvp(G)$ (the \emph{minimum vertex planarization} number) to denote the minimum size of planar modulator in~$G$. The formal definition of a lossy kernelization is given in Section~\ref{sec:prelims}. Our main result is the following theorem.

\begin{restatable}{thm}{restThmMain}
\label{thm:main}
The \planardel problem admits a \bmp{deterministic} 510-approximate kernelization of polynomial size. 
There is a deterministic polynomial-time algorithm that, given a graph~$G$ and integer~$k$, outputs a graph~$G'$ on~$k^{\Oh(1)}$ vertices and integer~$k' \leq k$ such that any planar modulator~$S'$ in~$G'$ can be lifted in polynomial time to a planar modulator~$S$ in~$G$ satisfying:
\begin{equation*}
\frac{\min(|S|,k+1)}{\min(\mvp(G),k+1)} \leq 510 \cdot \frac{\min(|S'|,k'+1)}{\min(\mvp(G'),k'+1)}.
\end{equation*}
\end{restatable}

To understand \mic{this} statement, compare the guarantee to the statement~$\frac{|S|}{\mvp(G)} \leq 510 \cdot \frac{|S'|}{\mvp(G')}$, which would simply say that the approximation factor of~$S$ is at most~$510$ times worse than that of~$S'$. The lossy kernelization framework caps the solution cost at~$k+1$ to ensure that~$k$ captures the complexity of the task in a meaningful way. Such a requirement is natural for a framework aimed at obtaining approximate kernelizations whose approximation factors are better than the best-known polynomial-time approximations: the existence of a constant-factor approximate lossy kernel without capping the cost of a solution in terms of~$k$ would imply the existence of a polynomial-time constant-factor approximation algorithm, by simply running the lossy kernel for~$k \in \Oh(1)$, solving the resulting instance~$G'$ optimally by brute force in~$\Oh(1)$ time, and lifting the resulting solution back to~$G$. In terms of the capped solution cost, the theorem guarantees that a $\beta$-approximation on~$G'$ can be lifted to a~$(510 \cdot \beta)$-approximation on~$G$.

Theorem~\ref{thm:main} resolves an open problem posed by Daniel Lokshtanov during the 2019 Workshop on Kernelization~\cite{OpenProblems2}. To the best of our knowledge, it is the first result making progress towards a polynomial kernelization for \textsc{Vertex planarization}. To prove Theorem~\ref{thm:main} we develop a substantial set of tools to reduce and decompose instances of the problem, which we believe will be useful in settling the existence of an exact kernelization of polynomial size. By pipelining our lossy kernelization with the approximation algorithms by Kawarabayashi and Sidiropoulos~\cite{kawarabayashi2017polylogarithmic}, we can effectively reduce the dependence on~$n$ to~$\mathsf{OPT}$ in their approximation guarantees.

\begin{restatable}{thm}{approximationThm}
\label{thm:approximation}
\planardel admits \bmp{the} following randomized approximation algorithms.
\begin{enumerate}
    \item $\Oh(\mathsf{OPT}^{\eps})$-approximation in time $n^{C(\eps)}$, \mic{where $C(\eps)$ is a constant depending on $\eps$,} for sufficiently small $\eps > 0$,
    \item $\Oh((\log \mathsf{OPT})^{32})$-approximation in time $n^{\Oh(1)} + \mic{\mathsf{OPT}^{\Oh\left(\frac{\log \opt}{\log \log \opt} \right)}}$.
\end{enumerate}
\end{restatable}

This constitutes the best polynomial-time approximation factor achieved for \planardel so far.
The randomization occurs within the algorithms by Kawarabayashi and Sidiropoulos~\cite{kawarabayashi2017polylogarithmic} which are invoked to process the compressed instance. 
The algorithm might return a solution larger than the approximation guarantee with probability at most $n^{-c}$, where $c$ is an arbitrarily large constant.

\subparagraph*{Techniques} \label{sec:techniques}
While we defer a detailed technical outline of our algorithm to Section~\ref{sec:outline}, we discuss {the overall algorithmic strategy}
here. At a high level, the algorithm starts by reducing to an input which has a planar modulator of size~$k^{\Oh(1)}$. We refine the properties of this modulator in several phases to arrive at a decomposition of the instance into~$k^{\Oh(1)}$ parts, each of which corresponds to a planar subgraph which can be embedded so that it only communicates \mic{with} the rest of the graph via~$k^{\Oh(1)}$ vertices which lie on the outer face.
Then we process each part independently, and find a \emph{solution preserver} of size~$k^{\Oh(1)}$ which, intuitively, is a vertex set containing a constant-factor approximate solution for each of the exponentially many ways in which the remainder of the graph can be planarized and embedded. The union over these solution preservers has size~$k^{\Oh(1)}$ and is guaranteed to contain a constant-factor approximate solution to \textsc{Vertex planarization}. We can therefore interpret the remaining vertices as ``forbidden to be deleted'' without significantly affecting the optimum. With this interpretation, we can replace these forbidden vertices by \mic{gadgets of total size}~$k^{\Oh(1)}$ which enforces the same conditions on the behavior of solutions, thereby leading to the approximate kernelization of size~$k^{\Oh(1)}$. 

Although our algorithm consists of several phases, we consider the computation of \emph{solution preservers}, for planar subgraphs
 with a~$k^{\Oh(1)}$-size set of terminals,
the most crucial. Roughly speaking, given a plane graph~$C$ with some terminal set~$T \subseteq V(C)$ of size~$k^{\Oh(1)}$, 
we want to compute a \bmp{set~$M \subseteq V(C)$} of size~$k^{\Oh(1)}$ with the following guarantee: for every graph~$G$ with~$\mvp(G) \leq k$ that can be obtained from~$C$ by inserting vertices and edges which are not incident on any vertex of~$V(C) \setminus T$, there should be an $\Oh(1)$-approximate planar modulator~$S$ in~$G$ satisfying~$S \cap V(C) \subseteq M \cup T$.
In particular, it might be necessary for the solution to disconnect some subsets of the terminals.
This task resembles \textsc{Vertex multiway cut} {on planar graphs.}
However, the kernelization task for \textsc{Vertex planarization} is significantly harder because the \emph{obstacles} which have to be intersected in the planarization problem (\mic{partial} minor models of~$K_5$ and~$K_{3,3}$) are significantly more complex than simply paths between terminals which have to be intersected in \textsc{Multiway cut}.
\mic{See Figure~\ref{fig:intro} for an example how removing some vertices 
allow the terminals to be drawn in a different arrangement, which may play a role in merging the embedding of $C$ with the remainder of the graph.}
We adapt the techniques from the (exact) kernelization for 
\textsc{Vertex multiway cut} {on planar graphs} by Jansen, Pilipczuk, and van Leeuwen~\cite{JansenPvL19} to reduce the problem to the case where all the terminals from \bmp{the} set $T$ are embedded on the outer face of the plane graph $C$.
The phenomenon that planar graph problems become easier when all terminals lie on the outer face has been exploited in multiple settings~\cite{Bentz19,BienstockM89,EricksonMV87,JansenPvL19,PilipczukPSvL18}. 

\mic{
As one of our main algorithmic contributions, 
we develop a theory of \emph{inseparability} 
to capture graphs with terminal vertices, in which no vertex set~$S$ can separate two subsets of~$\Omega(|S|)$ terminals each.
We show that whenever a plane graph $C$, whose terminals lie on the outer face, is inseparable then deleting any vertex set within $C$ can be simulated by removing a set of terminals of comparable size.
If $C$ is not inseparable,
we decompose it recursively into
$k^{\Oh(1)}$ inseparable subgraphs and collect small separators {violating inseparability}.
These separators suffice to build approximate solutions to \textsc{Vertex planarization}.}

\begin{figure}
    \centering
    \makebox[\textwidth][c]{\includegraphics[scale=1]{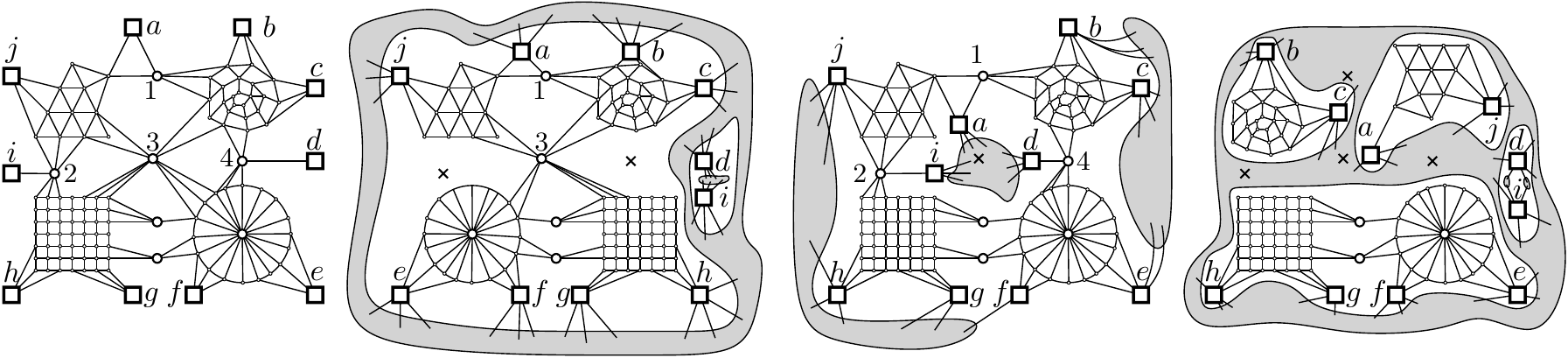}}
    \caption{The leftmost figure shows a vertex set~$D$ (circles) in an input graph~$G$ for \textsc{Vertex planarization} which induces a planar subgraph together with its neighborhood (squares). The vertices~$N_G(D)$ have neighbors in the remainder of the graph. Planarizing~$G$ may require removing some vertices of~$D$, to allow the remainder of~$D$ to be re-embedded consistently with the remaining graph (visualized schematically in gray, in the remaining figures). The second figure shows that when \mic{removing~$\{2,4\}$}, the embedding of terminals~$\{h,g,f,e\}$ can be flipped, potentially allowing an embedding compatible with the remainder. In the third figure, removing vertex~$3$ allows part of the remainder to be drawn inside, and part of the remainder outside the subgraph. In the last figure, removing~$\{1,2,3,4\}$ allows different subsets of terminals to be embedded in different regions of the remaining graph.}
    \label{fig:intro}
\end{figure}

\bmp{Compared to planar multiway cut,} another challenge comes from the fact that the graph only becomes planar \emph{after} removing a solution, which means that first we need to decompose the graph into planar parts. Here we face the additional difficulty that no deterministic polynomial-time~$\mathsf{OPT}^{\Oh(1)}$-approximation is known for \textsc{Vertex planarization}, while many kernelization algorithms~\cite{AgrawalM0Z19,BodlaenderD10,fomin2012planar,JansenP18,JansenPvL19,PilipczukPSvL18}, start by applying such an algorithm to learn the rough structure of the graph and apply reduction rules.

Finally, we remark why the existing techniques from kernelization for
\textsc{$\mathcal{F}$-minor-free deletion}~\cite{fomin2012planar} are  not applicable when $\mathcal{F}$ does not contain a planar graph.
If $\mathcal{F}$ contains a planar graph and~$G - S$ is $\mathcal{F}$-minor-free, then~$G-S$ has constant treewidth. Using an approximate solution as~$S$, this allows the graph~$G-S$ to be decomposed into \emph{near-protrusions}, subgraphs~$C_i$ of constant treewidth of which any optimal solution removes all but~$\Oh(1)$ neighbors. The characteristics of each bounded-treewidth subgraph~$C_i$ can be analyzed exactly in polynomial time using Courcelle's theorem. The fact that only~$\Oh(1)$ neighbors survive after removing an optimal solution implies that the number of relevant deletion sets in~$C_i$ which may be needed to accommodate a choice of solution outside~$C_i$, can be bounded by~$k^{\Oh(1)}$. This allows an irrelevant vertex or edge to be detected if~$|V(C_i)| > k^{\Oh(1)}$, and  leads to a polynomial kernelization. 
In \textsc{Vertex planarization} we work with a \bmp{target} graph class of unbounded treewidth. \bmp{This means that we cannot use Courcelle's theorem, and that} a totally new approach is needed to \bmp{control how forbidden minor models are formed by connections through terminals}.

\subparagraph*{Related work}
Apart from the mentioned FPT algorithms~\cite{MarxS12,Kawarabayashi09,JansenLS14} for \textsc{Vertex planarization}, there is an FPT algorithm for the generalization of  deleting vertices to obtain a graph of genus~$g$ for fixed~$g$ due to Kociumaka and Pilipczuk~\cite{KociumakaP19}
\mic{and an FPT algorithm for general \textsc{$\mathcal{F}$-minor-free deletion} by Sau, Stamoulis, and Thilikos \cite{sau20apices}.}
An exact algorithm for \textsc{Vertex planarization} with running time~$1.7347^n$ was given by Fomin and Villanger~\cite{FominTV11}. 

The edge deletion version of the planarization problem, which is related to the crossing number~\cite{ChuzhoyMT20}, has also been studied intensively. Chuzhoy, Makarychev and Sidiropoulos~\cite{ChuzhoyMS11} gave approximation algorithms for bounded-degree graphs and inputs which are close to planar. Approximation algorithms for the dual problem, finding a planar subgraph with a maximum number of edges, were given by Calinescu et al.~\cite{CalinescuFFK98}.

\subparagraph*{Organization}
In Section~\ref{sec:outline} we give a detailed outline of our algorithm. A substantial amount of technical preliminaries is given in Section~\ref{sec:prelims}, after which we devote one section to each \mic{step} of the algorithm until we combine all the pieces in Section~\ref{sec:wrapup}.
\mic{Section~\ref{app:planar} contains proofs of the planarization \bmp{criteria} formulated in the preliminaries.}
We conclude in Section~\ref{sec:conclusion}.

\section{Outline} \label{sec:outline}

We present the algorithm as a series of graph modifications $G = G_0 \to G_1 \to \dots \to G_\ell$, where $G_\ell$ is a graph of size $k^{\Oh(1)}$.
In each step we need to show two types of safeness: (a) the forward safeness which links $\mvp(G_{i+1})$ to $\mvp(G_i)$, and (b) the backward safeness which allows us to \emph{lift} a given solution in $G_{i+1}$ of size at most $k$ to a solution in $G_{i}$ of comparable size, in polynomial time.
The simplest type of modification involves removal of an \emph{irrelevant} vertex.
A vertex $v \in V(G)$ is called $k$-irrelevant if for any set $S \subseteq V(G) \setminus \{v\}$ of size at most $k$, when $G - (S \cup \{v\})$ is planar then also $G-S$ is planar.
When such a vertex $v$ is identified, it can be safely removed from the graph, because any solution in $G-v$ of size at most $k$ is also valid in $G$.
This provides the backward safeness of this modification, whereas forward safeness here \bmp{follows since} removal of a vertex cannot increase the optimum.
The irrelevant vertex technique has found use in many previous \bmp{algorithms} for \planardel~\cite{JansenLS14,Kawarabayashi09,MarxS12}, as well as other problems~\cite{AdlerKKLST17,FominGT19,FominLP0Z20,KociumakaP19,LindermayrSV20,RobertsonS12}.
We \bmp{start by explaining how to certify planarity}, which is needed in proving the irrelevance of a vertex and other properties used in proving the backward safeness.

\paragraph*{Planarization criterion}

\mic{
Let $G$ be a planar graph and $C$ be a cycle in $G$.
We consider the set of $C$-bridges in $G$ which is given by all the connected components of $G-V(C)$, together with their attachments at $C$, and all chords of $C$.
If two $C$-bridges are drawn on the same side of the Jordan curve given by the embedding of $C$, then they cannot ``overlap'', that is, their attachments at $C$ cannot be intrinsically crossed.
This gives a necessary and sufficient condition for
a graph $G$ with a cycle $C$ to be planar: (1) the graph given by the pairs of overlapping $C$-bridges must be bipartite, and (2) each $C$-bridge, when considered together with $C$, must induce a planar graph.
This criterion allows us to examine the planarity of $G$ {in} some sense locally and it has been leveraged in several contexts~\cite{DiETT99, HopcroftT74, JansenLS14}.
While this proof technique has already been used for showing that the center of a flat $\Theta(k) \times \Theta(k)$-grid is $k$-irrelevant, we utilize the criterion to derive irrelevance in less restrictive settings.}

\paragraph*{Finding a size-$k^{\Oh(1)}$ planar modulator}
A common approach in designing kernelization algorithm for vertex-deletion problems is to bootstrap with an approximate solution: a modulator $X$ in graph $G$ to the considered graph property, of size $\Oh(k)$ or even $k^{\Oh(1)}$; and then focus on compressing the connected components of $G-X$.
Unfortunately, there is no such approximation algorithm known for \planardel.
The best available approximation factor is $(\log n)^{\Oh(1)}$~\cite{kawarabayashi2017polylogarithmic} but, even though we are able to bound $(\log n)$ by $k^{\Oh(1)}$, it works in randomized quasi-polynomial time.
As we aim at deterministic polynomial time, we need another approach.

It suffices to reduce the treewidth of $G$ to $k^{\Oh(1)}$ and then
a greedy argument allows us to either construct a planar modulator $X$ of size $k^{\Oh(1)}$ or conclude that there is no solution of size at most $k$.
Bounding treewidth plays an important role in {parameterized} algorithms for \planardel 
and {may be performed by combining the irrelevant vertex technique} with either recursion~\cite{JansenLS14, Kawarabayashi09} or iterative compression~\cite{cygan2015parameterized, MarxS12}.
\mic{None of these techniques are compatible with polynomial-time kernelization though.}
We reduce the treewidth differently by detecting a large grid minor and inspecting its non-planar regions. 

\mic{
Grid minors are known to form obstructions to having small treewidth and vice versa: a graph with treewidth $\Omega(t^c)$, for a sufficiently large constant $c$, contains a $t \times t$-grid as a minor~\cite{ChuzhoyT21}.
In particular, there is a randomized polynomial-time algorithm that given a graph $G$ either returns a minor model of the $t \times t$-grid in $G$ or a tree decomposition of $G$ of width $t^{\Oh(1)}$~\cite{ChekuriC16}.
However, as we want to avoid randomization, we show that in our setting we can find either a $t \times t$-grid minor or a tree decomposition of width $t^{\Oh(1)}$ deterministically.
First, we can assume that $\log n \le k \log k$ as otherwise the FPT algorithm for \planardel with running time $2^{\Oh(k\log k)}n$~\cite{JansenLS14} works in polynomial time.
Next, we find a subgraph of relatively large treewidth and maximal degree $\Delta = 4$~\cite{KreutzerT10} and exploit the $poly(\log n, \Delta)$-approximation algorithm for \planardel~\cite{KawarabayashiS19} to find a planar subgraph of large treewidth, where a large grid minor can be found deterministically.

Supplied with a minor model of the $t \times t$-grid in $G$, where $t = \Omega(k^3)$, we look for a subgrid of size $\Theta(k)\times\Theta(k)$ which (a) represents a planar subgraph of $G$ and (b) is not incident to any \emph{non-local} edges, i.e., edges connecting non-adjacent branch sets of the grid.
A vertex located in the center of such a subgrid must be $k$-irrelevant.
Unfortunately, \bmp{such a subgrid may not exist} when there are many pairs of branch sets connected by non-local edges.
If these edges are sufficiently scattered over the grid,
we prove that some non-local edges remain untouched 
by any set of $k$ vertices which implies
that $\mvp(G) > k$.
Otherwise, we inspect the spanning trees inside the branch sets of the minor model to identify a vertex $v$ which must be present in every solution of size at most $k$---we call such a vertex $k$-necessary---which again can be safely removed from $G$.
}

\paragraph*{Strengthening the planar modulator}
From now on we can assume that
we are given a planar modulator $X$ in $G$ of size $k^{\Oh(1)}$ and a tree decomposition of $G-X$ of width $k^{\Oh(1)}$.
Note that $X$ might not be a valid solution for the original instance because the conditions of the backward safeness apply only to solutions of size at most $k$.
We would like to augment the modulator $X$ to satisfy the following property: when $C$ is a connected component of $G-X$, then the graph induced by $N_G[C]$ (the set $C$ together with its neighborhood) is planar.
Such a modulator is called \emph{strong}.
This concept is inspired by the idea of tidying sets~\cite{BevernMN12} (cf.~\cite{JansenP18}), introduced in the context of kernelization for different graph problems, aimed at finding a possibly larger modulator which remains valid after removing a single vertex from it.
Observe that when the maximal degree in $G$ is bounded by $k^{\Oh(1)}$ then $N_G[X]$ is a strong modulator of size $k^{\Oh(1)}$ but we cannot make such \bmp{an} assumption on the vertex degrees.
Unlike the previous steps, constructing a strong modulator requires graph modifications which are lossy, i.e., the lifting algorithm might increase the size of the solution by a~constant factor.

{
We construct a vertex set $Y \subseteq V(G) \setminus X$ of size $k^{\Oh(1)}$, such that for any connected component of $G-(X\cup Y)$ it holds that $|N_G(C) \cap X| \le 2$.
First, we try to construct such {a} set $Y$ greedily, marking maximal subtrees in the tree decomposition of $G-X$ which represent connected subgraphs of $G-X$ with at least 3 neighbors in $X$.
If this procedure terminates in $k^{\Oh(1)}$ steps, we mark a set of bags in the tree decomposition, whose union gives the set $Y$.
Otherwise, we identify a triple $x_1, x_2,x_3 \in X$ and a family of $k+3$ disjoint subsets $V_i \subseteq V(G) \setminus X$ such that each $G[V_i]$ is connected and $\{x_1, x_2,x_3\} \subseteq N_G(V_i)$.
Suppose that a set $S \subseteq V(G)$ is a planar modulator of size at most $k$.
There must be 3 indices $i_1, i_2, i_3$ such that $V_{i_j} \cap S = \emptyset$ for $j=1,2,3$.
If $\{x_1, x_2,x_3\} \cap S = \emptyset$ then $G-S$ contains a $K_{3,3}$-minor so it is not planar.
This means that any solution of size at most $k$ must contain at least one of $x_1, x_2, x_3$.
We can therefore remove $x_1,x_2,x_3$ from the graph and decrease the value of parameter to $k-1$.
This forms a lossy reduction in which we decrease the parameter at the expense of approximation factor 3.}
\mic{To get a strong modulator, we need to begin this process with a modulator $X$ satisfying the following: inserting any two vertices from $X$ to any connected component of $G-X$ does not affect its planarity.
This can be enforced with a lossy reduction of \bmp{a} similar kind.
}

When no reductions are applicable anymore, we are able to find a strong planar modulator $X$ of size $k^{\Oh(1)}$
and we proceed by reducing the size of the components of $G-X$.
However, the number of such components can be also large.
We observe that the components with just one or two neighbors in $X$ can be neglected, as they can be easily compressed in a later step of the algorithm.
To reduce the number of the remaining components, we proceed as before: if there are too many components with at least 3 neighbors in $X$, then we can locate 3 vertices in $X$ such every solution of size at most $k$ must contain at least one of them; hence all three vertices can be removed.

\paragraph*{Reducing the radial diameter}

Let $G\brb C$ denote the subgraph of $G$ induced by the set $N_G[C]$ with the vertices from $N_G(C)$ marked as the \emph{boundary vertices}, also referred to as \emph{terminals}.
As \bmp{the strong modulator allows us to} assume that $G\brb C$ is planar for each connected component $C$ of $G-X$, we can consider some plane embedding  of $G\brb C$.
This allows us to unlock the vast algorithmic toolbox available for problems on plane graphs.
The current task can be regarded as processing a plane graph with $k^{\Oh(1)}$ terminals and mimicking the behavior of any size-$k$ solution over $C$ with another plane graph of size $k^{\Oh(1)}$ with the same set of terminals, by possibly paying a constant factor in the approximation ratio.
A landmark result of this kind is a polynomial (exact) kernelization for \textsc{Steiner tree} on planar \bmp{graphs} parameterized by the number of edges in the solution \cite{PilipczukPSvL18} which also found application in designing a polynomial kernelization
for \textsc{Vertex multiway cut} on planar graph parameterized by the size of the solution~\cite{JansenPvL19}.
Observe that our task is at least as hard as \textsc{Vertex multiway cut} on planar graphs (see \cite[Lemma 6.3]{JansenPvL19}) because the behavior of a solution over $C$ might be to disconnect some sets of terminals.
However, preserving multiway cuts is not enough because removing the solution vertices can, e.g., reduce the size of a minimal separator between some terminals in $G\brb C$ from 3 to 2, which 
may play a role when considering the remaining vertices outside $N_G[C]$.
{Observe that different biconnected components can be re-arranged in an embedding by pivoting around an articulation point that separates them, while different triconnected components attaching to the same separation pair can also be reordered among each other \mic{(see Figure~\ref{fig:intro})}. Hence (nearly) separating vertices leads to more freedom in making an embedding and therefore understanding the separation properties of vertex sets \bmp{is} paramount to understand their role in potential solutions. }

An important special case for terminal-based problems on plane graphs occurs when all the terminals are located on a common face.
We would like to reduce our task to this special case but a~priori we have no control over the embedding of terminals.
As a step in this direction, we reduce the \emph{radial diameter} to $k^{\Oh(1)}$, where the radial diameter is the longest sequence of incident faces necessary to make a connection between a pair of vertices in the plane graph.

Following the idea \bmp{for \textsc{Vertex multiway cut} on planar graphs}~\cite{JansenPvL19}, we decompose a plane graph into outerplanarity layers: \bmp{a vertex belongs to layer~$i$ if it lies on the outer face of the graph obtained by~$i-1$ rounds of removing all vertices on the outer face}.
As \mic{a} worst-case scenario, consider one group $T_1$ of terminals located on the outer face, a deep ``well'' of nested cycles from different outerplanarity layers, each of length $o(k)$, and another group $T_2$ of terminals located inside the well.
Note that the terminals are connected in $G$ to other parts of the graph so in this situation we cannot find any $k$-irrelevant vertices.
We show that only \bmp{the} $\Oh(k)$ outerplanarity layers closest to $T_1$ or $T_2$ are relevant and in the middle layers it suffices to preserve \bmp{a} maximal family of vertex-disjoint $(T_1,T_2)$-paths $\mathcal{P}$.
Assume for a moment that the number of paths in $\mathcal{P}$ is at least 4.
\mic{First we give an argument to ensure that these paths never go back and forth through the outerplanarity layers. Then,}
we remove all the vertices and edges in the middle layers leaving out only subdivided edges corresponding to subpaths from $\mathcal{P}$.
Now we need to argue for the backward safeness of this modification.
When $S$ is some solution in $G$ of size at most $k$, it must leave untouched at least 2 cycles in the relevant part closer to $T_1$ and 2 cycles on the side of $T_2$.
If $G-S$ contains at least 4 paths from $\mathcal{P}$, we use them together with the 4 cycles, to surround each removed vertex/edge with 2 cycles separating it from the terminals, thus showing that it is irrelevant to the solution $S$.
If the $(T_1,T_2)$-flow is less than 4 after removing $S$, then $S$ must contain sufficiently many vertices in the considered subgraph so they may be ``charged'' to pay for adding to $S$ a $(T_1,T_2)$-separator in $G\brb C$, which again allows is to apply the planarity criterion.
In the latter case the lifting algorithm increases the size of the solution by a factor 3\footnote{The factor 3 is obtained by choosing the constants more carefully. We omit the details here for the sake of keeping the outline simple. }.

To justify the previous argument we have assumed the value of $(T_1,T_2)$-flow in $G\brb C$ to be at least 4.
However, it might be the case that there exists a long sequence of $(T_1,T_2)$-separators, each of size less than 4.
In this case, we cannot decompose $G\brb C$ into a bounded number of subgraphs to which the given reduction applies.
However, we can then locate a large subset $A \subseteq C$ so that $A$ contains no terminals and $N_G(A)$ is of size $\Oh(1)$, so $N_G[A]$ induces a large planar subgraph of $G$ with small boundary.
To deal with such subgraphs, we take advantage of the technique known as
protrusion replacement.

\paragraph*{Protrusion replacement}
{For a graph $G$ and $X \subseteq V(G)$, we define the boundary of $X$, denoted $\partial_G(X)$, as the set of vertices from $X$ adjacent to some vertex outside $X$.}
An $r$-protrusion in a graph $G$ is a set $X \subseteq V(G)$ so that $|\partial_G(X)| \le r$ and \bmp{the} treewidth of $G[X]$ is at most $r$. 
The standard protrusion replacement technique~\cite{bodlaender2016meta} is aimed at replacing a large $r$-protrusion with a subgraph of size at most $\gamma(r)$, where $\gamma$ is some function, while preserving the optimum value of the instance for the considered optimization problem.
The function $\gamma$ might be growing very rapidly, so we can apply this technique only for constant $r$.
Since we want to perform lifting of a potentially non-optimal solution, we use the lossless variant of the technique~\cite{fomin2012planar}, introduced for the sake of designing approximation algorithms for \bmp{\textsc{$\mathcal{F}$-minor-free deletion} for families~$\mathcal{F}$ containing a planar graph}, which provides us with such a mechanism.

We are interested in compressing a planar subgraph $G[X]$ of boundary $r = \Oh(1)$ which may however have large treewidth.
To make it amenable to protrusion replacement, we first need to reduce \bmp{the} treewidth of $G[X]$ to a constant.
It is easy to show that for any optimal planar modulator $S$ it holds that $|S \cap X| \le r$, as otherwise the solution $(S \setminus X) \cup \partial_G(X)$ would be smaller.
By an argument based on grid minors, {if} $\tw(G[X])$ is larger than some function of $r$, then $X$ contains a vertex irrelevant for any optimal solution.
As we work with solutions which are not necessarily optimal,
we show that the lifting algorithm can perform a preprocessing step  
which still allows us to assume that  $|S \cap X| \le r$ and the argument above remains valid.

\paragraph*{Solution preservers}
At this point of the algorithm we can assume that for each connected component $C$ of $G-X$ the boundaried graph $G\brb C$ has radial diameter bounded by $k^{\Oh(1)}$.
\bmp{Towards the final goal of reducing each component~$G \brb C$ to size~$k^{\Oh(1)}$,} 
we take a two-step approach.
First, we mark a vertex set $S_C \subseteq C$ of size $k^{\Oh(1)}$ so that we can assume that there is an $\alpha$-approximately optimal planar modulator $S$ such that $S \cap C \subseteq S_C$, where $\alpha$ is some constant.
Such a set $S_C$ is called an $\alpha$-preserver for $C$. 
For example, if $|N_G(C)| \le \alpha$, then $\emptyset$ is an $\alpha$-preserver for $C$ because we can replace any solution $S$ with $|S \cap C| \ge 1$ with a solution $(S \setminus C) \cup N_G(C)$.
\bmp{Afterwards}, we will compress the subgraphs within $C \setminus S_C$ for which we know that the approximately optimal solution does not remove any vertices inside them.

We show how to reduce the task of constructing a solution preserver for $C$, given an embedding of $G\brb C$ with bounded radial diameter,
to the case where in the embedding of \bmp{$G\brb C$} all the terminals lie on the outer face.
Such a boundaried plane graph is called \emph{\nice}. 
\bmp{Bounded radial diameter in a plane graph~$H$ implies that any pair of vertices can be \emph{radially connected} by a curve that intersects only a bounded number of vertices of the drawing and does not intersect any edges. Using this property, given a set of terminals~$T \subseteq V(H)$ we can mark a bounded-size set of to-become terminals~$T' \supseteq T$ such that any pair of vertices in~$T'$ can be connected by a sequence of vertices of~$T'$ such that successive vertices in the sequence lie on a common face. This implies that there is a single face in the embedding of~$H - T'$ that contains the images of all boundary vertices~$T'$.}
\bmp{A boundaried plane graph with this property} is called \emph{\circum}.
There is a subtle difference between \circum and \nice boundaried plane graphs (see Figure~\ref{fig:boundariedgraph} on page~\pageref{fig:boundariedgraph}).
If $H$ is \circum, then the edges incident to the boundary of $H$ may separate some boundary vertices from the outer face of $H$. 
We decompose $H$, {by marking some vertices as to-become terminals}, into boundaried subgraphs adjacent to at most two {terminals of~$H$}. 
We want to say that such a boundaried subgraph becomes \nice after discarding these two special vertices.
This is done by exploiting properties of inclusion-wise minimal separators in a plane graph.
Such a separator can be represented by a closed curve intersecting some vertices.
By removing one side of this separator and treating the separator as a boundary for the other side, we obtain a \nice boundaried plane graph.
In order to avoid tedious topological arguments, we give a contraction-based criterion to show that a given boundaried graph admits a \nice embedding. 

This construction may require removal of at most two vertices from the obtained boundaried subgraph $G \brb {C'}$ to make it \nice.
However, if $S$ is a solution and $S \cap C'$ is non-empty we can ``charge'' this part of the solution to pay for the removal of the two vertices, turning an $\alpha$-preserver into an $(\alpha+2)$-preserver.

\paragraph*{Single-faced boundaried graphs and inseparability}
By the arguments above, it suffices  to give a construction of an $\alpha$-preserver for $C$ in the case where $G\brb C$ is equipped with a \nice embedding.
This is the crucial part of the proof.
We introduce the notion of $c$\emph{-inseparability} for boundaried graphs
and a decomposition into $c$-inseparable parts.
A $c$-separation in a boundaried graph $H$ is a partition of $V(H)$ into sets $(A,S,B)$ such that $S$ is an $(A,B)$-separator and each of $A \cup S, B \cup S$ contains at least $c\cdot|S|$ terminals.
A boundaried graph $H$ is $c$-inseparable if it admits no $c$-separation.
This can be compared to the concept of $(s,c)$-unbreakable graphs and 
the $(s,c)$-unbreakable decomposition~\cite{ChitnisCHPP16, LokshtanovRSZ18}. 
A graph $G$ is $(s,c)$-unbreakable if there is no partition of $V(G)$ into sets $(A,S,B)$ such that $S$ is an $(A,B)$-separator, $|S| \le s$ and each of $A, B$ contains at least $c$ vertices.
The main difference is that $c${-inseparability} concerns separators of arbitrarily large size as long as both sides of the separation contain relatively many terminals.
Moreover, this framework works particularly well with \nice boundaried plane graphs.

Suppose that we are given a set $C \subseteq V(G)$ so that $G\brb C$ is $c$-inseparable for some constant $c > 1$ and $G\brb C$ admits a fixed \nice embedding.
Let $S$ be some solution in $G$ and $X = S \cap C$.
If the set $X$ separates \mic{some subsets of} terminals (that is, vertices from $N_G(C)$) in $G\brb C$ then, due to $c$-inseparability,
\mic{one of these subsets has less than $c\cdot |X|$ terminals.
If we only cared about preserving separations, we could
add these terminals to the solution instead of using $X$. }
As mentioned before, just preserving separations is not enough because we also care about sets which ``almost separate'' some terminals.
Let us consider the set $X'$ given by all the vertices in $G\brb C$ within a reach of three internal faces from some vertex from $X$.
We show that the behavior of $X'$ is similar to a separator of size $\Oh(|X|)$ and there must be one large connected component of $G\brb C - X'$ containing all but $\Oh(|X|)$ terminals.
We define a new solution $S'$ by adding to $S$ all vertices from $N_G(C)$ that do not belong to the large component.
Next, we argue that vertices from $X$ can be removed from $S'$ without invalidating the solution.
To justify this, we show for each $x \in X$ the existence of three nested connected sets in $X' \setminus X$ that separate $x$ from the terminals not contained in $S'$.
This suffices to apply the planarity criterion and to show that $S' \setminus X$ is a planar modulator in $G$ implying that there is an $\Oh(1)$-approximate solution disjoint from $C$.

{To apply the approach \mic{for $c$-inseparable} graphs described above,} we need to decompose a \nice boundaried plane graph $H$ into $c$-inseparable parts.
We show that whenever $H$ admits some $(2c)$-separation, then there must exist a $c$-separation $(A,S,B)$ which separates the set of terminals into two connected segments along the outer face.
The number of choices for such a separation is quadratic in the number of terminals, so it can be found in polynomial time.
Furthermore, by exploiting the minimality of the~separator $S$ we show that $S$ splits $H$ into two smaller \nice boundaried plane graphs $H_1, H_2$ contained respectively in $A\cup S$ and $B \cup S$.
We proceed by marking the vertices from $S$ as additional terminals and decomposing $H_1, H_2$ recursively. 
A priori this construction may mark a large number of vertices so we need some kind of a measure to ensure that the number of splits will be only $k^{\Oh(1)}$.
Let $t$ be the number of terminals in $H$, $s$ be the size of the separator $S$, and $t_1,t_2$ be respectively the number of terminals in $H_1$ and $H_2$. 
Because we work with a $c$-separation we have that $c\cdot s \le \min(t_1,t_2)$.
This implies that even though $t_1 + t_2 > t$, these quantities are in some sense balanced.
It turns out that it suffices to take $c=4$ to obtain an inequality  $t_1^2 + t_2^2 < t^2$.
We consider a semi-invariant given by the sum of squares of the boundary sizes in the family of constructed boundaried graphs and prove that this measure decreases at every split, providing us with the desired guarantee.
After marking $k^{\Oh(1)}$ vertices, we decompose $H$ into \nice boundaried plane graphs which are 8-inseparable (due to the 2-approximation for finding separations).
By the arguments from the previous paragraph,
we prove that the set of marked vertices forms
a 168-preserver for $C$.
This leads to a 170-preserver for the general case.

\paragraph*{Compressing the negligible components}
After collecting the solution preservers, we obtain a planar modulator $X' \subseteq G $ with the following property:  there exists a 170-approximate solution $S'$ such that $S' \subseteq X'$.
In the final step of the algorithm we want to compress each connected component $C$ of $G-X'$.
By the arguments given so far we can assume that the boundaried \bmp{graphs} $G\brb C$ \bmp{are} \circum.
Note that whenever there exists a cycle $C'$ within the \circum embedding of $G\brb C$ which separates $N_G(C)$ from some vertices of $C$,
we can modify the $C'$-bridges disjoint from $N_G(C)$ without violating the solution $S'$ which is disjoint from $C$.
We cannot however just contract these $C'$-bridges because we need to also provide the backward safeness: we should ensure that no new size-$(\leq k)$ solutions~$S \not \subseteq X'$ are introduced which cannot be lifted back to solutions in~$G$. 
Let $B$ be a $C'$-bridge so that $V(B) \cap N_G(C) = \emptyset$.
We show how to replace $B$ with a \emph{well} of depth $\Oh(k)$ and perimeter $k^{\Oh(1)}$ in a safe way.
In order to advocate safeness, we prove a new criterion for an edge $e$ to be $k$-irrelevant, applicable when $e$ is surrounded by $k+3$ cycles, separating $e$ from the non-planar part of the graph but intersecting each other at some vertex.
We show that we can identify a bounded number of subgraphs in $C$, such that after applying the described modification to each of them we can mark  $k^{\Oh(1)}$ vertices in $C$ to split it into subgraphs of neighborhood size $\Oh(1)$.
These subgraphs can be compressed using protrusion replacement\bmp{, which concludes the algorithm.}

\begin{figure}[hb] 
\centering
\includegraphics{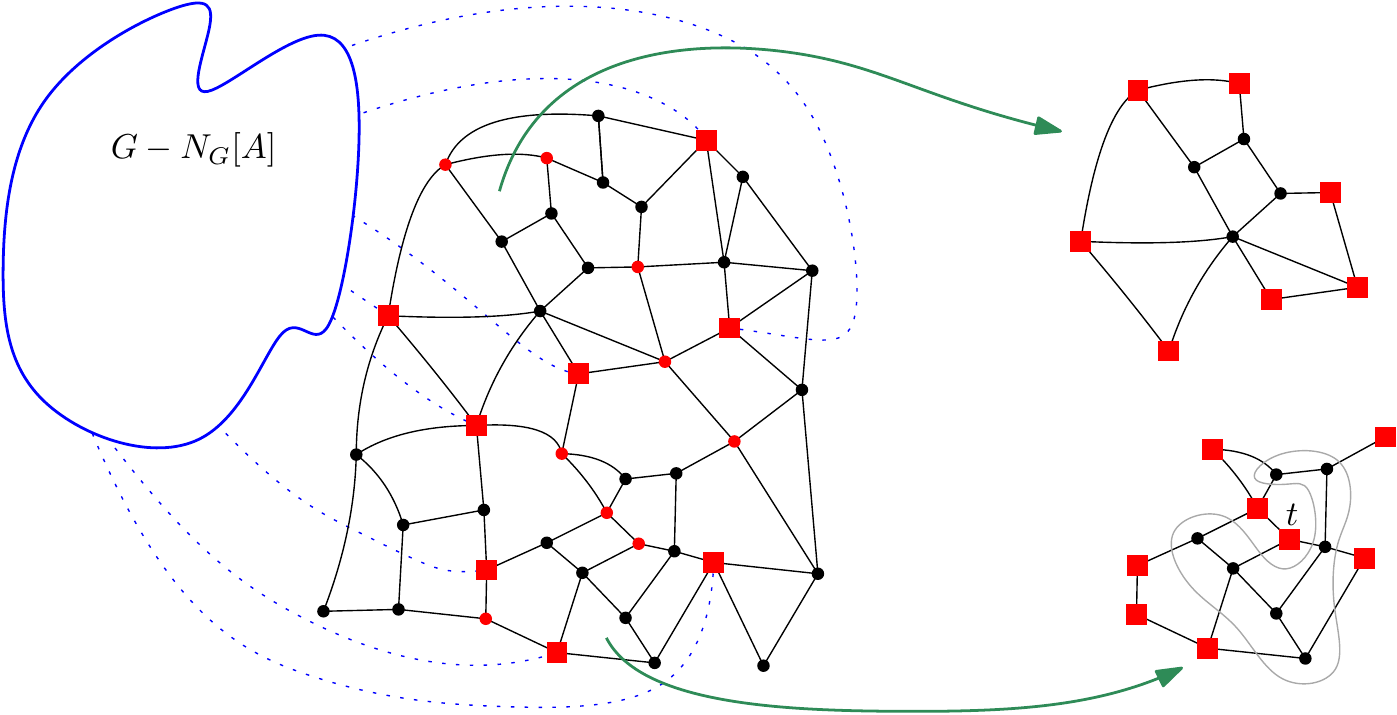}
\caption{Illustration of the various notions of boundaried graphs. On the left, a plane boundaried graph~$G \brb A$ is drawn, together with a schematic representation of the remaining graph~$G - V(G \brb A) = G - N_G[A]$. The squares represent the vertices of~$N_G(A)$, which form the boundary of the boundaried graph~$G \brb A$ and interact with the rest of the graph as indicated by the dotted blue lines. The circles (either red or black) represent vertices of~$A$. By treating the red vertices as additional terminals, the boundaried plane graph decomposes further into connected components formed by black vertices. Two are highlighted on the right, together with their neighborhoods which form their boundaries. The one on the top right is \nice since all boundary vertices lie on the outer face. The one the bottom right is not, since~$t$ does not lie on the outer face. However, this boundaried graph is still \circum since the images of all terminals lie in the infinite region of the plane graph obtained by removing the terminals. The term circumscribed comes from the fact that there exists a curve intersecting the drawing only at edges incident on the boundary, with all boundary vertices on one side and all remaining vertices on the other side. Such a curve always exists for a \circum plane graph. One is depicted in the figure, in gray.} \label{fig:boundariedgraph}
\end{figure}

\section{Preliminaries} \label{sec:prelims}

\subparagraph*{Lossy kernelization}

We give the definitions forming the framework of lossy kernelization. Following \cite{lossy}, we define
a parameterized minimization problem $\Pi$ as a
computable function
$\Pi \colon \Sigma^* \times \mathbb{N} \times \Sigma^* \to \mathbb{R} \cup \{\infty\}$.
For an instance $(I,k) \in \Sigma^* \times \mathbb{N}$,
a solution is a string of length at most $|I| + k$.
The value of a solution $s \in \Sigma^*$ equals $\Pi(I,k,s)$.
The objective is to find a solution $s \in \Sigma^*$ minimizing $\Pi(I,k,s)$.
We define $\opt_\Pi(I,k) = \min_{s \in \Sigma^*,\, |s| \le |I| + k} \Pi(I,k,s)$.

\begin{definition}[\cite{lossy}]
Let $\alpha > 1$ be a real number and $\Pi$ be a parameterized minimization problem.
An \textbf{$\alpha$-approximate polynomial-time  pre-processing algorithm} $\mathcal{A}$ for $\Pi$
is a pair of polynomial-time algorithms. The first one is called the \textbf{reduction algorithm} and computes a map $\mathcal{R}_\mathcal{A} \colon \Sigma^* \times \mathbb{N} \to \Sigma^* \times \mathbb{N}$. Given as input an~instance
$(I,k)$ of $\Pi$, the reduction algorithm outputs another instance $(I', k') = \mathcal{R}_\mathcal{A}(I, k)$.

The second algorithm is called the \textbf{solution lifting algorithm}. This algorithm takes as
input an~instance $(I, k) \in \Sigma^* \times \mathbb{N}$ of $\Pi$, the output pair $(I', k') = \mathcal{R}_\mathcal{A}(I, k)$ of the reduction algorithm, a solution $S'$ to
$(I',k')$, and outputs a solution $S$ to $(I,k)$ such that

\[
\frac{\Pi(I,k,S)}{\mathsf{OPT}_\Pi(I,k)} \le \alpha \cdot \frac{\Pi(I',k',S')}{\mathsf{OPT}_\Pi(I',k')}.
\]
\end{definition}

\begin{definition}[\cite{lossy}]
Let $\alpha > 1$ be a real number, $\Pi$ be a parameterized minimization problem and $f \colon \mathbb{N} \to \mathbb{N}$ be a function.
An $\alpha$-approximate kernelization of size $f$
is an {$\alpha$-approximate polynomial-time  pre-processing algorithm} $\mathcal{A}$ for $\Pi$,
such that if $(I', k') = \mathcal{R}_\mathcal{A}(I, k)$ then $|I'| + k' \le f(k)$.
\end{definition}

\subparagraph*{Formalization of the problem}
In general, the role of the parameter $k$ is to measure the difficulty of the optimization task.
Similarly as in the standard kernelization theory, we are often interested in detecting solutions of size at most $k$. 
A natural way to implement this concept is to treat the solutions to $(I,k)$ of value larger than $k$ as equally bad, assigning them value $k+1$ in the function~$\Pi$.
We refer to \cite[\S 3.2]{lossy} for an extended discussion on ``capping the objective function at $k+1$''. 

Following this convention, we define a parameterized minimization problem $\Pi_\vp$.
Let $G$ be a graph, $k$ be an integer, and $S \subseteq V(G)$.
We assume an arbitrary string encoding of graphs and sets and define
$\Pi_\vp(G,k,S) = \min(|S|,k+1)$ if $G-S$ is planar, and $\Pi_\vp(G,k,S) = \infty$ otherwise.

\subparagraph*{Graph theory}
The set $\{1,\ldots,p\}$ is denoted by $[p]$. \bmp{For a finite set~$S$ and non-negative integer~$i$ we use~$\binom{S}{i}$ to denote the collection of size~$i$ subsets of~$S$.}
We consider simple undirected graphs without self-loops. A~graph $G$ has vertex set $V(G)$ and edge set \bmp{$E(G) \subseteq \binom{V(G)}{2}$}. We use shorthand $n = |V(G)|$ and $m = |E(G)|$. For ({not necessarily disjoint}) $A,B \subseteq V(G)$, we define $E_G(A,B) = \{uv \mid u \in A, v \in B, uv \in E(G)\}$.
 The {open} neighborhood of $v \in V(G)$ is $N_G(v) := \{u \mid uv \in E(G)\}$, where we omit the subscript $G$ if it is clear from context. {For a vertex set~$S \subseteq V(G)$ the open neighborhood of~$S$, denoted~$N_G(S)$, is defined as~$\bigcup _{v \in S} N_G(v) \setminus S$. The closed neighborhood of a single vertex~$v$ is~$N_G[v] := N_G(v) \cup \{v\}$, and the closed neighborhood of a vertex set~$S$ is~$N_G[S] := N_G(S) \cup S$.} 
{The boundary of a vertex set $S \subseteq V(G)$ is the set $\partial_G(S) = N_G(V(G) \setminus S)$.}
For $A \subseteq V(G)$, the graph induced by $A$ is denoted by $G[A]$ and {we say that the vertex set $A$ is connected in $G$ if the graph $G[A]$ is connected.}
For $v \in V(G), S \subseteq V(G) \setminus \{v\}$, let $R_G(v,S)$ be the set of vertices reachable from $v$ in $G-S$, 
and $R_G[v,S] = R_G(v,S) \cup S$.
{We sometimes omit the subscript when it is clear from context.}
We use shorthand $G-A$ for the graph $G[V(G) \setminus A]$. For $v \in V(G)$, we write $G-v$ instead of $G-\{v\}$.
{For~$A \subseteq E(G)$ we denote by~$G \setminus A$ the graph with vertex set~$V(G)$ and edge set~$E(G) \setminus A$. For~$e \in E(G)$ we write~$G \setminus e$ instead of~$G \setminus \{e\}$.}
If $e = uv$, then $V(e) = \{u,v\}$.

For two sets $X,Y \subseteq V(G)$, a set $S \subseteq V(G)$ is an unrestricted $(X,Y)$-separator if no connected component of $G-S$ contains a vertex from both {$X \setminus S$} and {$Y \setminus S$}. Note that such a separator may intersect $X \cup Y$. Equivalently each $(X,Y)$-path (that is, a path with one endpoint in $X$ and the other one in $Y$) contains a vertex of $S$.
{By Menger's theorem, the minimum cardinality of such a separator is equal to the maximum cardinality of a set of pairwise vertex-disjoint $(X,Y)$-paths.}
A restricted~$(X,Y)$-separator~$S$ is an unrestricted $(X,Y)$-separator which additionally satisfies~$S \cap (X \cup Y) = \emptyset$.
{\bmp{By} default, by an $(X,Y)$-separator we mean a restricted $(X,Y)$-separator.
Again by Menger's theorem, the minimum cardinality of such a separator \mic{for non-adjacent $X,Y$} is equal to the maximum cardinality of a set of pairwise \emph{internally} vertex-disjoint $(X,Y)$-paths.}
A restricted (resp. unrestricted) $(X,Y)$-separator $S$ is called minimal if no proper subset of $S$ is a restricted (resp. unrestricted) $(X,Y)$-separator.
If $X = \{x\}$, $Y = \{y\}$, we simply refer to $(x,y)$-paths and $(x,y)$-separators.

\bmpr{It would be good to define biconnected component and articulation point in the prelims; definitions for them are commented out in TeX, and they do not appear elsewhere in prelims I think.}

{A tree is a connected graph that is acyclic. \bmp{A \emph{rooted} tree~$T$ is a tree together with a distinguished vertex~$r \in V(T)$ that forms the root.} A (rooted) forest is a disjoint union of (rooted) trees. In tree $T$ with root $r$, we say that $t \in V(T)$ is an ancestor of $t' \in V(T)$ (equivalently $t'$ is a descendant of $t$)} {if $t$ lies on the (unique) path from $r$ to $t'$.}
A vertex set $A \subseteq V(G)$ is an independent set in $G$ if $E_G(A,A) = \emptyset$.
A graph {$G$} is \emph{bipartite} if there is a partition of {$V(G)$} into two independent sets $A, B$.
For an integer $q$, the graph $K_q$ is the complete graph on $q$ vertices.
For integers $p,q$, the graph $K_{p,q}$ is the bipartite graph $(A \cup B, E)$, where $|A|=p$, $|B|=q$,
and $uv \in E$ whenever $u \in A, v \in B$.
A~Hamiltonian path (resp. cycle) in a graph $G$ is a path (resp. a cycle) whose vertex set equals $V(G)$.
The diameter of a connected graph is the maximum distance between a pair of vertices, where the distance is given by the shortest path metric.
For a graph $G$, we define $\cc\cc(G)$ to be the family of subsets of $V(G)$ which induce the connected components of $G$.

\mic{
\begin{definition} \label{def:fragmentation}
The fragmentation of a set $X \subseteq V(G)$ is defined as \bmp{$\max(|X|, \ell)$, where~$\ell$ is the} number of connected components of $G-X$ with at least 3 neighbors \bmp{in~$X$}.
\end{definition}
}

\subparagraph*{Contractions and minors}
The operation of contracting an edge $uv \in E(G)$ introduces a~new vertex adjacent to all of {$N_G(\{u,v\})$}, after which $u$ and $v$ are deleted. The result of contracting $uv \in E(G)$ is denoted $G / uv$. For $A \subseteq V(G)$ such that $G[A]$ is connected, we say we contract $A$ if we simultaneously contract all edges in $G[A]$ and introduce a single new vertex.
We say that $H$ is a contraction of $G$, if we can turn $G$ into $H$ by a (possibly empty) series of edge contractions.
We can represent the result of such a process with a mapping $\Pi \colon V(H) \to 2^{V(G)}$, such that the sets $(\Pi(h))_{h\in V(H)}$ form a partition of $V(G)$, induce connected subgraphs of $G$, 
and $E_G(\Pi(h_1), \Pi(h_2)) \ne \emptyset$ if and only if $h_1h_2 \in E(H)$.
This mapping is called  a contraction-model of $H$~in~$G$ and the sets $\Pi(h)$ are called branch sets.

We say that $H$ is a {minor} of $G$, if we can turn $G$ into $H$ by a (possibly empty) series of edge contractions, edge deletions, and vertex deletions.
The result of such a process is given by a minor-model, i.e., a~mapping $\Pi \colon V(H) \to 2^{V(G)}$, such that
the branch sets $(\Pi(h))_{h\in V(H)}$ are pairwise disjoint, induce connected subgraphs of $G$, 
and $h_1h_2 \in E(H)$ implies that $E_G(\Pi(h_1), \Pi(h_2)) \ne \emptyset$.

\subparagraph*{Planar graphs and modulators} \label{sec:prelims:planargraphs}
A plane embedding of graph $G$ is given by a mapping from $V(G)$ to $\mathbb{R}^2$ and a mapping that associates with each edge $uv \in E(G)$ a simple curve on the plane connecting the images of $u$ and $v$, such that the curves given by two distinct edges {can intersect only at the image of a vertex that is a common endpoint of both edges}.
{A plane graph is a graph with a fixed planar embedding, in which we identify the set of vertices with the set of their images on the plane.} 

A face in a plane embedding of a graph $G$ is a \mic{maximal connected subset of the plane minus the image of $G$}.
We say that a vertex or an edge lies on a face $f$ if its images  belongs to the closure of $f$.
In every plane embedding there is exactly one face of infinite area, referred to as the outer face.
All other faces are called interior.
Let $F$ denote the set of faces in a plane embedding of \bmp{$G$}.
For a plane graph $G$,
we define the radial graph $\widehat{G}$ as a bipartite graph with the set of vertices $V(\widehat{G}) = V(G) \cup F$ and edges given by pairs $(v,f)$ where $v\in V(G)$, $f \in F$, and $v$ lies on the face $f$.

Given a plane graph $G$ and two vertices $u$ and $v$, we define the radial distance $d_G(u, v)$ 
between $u$ and $v$ in $G$ to be one less than the minimum length of a sequence of vertices that
starts at $u$, ends in $v$, and in which every two consecutive vertices lie on a common face. We define $\mathcal{R}_G^\ell(v) = \{u \in V(G) \mid d_G(u,v) \le \ell\}$\bmp{, that is, the ball of radius~$\ell$ around~$v$}. 
The radial diameter of a plane graph $G$ equals $\max_{u,v\in V(G)} d_G(u,v)$. 
We also define the weak radial distance $w_G(u, v$)
between $u$ and $v$ to be one less than the minimum length of a sequence of vertices that
starts at $u$, ends in $v$, and in which every two consecutive vertices either share an edge or lie on a common \emph{interior} face.
\mic{If $u,v$ lie in different connected components of $G$ then $w_G(u,v) = \infty$.}
Observe that $w_G(u,v) \ge d_G(u,v)$.
We define $\mathcal{W}_G^\ell(v) = \{u \in V(G) \mid w_G(u,v) \le \ell\}$.

A graph is called planar if it admits a plane embedding.
Planarity of an $n$-vertex graph can be tested in time $\Oh(n)$~\cite{HopcroftT74}.
If a graph $G$ is planar and $H$ is a minor of $G$, then $H$ is also planar.
By Wagner's theorem, a graph $G$ is planar if and only if $G$ \mic{contains neither} $K_5$ nor $K_{3,3}$ as a minor. 
A~graph $G$ is planar if and only if every biconnected component in $G$ induces a planar graph.
A~planar modulator in $G$ is a vertex set $X \subseteq V(G)$ such that $G-X$ is planar.
The minimum vertex planarization number of $G$, denoted $\mvp(G)$, is the minimum size of a  planar modulator in $G$.

\begin{definition}\label{def:prelim:strong}
A vertex \bmp{set} $X \subseteq V(G)$ in $G$ is called a strong planar modulator in $G$ if for every connected component $C$ of $G-X$ the subgraph induced by $N_G[C]$ is planar.
A planar modulator $X$ is called $r$-strong if additionally for every connected component $C$ of $G-X$ it holds that $|N_G(C)| \le r$.
\end{definition}

\subparagraph*{Treewidth and the LCA closure}
A tree decomposition of graph $G$ is a pair $(T, \chi)$ where~$T$ is a rooted tree, and~$\chi \colon V(T) \to 2^{V(G)}$, such that:
\begin{enumerate}
    \item For each~$v \in V(G)$ the nodes~$\{t \mid v \in \chi(t)\}$ form a {non-empty} connected subtree of~$T$. 
    \item For each edge~$uv \in E(G)$ there is a node~$t \in V(G)$ with~$\{u,v\} \subseteq \chi(t)$.
\end{enumerate}
The \emph{width} of a tree decomposition is defined as~$\max_{t \in V(T)} |\chi(t)| - 1$. The \emph{treewidth} of a graph~$G$, denoted $\tw(G)$, is the minimum width of a tree decomposition of~$G$.

{For a tree decomposition~$(T,\chi)$ of a graph~$G$ and a node~$t \in V(T)$, we use~$T_t$ denote the subtree of~$T$ rooted at~$t$. For a subtree~$T'$ of~$T$, we use~$\chi(T') = \bigcup_{t \in V(T')} \chi(t)$. We use the following standard properties of tree decompositions.

\begin{observation} \label{obs:bag:separator}
If~$(T,\chi)$ is a tree decomposition of a graph~$G$ and~$t \in V(T)$, then the set~$\chi(t)$ is an unrestricted~$(\chi(T_t), V(G) \setminus \chi(T_t))$-separator.
\end{observation}

\begin{observation} \label{obs:treedec:subgraph:subtree}
If~$(T,\chi)$ is a tree decomposition of a graph~$G$ and~$H$ is a connected subgraph of~$G$, then the nodes~$\{t \in V(T) \mid \chi(t) \cap V(H) \neq \emptyset\}$ form a connected subtree of~$T$.
\end{observation}
}

\bmp{The following lemma, originating in the work on protrusion decompositions of planar graphs~\cite{bodlaender2016meta}, shows how a vertex set whose removal bounds the treewidth of a planar graph can be used to divide the graph into pieces whose neighborhood-size can be controlled.}

\begin{lemma}[{\cite[Lem. 15.13]{fomin2019kernelization}}]
\label{lem:undeletable:treewidth-modulator}
Suppose that graph $G$ is planar, $S \subseteq V(G)$, and $\tw(G-S) \le \eta$.
Then there exists $Y \subseteq V(G) \setminus S$, such that\footnote{In the original statement of the lemma, $Y$ is a superset of $S$ but we choose this formulation, with a slightly weaker size guarantee, for our convenience.}
\begin{itemize}
    \item $|Y| \le 4(\eta +1) \cdot |S|$,
    \item each connected component of $G - (S \cup Y)$ has at most 2 neighbors in $S$ and at most $2\eta$ neighbors in $Y$.
\end{itemize}
Furthermore, given $G, S$, and a tree decomposition of $G-S$ of width $\eta$, \bmp{such a} set $Y$ can be computed in polynomial time.
\end{lemma}

\bmp{We also use the least common ancestor closure and its properties.}

\begin{definition}\label{def:treewidth:lca}
Let $T$ be a rooted tree and $S \subseteq V(T)$ be a set of vertices in $T$. 
We define the least common ancestor of (not necessarily distinct) $u, v \in V(T)$, denoted as {$\mathsf{LCA}(u, v)$}, to be the deepest node $x$ which is an ancestor of both $u$ and $v$.
The LCA closure of $S$ is the set
\[
\overline{\mathsf{LCA}}(S) = \{\mathsf{LCA}(u, v): u, v \in S\}.
\]
\end{definition}

\begin{lemma}[{\cite[Lem. 9.26, 9.27, 9.28]{fomin2019kernelization}}]\label{lem:treewidth:lca}
Let $T$ be a rooted tree, $S \subseteq V(T)$, and $M = \overline{\mathsf{LCA}}(S)$. {All of the following hold.}
\begin{enumerate}
    \item Each connected component $C$ of $T - M$ satisfies $|N_T(C)| \le 2$.
    \item $|M| \le 2\cdot |S| - 1$.
    \item $\overline{\mathsf{LCA}}(M) = M$.
\end{enumerate}
\end{lemma}

\paragraph*{Boundaried graphs}

A boundaried graph is a pair $(G,T)$, where $G$ is a graph and $T \subseteq V(G)$ is the set of the boundary vertices.
When $H = (G,T)$ is a boundaried graph, we write $V(H) = V(G)$ and $\partial(H) = T$.
A labeled boundaried graph is a triple $(G,T,\lambda)$, where $(G,T)$ is a boundaried graph and $\lambda \colon [|T|] \to T$ is a bijection assigning an integer label to each boundary vertex.
A boundaried graph $H$ is called $r$-boundaried if $|\partial(H)|=r$.  \bmp{Note that in some other papers working with boundaried graphs, the term~$r$-boundaried graph is used to refer to what we would call a labeled $r$-boundaried graph.} A boundaried graph $H$ is called $(\le r)$-boundaried if it is $r'$-boundaried for some $r' \le r$.

\bmp{Given two labeled $r$-boundaried graphs $H_1 = (G_1,T_1,\lambda_1)$ and $H_2 = (G_2,T_2,\lambda_2)$, we can glue the graphs along their boundaries to form a (non-boundaried) graph $H = H_1 \oplus H_2$.}
The gluing operation takes the disjoint union of $G_1$ and $G_2$ and identifies, for each~$i \in [r]$, vertex~$\lambda_1(i)$ with~$\lambda_2(i)$.  
If there are vertices $u_1, v_1 \in T_1$ and $u_2, v_2 \in T_2$ so that $u_1,u_2$ are identified and  $v_1,v_2$ are identified, then they are adjacent in $H_1 \oplus H_2$ if $u_1v_1 \in E(G_1)$ or $u_2v_2 \in E(G_2)$.
If two $r$-boundaried graphs $H_1, H_2$ are subgraphs of a common graph \bmp{$G$ with~$\partial(H_1) = \partial(H_2) = V(H_1) \cap V(H_2)$}, then we can glue them by picking an arbitrary~$\lambda \colon [r] \to \partial(H_1) = \partial(H_2)$, resulting in the graph~$(H_1, \partial(H_1), \lambda) \oplus (H_2, \partial(H_2), \lambda)$. Since the result is independent of the choice of~$\lambda$, we simply denote the graph resulting from this gluing process as~$H_1 \oplus H_2$. 

\begin{definition} \label{def:induced:boundaried}
For $A \subseteq V(G)$ we denote by
$G_\partial[A] = (H,T)$
the boundaried graph induced by $A$ where $H = G[A]$ and $T = \partial_G(A)$.
Furthermore
we define $G\brb A = (H, T)$ as the boundaried graph given by $H = G[N_G[A]]$ and $T = N_G(A)$.
\end{definition}

For a vertex set~$A \subseteq V(G)$ we use~$G_\partial[\overline A]$ as a shorthand for~$G_\partial[V(G) \setminus A]$. Observe that for any $A \subseteq V(G)$ it holds that $G = G\brb A \oplus G_\partial[\overline A]$. 

Consider a graph $G$, \bmp{a subset} $A \subseteq V(G)$, and
\bmp{the} boundaried graph $H = G\brb A$.
When $B \subseteq V(H) \setminus \partial(H) = A$, then $H \brb B = G \brb B$ \bmp{since~$N_H[B] = N_G[B]$}.
Following the convention for (non-boundaried) graphs, we define a boundaried plane graph as a pair $(G,T)$, where $G$ is a plane graph and $T \subseteq V(G)$.
We will be interested in boundaried graphs admitting well-behaving plane embeddings, in which the boundary vertices are lying ``on the outside''.

\begin{definition}\label{def:prelim:circum}
A boundaried plane graph~$\bmp{H}$ is called \circum if the graph $H-\partial(H)$ is connected and all vertices of $\partial(H)$ are embedded in the outer face of $H-\partial(H)$.
\end{definition}

\begin{definition}\label{def:prelim:single-faced}
A connected boundaried plane graph $H$ is called {\nice} if $\partial(H)$ is non-empty and 
all the vertices from $\partial(H)$ lie on the outer face of $H$.
A disconnected boundaried plane graph $H$ is called {\nice} if each connected component of $H$ is \nice.
A boundaried graph is called \nice if it admits an embedding as a \nice boundaried plane graph.
\end{definition}

\bmp{Note that a connected boundaried plane graph~$H$ can be circumscribed but not \nice, which happens if edges between vertices of~$\partial(H)$ separate a vertex of~$\partial(H)$ from the outer face.}
\mic{We use the fact that one can always augment a plane graph by inserting a new vertex within some face and making it adjacent to all vertices lying on that face. }

\begin{observation}\label{obs:prelim:single-faced}
A boundaried graph is \nice if and only if 
the graph $\widehat H$, obtained from $H$ by adding a new vertex $u_H$ adjacent to each vertex in $\partial(H)$, is planar and connected. 
The graph $\widehat H$ is called the closure of $H$.
\end{observation}

\paragraph*{Planarity criteria} \label{sec:prelims:criteria}

Let $G$ be a graph, $v \in V(G)$, and $S \subseteq V(G) \setminus \{v\}$.
We say that~\emph{$S$ planarizes~$v$} \mic{or $S$ is $v$-planarizing} if~$R_S[v,S]$ induces a  planar subgraph of $G$. 

Let~$C$ be a cycle in graph~$G$. A \emph{$C$-bridge} in~$G$ is a subgraph of~$G$ which is either a chord of~$C$, or a connected component~$B$ of~$G - V(C)$ together with all edges between~$B$ and~$C$, and their endpoints. 
If~$B$ is a $C$-bridge, then the vertices~$V(B) \cap V(C)$ are the \emph{attachments of~$B$}. Two $C$-bridges~$B_1, B_2$ \emph{overlap} if at least one of the following conditions is satisfied: (a)~$B_1$ and~$B_2$ have at least three attachments in common, or (b) the cycle~$C$ contains distinct vertices~$a,b,c,d$ (in this cyclic order) such that~$a$ and~$c$ are attachments of~$B_1$, while~$b$ and~$d$ are attachments of~$B_2$.

\begin{observation}\label{lem:prelim:overlap}
Let $G$ be a plane graph and $C$ be a cycle in $G$.
If two $C$-bridges in $G$ overlap, then one of them is embedded entirely in the interior of $C$ and the other one is embedded entirely in the exterior of $C$.
\end{observation}

For a graph~$G$ with a cycle~$C$, the corresponding \emph{overlap graph}~$O(G,C)$ has the~$C$-bridges in~$G$ as its vertices, with an edge between two bridges if they overlap. The following characterization of planar graphs will be essential to our arguments.

\begin{lemma}[{\cite[Thm.~3.8]{DiETT99}}] \label{lem:planarity:characterization}
A graph~$G$ with a cycle~$C$ is planar if and only if the following two conditions hold:
\begin{itemize}
	\item For each $C$-bridge~$B$ in~$G$, the graph~$B \cup C$ is planar.
	\item The overlap graph~$O(G,C)$ is bipartite.
\end{itemize}
\end{lemma}
Although the preceding lemma was originally stated only for biconnected graphs, it is easily seen to be true for all graphs.

{Similarly as in previous work~\cite{JansenLS14}, we will use Lemma~\ref{lem:planarity:characterization} to argue that a vertex~$v$ which is surrounded by~$\Omega(k)$ disjoint separators, each of whose removal leaves~$v$ in a planar connected component, is irrelevant with respect to planar deletion sets of size~$k$. The exact type of separators employed  will differ from earlier work, though. We introduce some terminology to formalize these concepts.}

\bmp{In a graph~$G$ we say that a sequence of pairwise disjoint vertex sets $(S_1, \dots, S_m)$ is \emph{nested} with respect to $v$ if for each~$i \in [m-1]$ the set~$S_i$ is a restricted~$(v,S_{i+1})$-separator. Note that this is equivalent to stating that~$v \notin \bigcup_{i=1}^m S_i$ and $R_G(v,S_i) \subseteq R_G(v,S_{i+1})$ for each $i \in [m-1]$.}

\bmp{Similarly, for~$\ell \in \mathbb{N}$ and two vertices~$v_0, v_{\ell+1}$ in a graph~$G$, a sequence of vertex-disjoint restricted $(v_0, v_{\ell+1})$-separators~$(S_i)_{i=1}^\ell$ is \emph{nested} if~$S_i$ is a $(S_{i-1}, S_{i+1})$-separator \bmp{for each~$i \in [\ell]$}, interpreting~$S_0 = \{v_0\}$ and~$S_{\ell+1} = \{v_{\ell+1}\}$. This implies that~$(S_1, \ldots, S_m)$ is nested with respect to~$v_0$, while the reverse sequence is nested with respect to~$v_{\ell+1}$.}

\bmp{The following two planarization criteria show how the planarity of subgraphs formed by reachability sets of nested separators can ensure planarity of the entire graph. The first one was already used in prior work~\cite{JansenLS14}; the second one is new. Their proofs are deferred to Section~\ref{app:planar}.}

\begin{lemma}[Cf.~\cite{JansenLS14}]\label{lem:prelim:criterion:old}
Let $v_0,v_3 \in V(G)$ and let $S_1,S_2 \subseteq V(G)$ be disjoint nested restricted $(v_0,v_3)$-separators \bmp{in a connected graph~$G$}, so that $G[S_i]$ has a Hamiltonian cycle for each~$i \in [2]$. If $R_G[v_0,S_2]$ and $R_G[v_3,S_1]$ induce planar graphs, then~$G$ is planar.
\end{lemma}

\begin{restatable}{lemma}{restPrelimCriterionNew}
\label{lem:prelim:criterion:new}
Let $v_0,v_4 \in V(G)$ and let $S_1,S_2,S_3 \subseteq V(G)$ be disjoint nested restricted $(v_0,v_4)$-separators \bmp{in a connected graph~$G$}, so that $G[S_i]$ is connected for each~$i \in [3]$. If $R_G[v_0,S_3]$ and $R_G[v_4,S_1]$ induce planar graphs, then~$G$ is planar.
\end{restatable}

As a corollary of the two lemmas above, we get that when $v \in V(G)$ is separated from some other vertex by either 2 cyclic or 3 connected nested $v$-planarizing separators and $G-v$ is planar, then $G$ is planar as well.
\mic{This corollary is sufficient for our purposes but we provide the more general proofs for completeness.}

\bmp{We will also use the following observation to argue for planarity of a graph.}

\begin{observation} \label{obs:planar:cutvertex}
Let~$G$ be a graph and let~$v \in V(G)$. Then~$G$ is planar if and only if for each connected component~$C_i$ of~$G - v$, the graph~$G[V(C_i) \cup \{v\}]$ is planar.
\end{observation}

Intuitively, planar drawings of the graphs~$G[V(C) \cup \{v\}]$ can be merged into a planar drawing of~$G$ by identifying the drawings of~$v$.

\paragraph*{Irrelevant and necessary vertices}

Let $G$ be a graph.
We say that a vertex $v \in V(G)$ (resp. an edge $e \in E(G)$) is $k$-irrelevant in $G$ if for every vertex set $S \subseteq V(G)$ of size at most $k$, we have that $G-S$ is planar if and only if $(G-v) - S$ 
(resp. $(G \setminus e) - S$)
is planar.

The proofs of the following \bmp{criteria} are postponed to Section~\ref{app:planar}.

\begin{restatable}{lemma}{restPrelimConnectedSeparators}
\label{lem:prelim:connected-separators}
Let $G$ be a connected graph, $k \in \mathbb{N}$, $v \in V(G)$, and $S_1, S_2, \dots, S_{k+3} \subseteq V(G) \setminus \{v\}$ be disjoint, connected, nested (with respect to $v$), and $v$-planarizing.
Then $v$ is $k$-irrelevant in $G$ and every edge incident to $v$ is $k$-irrelevant in $G$.
\end{restatable}

\bmp{Finally, we introduce a variant of the previous criterion based on a sequence of vertex sets which are disjoint except intersecting at a single vertex~$t$.}

\begin{definition}\label{def:prelim:pseudo-nested}
Let $G$ be a graph and $v \in V(G)$.
We say that a sequence of vertex sets $(S_1, S_2, \dots, S_m)$ from $V(G)$ is a pseudo-nested sequence \bmp{of cycles} with respect to $v$ if there exists a vertex $t \in V(G)$ such that the following conditions hold:
\begin{enumerate}
    \item for every $i \in [m]$ we have that $v \not\in S_i$ and the graph $G[S_i]$ contains a Hamiltonian cycle,
    \item for every $1 \le i < j \le m$ \mic{it holds that $R_G(v,S_i) \cap S_j = \emptyset$,}
    \item $S_i \cap S_j = \{t\}$ for all $i\ne j \in [m]$.
\end{enumerate}
\end{definition}

\begin{restatable}{lemma}{restPrelimPseudoNested}
\label{lem:prelim:pseudo-nested}
Let $G$ be a connected graph \bmp{containing} $v \in V(G)$, let~$k \in \mathbb{N}$ and let $(S_1, S_2, \dots, S_{k+3})$ be a \bmp{pseudo-nested sequence of cycles} with respect to $v$, such that each set~$S_i$ is $v$-planarizing.
Then $v$ is $k$-irrelevant in $G$ and every edge incident to $v$ is $k$-irrelevant in $G$.
\end{restatable}

We say that a vertex $v \in V(G)$ is $k$-necessary in $G$ if for every planar modulator $S$ in $G$ of size at most $k$ it holds that $v \in S$.
Note that when $\mvp(G) > k$, then there are no such modulators, and according to the definition every vertex in $G$ is $k$-necessary. 

\section{Treewidth reduction}
\label{sec:grid}

The first step of the algorithm aims at reducing the treewidth of the input graph to be $k^{\Oh(1)}$ or concluding that $\mvp(G) > k$.
While the treewidth is large we are going to find a vertex that is either $k$-irrelevant or $k$-necessary.
Such a vertex can be safely removed from the graph and we can continue this process until no such vertex can be found anymore.
We take advantage of the fact that a graph with large treewidth must contain a large \bmp{grid} as a minor.
However, instead of working \bmp{with} grid minors directly, it will be more convenient to work with contraction models of graphs obtained from a grid by adding edges.
First, we define two types \bmp{of} triangulated grids: $\boxtimes_k$ and $\Gamma_k$, where $\Gamma_k$ has additional edges ``outside'' the grid. 

\begin{definition}
Let $k \in \mathbb{N}$.
The triangulated grid $\boxtimes_k$ is a graph with the vertex set $V = [k] \times [k]$ and the edge set $E$ defined by the following rules: 
\begin{enumerate}
    \item if $1 \le i < k$, $1 \le j \le k$, then $(i,j)(i+1,j) \in E$,
    \item if $1 \le i \le k$, $1 \le j < k$, then $(i,j)(i,j+1) \in E$,
    \item if $1 \le i < k$, $1 < j \le k$, then $(i,j)(i+1,j-1) \in E$.
\end{enumerate}
The graph $\Gamma_k$ is obtained from $\boxtimes_k$ by connecting the vertex $(k,k)$ to all vertices $(i,j) \in [k] \times [k]$ for which $i \in \{1,k\}$ or $j \in \{1,k\}$.
\end{definition}

We say that $H$ is an extension of $\boxtimes_k$ if $V(H) = [k] \times [k]$ and $E(\boxtimes_k) \subseteq E(H)$.
The graph $H$ can be represented as $\boxtimes_k \uplus \widehat H$, where
$\widehat H$ is another graph on vertex set $[k] \times [k]$
and $E(\boxtimes_k) \cap E(\widehat H) = \emptyset$.

We say that a vertex $(i,j) \in [k] \times [k]$ is internal if $\{i,j\} \cap \{1,2,k-1,k\} = \emptyset$, i.e., it does not lie in the two outermost layers of $\boxtimes_k$.
We denote the set of internal vertices in $\boxtimes_k$ by $Int(k)$.
For an internal \bmp{vertex} $x \in [k] \times [k]$, we define $T[x]$ to be \mic{the set of 25 elements from $[k] \times [k]$ forming a} $5 \times 5$ grid around $x$.

\subsection{Irrelevant vertices in a grid minor}

The reason why working with contraction models is more convenient than working with minor models is that a contraction model preserves \bmp{separation} properties of a graph.

\begin{observation}\label{lem:grid:contraction-separator}
Let $\Pi$ be a contraction model of $H$ in $G$.
If $x,y \in V(H)$, $u \in \Pi(x), v \in \Pi(y)$,
and $S \subseteq V(H)$ is a restricted $(x,y)$-separator in $H$, then $\Pi(S)$ is a restricted $(u,v)$-separator in $G$.
\end{observation}

Let $H = \boxtimes_\lambda \uplus \widehat{H}$ be an extension of $\boxtimes_\lambda$, where $\lambda = \Omega(k^3)$ is sufficiently large compared to $k$.
When given a contraction model of $H$ in $G$, ideally we would like to
{find a subgrid of size $(2k+7)\times(2k+7)$ which is not incident to any edges from $\widehat{H}$.} 
\mic{The quantity $2k+7$ is chosen to ensure that the central grid vertex is surrounded by $k+3$ layers. }
If this subgrid represents a planar subgraph of $G$, then \cref{lem:grid:contraction-separator} implies that there exists a vertex $v \in V(G)$ surrounded by $k+3$ connected $v$-planarizing separators in $G$, so $v$ is $k$-irrelevant.
However, it might be the case that every $(2k+7)\times(2k+7)$-subgrid is incident to some edge from $\widehat{H}$.
We show that this imposes strong conditions on any potential planar modulator of small size.

\begin{restatable}{lemma}{restGridIsolated}
\label{lem:grid:isolated}
Let $\lambda > k \in \mathbb{N}$, $H = \boxtimes_\lambda \uplus \widehat{H}$ be an extension of $\boxtimes_\lambda$, and $\Pi$ be a contraction model of $H$ in $G$.
Suppose there are $4\cdot 31^2 \cdot (k+2)^4$ internal vertices which are non-isolated in $\widehat{H}$.
Then either $\mvp(G) > k$ or $G$ contains a $k$-necessary vertex.
Furthermore, there is a~polynomial-time algorithm that
either finds \bmp{a $k$-necessary vertex} or correctly concludes that $\mvp(G) > k$.
\end{restatable}

\mic{Although the lemma is formulated for any integers $\lambda > k$, its preconditions give an implicit lower bound on $\lambda$ in terms of $k$.} 
Before proving \cref{lem:grid:isolated} we show how it helps in finding irrelevant or necessary vertices.

\begin{lemma}\label{lem:grid:irrelevant}
There is a polynomial-time algorithm that, given a graph $G$, an integer $k$, and a~minor model $\Pi$ of $\boxtimes_\lambda$ in $G$, where $\lambda = 31^2 \cdot (2k+7)^3 + 4$, either outputs a $k$-necessary vertex or a $k$-irrelevant vertex, or correctly concludes that $\mvp(G) > k$.
\end{lemma}
\begin{proof}
We can assume that $G$ is connected, as otherwise we just consider the connected component of $G$ containing the branch sets of $\Pi$.
We turn $\Pi$ into a contraction model of an extension of $\boxtimes_\lambda$ as follows.
Initialize $H = \boxtimes_\lambda$.
While there is a vertex $v \in V(G)$ not contained in any branch set but adjacent to one, we insert it \bmp{into} this branch set.
Since $G$ is connected, when this process terminates then the union of the branch sets is $V(G)$. 
Next, if there are vertices $u,v \in V(G)$, $uv \in E(G)$, belonging to branch sets of $x, y$, $xy \not \in E(H)$, insert $xy$ into $H$.
After \bmp{handling} all edges in $G$, we \mic{turn} $\Pi$ into a contraction model of $H$, being an extension of $\boxtimes_\lambda$.
Let us keep the variable $\Pi$ to denote this contraction model and represent $H$ as $\boxtimes_\lambda \uplus \widehat{H}$.

We \bmp{partition} $Int(\lambda)$ into \bmp{the vertex sets of} $31^4 \cdot (2k+7)^4 > 4\cdot 31^2 \cdot (k+2)^4 + k$ many disjoint $(2k+7)\times(2k+7)$-subgrids $A_i \subseteq [\lambda] \times [\lambda]$. 
Each \bmp{subgrid} contains $k+3$ layers surrounding a central vertex.
For each of these subgrids, we check whether $\Pi(A_i)$ induces a planar subgraph of $G$ using the classic linear-time algorithm~\cite{HopcroftT74}.
If \bmp{at least} $k+1$ of these subgraphs are non-planar, then clearly $\mvp(G) > k$ and we can terminate the algorithm.
Next, we check if the number of internal vertices which are  non-isolated in  $\widehat{H}$
is at least $4\cdot 31^2 \cdot (k+2)^4$.
If \bmp{so}, we use \cref{lem:grid:isolated} to either find a $k$-necessary vertex or conclude that $\mvp(G) > k$.
Otherwise, there is at least one $(2k+7)\times(2k+7)$-subgrid $A \subseteq [\lambda] \times [\lambda]$ which (1)~is not incident to any edge from $\widehat{H}$
\bmp{and for which} (2)~$\Pi(A)$ induces a planar subgraph of $G$.
Let $x \in A$ be the central vertex in $A$ and $v$ be any vertex from $\Pi(x)$.
Next, let $S^H_i$ be the $(k+4-i)$-th outermost layer of $A$ for $i \in [k+3]$, so we count the layers starting from $x$. 
From (1) we get that $R_H(x, S^H_i) \subseteq R_H(x, S^H_{i+1})$ for $i\in [k+2]$ and $R_H[x, S^H_{k+3}] \subseteq A$.
Since $\Pi$ is a contraction model, the \bmp{separation} properties can be lifted from $H$ to $G$ (see \cref{lem:grid:contraction-separator}).
Therefore
$R_G(v, \Pi(S^H_i)) \subseteq R_G(v, \Pi(S^H_{i+1}))$ for $i\in [k+2]$ and $R_G[x, \Pi(S^H_{k+3})] \subseteq \Pi(A)$.
We have thus constructed a sequence of $k+3$ sets which are disjoint, connected, and nested with respect to $v$.
By (2) these sets are $v$-planarizing.
Finally \cref{lem:prelim:connected-separators} implies that $v$ is $k$-irrelevant.
\end{proof}

In order to prove \cref{lem:grid:isolated} we need several observations about how inserting additional edges \bmp{into} $\boxtimes_k$ affects its susceptibility to planarization.
First we show that under some conditions inserting an edge \bmp{into} a subgraph of $\boxtimes_k$ makes it non-planar.
It is an adaptation of an argument from \cite[\S 7.8]{cygan2015parameterized}.

\begin{lemma}\label{lem:grid:k5}
Let $H$ be an extension of $\boxtimes_k$ \bmp{and let} $U \subseteq [k] \times [k]$ induce a connected subgraph of $\boxtimes_k$.
Suppose that there exist distinct $x,y \in U$ such that $x$ is internal, $T[x] \subseteq U$, and $xy \in E(H) \setminus E(\boxtimes_k)$.
Then $H[U]$ is not planar.
\end{lemma}
\begin{proof}
We will construct a $K_5$-minor within $H[U]$.
Let $x = (i,j)$.
The first branch set $B_1$ is just $\{x\}$.
Next, we divide the neighborhood of $x$ in $\boxtimes_k$ into sets $B_2 = \{(i, j-1), (i+1, j-1)\}$, $B_3 = \{(i+1, j), (i, j+1)\}$,  $B_4 = \{(i-1, j), (i-1, j+1)\}$.
Each $B_i$ for $i \in [4]$ induces a connected subgraph of $\boxtimes_k$ and each pair of them is connected by an edge from $E(\boxtimes_k)$.
Let $B_{1,2,3,4} = \bigcup_{i=1}^4 B_i$ and
$B'_5 = T[x] \setminus B_{1,2,3,4}$.
Then $B'_5$ is adjacent in $\boxtimes_k$ to $B_2,B_3,B_4$, but not to $B_1$.
Since $xy \not\in E(\boxtimes_k)$, the vertex $y$ is not contained in $B_{1,2,3,4}$.
Furthermore, $B'_5$ is an unrestricted $(y, B_{1,2,3,4})$-separator in $\boxtimes_k$.
By assumption $\boxtimes_k[U]$ is connected, so
 there exists a path $P$ within $U$ connecting $y$ to $B'_5$, which is disjoint from $B_{1,2,3,4}$.
This path can be a singleton when  $y \in B'_5$.
We set $B_5 = B'_5 \cup V(P)$.
Then $E_H(B_1, B_5) \ne \emptyset$ because $xy \in E(H)$.
Therefore $B_1,B_2,B_3,B_4,B_5$ form branch sets of a minor model of $K_5$ in $H[U]$.
\end{proof}

\begin{lemma}\label{lem:grid:cc}
Let $S \subseteq [k] \times [k]$ for some $k \in \mathbb{N}$.
Then the graph $\boxtimes_k - S$ has at most $|S|+1$ connected components.
\end{lemma}
\begin{proof}
If follows from the definition that $\boxtimes_k$ admits a Hamiltonian path; let $P$ denote \bmp{such a} path. 
Removing $|S|$ vertices can divide the path $P$ into at most $|S|+1$ connected components. 
The graph $P-S$ can be turned into $\boxtimes_k-S$ by adding edges and
this operation can only decrease the number of connected components, which implies the claim.
\end{proof}

Next, we show that if some branch set \bmp{of a contraction model of an extension of some $\boxtimes_\ell$} contains a connected subgraph $C$ with edges to sufficiently many other branch sets, which are represented by internal vertices from the grid, then any solution of size at most $k$ must intersect $C$.
This will imply a criterion for \bmp{the} necessity of a vertex.

\begin{lemma}\label{lem:grid:not-planar}
Let $k, \lambda \in \mathbb{N}$, $H$ be an extension of $\boxtimes_\lambda$, and $\Pi$ be a contraction model of $H$ in $G$.
Furthermore, let $S \subseteq V(G)$ be of size at most $k$.
Suppose there exists $x \in V(H)$ and a vertex set $C \subseteq \Pi(x)$ such that $G[C]$ is connected, $C \cap S = \emptyset$ and there are $\ell = 30(k+2)$ \bmp{distinct} internal vertices $x_1, \dots, x_\ell$ from $V(H)$ such that $x_i \ne x$ and $N_G(C) \cap \Pi(x_i) \ne \emptyset$ \bmp{for each~$i \in [\ell]$}.
Then $G-S$ is not planar.
\end{lemma}
\begin{proof}
Let $H_S \subseteq [\lambda] \times [\lambda]$ be \bmp{the} vertices that contain a vertex from $S$ in their branch sets, together with $x$.
Clearly $|H_S| \le k+1$, where the additive term accounts for~$x$.
Let $I_S = \{y \in Int(\lambda) \mid T[y] \cap H_S \ne \emptyset\}$.
Since any vertex $z \in [\lambda] \times [\lambda]$ can belong to at most 25 different sets $T[y]$, we have $|I_S| \le 25 \cdot (k+1)$.
Let $X = \{x_1, \dots, x_\ell\} \setminus I_S$.
We have $|X| \ge 5(k + 2)$.
By \cref{lem:grid:cc}, the graph $\boxtimes_\lambda - H_S$ has at most $k+2$ connected components, so there is
a set $U \subseteq [\lambda] \times [\lambda]$
inducing a connected component of $\boxtimes_\lambda - H_S$ which contains at least 5 vertices from $X$. \bmp{Note that~$x \notin U$ since~$x \in H_S$.} 
Since the graph $\boxtimes_\lambda$ does not contain $K_5$ as a subgraph,
there exists distinct vertices $x_1, x_2 \in X \cap U$ such that $x_1x_2 \not\in E(\boxtimes_\lambda)$.
Note that $T[x_1], T[x_2] \subseteq U$ \bmp{by definition of~$U$ and~$I_S$.}
By assumption, there \bmp{exist} vertices $v_1 \in \Pi(x_1)$, $v_2 \in \Pi(x_2)$, contained in $N_G(C)$.
The last argument is based on contracting $C$ into $x_1$, \bmp{which creates the edge~$x_1x_2$ if it does not already exist.} 
Let us consider a graph $H'$ obtained from $H$ by inserting the edge $x_1x_2$ (if this edge is present in $H$, then $H' = H$).
Since $C \cap S = \emptyset$, $x \not\in U$, and $C \subseteq \Pi(x)$, the graph $H'[U]$ is a minor of $G-S$.
From \cref{lem:grid:k5} we get that $H'[U]$ is not planar and the same holds for $G-S$.
\end{proof}

{We now use the lemma above to argue that if some branch set~$\Pi(x)$ of a contraction model of an extension of some $\boxtimes_\ell$ contains a star-like structure of paths leading to neighbors in~$\Pi(x)$ of vertices in different branch sets, the central vertex of this star is $k$-necessary.}

\begin{lemma}\label{lem:grid:paths}
Let $k, \lambda \in \mathbb{N}$, $H$ be an extension of $\boxtimes_\lambda$, and $\Pi$ be a contraction model of $H$ in $G$.
Suppose that $x \in V(H)$, $v \in \Pi(x)$, and there exists a family of $\ell = 31(k+2)$ paths $P_1, P_2, \dots, P_\ell$ \bmp{in~$G$}, such that
\begin{enumerate}
    \item for \bmp{each} $i \in [\ell]$, $P_i$ is a $(v,v_i)$-path, where $v_i \in \Pi(x_i)$ for some internal $x_i \ne x$, and all interior vertices of $P_i$ are contained in $\Pi(x)$,
    \item for \bmp{each} $i \ne j \in [\ell]$ it holds that $x_i \ne x_j$ and $V(P_i) \cap V(P_j) = \{v\}$.
\end{enumerate}
Then $v$ is $k$-necessary.
\end{lemma}
\begin{proof}
Consider any vertex set $S \subseteq V(G) \setminus v$ of size at most $k$.
We are going to show that $G-S$ cannot be planar.
Since the paths $P_1, P_2, \dots, P_\ell$ are vertex-disjoint except for $v$, the set $S$ can intersect at most $k$ of them.
Let $C \subseteq \Pi(x)$ be the union of the interiors of the paths \bmp{not} intersecting $S$, together with $v$.
Then $G[C]$ is connected, $C \cap S = \emptyset$ and there are $\ell - k \ge 30(k+2)$ different internal vertices $x'_i$ for which $N_G(C) \cap \Pi(x'_i) \ne \emptyset$.
Therefore the assumptions of \cref{lem:grid:not-planar} are satisfied and thus $G-S$ is not planar.
\end{proof}

When a vertex in $H$ has sufficiently many internal neighbors, then its branch set contains a tree spanning vertices adjacent to their branch sets, which is difficult to \bmp{disconnect} by
any solution of size at most $k$.
The next lemma shows that in this case either $\mvp(G) > k$ or some vertex meets the condition of \cref{lem:grid:paths}.

\begin{lemma}\label{lem:grid:large-neighborhood}
Let $k, \lambda \in \mathbb{N}$, $\ell = 2\cdot 31^2 \cdot (k+2)^3$, $H$ be an extension of $\boxtimes_\lambda$, and $\Pi$ be a contraction model of $H$ in $G$.
Suppose that there exists distinct vertices $x_0, x_1, \dots, x_\ell$ in $[\lambda] \times [\lambda]$ such that for \bmp{each} $i > 0$ vertex $x_i$ is internal and $x_0x_i \in E(H)$.
Then either $\mvp(G) > k$ or $G$ contains a $k$-necessary vertex.
Furthermore, there is a~polynomial-time algorithm that
either finds \bmp{a $k$-necessary vertex} or correctly concludes that $\mvp(G) > k$.
\end{lemma}
\begin{proof}
By assumption there is a mapping $\rho \colon [\ell] \to \Pi(x_0)$ such that for $i \in [\ell]$ we \bmp{have} $N_G(\rho(i)) \cap \Pi(x_i) \ne \emptyset$, i.e., $\rho(i)$ is a neighbor of $\Pi(x_i)$ in $\Pi(x_0)$.
Let $R = \{\rho(1), \dots, \rho(\ell)\}$.  

Since $G[\Pi(x)]$ is connected, there exists a tree $T$ in $G[\Pi(x)]$ connecting vertices from $R$
such that each leaf of $T$ belongs to $R$.
Such a tree can be computed in polynomial time as follows.
We initialize $T = G[\Pi(x)]$ and while there is an edge or a vertex from $\Pi(x) \setminus R$ whose removal maintains connectivity of $T$, we remove it.
While $T$ is not a tree with leaves in $R$, then we can always find a removable edge or vertex.

Suppose there exists a vertex $v \in R$ such that $\rho(i) = v$ for at least $31(k+2)$ indices $i \in [\ell]$.
Then $v$ trivially satisfies the conditions of \cref{lem:grid:paths} and hence $v$ is $k$-necessary.
Suppose now there exists a vertex $v \in V(T)$ such that the degree of $v$ in $T$ is at least $31(k+2)$.
Since every leaf of $T$ belongs to $R$, there at $31(k+2)$ paths in $G[\Pi(x)]$ starting at $v$ and ending at distinct vertices from $R$, so that these paths are vertex-disjoint except for their starting point $v$. \bmp{Since each such path can be extended with an edge from~$R$ to a unique branch set~$x_i$,} this again meets the conditions of \cref{lem:grid:paths}, so $v$ is $k$-necessary.

Suppose none of the previous two cases hold and consider any vertex set $S \subseteq V(G) \setminus \{v\}$ of size at most $k$.
Since we have $|\rho^{-1}(u)| \le 31(k+2)$ \bmp{for each~$u \in R$}, there \bmp{are} at least $\ell - 31 \cdot k(k+2) \ge 31^2 \cdot k(k+2)^2$ indices for which $\rho(i) \not\in S$.
Next, observe that removing $k$ vertices from a tree of maximal degree bounded by $31(k+2)$ can split it into at most $31 \cdot k(k+2)$ connected components.
Therefore, there is a vertex set $C$ inducing a connected component of $T-S$, such that $|\rho^{-1}(C)| \ge \frac{31^2 \cdot k(k+2)^2}{31 \cdot k(k+2)} > 30(k+2)$.
In particular this means that $G[C]$ is connected, $C \cap S = \emptyset$ and \bmp{$N_G(C) \cap \Pi(x_i) \neq \emptyset$} for \bmp{at least~$30(k+2)$ distinct vertices} $x_i$.
We take advantage of \cref{lem:grid:not-planar} to infer that $G-S$ is not planar.
Since the choice of $S$ was arbitrary, this implies that $\mvp(G) > k$.
\end{proof}

We are ready to prove \cref{lem:grid:isolated} (restated below) and thus finish the proof of \cref{lem:grid:irrelevant}. 
When there is a vertex in $H$ with sufficiently many internal neighbors then the proof reduces to \cref{lem:grid:large-neighborhood}.
Otherwise, as $\widehat{H}$ contains many non-isolated internal vertices, removing any $k$ of them cannot hit all the edges in $\widehat{H}$.
Therefore any potential solution of size at most $k$ leaves out some subgraph to which \cref{lem:grid:k5} is applicable and thus it cannot be a planar modulator.

\restGridIsolated*
\begin{proof}
We denote by
$A \subseteq Int(\lambda)$ the set of internal vertices which are non-isolated in $\widehat{H}$.
If there is a vertex $x \in [\lambda] \times [\lambda]$ (not necessarily internal) which is adjacent in $\widehat{H}$ to at least $\ell = 2\cdot 31^2 \cdot (k+2)^3$ vertices from $A$, then we arrive at the case considered in \cref{lem:grid:large-neighborhood}, so we are done. 
Suppose for the rest of the proof that there is no such vertex and consider any vertex set $S \subseteq V(G)$ of size at most $k$.
We are going to show that $G-S$ is not planar, implying $\mvp(G) > k$.

We proceed similarly as in the proof of \cref{lem:grid:not-planar}.
Let $H_S \subseteq [\lambda] \times [\lambda]$ be \bmp{the} vertices that contain a vertex from $S$ in their branch \bmp{set}.
Clearly $|H_S| \le k$.
Let $I_S = \{x \in Int(\lambda) \mid T[x] \cap H_S \ne \emptyset\}$.
Since any vertex $y \in [\lambda] \times [\lambda]$ can belong to at most 25 different sets $T[x]$, we have $|I_S| \le 25 \cdot k$.
Let $A_S = A \setminus (I_S \cup N_{\widehat{H}}[H_S]$).
Because we have assumed an upper bound on $|N_{\widehat{H}}(x) \cap A|$ \bmp{for each~$x \in [\lambda] \times [\lambda]$}, 
we have $|A_S| \ge 4\cdot 31^2 \cdot (k+2)^4 - 25 \cdot k - k \cdot (\ell + 1) \ge 31^2 \cdot (k+2)^4 > 5 \cdot (k+1)^2$.
Note that every vertex $x \in A_S$ is internal, $T[x] \cap H_S = \emptyset$, and if $xy \in E(\widehat{H})$, then  $y \not\in H_S$.

By \cref{lem:grid:cc}, the graph $\boxtimes_\lambda - H_S$ has at most $k+1$ connected components, so there is
a set $U \subseteq [\lambda] \times [\lambda]$
inducing a connected component of $\boxtimes_\lambda - H_S$ containing at least $5(k+1)$ vertices from $A_S$.
\bmp{This} gives  $5(k+1)$ vertices from $U \cap A_S$ adjacent in $\widehat{H}$ to vertices from $[\lambda] \times [\lambda] \setminus H_S$. \bmp{We distinguish two cases, based on whether there is a~$\widehat{H}$-neighbor of~$U \cap A_S$ that lies in the same component~$U$ or not}.
%
Suppose first that there are vertices $x_1,x_2 \in U$
such that $x_1 \in A_S$ and
$x_1x_2 \in E(\widehat{H})$.
By the assumption, $E(\boxtimes_k) \cap E(\widehat{H}) = \emptyset$ so $x_1x_2 \not\in E(\boxtimes_k)$.
As $T[x_1] \subseteq U$, this meets the conditions of \cref{lem:grid:k5}, which means that $H[U]$ is not planar and, since it is a minor of $G-S$, the graph $G-S$ is also not planar.

In the remaining case, \bmp{each vertex of~$U \cap A_S$ has a neighbor in~$\widehat{H}$ that does not lie in~$U$. Since there are at least~$5(k+1)$ vertices having such a neighbor and at most~$k+1$ components in~$\boxtimes_\ell - H_S$,} 
by a counting argument there is a vertex set $U' \subseteq [\lambda] \times [\lambda]$ inducing another connected component of $\boxtimes_\lambda - H_S$ and
at least 5 vertices in $U \cap A_S$ adjacent in $\widehat{H}$ to $U'$.
Since the graph $\boxtimes_\lambda$ does not contain $K_5$ as a subgraph, there must be a pair of them which is non-adjacent in $\boxtimes_\lambda$.
Therefore,
there exist vertices $x_1, x_2 \in U \cap A_S$, $y_1, y_2 \in U'$ such that $x_1, x_2$ are distinct and internal, $y_1,y_2$ are not necessarily internal nor distinct, $x_1x_2 \not\in E(\boxtimes_\lambda)$, $x_1y_1, x_2y_2 \in E(\widehat{H})$.
Our final argument is based on contracting $U'$ into $x_1$.
Let us consider a graph $H'$ obtained from $H$ by inserting the edge $x_1x_2$.
As $(U \cup U') \cap H_S = \emptyset$, we get that $H'[U]$ is a minor of $G-S$.
By the definition of $A_S$ we have that $T[x_1] \subseteq U$, so we can use \cref{lem:grid:k5} to conclude that $H'[U]$ (and so $G-S$) is not planar.
The claim follows.
\end{proof}

\subsection{Finding a large grid minor}

In order to take advantage of \cref{lem:grid:irrelevant}, we need a way to find a large \bmp{triangulated} grid minor when the input graph has large treewidth. \bmp{Since the triangulated grid~$\boxtimes_k$ is a minor of a~$\Oh(k) \times \Oh(k)$-grid, for this purpose it suffices to find a large grid minor.} 
It is known that any graph $G$ excluding a $k \times k$-grid minor has treewidth bounded by $k^{\Oh(1)}$.
The best known bound on the exponent at $k$ is 9 (modulo polylogarithmic factors)~\cite{ChuzhoyT21}.
Furthermore, there is a randomized polynomial-time algorithm that outputs either a minor model of a $k \times k$-grid in $G$ or a tree decomposition of width $\Oh(k^{98} (\log k)^{\Oh(1)})$~\cite{ChekuriC16}.
However, we would like to avoid randomization and instead we show that in our setting we can rely on the deterministic algorithm finding a~large grid minor in a planar graph of large treewidth~\cite{GuT12, fomin2019kernelization}.
For our convenience we state a~version of the planar grid theorem that already provides us with a~triangulated grid (and even with a contraction model which will come in useful in further applications).
All other algorithms summoned in this section are also deterministic.

\begin{proposition}[{\cite[Thm. 14.58]{fomin2019kernelization}}]
\label{lem:grid:contraction-model}
There is a polynomial-time algorithm that, given a connected planar graph $G$ and an integer $t > 0$, outputs either a tree decomposition of $G$ of width $27 \cdot t$ or a contraction model of $\Gamma_t$ in~$G$.
\end{proposition}

In the original statement the guarantee on the decomposition width is $\frac{27}{2} \cdot (t+1)$, but for the sake of keeping the formulas simple, we trivially upper bound $t+1$ by $2t$. 

A concept more general than a grid minor is a bramble~\cite{SeymourR93} which is, similarly as a grid minor, an object dual to a tree decomposition.
Kreutzer and Tazari~\cite{KreutzerT10} introduced the notion of a perfect bramble which enjoys some interesting properties.
First,
if a graph $G$ admits a perfect bramble of order $k$, then $\tw(G) \ge \lceil{\frac {k} {2}}\rceil - 1$, and if $\tw(G) = \Omega(k^4 \sqrt{\log k})$ then $G$ admits a perfect bramble of order $k$.
What is particularly interesting for us, the union of subgraphs constituting a perfect bramble forms a subgraph of maximal degree 4.
This allows us to find a subgraph of constant degree whose treewidth is still large. 

\begin{thm}[{\cite[Thm. 5.1(c), Cor. 5.1]{KreutzerT10}}]\label{lem:grid:bramble}
There is a constant $c_1>0$ and a polynomial-time algorithm that, given a graph $G$ with $\tw(G) \ge c_1 \cdot t^7$,
finds a subgraph $H$ of $G$, such that the maximal degree in $H$ is at most 4 and $\tw(H) \ge t$.
\end{thm}

Given an $n$-vertex subgraph $H$ of constant degree, we can take advantage of the approximation algorithm for \planardel by \mic{Kawarabayashi and Sidiropoulos~\cite{KawarabayashiS19}
whose approximation factor depends on the maximal degree in $H$ and $\log n$, which we will later bound by $k^{\Oh(1)}$.
We state it below in a form with two possible outcomes which is convenient for our application.
The first algorithm of this kind was given by Chekuri and Sidiropoulos~\cite{{ChekuriS18}} but with larger exponents in the approximation factor.}
If we conclude that $\mvp(H) > k$ and $H$ is a subgraph of $G$, then also $\mvp(G) > k$ and we can terminate the algorithm.
Otherwise, by adjusting the constants, we can find another (\bmp{possibly} large) subgraph that is planar and has large treewidth.
This allows us to use the planar grid theorem on this subgraph to find a grid minor, which is also a grid minor in the original graph~$G$.

\begin{thm}[{\cite[Cor. 1.2]{KawarabayashiS19}}]\label{lem:grid:bounded-degree}
There is a constant $c_2>0$ 
and a polynomial-time algorithm that, given an $n$-vertex graph $G$ of maximum degree $\Delta$ and an integer $k$, either correctly concludes that $\mvp(G) > k$ or returns a planar modulator of size at most $c_2 \cdot \Delta^3\cdot (\log n)^\frac{7}{2} \cdot k$.
\end{thm}

If the treewidth of $G$ is not too large, we would like to efficiently find a tree decomposition of moderate width.
To this end, we will use the polynomial-time approximation algorithm for treewidth by Feige et al.~\cite{FeigeHL08}.

\begin{thm}[{\cite[Thm. 6.4]{FeigeHL08}}]\label{lem:grid:feige}
There is a polynomial-time algorithm that, given a graph $G$ such that $\tw(G) \le t$, outputs a tree decomposition of $G$ of width $\Oh(t \sqrt{\log t})$.
\end{thm}

We wrap up these ingredients into the following algorithm. \micr{I have updated constants}

\begin{lemma}\label{lem:grid:minor-find}
Let $d > 0$ be any integer constant. 
There is a polynomial-time algorithm that, given an $n$-vertex graph $G$ and integer $k > 0$ such that $\log n \le k \log k$, either outputs a tree decomposition of $G$ \mic{of width $\Oh(k^{36})$, or outputs a minor model of $\boxtimes_\lambda$, where $\lambda = d \cdot k^{5}$,} or correctly concludes that $\mvp(G) > k$.
\end{lemma}
\begin{proof}
We refer to constants $c_1, c_2$ from Theorems \ref{lem:grid:bramble} and \ref{lem:grid:bounded-degree}.
Let $d_1 = 28 \cdot d$ and $d_2$ be a constant for which $(\log k)^{\frac 7 2} \le d_2 \cdot \sqrt{k}$ holds for all $k \ge 1$.
We begin with \cref{lem:grid:feige} for $t = c_1 \cdot (4^3 \cdot c_2 \cdot d_2 + d_1)^7 \cdot k^{5 \cdot 7}$.
If $\tw(G) \le t$, we obtain a tree decomposition of $G$ of width $\Oh(k^{5 \cdot 7} \sqrt{\log k}) = \Oh(k^{36})$ and we are done.
Otherwise $\tw(G) > c_1 \cdot (\bmp{c_2 \cdot 4^3} \cdot d_2 + d_1)^7 \cdot k^{5 \cdot 7}$ and we can use the algorithm from \cref{lem:grid:bramble} to find a subgraph $H$ of $G$, such that the maximal degree in $H$ is $\Delta \le 4$ and $\tw(H) \ge (c_2 \cdot 4^3 \cdot d_2 + d_1) \cdot k^{5}$.
As the next step, we execute the algorithm from \cref{lem:grid:bounded-degree} for the graph $H$ with $\Delta \le 4$.
In the first case, it
concludes that $\mvp(H) > k$, which implies that $\mvp(G) > k$ and we are done. 
Otherwise it
 returns a planar modulator $S_H \subseteq V(H)$ of size at most 
 \[
 c_2 \cdot \Delta^3  \cdot (\log n)^\frac{7}{2} \cdot k \le c_2 \cdot {4^3} \cdot (k\log k)^\frac{7}{2} \cdot k \le c_2 \cdot 4^3 \cdot d_2 \cdot k^{5}.
 \]
The graph $H' = H - S_H$ is planar and $\tw(H') \ge \tw(H) - |S_H| \ge d_1 \cdot k^{5} > 27 \cdot d \cdot k^{5}$ because removing a single vertex can decrease the treewidth by at most one.
For the graph $H'$ we can take advantage of the planar grid theorem.
If it is disconnected, we focus on the connected component having large treewidth.
We use \cref{lem:grid:contraction-model} with $t = \lambda = d \cdot k^{5}$ and, since treewidth of $H'$ is sufficiently large, the algorithm must return a contraction model of $\Gamma_\lambda$ in $H'$.
Since $H'$ is a subgraph of $G$, we can turn it into a~minor model of $\boxtimes_\lambda$ in $G$.
\end{proof}

Finally, we need to justify the assumption that $\log n$ is small.
If $\log n$ is larger than $k \log k$, then the exponential term of the form $2^{\Oh(k \log k)}$ becomes polynomial.
In such a case, we can use the known FPT algorithm for \planardel.

\begin{thm}[\cite{JansenLS14}]\label{lem:grid:fpt}
There is an algorithm that, given \bmp{an $n$-vertex} graph $G$ and an integer~$k$,
runs in time $2^{\Oh(k \log k)}\cdot n$ and either correctly concludes that $\mvp(G) > k$ or outputs a planar modulator of size $\mvp(G) \le k$.
\end{thm}

After a large grid minor is identified, we invoke \cref{lem:grid:irrelevant} to find either a necessary vertex or an irrelevant vertex.
Once such a vertex is located we can remove it from the graph and proceed recursively.
In the end, we either solve the problem or arrive at the case where treewidth is bounded.

\begin{proposition}\label{lem:grid:final}
There is a polynomial-time algorithm that, given a graph $G$ and an integer $k$,
either
\begin{enumerate}
    \item 
outputs a planar modulator in $G$ of size $\mvp(G) \le k$, or 
\item correctly concludes that $\mvp(G) > k$, or
\item 
outputs a graph $G'$ together with a tree decomposition of width $\Oh(k^{36})$ and an integer $k' \le k$, such that if $\mvp(G) \le k$, then $\mvp(G')= \mvp(G) - (k - k')$.
Furthermore, there is a polynomial-time algorithm that, given $G, k, G'$, and a planar modulator $S'$ in $G'$ of size at most $k'$,
outputs a~planar modulator $S$ in $G$ such that
$|S| = |S'| + (k - k')$.
\end{enumerate}
\end{proposition}
\begin{proof}
First, we check if $k \log k < \log n$.
If yes, then the running time $2^{\Oh(k \log k)}\cdot n$ from \cref{lem:grid:fpt} becomes polynomial.
In this is the case, we take advantage of the exact FPT algorithm to either conclude that $\mvp(G) > k$ or find a planar modulator of size $\mvp(G) \le k$.
Otherwise $\log n \le k \log k$ and we can execute the algorithm from
\cref{lem:grid:minor-find}.
We choose a sufficiently large constant $d$ to ensure that
$31^2 \cdot (2k+7)^3 + 4 \le d \cdot k^{5}$ for all $k \ge 1$.
If the algorithm returns a tree decomposition of width $\Oh(k^{36})$ or concludes that $\mvp(G) > k$, we are done.
Otherwise we obtain a minor model of $\boxtimes_\lambda$
in $G$, where $\lambda = 31^2 \cdot (2k+7)^3 + 4$ (if the returned grid is larger, we can easily truncate the minor model).
We supply the computed minor model to the algorithm from \cref{lem:grid:irrelevant}.
It either outputs a $k$-irrelevant or $k$-necessary vertex $v$, or concludes that $\mvp(G) > k$.
In the first case we recurse on the instance $(G-v,k)$, in the second case we recurse on $(G-v,k-1)$, and on the last case we terminate the algorithm.

If $v$ is $k$-irrelevant then any planar deletion set of size at most $k$ in $G-v$ is also valid in $G$.
This entails that whenever $\mvp(G) \le k$ then $\mvp(G-v) = \mvp(G)$.
If $v$ is $k$-necessary then $\mvp(G) \le k$ implies $\mvp(G-v) = \mvp(G) - 1$ and so $\mvp(G-v) > k-1$ implies $\mvp(G) > k$.
Furthermore,
for any planar deletion set $S$ in $G-v$, clearly $S \cup \{v\}$ is a planar deletion set in $G$.

Let $(G', k')$ be the instance on which the recursion has stopped.
Observe that when we are given 
a planar modulator $S'$ in $G'$ of size at most $k'$
then by inserting the removed necessary vertices to $S'$ we can turn it into a planar modulator in $G$ of size $|S'| + (k - k')$.
If the algorithm has terminated with 
a planar modulator in $G'$ of size $\mvp(G') \le k'$
then we turn it into a planar modulator in $G$ of size $\mvp(G') + (k - k') = \mvp(G)$.
If we received a message that $\mvp(G') > k'$ then it implies that $\mvp(G) > k$.
Finally, \bmp{if} the algorithm terminated with a tree decomposition of $G'$  of width $\Oh((k')^{36})$ we can return it and whenever the lifting algorithm receives a planar modulator $S'$ in $G'$ of size $\mvp(G') \le k'$, it can be turned into a planar modulator in $G$ of size $|S'| + (k - k')$.
\end{proof}

\section{Constructing a strong modulator}
\label{sec:treewidth}

In this section we work with a given tree decomposition of $G$ of width $k^{\Oh(1)}$.
We are going to process such a decomposition several times and so we begin with presenting a general procedure which scans the tree decomposition bottom-up and detects subgraphs with certain properties. \bmp{The procedure is based on the well-known covering/packing duality for connected subgraphs of a bounded treewidth host graph (cf.~\cite[Lemma 3.10]{RaymondT17}). While such arguments are typically applied to pack and cover cycles or minor models, we will need it in a slightly more general form. In the lemma below, we therefore work with a family~$\mathcal{F}$ of vertex subsets of~$G$ which is hereditary (so that~$\mathcal{F}$ is closed under taking subsets), requiring only a polynomial-time subroutine to recognize whether a set belongs to~$\mathcal{F}$ or not.}

\begin{lemma}\label{lem:treewidth:meta}
There is a polynomial-time algorithm that, given a graph $G$, a tree decomposition $(T,\chi)$ of $G$, an integer~$k$, and a polynomial-time subroutine that tests membership of vertex subsets~$S \subseteq V(G)$ in a hereditary family~$\mathcal{F} \subseteq 2^{V(G)}$, 
returns one of the following:
\begin{enumerate}
    \item a subset $T_0 \subseteq V(T)$
    of size at most $k$, such that
    every set $C \subseteq V(G) \setminus \bigcup_{t \in T_0} \chi(t)$, for which $G[C]$ is connected, belongs to $\mathcal{F}$, or
    \item a collection of $k+1$ disjoint vertex subsets $(V_i)_{i=1}^{k+1}$, such that \bmp{each} $G[V_i]$ is connected and $V_i \not\in \mathcal{F}$.
\end{enumerate}
\end{lemma}
\begin{proof}
Recall that notation $T_t$ stands for the subtree of $T$ rooted at $t \in V(T)$.
Observe that for a given set $U \subseteq V(G)$, we can check if all its connected subsets belong to $\mathcal{F}$ by checking if each connected component of $G[U]$ belongs to $\mathcal{F}$.

We consider an iterative process: in each step we identify a node $t \in V(T)$ whose bag separates some connected subsets which belong to $\mathcal{F}$ from the rest of the graph
but their union together with the bag of $t$ contains a connected subset which does not belong to $\mathcal{F}$.

Initialize $S_0 = \emptyset$ and $U_0 = \emptyset$.
In the $i$-th step, $i \ge 1$, we consider the induced subgraph $G - U_{i-1}$.
For $t \in V(T)$, let $\chi_i(t) = \chi(t) \setminus U_{i-1}$.
The pair $(T,\chi_i)$ is a tree decomposition of $G - U_{i-1}$.
We choose $t_i \in V(T)$ to be \bmp{a} deepest node (breaking ties arbitrarily) so that 
$\chi_i(T_{t_i})$ contains a connected vertex set which does not belong to $\mathcal{F}$ -- let us refer to this set as $V_i$.
Let $U_i = U_{i-1} \cup \chi_i(T_{t_i})$ and $S_i = S_{i-1} \cup \chi_i(t_i)$.
By \cref{obs:bag:separator}, the set $\chi_i(t_i)$ separates $G - U_{i-1}$ into $V(G) \setminus U_i$ and $\chi_i(T_{t_i}) \setminus \chi_i(t_i)$.
By the choice of $t_i$, we know that all the connected subsets of $\chi_i(T_{t_i}) \setminus \chi_i(t_i)$ belong to $\mathcal{F}$.
We continue this process until all connected subsets of $V(G) \setminus U_i$ belong to $\mathcal{F}$,
so there is no way to choose~$t_{i+1}$.

If we have performed at least $k+1$ iterations, we return the sets  $V_1, V_2, \dots, V_{k+1}$.
Since $V_i \subseteq U_i$ and $V_j \cap U_i = \emptyset$ for $j > i$, these sets are  disjoint.
Furthermore, each of them is connected and does not belong to $\mathcal{F}$.

If we have terminated in at most $\ell \le k$ iterations, then the set $S_\ell$ is a~union of $\ell$ sets $\chi_i(t_i)$
and $\chi_i(t_i) \subseteq \chi(t_i)$.
Observe that $V(G) \setminus U_\ell$ and sets  $\chi_i(T_{t_i}) \setminus \chi_i(t_i)$, $i \in [\ell]$,
are pairwise separated by $S_\ell$ in $G$.
These sets contain only connected subsets from $\mathcal{F}$.
In this case we return the set of nodes $\{t_1, \dots, t_\ell\}$. 
\end{proof}

\subsection{Planarizing small neighborhoods}

\bmp{In this section we work towards proving Lemma~\ref{lem:treewidth:planar-neighborhood-reduction}, which essentially gives a reduction to problem instances with a \emph{tidy} planar modulator~$U$ such that each component of~$G-U$ remains planar even when reinserting two arbitrary vertices from~$U$. To get there, we need two ingredients.}

\bmp{The following simple lemma uses the fact that obstructions to planarity are connected together with covering/packing duality in a tree decomposition to either find a planar modulator of bounded size or conclude that any planar modulator must be large.}

\begin{lemma}\label{lem:treewidth:modulator}
There is a polynomial-time algorithm that, given a graph $G$ with a tree decomposition of width $t-1$ and an integer $k$,
either returns a planar modulator in $G$ of size at most $k\cdot t$ or correctly concludes that $\mvp(G) > k$.
\end{lemma}
\begin{proof}
We apply Lemma~\ref{lem:treewidth:meta} for the family $\mathcal{F}$ of vertex subsets which induce planar subgraphs of $G$.
This property is hereditary and decidable in polynomial time~\cite{HopcroftT74}, so the conditions of the lemma are satisfied.
In the first scenario,
the algorithm returns a~set of
at most $k$ nodes, so that the union of their bags is a
planar modulator of size at most $k\cdot t$.
Otherwise,
the algorithm returns $k+1$ vertex-disjoint non-planar subgraphs.
This testifies that every planar modulator must have at least $k+1$ vertices.
\end{proof}

\bmp{The next ingredient needed for Lemma~\ref{lem:treewidth:planar-neighborhood-reduction} is the following encapsulation of Lemma~\ref{lem:treewidth:meta}.}

\begin{lemma}\label{lem:treewidth:planar-neighborhood}
There is a polynomial-time algorithm that, given a graph $G$, sets $X_0 \subseteq X \subseteq V(G)$, such that $G-X$ is planar, a tree decomposition of $G-X$ of width $t-1$, and integer $k$, returns one of the following:
\begin{enumerate}
    \item a set $Y \subseteq V(G) \setminus X$ of size at most $k\cdot t$, such that for each vertex set $C \subseteq V(G) \setminus (X\cup Y)$, for which $G[C]$ is connected, the graph $G[C \cup X_0]$ is planar, or
    \item a collection of $k+1$ disjoint vertex subsets $(V_i)_{i=1}^{k+1}$ such that $V_i \subseteq V(G) \setminus X$, $G[V_i]$ is connected, and $G[V_i \cup X_0]$ is non-planar. 
\end{enumerate}
\end{lemma}
\begin{proof}
We apply Lemma~\ref{lem:treewidth:meta} for the graph $G-X$ and family $\mathcal{F}$ of its vertex subsets defined
as follows: $C \in \mathcal{F}$ if $G[C \cup X_0]$ is planar.
In the first case, the union of bags of the returned nodes has size at most $k\cdot t$ and satisfies condition (1).
Otherwise, we obtain $k+1$ disjoint connected vertex sets that do not belong to $\mathcal{F}$.
\end{proof}

\bmp{Using these ingredients we now present the lemma that reduces an instance to one with a suitably tidy modulator, or detects a small vertex set that allows an approximation-preserving reduction to a smaller graph.}

\begin{lemma}\label{lem:treewidth:planar-neighborhood-reduction}
There is a polynomial-time algorithm that, given a graph $G$, a~set $X \subseteq V(G)$ of size~$m$, such that $G-X$ is planar, a tree decomposition of $G-X$ of width $t-1$, and an integer $k$, returns one of the following:
\begin{enumerate}
    \item a set $U \supseteq X$ of size at most $2^3\cdot m^2\cdot k^3\cdot t^3$, such that for each vertex set $C \subseteq V(G) \setminus U$, for which $G[C]$ is connected,
    and for each $X_0 \subseteq U$ of size at most 2,
    the graph $G[C \cup X_0]$ is planar, or
    \item a set $X_0 \subseteq V(G)$ of size at most 2,
    such that $\mvp(G) \le k$ implies $\mvp(G-X_0) \le \mvp(G) - 1$.
\end{enumerate}
\end{lemma}
\begin{proof}
For every $x \in X$ we
apply Lemma~\ref{lem:treewidth:planar-neighborhood} with $X_0 = \{x\}$ and remaining parameters as in the statement of this lemma.

Suppose that for at least one call we obtain answer (2): a collection of $k+1$ disjoint vertex subsets $(V_i)_{i=1}^{k+1}$ such that $V_i \subseteq V(G) \setminus X$ and $G[V_i \cup \{x\}]$ is non-planar.
Assume that $G$ admits a planar modulator $S$ of size $\ell \le k$.
Then for some $i \in [k+1]$ it holds that $V_i \cap S = \emptyset$.
Since $G[V_i \cup \{x\}]$ is non-planar, \bmp{we have} $x \in S$.
Therefore $S \setminus \{x\}$ is a planar modulator for $ G - \{x\}$ of size $\ell - 1$ and the set $\{x\}$ satisfies condition (2).

Suppose otherwise, that for each $x \in X$ we computed a set $Y_x \subseteq V(G) \setminus X$ of size at most $k \cdot t$, such that for each vertex set $C \subseteq V(G) \setminus (X\cup Y_x)$, for which $G[C]$ is connected, the graph $G[C \cup \{x\}]$ is planar.
Let $Y = \bigcup_{x \in X} Y_x$.
For each $x \in X$ it holds that every connected component of $G - (X\cup Y)$ is contained in some connected component of $G - (X\cup Y_x)$.
Therefore for each $x \in X$ and each $C \subseteq V(G) \setminus (X\cup Y)$, for which $G[C]$ is connected, the graph $G[C \cup \{x\}]$ is planar.

Next, we iterate over all subsets $X_0 \subseteq X \cup Y$ of size 2 and for each we apply Lemma~\ref{lem:treewidth:planar-neighborhood} to the tree decomposition of $G - (X\cup Y)$ and sets $X_0, X\cup Y$.
Note that we can easily turn a tree decomposition of $G-X$ into a tree decomposition of $G-(X\cup Y)$ \bmp{without increasing the width}.

Again, suppose that for at least one call we obtain answer (2): a collection of $k+1$ disjoint vertex subsets $(V_i)_{i=1}^{k+1}$ such that $V_i \subseteq V(G) \setminus X$ and $G[V_i \cup X_0]$ is non-planar.
Assume that $G$ admits a planar modulator $S$ of size $\ell \le k$.
Then for some $i \in [k+1]$ it holds that $V_i \cap S = \emptyset$.
Since $G[V_i \cup X_0]$ is non-planar, \bmp{we have} $X_0 \cap S \ne \emptyset$.
Therefore $S \setminus X_0$ is a planar modulator for $ G - X_0$ of size at most $\ell-1$ and the set $X_0$ satisfies condition (2).

In the remaining case, for each considered $X_0 = \{x,y\}$ we have obtained a set $Z_{\{x,y\}} \subseteq V(G) \setminus (X \cup Y)$ of size at most $k \cdot t$, such that for each vertex set $C \subseteq V(G) \setminus (X\cup Y \cup  Z_{\{x,y\}})$, which is connected in $G$, the graph $G[C \cup \{x,y\}]$ is planar.
Let $Z = \bigcup_{{\{x,y\}} \subseteq X \cup Y} Z_{\{x,y\}}$ and $U = X \cup Y \cup Z$.
For each set ${\{x,y\}} \subseteq X \cup Y$ it holds that every connected component of $G - U$ is contained in some connected component of $G - (X\cup Y \cup Z_{\{x,y\}})$.
Therefore for each ${\{x,y\}} \subseteq X \cup Y$ and each $C \subseteq V(G) \setminus U$, which is connected in $G$, the graph $G[C \cup \{x,y\}]$ is planar.

\bmp{To argue that the constructed set~$U$ satisfies condition~(2),} it remains to handle pairs ${\{x,y\}} \subseteq U$, which are not contained in $X \cup Y$.
Let $C \subseteq V(G) \setminus U$ be a vertex set connected in $G$.
If $x,y \in Y \cup Z$, then
$C \cup \{x,y\} \subseteq V(G) \setminus X$ and this set induces a planar graph because $G-X$ is planar by the assumption.
The remaining case is when 
$x \in X$ and $y \in Z$.
If $y \in N_G(C)$,
then $C \cup \{y\}$ is a connected subset of $V(G) \setminus (X \cup Y)$ and it remains planar with $x$ attached due to the first round of the construction.
If $y \not\in N_G(C)$, then $G[C \cup \{x\}]$ is planar by the same argument as above and adding a vertex which has degree 1 or is isolated preserves planarity.

We established that attaching any pair of vertices from $U$ to any connected component of $G-U$ preserves its planarity.
It follows that the same holds if we consider attaching just a single vertex from  $U$.
Finally, we estimate the size of $U$.
The set $Y$ is a union of $m$ sets $Y_x$ and so $|Y| \le m\cdot k\cdot t$.
The set $Z$ is a union of ${\binom{X\cup Y}{2}}
\le (2|Y|)^2 \le 4 \cdot (m \cdot k\cdot t)^2$
many sets $Z_{\{x,y\}}$ and so
$|Z| \le 4 \cdot m^2\cdot k^3\cdot t^3$.
As $|X \cup Y| \le |Z|$, the upper bound follows.
\end{proof}

\subsection{Reducing the neighborhood of planar components}

{Using Lemma~\ref{lem:treewidth:planar-neighborhood-reduction} we can reduce to an instance which has a tidy planar modulator~$U$ such that each component of~$G-U$ remains planar when reinserting two vertices from~$U$. In this section we work towards obtaining even \emph{stronger} properties for components outside the modulator as described in Definition~\ref{def:prelim:strong}, so that each component outside the modulator induces a planar subgraph when taken together with its entire neighborhood. }


\bmp{The following lemma is inspired by existing tools to compute protrusion decompositions for problems on planar graphs (cf.~\cite[Lemma 15.13]{fomin2019kernelization}). It is used to either build an enhanced modulator which gives a bound on the size of the neighborhood of each remaining component, or identify a triple of vertices that allows an approximation-preserving graph reduction step by detecting three vertices which are part of many minor models of~$K_{3,3}$.}

\begin{lemma}\label{lem:treewidth:small-neighborhood}
There is a polynomial-time algorithm that, given a graph $G$, set $X \subseteq V(G)$ of size~$m$ such that $G-X$ is planar, a tree decomposition $(T,\chi)$ of $G-X$ of width $t-1$, and an integer $k>0$, returns one of the following:
\begin{enumerate}
    \item a set $Y \subseteq V(G) \setminus X$ of size at most $2^3\cdot m^3\cdot k \cdot t$, such that for each connected vertex set $C$ of $G-(X\cup Y)$ it holds that $|N_G(C) \cap X| \le 2$ and $|N_G(C) \cap Y| \le 2t$,
    \item a triple $\{v_1, v_2, v_3\} \subseteq X$ and a~collection of $k+3$ disjoint vertex subsets $(V_i)_{i=1}^{k+1}$ such that $V_i \subseteq V(G) \setminus X$, $G[V_i]$ is connected, and $E({V_i},v_j) \ne \emptyset$ for all $i \in [k+3],\, j \in [3]$.
\end{enumerate}
\end{lemma}
\begin{proof}
We apply Lemma~\ref{lem:treewidth:meta} for the graph $G-X$, parameter $k' = m^3\cdot(k+2)$ and family $\mathcal{F}$ of its vertex subsets defined
as follows: $C \in \mathcal{F}$ if $|N_G(C) \cap X| \le 2$.

In the first scenario,
the algorithm returns a set $S$
of at most $k'$ nodes from $V(T)$, so that
for each connected \bmp{component} $C$ of $G-(X\cup \bigcup_{t \in S} \chi(t))$ it holds that $|N_G(C) \cap X| \le 2$, \bmp{which implies the same statement for each connected vertex set~$C$}.
Let $M =\overline{\mathsf{LCA}}(S)$ and $Y = \bigcup_{t \in M} \chi(t))$.
By Lemma~\ref{lem:treewidth:lca}, point (b), $|M| \le 2|S| - 1$
and so $|Y| \le 2\cdot m^3\cdot(k+2)\cdot t$.
We upper bound $k+2$ with $3k \leq 2^2 k$ to get the formulated estimation. 
By Lemma~\ref{lem:treewidth:lca}, point (a), for each connected vertex set $C$ of $G- (X \cup Y)$ (which is contained in a union of bags from a connected subtree of $T-M$ by Observation~\ref{obs:treedec:subgraph:subtree}) the set $N_G(C) \cap Y$ is contained in at most 2 bags, hence $|N_G(C) \cap Y| \le 2t$.
We have $M \supseteq S$, so it also holds that $|N_G(C) \cap X| \le 2$.

In the second scenario, we obtain a~collection of $m^3\cdot(k+2) + 1$ disjoint vertex sets $V_i$ which induce connected subgraphs of $G$ and do not belong to $\mathcal{F}$.
Each $V_i$ contains some triple from $X$ in its neighborhood and,
since there are $\binom{m}{3} \le m^3$ distinct triples from $X$, one triple has to appear in \bmp{the common neighborhood} of at least $(k+3)$ sets $V_i$.
It can be identified in polynomial time by enumerating all the triples from~$X$.
\end{proof}

\bmp{Using the preceding lemma, the next algorithm makes progress towards obtaining a strong modulator by ensuring planarity for each connected component outside the modulator, even when taken together with its neighborhood.}

\begin{lemma}\label{lem:treewidth:strong-modulator}
There is a polynomial-time algorithm that, given a graph $G$, set $X \subseteq V(G)$ of size~$m$ such that $G-X$ is planar, a tree decomposition of $G-X$ of width $t-1$, and an integer $k>0$, returns one of the following:
\begin{enumerate}
    \item a set $Z \supseteq X$ of size at most $2^{13} \cdot m^6\cdot k^{10}\cdot t^{10}$, such that for each connected vertex set $C$ in $G-Z$ the graph $G[N_G(C)]$ is planar and $|N_G(C) \cap Z| \le 2t+2$, or 
    \item a set $X_0 \subseteq V(G)$ of size at most 3,
    such that $\mvp(G) \le k$ implies $\mvp(G-X_0) \le \mvp(G) - 1$.
\end{enumerate}
In particular, when $G$ is planar (so $\mvp(G) = 0$) then the algorithm must return answer (1).
\end{lemma}
\begin{proof}
We apply Lemma~\ref{lem:treewidth:planar-neighborhood-reduction} for the same arguments as in the statement of this lemma.
We either obtain a set $X_0 \subseteq V(G)$ of size at most 2, which satisfies condition (2) or a set $U \supseteq X$ of size at most $2^3\cdot m^2\cdot k^3\cdot t^3$, such that each connected vertex set $C \subseteq V(G) \setminus U$ remains planar after attaching any vertex or a pair of vertices from $U$.
In the first case we are done and
in the second case, we proceed with Lemma~\ref{lem:treewidth:small-neighborhood} by substituting $X \leftarrow U$. \bmp{We distinguish two cases, depending on the outcome of Lemma~\ref{lem:treewidth:small-neighborhood}.}

Suppose first that we have obtained a triple $\{v_1, v_2, v_3\} \subseteq U$ and a~collection of $k+3$ disjoint vertex subsets $(V_i)_{i=1}^{k+1}$ such that $V_i \subseteq V(G) \setminus U$, $G[V_i]$ is connected, and $E({V_i},v_j) \ne \emptyset$ for all $i \in [k+3],\, j \in [3]$.
Assume that $G$ admits a planar modulator $S$ of size $\ell \le k$.
Then for some $\{i_1, i_2, i_3\} \in [k+3]$ it holds that $V_{i_j} \cap S = \emptyset$ for $j \in [3]$.
If $\{v_1, v_2, v_3\} \cap S = \emptyset$, then $\{v_1, v_2, v_3, V_{i_1}, V_{i_2}, V_{i_3}\}$ would form branch sets of a~$K_{3,3}$ minor in $G-S$, which is not possible.
Hence, $\{v_1, v_2, v_3\} \cap S \ne \emptyset$
and $S \setminus \{v_1, v_2, v_3\}$ is a planar modulator in $G - \{v_1, v_2, v_3\}$ of size at most $\ell - 1$.
We can thus return $X_0 = \{v_1, v_2, v_3\}$ to satisfy condition (2).

Now suppose that the algorithm from Lemma~\ref{lem:treewidth:small-neighborhood} \bmp{terminated} with scenario (1).
We have obtained a~set $Y$ of size at most $2^3 \cdot |U|^3\cdot k \cdot t \le 2^{12} \cdot m^6\cdot k^{10}\cdot t^{10}$. 
\bmp{Let~$Z = U \cup Y$, which satisfies the claimed size bound since~$|U| \leq |Y|$. For each connected vertex set $C$ of $G-(U\cup Y)$ Lemma~\ref{lem:treewidth:small-neighborhood} guarantees that $|N_G(C) \cap U| \le 2$ and $|N_G(C) \cap Y| \le 2t$, so~$|N_G(C) \cap Z| \leq 2t+2$.}
The set $N_G[C] \setminus U$ is a connected vertex subset of $V(G) \setminus U$ and
by the properties of $U$ guaranteed by Lemma~\ref{lem:treewidth:planar-neighborhood-reduction}
for $X_0 = N_G(C) \cap U$,
this implies that $G[(N_G[C] \setminus U) \cup (N_G(C) \cap U)]$ is planar, and so
$G[N_G[C]]$ is planar. 
\end{proof}

\bmp{We are ready to state and prove the main result of this subsection, which encapsulates the previous reduction steps using the notion of strong modulators (Definition~\ref{def:prelim:strong}). In addition, it gives a bound on the number of connected components with at least three neighbors which occur outside the new modulator~$X'$, using the terminology of fragmentation (Definition~\ref{def:fragmentation}).}

\begin{proposition}
\label{lem:treewidth:final}
There is a polynomial-time algorithm that, given a graph $G$ with a tree decomposition of width $t-1$, and an integer $k>0$, either correctly concludes that $\mvp(G)>k$, or outputs a graph $G'$ with a $(2t+2)$-strong planar modulator $X'$, and an integer $k' \le k$, so that:
\begin{enumerate}
    \item if $\mvp(G) \le k$, then $\mvp(G') \le \mvp(G) - (k - k')$,
    \item the fragmentation of $X'$ in $G'$ is $\Oh(k^{49}t^{48})$,
    \item given a planar modulator $S'$ in $G'$, one can turn it into a planar modulator in $G$ of size at most $|S'| + 3\cdot (k - k')$ in polynomial time.
\end{enumerate}
\end{proposition}
\micr{new recursive proof}
\mic{
\begin{proof}
We begin with Lemma~\ref{lem:treewidth:modulator} to either find a planar modulator $X$ of size at most $\lambda = k\cdot t$ or conclude that $\mvp(G) > k$.
In the first case,
we  can easily turn the given tree decomposition of $G$ into a tree decomposition of $G-X$ of no greater width.
We keep the values of $t, \lambda$ intact within this proof, while $k$ is a free variable in the recursion.

We present the graph modification as a recursive process in which we are given a tuple $(G,X, T, \chi, k)$ where $X$ is a planar modulator in $G$ of size at most $\lambda$ and $(T,\chi)$ is a tree decomposition of $G-X$ of width at most $t$.
We either conclude that $(G,X,k)$ satisfies the claim or 
recurse on some tuple $(G', X',k-1, T', \chi')$ so that $|X'| \le |X|$.
In order to prove the forward safeness we will show that when $\mvp(G) \le k$ then $\mvp(G') \le \mvp(G) - 1$. 
To prove the backward safeness we show that when given a planar modulator $S'$ in $G'$ 
we can turn in into a planar modulator in $G$ of size at most $|S'|+3$.

Given a tuple $(G,X, T, \chi, k)$,
we apply \cref{lem:treewidth:strong-modulator}.
If we get the outcome (2), we obtain a set $X_0 \subseteq V(G)$ of size at most 3 such that $\mvp(G) \le k$ implies $\mvp(G-X_0) \le \mvp(G) - 1$: this yields the forward safeness.
The backward safeness follow from the fact that
when $S$ is a planar modulator in $G-X_0$ 
then $S \cup X_0$ is a planar modulator in $G$.
Since $G-(X\cup X_0)$ is a subgraph of $G-X$, we can turn $(T,\chi)$ into a tree decomposition of $G-(X\cup X_0)$ of no greater width.

If we get the outcome (1) from \cref{lem:treewidth:strong-modulator}, we obtain  a $(2t+2)$-strong planar modulator $Z \supseteq X$.
We have $|Z| \le 2^{13} \cdot |X|^6\cdot k^{10}\cdot t^{10} \le 2^{13} \cdot k^{16}\cdot t^{16}$.
We stop the recursion at this point.

In order to bound the fragmentation of $Z$, let us consider
another recursive process,
in which we process a triple $(G, Z, k)$ where $Z$ is a $(2t+2)$-strong planar modulator in $G$.
Whenever we recurse on $(G', Z', k-1)$ we demand $|Z'| \le |Z|$ and consider the same safeness conditions as before: (a) $\mvp(G) \le k$ implies $\mvp(G') \le \mvp(G) - 1$, and (b) when given a planar modulator $S'$ in $G'$ 
we can turn in into a planar modulator in $G$ of size at most $|S'|+3$.

Given a triple $(G, Z, k)$ we enumerate all size-3 subsets of $Z$ and look for a subset $Z_0 := \{v_1, v_2, v_3\} \subseteq Z$, such that for at least $k+3$ connected components $(C_i)_{i=1}^{k+3}$ in $\cc\cc(G-Z)$ we have $Z_0 \subseteq N_{G}(C_i)$.
Consider any planar modulator $S$ in $G$ of size $\ell \le k$.
Analogously as in Lemma~\ref{lem:treewidth:planar-neighborhood},
we observe that for some $\{i_1, i_2, i_3\} \in [k+3]$ it must hold that $C_{i_j} \cap S = \emptyset$ for {each} $j \in [3]$.
If $Z_0 \cap S = \emptyset$, then $\{v_1, v_2, v_3, C_{i_1}, C_{i_2}, C_{i_3}\}$ would form branch sets of a~$K_{3,3}$ minor in $G-S$, which is not possible.
Hence, $Z_0 \cap S \ne \emptyset$
and $S \setminus Z_0$ is a planar modulator in $G - Z_0$ of size at most $\ell - 1$.
In such a scenario, we recurse on $(G-Z_0, Z\setminus Z_0, k-1)$.
We have $\mvp(G)\le k$ implies $\mvp(G-Z-0)\le \mvp(G)-1$.
Again,
when $S$ is a planar modulator in $G-Z_0$ 
then $S \cup Z_0$ is a planar modulator in $G$.
Furthermore, $Z\setminus Z_0$ is a $(2t+2)$-strong planar modulator in $G-Z_0$.

When we can no longer find such a set $Z_0 \subseteq Z$, we stop the recursion and return $G$ and $X' = Z$.
By a counting argument, the number of connected components in $G - Z$ with at least 3 neighbors cannot be larger than $|Z|^3\cdot(k+2) = \Oh(k^{49}t^{48})$, which gives a bound on {the} fragmentation of $Z$.
\end{proof}}

\section{Reducing the radial diameter}
\label{sec:diameter}

From now on we \bmp{work with an extended modulator so} that $G\brb C$ is planar for each connected component $C$ of $G-X$.
We will consider a boundaried plane graph $H$ given by a plane embedding of $G\brb C$; we begin with an arbitrary embedding.
We can describe modifications to the graph $G$ by modifications to $H$ as every vertex or edge in $H$ represents a vertex or edge in $G$.
To make the arguments clear, we avoid modifying vertices from $\partial (H)$ or edges incident to $\partial (H)$, so no such modification to $H$ affects other parts of the graph $G$ outside $C$.

In order to replace $H$ with a smaller boundaried graph $H'$, which approximately preserves its structure, we are going to reduce the task to the case where the vertices of $\partial (H)$ are embedded in a convenient way.
In particular, we are interested in the case there  $\partial (H)$ is located on a single face of $H$.
As a step in this direction, we first reduce the {radial diameter} of $H$ to $k^{\Oh(1)}$.
One obstacle towards this is given by the existence of a large subset $X \subseteq V(H) \setminus \partial(H)$, with a small (constant) boundary, that separates some vertices from $\partial(H)$ in $H$.
We are going to compress such subsets using the technique of protrusion replacement.
This technique will also come in useful in the final part of the proof.

\subsection{Protrusion replacement}

An $r$-protrusion in a graph $G$ is a vertex set $X \subseteq V(G)$, such that $|\partial_G(X)| \le r$ and $\text{tw}(G[X]) \le r$.
\mic{For a vertex set $X \subseteq V(G)$ with $|\partial_G(X)| = r$, a bijection $\lambda \colon [r] \to \partial_G(X)$, and a labeled $r$-boundaried graph $H$, the act of replacing $X$ by $H$ through $\lambda$ means (1) creating labeled $r$-boundaried graphs $H^X = (G[X], \partial_G(X), \lambda)$, $H^{V(G) \setminus X} = (G[(V(G) \setminus X) \cup \partial_G(X)], \partial_G(X), \lambda)$, and (2)
replacing $G =H^{V(G) \setminus X} \oplus H^X $ by $G' = H^{V(G) \setminus X} \oplus H$.
The graph $H^X$ is introduced above only for the sake of giving a more detailed view on the process.
We say that $G'$ is obtained from $G$ by replacing a set $X \subseteq V(G)$ if there exist $H, \lambda$ so that $X$ is replaced  by $H$ through $\lambda$.}

The standard protrusion replacement technique~\cite{bodlaender2016meta,fomin2012planar} is aimed at reducing large $r$-protrusions, where $r = \Oh(1)$, while preserving the optimum value of the instance for the considered optimization problem.
Since we want to perform lifting of a potentially non-optimal solution, we use the lossless variant of the technique which provides us with such a mechanism.
The definition below can be instantiated for any
\textsc{min-CMSO} (or \textsc{max-CMSO})
 vertex subset problem $\Pi$.
In such a problem we are given a graph~$G$ 
and the objective is to find a set $S \subseteq V(G)$ minimizing $|S|$ such that the
predicate $P_\Pi(G, S)$, expressible in \bmp{Counting Monadic Second-Order} (\textsc{CMSO}) logic, is satisfied.
For our purposes it is only relevant that this definition captures the \planardel problem. 
\mic{We specify the labeling $\lambda$ directly in the definition whereas the formulation in \cite{fomin2020planar-arxiv} assumes an implicit labeling on the boundary of $X$.}

\begin{definition}[Lossless Protrusion Replacer \cite{fomin2012planar, fomin2020planar-arxiv}]
\label{def:protrusion:replacer}
A lossless protrusion replacer for a \textsc{min-CMSO}
 vertex subset problem $\Pi$ is a family of algorithms, with one algorithm for every
constant $r$. The $r$-th algorithm has the following specifications. There exists a constant $\gamma(r)$
such that given an instance $G$ and an $r$-protrusion $X$ in $G$ containing at least $\gamma(r)$ vertices, $|\partial_G(X)| = r' \le r$, the
algorithm runs in polynomial time and outputs an instance $G'$ with the following properties.
\begin{itemize}
    \item $G'$ is obtained from $G$ by replacing $(G[X], \partial_G(X))$ by a labeled $r'$-boundaried graph $H$ with less than $\gamma(r)$ vertices, through some labeling $\lambda \colon [r'] \to \partial_G(X)$, and thus $V(G') < V(G)$.
    \item $\mathsf{OPT}(G') \le \mathsf{OPT}(G)$.
    \item There is an algorithm that runs in polynomial time and, given a feasible solution $S'$
to $G'$,
outputs a~set $X^* \subseteq  X$ such that $S = (S' \setminus V(H)) \cup X^*$ is a feasible solution to $G$ and
$|S| \le |S'| + \mathsf{OPT}(G) - \mathsf{OPT}(G')$.
\end{itemize}
\end{definition}

In the original statement, the lossless protrusion replacer is defined to run in time $O(|X|)$.
This is important if we aim at optimizing the running time of the algorithm but requires some care with specifying the data structures for maintaining the decomposition into $r$-protrusion.
For the sake of keeping the presentation simple, we have chosen a weaker specification with just a~polynomial running time guarantee.

\begin{lemma}\label{lem:protrusion:planar-deletion}
The \planardel problem admits a lossless protrusion replacer.
Moreover, the choice of the replacement $H$ and the transposition value $\mvp(G) - \mvp(G')$ depends only on the labeled boundaried graph $(G[X], \partial_G(X), \lambda)$.
\end{lemma}
\begin{proof}
The \planardel problem is a special case of \textsc{$\mathcal{F}$-Minor Deletion} for the family of forbidden minors $\mathcal{F} = \{K_5, K_{3,3}\}$.
Fomin et al. \cite{fomin2012planar, fomin2020planar-arxiv} have shown that every \textsc{min-CMSO} vertex subset problem $\Pi$ that is strongly
monotone admits a lossless protrusion replacer.
Without going into details, both conditions hold for  \textsc{$\mathcal{F}$-Minor Deletion} for every family $\mathcal{F}$ composed of connected graphs~\cite[Lem. 8.4]{bodlaender2016meta}.
We remark that the proof of the strong monotonicity property in \cite{bodlaender2016meta} is formulated for the case where the family $\mathcal{F}$ contains at least one planar graph, but this condition is not used in this particular proof.
Finally, by examining the proof of \cite[Thm. 9]{fomin2020planar-arxiv} it follows that the choice of the replacement $H$ and the transposition value $\mathsf{OPT}(G) - \mathsf{OPT}(G') = \mvp(G) - \mvp(G')$ depends only on the labeled boundaried graph $(G[X], \partial_G(X), \lambda)$ and is indifferent to the structure of $G - X$.
\end{proof}

In our applications, we need to replace a vertex set $X \subseteq V(G)$, such that $G[X]$ is planar \bmp{and} $|\partial_G(X)| \le r$ for constant $r$, but the treewidth of $G[X]$ is not necessarily bounded.
It is easy to show that for any minimum planar modulator $S$ it holds that $|S \cap X| \le r$, as otherwise the solution $(S \setminus X) \cup \partial_G(X)$ would be smaller.
By an argument based on grid minors, \bmp{if} $\tw(G[X])$ is larger than some function of $r$, then $X$ contains a vertex irrelevant for minimum solutions.
As we work with solutions which are not necessarily optimal, we need an additional argument that allows us to assume that  $|S \cap X| \le r$.

\begin{definition}
We say that a planar modulator $S$ in $G$ is $r$-locally optimal
if there are no sets $S' \subseteq S$ and $S^* \subseteq V(G) \setminus S$ such that $|S'| > r$, $|S^*| \le r$, and $(S \setminus S') \cup S^*$ is also a planar modulator.  
\end{definition}

\begin{lemma}\label{lem:protrusion:locally-optimal-find}
Given an $n$-vertex graph $G$ with a planar modulator $S \subseteq V(G)$, and an integer $r$, one can find, in time $n^{\Oh(r)}$, an $r$-locally optimal planar modulator $\widehat S$ of size at most $|S|$.
\end{lemma}
\begin{proof}
Suppose that $S$ is not $r$-locally optimal,
so there are sets $S' \subseteq S$ and $S^* \subseteq V(G) \setminus S$ such that $|S'| > r$, $|S^*| \le r$, and $(S \setminus S') \cup S^*$ is also a planar modulator.
Let $S''$ be an arbitrary subset of $S'$ of size $r+1$.
Then $(S \setminus S'') \cup S^*$ is a superset \bmp{of} $(S \setminus S') \cup S^*$, so it is also a planar modulator, and it is smaller than $S$.
The algorithm checks all $\Oh(n^{r+1})$ candidates for $S''$, and all $\Oh(rn^r)$ candidates for $S^*$.
If for any pair $(S'', S^*)$, the set $(S \setminus S'') \cup S^*$ is also a planar modulator (this can be checked in time $\Oh(n+m)$~\cite{HopcroftT74}), 
we set $S \leftarrow (S \setminus S'') \cup S^*$ and repeat the algorithm.
In every step we decrease the size of $S$,
so after at most $n$ steps the process must terminate.
Since at this moment we cannot find any pair $(S'', S^*)$ that would provide further refinement, the obtained planar modulator must be $r$-locally optimal.
\end{proof}

Consider a set $X = G\brb A$ for some $A$, for which $|N_G(A)| \le r$.
The lifting algorithm can use \cref{lem:protrusion:locally-optimal-find} to find a solution $S$ satisfying  $|S \cap X| \le r$.
Therefore any vertex $v \in X$ which can be surrounded by $r+3$ connected $(v,\partial_G(X))$ separators is irrelevant for $S$.
We follow the standard argument based on a grid contraction to find such a vertex whenever $\tw(G[X])$ is large. 

\begin{lemma}\label{lem:protrusion:reduce-treewidth}
Let $G$ be a graph, $A \subseteq V(G)$, and $G\brb A$ be a $(\le r)$-boundaried connected planar graph.
If $\text{tw}(G \brb A) > 27(2r+7)(r+1)$, then there exists a
vertex $v \in A$ and a sequence of $r+3$ connected pairwise-disjoint $(v,N_G(A))$-separators in $G\brb A$.
Such \bmp{a} vertex $v$ can be found in polynomial time, when given $G$ and $A$.
Furthermore, given $G,A,v$, and a planar modulator $S$ in $G-v$, one can find, in  time $n^{\Oh(r)}$, a planar modulator in $G$ of size at most $|S|$. 
\end{lemma}
\begin{proof}
We apply~\cref{lem:grid:contraction-model} for $t = (2r+7)(r+1)$.
As \bmp{the} treewidth of $G\brb A$ is larger than $27\cdot t$,
we can find, in polynomial time,
a contraction model of $\Gamma_t$ in $G\brb A$.
We have $(r+1)^2 > r$ trivially, so this contraction model
contains a $(2r+7) \times (2r+7)$-subgrid without vertices from $N_G(A)$.
A $(2r+7) \times (2r+7)$-subgrid contains $r+3$  nested connected separators $S'_1, \dots, S'_{r+3}$ between the central branch set and the rest of the grid.
Let $v$ be any vertex from the central branch set of this subgrid
and $S_i$ be the union of vertices contracted into $S'_i$ for $i \in [r+3]$.
By \cref{lem:grid:contraction-separator},
the sets $S_1, \dots, S_{r+3}$ are nested connected  $(v,N_G(A))$-separators.
Therefore each set $S_i$, $i \in [r+3]$, is $v$-planarizing.

Suppose now that one is given a planar modulator $S$ in $G-v$.
By \cref{lem:protrusion:locally-optimal-find}, one can find, in time $n^{\Oh(r)}$, an $r$-locally optimal planar modulator $S'$ in $G-v$ of size at most $|S|$.
Observe that $|S' \cap A| \le r$ because otherwise $(S' \setminus A) \cup N_G(A)$ would be a smaller planar modulator, contradicting $S'$ being $r$-locally optimal.
Therefore there exists distinct indices $i_1, i_2, i_3 \in[r+3]$
such that $S_{i_1}, S_{i_2}, S_{i_3}$ are disjoint from $S'$.
For each $j \in [3]$ we have that $R_{G-S'}[v, S_{i_j}] \subseteq A$ induces a planar subgraph of $G-S'$.
By the assumption $G-S'-v$ is planar, so together, by \cref{lem:prelim:criterion:new}, $G-S'$ is planar.
We have thus shown how to efficiently turn a planar modulator in $G-v$ into a planar modulator in $G$.
\end{proof}

Once we are able to reduce the treewidth of a planar subgraph with a small (constant) boundary, we can use protrusion replacement to reduce its size.
The setting in which we work implies some convenient corollaries about the result of the replacement act, of which we make note in the following lemma.

\begin{lemma}\label{lem:protrusion:planar-replacement}
Let $\gamma$ be the function from \cref{def:protrusion:replacer} for the \planardel problem, $r$ be a constant and $r' = 27(2r+7)(r+1)$.
There is a polynomial-time algorithm that, given an $n$-vertex graph $G$, a vertex set $A \subseteq V(G)$, such that $G\brb A$ is planar and connected, $N_G(A) = r$ and $|A| \ge \gamma(r')$, outputs a graph $G'$ so the following hold.
\begin{enumerate}
    \item $G'$ is obtained from $G$ by replacing $N_G[A]$ with some labeled $r$-boundaried graph $H$ through some bijection $\lambda \colon [r] \to N_G(A)$.
    \item $H$
    has less than $\gamma(r')$ vertices.
    \item $\mvp(G') = \mvp(G)$.
     \item For any vertex set $X \subseteq V(G) \setminus N_G[A]$, if $G - X$ is planar and connected then $G' - X$ is planar and connected.\label{lem:protrusion:planar-replacement:planar-subgraph}
    \item Given $G, G'$ and a planar modulator $S'$ in $G'$ one can, in time $n^{\Oh(r)}$, turn it into a planar modulator in $G$ of size at most $|S'|$.
\end{enumerate}
\end{lemma}
\begin{proof}
First, we reduce the treewidth of $G\brb A$ to be at most $r'$ using
\cref{lem:protrusion:reduce-treewidth}.
While $\tw(G\brb A) > r'$, we can find a vertex $v \in A$ such that
any planar modulator $S$ in $G-v$ can be turned, in time $n^{\Oh(r)}$, into a planar modulator in $G$ of size at most $|S|$.
Hence, removing $v$ is safe and we can continue the process on the graph $G-v$.
Furthermore, there is at least one \mic{connected} restricted $(v,N_G(A))$-separator in $G\brb A$ so removal of $v$ does not disconnect any vertices from $N_G(A)$ within $G\brb A$.
If some vertices get disconnected from $N_G(A)$ then they form a planar connected component in $G$, which can be discarded.
Therefore, we maintain connectivity of the graph $G\brb A$.
Let $G_0$ be obtained from $G$ after the exhaustive application of \cref{lem:protrusion:reduce-treewidth} and $A_0 \subseteq A$ be obtained from $A$ by removing the irrelevant vertices.
Then the set $N_{G_0}[A_0]$ forms an $r'$-protrusion. 
We can treat $G_0$ as obtained from $G$ by replacing $A$.
Note that $\mvp(G_0) = \mvp(G)$ because the lifting algorithm can lift any solution losslessly.
If $|N_{G_0}[A_0]| < \gamma(r')$, then we are done.

Suppose that $|N_{G_0}[A_0]| \ge \gamma(r')$.
By \cref{lem:protrusion:planar-deletion},
the \planardel problem admits a lossless protrusion replacer,
so we can replace $N_{G_0}[A_0]$ with a labeled $r$-boundaried graph $H$, with at most $\gamma(r')$ vertices, through some labeling $\lambda$,
obtaining a graph $G'$.
From the definition of a lossless protrusion replacer
we have that $\mvp(G') \le \mvp(G_0) = \mvp(G)$ and any planar modulator $S'$ in $G'$ can be turned, in polynomial time, into a planar modulator $S_0$ in $G_0$ of size at most $|S'| + \mvp(G_0) - \mvp(G')$.
Thanks to \cref{lem:protrusion:reduce-treewidth}, we can compute a planar modulator $S$ in $G$ of size at most $|S_0|$, in time $n^{\Oh(r)}$.

We argue that actually $\mvp(G') = \mvp(G_0)$.
From \cref{lem:protrusion:planar-deletion} we know that the choice of the replacement $H$ and the transposition value $\mvp(G_0) - \mvp(G')$ depends only on the boundaried graph $G_0^{A_0,\lambda} = (G_0[N_{G_0}[A_0]], N_{G_0}(A_0), \lambda)$.
Let $I_r$ be a labeled $r$-boundaried graph with only the boundary vertices and no edges.
Consider graphs $J_0, J'$ obtained by gluing $I_r$ to respectively $G_0^{A_0,\lambda}$ and $H$ (equivalently, one can regard this operation as taking the underlying non-boundaried graphs).
The graph $J_0$ is planar so $\mvp(J_0) = 0$.
Because $\mvp(J') \le \mvp(J_0)$ we get that $\mvp(J') = \mvp(J_0)$.
Since the transposition value depends only on $G_0^{A_0,\lambda}$ we infer that $\mvp(G_0) - \mvp(G') = \mvp(J_0) - \mvp(J') = 0$.
If the graph $H$ contains a connected component disconnected from $\partial(H)$ then the argument above implies that this component must be planar.
Therefore if such a component appears, it can be safely removed as a planar component in $G_0$.

Finally \bmp{we} show that for any vertex set $X \subseteq V(G) \setminus N_G[A]$, if $G - X$ is planar and connected then $G' - X$ is planar and connected.
The first modification step does not affect these conditions so $G_0 - X$ is planar and connected.
Let $J_0^X$
be the labeled boundaried graph for which it holds $G_0 - X =  J_0^X \oplus G_0^{A_0,\lambda}$ and $G' - X = J_0^X \oplus H$.
Similarly as before, we use the fact that the relation between $G_0^{A_0,\lambda}$ and $H$ is indifferent to the boundaried graph being glued to them.
Therefore $\mvp(G'-X) = \mvp(G_0-X) = 0$, so $G'-X$ is 
planar.
Now suppose that $G' - X$ is not connected.
Because we have removed all the connected components of $H$ with empty intersection with $\partial(H)$, this means that there exist $u, v \in N_{G_0}(A_0) = \partial(H)$ which are connected in $G_0\brb {A_0}$ but not in $H$.
Let $\widehat{K_5}$ denote the graph $K_5$ with one edge removed.
Consider graphs $R_0, R'$ obtained by gluing $\widehat{K_5}$ to respectively $G_0^{A_0,\lambda}$ and $H$ so that the vertices $u,v$ are glued to the non-adjacent pair in $\widehat{K_5}$.
Since there is a $(u,v)$-path in $G\brb A$, the graph $R_0$ contains $K_5$ as a minor.
On the other \bmp{hand}, $u,v$ belong to different connected components of $H$, so every biconnected component of $R'$ is planar and so $R'$ is planar.
We get that $\mvp(R_0) > \mvp(R') = 0$, which contradicts the assumption that such vertices $u,v$ exist.
The claim follows.
\end{proof}

\subsection{Outerplanar decomposition}

In order to conveniently analyze plane graphs with large radial diameter
we take advantage of the concept of a decomposition into outerplanarity layers used by Jansen et al.~\cite{JansenPvL19} in their work on kernelization for \textsc{Multiway cut} on planar graphs.
We refer to the definitions from the full version~\cite{JansenPvL19-arxiv} of this article.


\newcommand{\forest}{\mathbb{F}}
\newcommand{\opindex}{\mathrm{\textsc{idx}}}
\renewcommand{\int}{\mathrm{\textsc{int}}}

\begin{definition}[{\cite[Definition 5.1]{JansenPvL19-arxiv}}]
\label{def:diameter:outerplanarity-layers}
Let~$G$ be a plane graph. The outerplanarity layers of~$G$ form a partition of~$V(G)$ into~$L_1, \ldots, L_m$ defined recursively. The vertices incident with the outer face of~$G$ belong to layer~$L_1$. If~$v \in V(G)$ lies on the outer face of the plane subgraph obtained from~$G$ by removing~$L_1, \ldots, L_i$, then~$v$ belongs to layer~$L_{i+1}$. For a vertex~$v \in V(G)$, the unique index~$i \in [m]$ for which~$v \in L_i$ is the \emph{outerplanarity index} of~$v$ and denoted~$\opindex_G(v)$.
\end{definition}

\bmp{A plane graph is $k$-outerplanar if this process partitions it into exactly~$k$ non-empty layers. A graph is $k$-outerplanar if it admits a $k$-outerplanar embedding.}

\begin{definition}[{\cite[Definition 5.2]{JansenPvL19-arxiv}}]
For a plane graph~$G$, let~$T$ be the simple graph obtained by simultaneously contracting all edges~$uv$ whose endpoints belong to the same outerplanarity layer, discarding loops and parallel edges. For a node~$u \in V(T)$, let~$\kappa(u) \subseteq V(G)$ denote the vertex set of~$G$ whose contraction resulted in~$u$. 
Let~$\opindex_G(u)$ denote the outerplanarity index that is shared by all nodes in~$\kappa(u)$.
\end{definition}

We make note of several properties of this definition.

\begin{lemma}[{\cite[Lemma 5.3]{JansenPvL19-arxiv}}] \label{lem:diamater:outerplanar-properties}
For a connected plane graph~$G$ and its outerplanar decomposition $(T,\kappa)$ based on the outerplanarity layers~$L_1, \ldots, L_m$, the following holds:
\begin{enumerate}
	\item $T$ is a tree.\label{prop:tree}
	\item For any node~$u \in V(T)$, the graph~$G[\kappa(u)]$ is connected.\label{prop:connected} 
	\item If~$x,y$ are distinct vertices of~$T$, then for any internal node~$z$ of the unique $(x,y)$-path in~$T$, the set~$\kappa(z)$ is a restricted $(\kappa(x), \kappa(y))$-separator in~$G$.\label{prop:cut}
	\item There is a unique node~$x_1 \in V(T)$ such that~$\kappa(x_1) = L_1$. \label{prop:root:layer}
	\item Root~$T$ at vertex~$x_1$. If~$u \in V(T)$ is a child of~$p \in V(T)$, then~$\opindex_G(u) = 1+\opindex_G(p)$.\label{prop:idx:child}
\end{enumerate}
\end{lemma}



We call the pair $(T,\kappa)$, where $T$ is the tree rooted at $x_1$ accordingly to \cref{lem:diamater:outerplanar-properties}, the \emph{outerplanar decomposition} of $G$. 
Note that it is uniquely defined for a given \bmp{connected} plane graph $G$.
Let $\delta_T$ denote the shortest path distance on the tree $T$.
For distinct $x,y \in V(T)$ we define $T^P_{x,y}$ to be the unique $(x,y)$-path in $T$ and for non-adjacent  $x,y \in V(T)$ we define $T_{x,y}$ to be the connected component of $T-\{x,y\}$ containing the internal vertices of $T^P_{x,y}$.
\mic{For a subtree $T'$ of $T$ we use the notation $\kappa(T')$ to denote the union of $\{\kappa(t) \mid t \in V(T')\}$.}
Note that $N_G(\kappa(T_{x,y})) \subseteq \kappa(x) \cup \kappa(y)$.

\begin{observation}[{\cite[Observation 5.4]{JansenPvL19-arxiv}}]
\label{lem:diameter:outerplanar-faces}
Let $G$ be a connected plane graph and $(T,\kappa)$ be its outerplanar decomposition.
If $u, v \in V(G)$ share a face, then either there is  $x \in V(T)$ such that $u, v \in \kappa(x)$
or there are $x, y \in V(T)$
which are adjacent in $T$ and $u \in \kappa(x), v \in \kappa(y)$.
Conversely, if $x$ is a child of~$y$ in~$T$ 
and $v \in \kappa(x)$, then there exists $u \in \kappa(y)$, so that $u,v$ share a face.
\end{observation}

This observation allows us to make a link between the depth of the outerplanar decomposition and the radial diameter of the graph.
\mic{We define the depth of a rooted tree as the longest root-leaf distance plus 1 (so a singleton tree has depth 1).}

\begin{lemma}\label{lem:diameter:outerplanar-diameter}
Let $G$ be a connected plane graph and $(T,\kappa)$ be its outerplanar decomposition.
If the maximal depth in $T$ is $d$, then the radial diameter of $G$ is at most $2d-1$.
\end{lemma}
\begin{proof}
By \cref{lem:diameter:outerplanar-faces} any vertex in $L_{i+1}$ shares a face with a vertex from $L_i$.
Let $u, v \in V(G)$ be any pair of vertices.
We have $\opindex_G(u), \opindex_G(v) \le d$ so they are at radial distance at most $d-1$ from some vertices from $L_1$ which are at radial distance 1, hence $d_G(u,v) \le 2d - 1$.
\end{proof}

Since $T$ is a contraction of $G$, when a set $S \subseteq V(G)$ separates $\kappa(x)$ from $\kappa(y)$, it also separates $\kappa(x)$ from the image of the entire subtree behind $y$.
We make note of this observation for the case where we separate  $\kappa(x)$ from two directions at once.

\begin{observation}
\label{lem:diameter:outerplanar-separator}
Let $G$ be a connected plane graph and $(T,\kappa)$ be its outerplanar decomposition.
Let $x, y, z \in V(T)$ be distinct pairwise non-adjacent vertices \bmp{such that} $z$ lies on the unique $(x,y)$-path in $T$.
Suppose that $S \subseteq V(G)$ is a restricted $(\kappa(z), \kappa(x) \cup \kappa(y))$-separator.
Then the connected component of $G - S$ containing $\kappa(z)$ is contained in $\kappa(T_{x,y})$.
\end{observation}

\mic{We also state 
a strengthening of \cref{lem:diamater:outerplanar-properties}(\ref{prop:cut}) which says that the separators within outerplanarity layers are in fact cycles. }

\begin{observation}\label{lem:diameter:outerplanar-cycle-separator}
Let $G$ be a connected plane graph, $(T,\kappa)$ be its outerplanar decomposition, and $x,y,z \in V(T)$ be distinct.
If $z$ lies on the unique $(x,y)$-path in $T$, then there is \bmp{a} cycle $C$ \bmp{in~$G$} such that $V(C) \subseteq \kappa(z)$ and one of $\kappa(x), \kappa(y)$ lies entirely in the interior of $C$ and the other one in the exterior of $C$.
\end{observation}

\subsection{Outerplanarity layers and vertex-disjoint paths}

We will be working with an outerplanar decomposition $(T,\kappa)$ of a plane graph $G$ and families of internally vertex-disjoint paths connecting two outerplanarity layers, so that one is drawn inside the other one.
First we show that paths in such a family can be assumed to cross each intermediate outerplanarity layer only once.

\begin{lemma}\label{lem:diameter:path-refinement}
Let $G$ be a connected plane graph, $L_1$ be the set of vertices lying on the outer face of $G$, and $F \subseteq V(G) \setminus L_1$ induce a connected subgraph of $G$.
Suppose that $u,v \in L_1$ are non-adjacent and connected by a path $P$ with all its internal vertices disjoint from $(L_1 \cup F)$.
Then there exists a $(u,v)$-path $P'$ such that $V(P') \subseteq L_1$ and the internal vertices of $P'$ belong
to a connected component of $G - V(P)$ which is disjoint from $F$. 
\end{lemma}
\begin{proof}
As $u,v$ are non-adjacent, the interior of $P$ is non-empty.
Since $G[L_1]$ is connected, there is a $(u,v)$-path within $L_1$, which is internally disjoint from $P$.
Therefore $u,v$ belong to a common biconnected component $B$ of $G$ and the subgraph $B$ contains a cycle $C_1$ such that $V(C_1) = L_1 \cap V(B)$.
Removing non-adjacent $u,v$ from $C_1$ separates $V(C_1)$ into two non-empty sets $S^1, S^2$, which end up in different connected components of $G - V(P)$ \bmp{since~$u$ and~$v$ lie on the outer face.}
Since $G[F]$ is connected, it is fully contained in one such component.
We can thus use either $S^1$ or $S^2$ to make the desired path $P'$.
\end{proof}

For a path $P$ and a vertex set $X$, a segment of $P$ in $X$ is an inclusion-wise maximal subpath $P'$ of $P$ such that $V(P') \subseteq X$.
We define $S(P,X)$ to be the set of segments of $P$ in $X$.

\begin{lemma}\label{lem:diameter:uncrossing}
Consider a connected plane graph $G$,
its outerplanar decomposition $(T,\kappa)$,
and non-adjacent $x,y \in V(T)$ such that $x$ is an ancestor of $y$.
Suppose we are given a family of internally vertex-disjoint $(\kappa(x),\kappa(y))$-paths $\mathcal{P}$.
Then we can compute, in polynomial time, a family $\mathcal{P}'$ of  internally vertex-disjoint $(\kappa(x),\kappa(y))$-paths such that $|\mathcal{P}'| = |\mathcal{P}|$ and for each $P \in \mathcal{P}'$ and $z \in T^P_{x,y}$ it holds that $|S(P,\kappa(z))| = 1$.
\end{lemma}
\begin{proof}
\bmp{By replacing each~$P \in \mathcal{P}$ by the $(\kappa(x),\kappa(y))$-subpath between the last visit to~$\kappa(x)$ and the first visit to~$\kappa(y)$,} we can assume that the internal vertices of the paths in $\mathcal{P}$ belong to $\kappa(T_{x,y})$.
For a family of internally vertex-disjoint $\kappa(x),\kappa(y)$-paths $\mathcal{P}$ we define the measure $\mu(\mathcal{P}) = \sum_{P \in \mathcal{P}} \sum_{z \in T^P_{x,y}} |S(P, \kappa(z))|$.
While $\mathcal{P}$ does not satisfy the claim, we are going to refine some path from $\mathcal{P}$ to decrease the measure $\mu(\mathcal{P})$.

Let $x' \in V(T)$
be the vertex on $T^P_{x,y}$ which is closest to $x$ among \bmp{those} for which there \bmp{exists} a path $P \in \mathcal{P}$ so that $|S(P,\kappa(x'))| > 1$.
Note that $x'$ is an ancestor of $y$.
Let us consider a plane graph $G'$ obtained by removing from $G$ all the vertices from $\kappa(z)$ over all $z \in V(T)$ outside the subtree of $T$ rooted at $x'$.
Note that $G'$ is connected and $\kappa(x')$ is the set of vertices lying on the outer face of $G'$.
\mic{By the choice of $x'$, for each $P \in \mathcal{P}$ there is exactly one segment of $P$ in $V(G')$: otherwise some path would have multiple segments in $\kappa(x'')$, for $x''$ being a parent of $x'$, which is a superset of $N_G(V(G'))$.}
Let $P'$ denote this segment of $P$.

Let $P \in \mathcal{P}$ be \bmp{a} path for which $|S(P,\kappa(x')| > 1$.
Then there \bmp{exist} vertices $u,v \in V(P') \cap \kappa(x')$ and a $(u,v)$-subpath $P_{u,v}$ of $P'$ whose internal vertices do not belong to $\kappa(x')$.
If $u,v$ are adjacent, then we can shortcut $P$ by replacing  $P_{u,v}$ with the edge $(u,v)$ thus decreasing $\mu(\mathcal{P})$.
Suppose that $u,v$ are non-adjacent.
We apply \cref{lem:diameter:path-refinement} to the graph $G'$ with the outer face lying on $\kappa(x')$, vertices $u,v$, path $P_{u,v}$, and $F = \kappa(y)$, which induces a connected subgraph of $G'$.
We obtain a $(u,v)$-path $P'_{u,v}$ such that $V(P'_{u,v}) \subseteq \kappa(x')$ and the internal vertices of $P'_{u,v}$ belong
to a connected component of $G' - V(P_{u,v})$ which is disjoint from $\kappa(y)$.
Since all the other paths from $\{P' \mid P \in \mathcal{P}\}$ connect $\kappa(x')$ to $\kappa(y)$ in $G'$, they cannot have any vertices in this component of $G' - V(P_{u,v})$.
Therefore by replacing $P_{u,v}$ with $P'_{u,v}$ in $P$
we obtain a new family of internally vertex-disjoint $(\kappa(x),\kappa(y))$-paths of the same size \bmp{for which the measure $\mu$ is} lower than  $\mu(\mathcal{P})$.

We apply this refinement exhaustively as long as \bmp{there exist} $P \in \mathcal{P}$ and $z \in T^P_{x,y}$ with $|S(P,\kappa(z))| > 1$.
It is clear that the refinement step can be performed in polynomial time.
At each step the value of $\mu$ decreases so at some point the process terminates and gives the desired path family.
\end{proof}

\mic{It is important that the sets considered in \cref{lem:diameter:uncrossing} are given by the consecutive outerplanarity layers.
For other choice of connected separators, such a claim might not be true.}

By inspecting the outerplanar decomposition $(T,\kappa)$ of an $r$-boundaried plane graph $H$, we can 
find the set of nodes in $T$ whose images cover the boundary $\partial(H)$, and then mark the LCA closure of this set (recall \cref{def:treewidth:lca}), thus marking at most $2r$ nodes.
We consider the nodes which were not marked and we want to compress the graph whenever there is a long path of non-marked nodes.
This translates to a situation
where one group $T_1$ of relevant vertices (that is, located in a marked bag) lies on one outerplanarity layer, there is a deep well of nested outerplanarity layers, and another group $T_2$ of relevant vertices is located inside the well.
We want to say that only \bmp{the} $\Oh(k)$ layers \bmp{closest} to $T_1$ and $\Oh(k)$ layers \bmp{closest} to $T_2$ are important and the middle layers can be compressed.

Let $S$ be some solution and
suppose that there at least 4 internally vertex-disjoint $(T_1,T_2)$-paths in $H-S$.
We need an argument that in this case, for any vertex~$v$ in the middle layers that does not lie on any of these four paths, we can draw two vertex-disjoint cycles around~$v$ that separate~$v$ from all the relevant vertices, and in particular from~$\partial(H)$.
We begin with a topological argument that such cycles exist.
We do not rely on \cref{lem:diameter:uncrossing} and give the argument in a general form.

\begin{lemma}\label{lem:diameter:topological}
Let $G$ be a plane graph, $C_1,\dots,C_6$ be \bmp{vertex-}disjoint cycles in $G$, so that $C_{i+1}$ is drawn inside $C_i$ for $i\in [5]$, and let \bmp{$P_A,P_B,P_C,P_D$} be vertex-disjoint $(C_1,C_6)$-paths.
Suppose that a vertex $v \in V(G)$ \bmp{with $v \not\in \bigcup_{i \in [4]} V(P_i)$} lies in the intersection of the proper interior of $C_3$ and the proper exterior of $C_4$.
Then there exist two \bmp{vertex-}disjoint cycles in $G$ which are nested with respect to $v$ and separate $v$ from \bmp{$V(C_1) \cup V(C_6)$}.
\end{lemma}
\begin{proof}
Without loss of generality, by shortcutting the paths if needed, we may assume that the paths~$P_A,\ldots,P_D$ do not contain any vertex from~$V(C_1) \cup V(C_6)$ in their interior.

For subsets of indices~$I \subseteq [6]$ and~$J \subseteq \{A,B,C,D\}$, let~$G^J_I$ be the plane subgraph of~$G$ obtained by restricting the drawing of~$G$ to only the vertices and edges which lie on a cycle~$C_i$ with~$i \in I$ or a path~$P_X$ with~$X \in J$. For ease of presentation, we abbreviate~$\{A,B,C,D\}$ by~$[D]$.

Since~$v$ does not lie on any of the four paths, while it lies in the intersection of the proper interior of~$C_3$ and proper exterior of~$C_4$, it follows that the image of~$v$ does not intersect the drawing of~$G_{3,4}^{[D]}$. Let~$F_{3,4}^{[D]}$ denote the face of~$G_{3,4}^{[D]}$ containing the image of~$v$. The face~$F_{3,4}^{[D]}$ is bounded since~$G_{3,4}^{[D]}$ contains~$C_3$ which has~$v$ in its proper interior, and due to the nesting of the cycles no vertex on~$F_{3,4}^{[D]}$ belongs to~$C_5$ or~$C_6$. Similarly, since~$v$ lies in the proper exterior of~$C_4$, face~$F_{3,4}^{[D]}$ does not contain any vertex of~$C_1$ or~$C_2$. Since~$G_{3,4}^{[D]}$ is a connected graph that does not have any articulation points, the boundary of face~$F_{3,4}^{[D]}$ is a simple cycle~$C'$ which contains~$v$ in its proper interior. As observed above,~$C'$ is disjoint from~$C_1 \cup C_2 \cup C_5 \cup C_6$. Note that since~$G_{3,4}^{[D]}$ contains cycles~$C_3$ and~$C_4$, the nesting of the cycles and placement of~$v$ in the interior of~$C_3$ and exterior of~$C_4$ ensures that~$V(C')$ separates~$v$ from~$V(C_1) \cup V(C_6)$.

\begin{figure}
\centering
\includegraphics{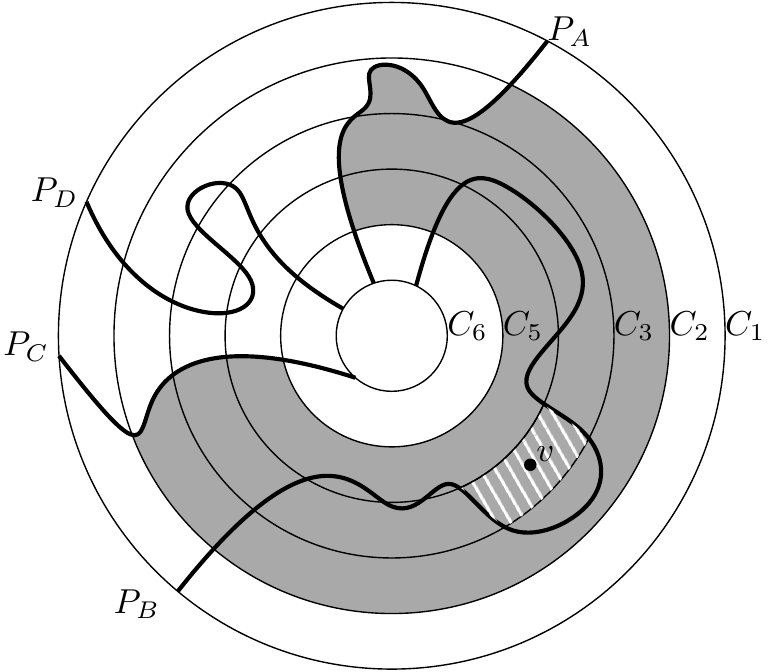}
\caption{Schematic representation of the graph~$G^{[D]}_{[6]}$. In this example, the face of~$G_{3,4}^{[D]}$ containing~$v$ is highlighted with dashed white lines. Its boundary cycle forms~$C'$ when applying the lemma for~$v$. It contains segments of~$C_3$,~$C_4$, and the single path~$P_B$. Additionally, the face of~$G_{2,5}^{[A,C]}$ containing~$v$ is shaded. It is bounded by segments of~$C_2$,~$C_5$, and the paths~$P_A$ and~$P_C$. Its boundary cycle forms~$C''$ when applying the lemma for~$v$.} \label{fig:topological}
\end{figure}

\mic{From the observations above, we get that $F_{3,4}^{[D]}$ is also a face in $G_{[6]}^{[D]}$.}
{Suppose that the boundary of~$F_{3,4}^{[D]}$ intersects at least three paths $P_X, P_Y, P_Z \in \{P_A, \ldots, P_D\}$.
Note that one can always add a vertex to a drawing inside some face and make it adjacent to all vertices on the boundary of this face.
Consider a plane graph $G'$ obtained from $G_{[6]}^{[D]}$ 
by inserting three new vertices $v_1, v_6, v_F$ adjacent respectively to all vertices of $C_1$, $C_6$ and $F_{3,4}^{[D]}$.
The sets $V(P_X), V(P_Y), V(P_Z)$, $\{v_1\}$, $\{v_6\}$, $\{v_F\}$ would form a minor model of $K_{3,3}$ which gives a contradiction.
}
Therefore the boundary of~$F_{3,4}^{[D]}$ intersects at most two paths~$P_X, P_Y \in \{P_A, \ldots, P_D\}$; see Figure~\ref{fig:topological} for an illustration. Note that it is possible that the boundary intersects only a single path. Let~$P_R, P_S$ be two paths not intersecting the boundary of $F_{3,4}^{[D]}$. 

In the graph~$G_{2,5}^{R,S}$, the face~$F_{2,5}^{R,S}$ containing the image of~$v$ is again bounded. Let~$C''$ denote the cycle bounding~$F_{2,5}^{R,S}$, which is simple since~$G_{2,5}^{R,S}$ is connected and has no articulation points. By the nesting of the cycles,~$C''$ is disjoint from~$C_1$ and~$C_6$. By definition of the graph~$G_{2,5}^{R,S}$, the cycle~$C''$ is disjoint from~$V(C_3) \cup V(C_4)$, and therefore~$V(C') \cap V(C'') = \emptyset$. Since~$G_{2,5}^{R,S}$ contains cycles~$C_2$ and~$C_5$, while the nesting and placement of~$v$ ensures that~$v$ is in the proper interior of~$C_2$ and proper exterior of~$C_5$, it follows that~$C''$ separates~$v$ from~$V(C_1) \cup V(C_6)$. Since both~$C'$ and~$C''$ have~$v$ in their interior and~$C_1$ in they exterior, they are nested. This concludes the proof.
\end{proof}

\subsection{Compressing the outerplanarity layers}

We focus on the case where
$x, y \in V(T)$ are some marked vertices, $x$ is an ancestor of $y$ and $T_{x,y}$ contains no marked vertices (and thus its image contains no boundary vertices). 
If $\delta_T(x,y)$ is large, we want to preserve just $\Oh(k)$ outerplanarity layers closest to $\kappa(x)$, $\Oh(k)$ layers closest to $\kappa(y)$, and in the middle layers we want to \bmp{preserve}  only the maximal family internally of vertex-disjoint $(\kappa(x),\kappa(y))$-paths, discarding the remaining vertices and edges.

We need to argue for the backward safeness of such a graph modification.
Consider a solution $S$ such that $|S \cap \kappa(T_{x,y})| \le k$.
If $H-S$ contains at least 4 internally vertex-disjoint $(\kappa(x),\kappa(y))$-paths, we can surround the modified areas with two cycles using \cref{lem:diameter:topological} and argue that inserting them back does not affect planarity.
Otherwise we want to say that $S \cap \kappa(T_{x,y})$ is sufficiently large so we can charge these vertices to ``pay'' for additional vertex removals which separate  $\kappa(x)$ from $\kappa(y)$ in $H$.
The following statement remains correct if we assume the existence of just 4 internally vertex-disjoint $(\kappa(x),\kappa(y))$-paths, instead of 6, but with a worse approximation factor of the lifting algorithm.
On the other hand, factor 2 can be improved to $(1+\eps)$ if we assume that even more paths exist but for the sake of keeping the presentation clear we do not optimize this constant.

\begin{lemma}\label{lem:diameter:middle-layers}
Let $G$ be a graph, $A \subseteq V(G)$, $H$ be a \mic{connected} boundaried plane graph given by some plane embedding of $G \brb A$,
$(T,\kappa)$ be the outerplanar decomposition of $H$, and $k$ be an integer.
Consider vertices $x,y \in V(T)$ such that $x$ is an ancestor of $y$, $\delta_T(x,y) > 4k + 11$, $\kappa(T_{x,y}) \cap \partial(H) = \emptyset$,
and there are 6 internally vertex-disjoint $(\kappa(x), \kappa(y))$-paths in the graph $H \brb {\kappa(T_{x,y})}$.
There is a polynomial-time algorithm that, given all the above, returns a vertex set $V_D \subseteq \kappa(T_{x,y})$ and an edge set $E_D \subseteq E_G(\kappa(T_{x,y}), \kappa(T_{x,y}))$ so that the following hold.
\begin{enumerate}
    \item For each vertex $v \in \kappa(y)$ there exists a vertex $u \in \kappa(x)$ so that the radial distance between $u$ and $v$ in $H' = (H \setminus E_D) - V_D$ is at most $4k + 11$.\label{lem:diameter:middle-layers:distance}
    \item The boundaried graph $H' = (H \setminus E_D) - V_D$ is connected.
    \item Given a planar modulator $S'$ in $G' = (G \setminus E_D) - V_D$, such that $|S' \cap \kappa(T_{x,y})| \le k$, one can turn it, in polynomial time, into a planar modulator $S$ in $G$ such that $|S| \le |S'| + 2 \cdot |S' \cap \kappa(T_{x,y})|$ and $S \setminus \kappa(T_{x,y}) = S' \setminus \kappa(T_{x,y})$.\label{lem:diameter:middle-layers:lifting}
\end{enumerate}
\end{lemma}
\begin{proof}

We set $D = \kappa(T_{x,y})$, $T_1 = \kappa(x)$, $T_2 = \kappa(y)$, and $J$ \mic{be the boundaried plane graph given by the subgraph of $H$ induced by $D \cup T_1 \cup T_2$ and the boundary $T_1 \cup T_2$.
Note that $N_G(D) \subseteq T_1 \cup T_2$ and $T_1,T_2$ are non-adjacent.}
Let $d$ be the minimum size of a~restricted $(T_1,T_2)$-separator in $J$.
By Menger's theorem there exist $d$~internally vertex-disjoint $(T_1, T_2)$-paths in $J$: $P_1, \dots P_d$.
By the assumption $d \ge 6$.
Furthermore, by \cref{lem:diameter:uncrossing} we can assume that each path $P_i$ has only one segment in each $\kappa(z)$ for $z \in T^P_{x,y}$.

\mic{Let $q_1 \in V(T)$ be the vertex on path $T^P_{x,y}$ with distance $\delta_T(x,q_1)$ exactly $2k+5$.
Similarly, $q_2 \in V(T)$ is the vertex on path $T^P_{x,y}$ with distance $\delta_T(y,q_2)$ exactly $2k+5$.
We define $Q = T_{q_1,q_2}$.}
By the assumption $Q$ is non-empty.
Let $V_Q = \kappa(Q)$; then $V_Q$ is a restricted $(T_1,T_2)$-separator in $H$ and also in $J$.
We define $E_D$  to be the set of edges with at least one endpoint in $V_Q$, which do \bmp{not} belong to any path $P_i$ for $i \in [d]$.
The set $V_D$ is the set of vertices in $V_Q$ which become isolated in $H \setminus E_D$. 
\mic{See Figure~\ref{fig:middle-layers} for an illustration.}
Let $G', H', J'$ denote the \bmp{graphs} obtained respectively from $G,H,J$ after removing the edges from $E_D$ and vertices from $V_D$.
We have that $G' = G_\partial[\overline A] \oplus H'$. 
Note that \mic{$H'$ is connected and} $J'$ still admits $d$ internally vertex-disjoint $(T_1,T_2)$-paths.

We will show that any planar modulator $S'$ in $G'$ with bounded intersection with $D$
can be lifted to a planar modulator in $G$, by possibly increasing its size slightly.
We consider two cases.
First, if removing vertices from $S' \cap D$ \bmp{preserves} at least 4 internally vertex-disjoint $(T_1, T_2)$-\bmp{paths} in $J'$,
then we show that any region in which the graph modifications occurred can be surrounded by 2 cycles separating this region from the boundary of $H$.
Then the planarity criterion can be applied to argue that $S'$ remains a planar modulator after undoing the graph modifications.
Otherwise, the set $S' \cap D$ is sufficiently large so we can charge it with adding to $S'$ a restricted $(T_1,T_2)$-separator in $J$.
Then each leftover connected component in $J$ is disconnected from either $T_1$ or $T_2$, \bmp{which} again allows us to apply the planarity criterion to justify undoing the graph modifications.

\begin{figure}
\centering
\includegraphics[width=0.75\linewidth]{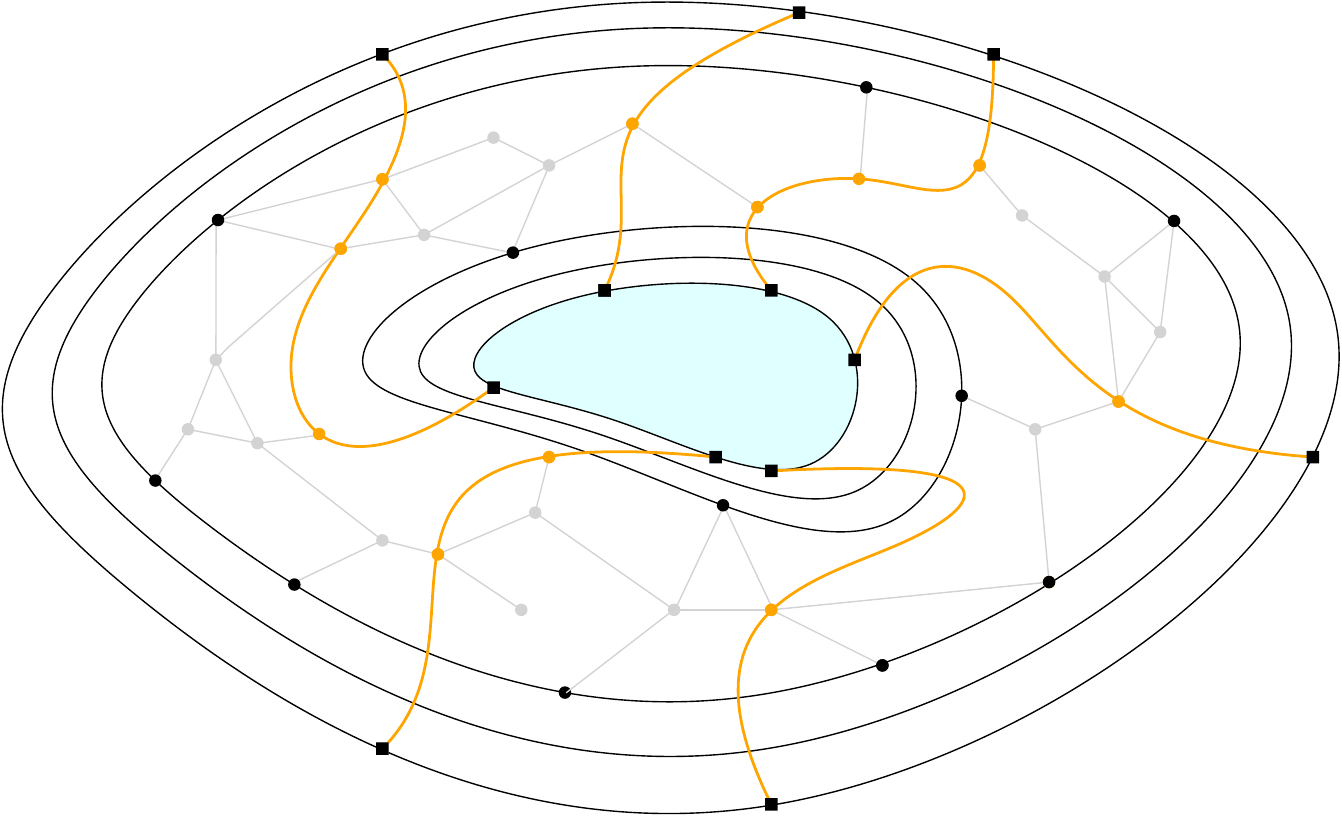}
\caption{Schematic view \bmp{of} the proof of \cref{lem:diameter:middle-layers}. The orange paths connect vertices from $T_1$ (on the outside) with vertices from $T_2$ (inside).
The middle layer corresponds to the area between $\kappa(q_1)$ and $\kappa(q_2)$.
The paths can be assumed to cross the middle layer only once due to \cref{lem:diameter:uncrossing}.
All the edges and vertices marked gray in the middle layer are being removed; these are the sets $E_D$ and $V_D$.
}
\label{fig:middle-layers}
\end{figure}

\begin{claim}
Let $S'$ be a planar modulator in $G'$ such that $|S' \cap D| \le k$ and the minimum size of \bmp{a} restricted $(T_1, T_2)$-separator in $J'-(S' \cap D)$ is at least 4.
Then $S'$ is also a planar modulator in~$G$.
\end{claim}
\begin{innerproof}
By Menger's theorem, there exist 4 internally vertex-disjoint $(T_1, T_2)$-paths in $J' - (S' \cap D)$: $R_1, R_2, R_3, R_4$.
We will transform the graph $J' - S'$ back into $J - S'$ by undoing the vertex and edge removals so that $G' - S'$ is transformed into $G - S'$, while maintaining planarity of the graph.
First, inserting an isolated vertex from $V_D$ does not affect planarity of a graph.
Now consider a single insertion of an edge $e \in E_D$.
If both endpoints of $e$ are disconnected from $(T_1 \cup T_2) \setminus S'$ in $J''-S'$ then this connected component of $G''$ after the modification 
is contained  in $D$, which induces a planar graph by the assumption
(recall that we assumed that $D \cap \partial(H) = \emptyset$).

Suppose otherwise.
Let $v \in V_Q \setminus S'$ be an endpoint of $e$.
Since $|S'| \le k$, there exist vertices $q'_1, \dots, q'_6 \in V(T)$ lying on this order on the path $T^P_{x,y}$
such that  $q'_1, q'_2, q'_3$ appear before $Q$ (counting from $x$) and $q'_4, q'_5, q'_6$ appear after $Q$ such that $\kappa(q'_j) \cap S' = \emptyset$ for $j \in [6]$.
By \cref{lem:diameter:outerplanar-cycle-separator} for each $j \in [6]$ there is a cycle $C'_j \subseteq \kappa(q'_j)$,
so that $C'_{j+1}$ lies in the interior of $C'_j$ in $J'-S'$.
Furthermore, $v$ lies in the intersection of the proper interior of $C_3$ and the proper exterior of $C_4$ in $J' - S'$
and each of the paths $R_i$, $i\in[4]$, contains a $(C'_1,C'_6)$-subpath $R'_i$ so that these subpaths are vertex-disjoint (not just internally vertex-disjoint).
This setting allows us to apply \cref{lem:diameter:topological} to obtain that there exist cycles $C_1, C_2$ in $J'-S'$ which are nested with respect to $v$ and separate $v$ from $C'_1 \cup C'_6$ and thus from $T_1 \cup T_2$.
Therefore $R_{H'-S'}[v,C_i] \cap \partial(H) = \emptyset$ for $i \in [2]$ and these cycles are $v$-planarizing.
By \cref{lem:prelim:criterion:old} we infer that if $(G'-S') \setminus e$ is planar then also $G'-S'$ is planar.
This allows us to insert back all the edges from $E_D$ into  $G'-S'$ while maintaining the planarity of the graph and show that $G-S'$ is planar.
The claim follows.
\end{innerproof}

\begin{claim}
Let $S'$ be a planar modulator in $G'$ such that $|S' \cap D| \le k$ and
there \bmp{exists} a restricted $(T_1, T_2)$-separator in $J'-(S' \cap D)$ of size at most $3$.
Let $R$ be a minimum-size restricted $(T_1, T_2)$-separator in $J$.
Then $S = S' \cup R$ is also a planar modulator in $G$ and $|S| \le |S'| + 2 \cdot |S' \cap D|$.
\end{claim}
\begin{innerproof}
Recall that we have preserved the paths $P_1,\dots,P_d$ during the graph modification and $J'$ still admits $d \ge 6$ internally vertex-disjoint $(T_1,T_2)$-paths.
The size of $S' \cap D$ is at least $d - 3 \ge 3$ because removing one vertex from $D$ can decrease the flow between $T_1$ and $T_2$ (which are disjoint from $D$) by at most one.
We have $|R| = d \le  |S' \cap D| + 3 \le 2 \cdot |S' \cap D|$
and so $|S| \le |S'| + |R| \le |S'| + 2 \cdot |S' \cap D|$, as intended.
Also, $|S| \le |S'| + |R| \le |S'| + |S' \cap D| + 3 \le 2k + 3$.
By the definition of $Q$, there exist vertices $q'_1, q'_2, q'_3, q'_4 \in V(T)$ lying \bmp{in} this order on the path $T^P_{x,y}$ 
such that  $q'_1, q'_2$ appear before $Q$ (counting from $x$) and $q'_3, q'_4$ appear after the subpath $Q$ such that $\kappa(q'_j) \cap S = \emptyset$ for $j \in [4]$.

We will transform the graph $J' - S$ back into $J - S$ by undoing the vertex and edge removals so that $G' - S$ is transformed into $G - S$, while maintaining planarity of the graph.
First, inserting an isolated vertex from $V_D$ does not affect planarity of a graph.
Now consider a single insertion of an edge $e \in E_D$ and let $J'',G''$ denote the graph obtained respectively from $J', G'$ after the insertions done so far, and assume inductively that $S$ is a planar modulator in $G'' \setminus e$.
We are going to prove that $S$ is a planar modulator in $G''$.
If both endpoints of $e$ are disconnected from $(T_1 \cup T_2) \setminus S$ in $J''-S$ then this connected component of $G''$ being modified
is contained  in $D$, which induces a planar graph by assumption.

Suppose otherwise.
Let $v \in V_Q$ be an endpoint of $e$.
Since $R$ is a restricted $(T_1, T_2)$-separator in $J$ (and so in $J''$) and $R \subseteq S$, $v$ can be connected to at most one of $T_1,T_2$ in $J'' - S$.
Suppose w.l.o.g. that $v$ is connected to some vertex from $T_1$ in $J''-S$ but none from $T_2$.
By the definition of $V_Q$, every $(T_1,v)$-path in $J$ must intersect $\kappa(q'_1)$ and $\kappa(q'_2)$.
These sets are disjoint from $S$
and by \cref{lem:diameter:outerplanar-cycle-separator} there are two cycles separating $v$ in $G''$ from any vertex from $T_1 \setminus S$.  
By \cref{lem:prelim:criterion:old} we conclude that $G'' - S$ is planar.
After undoing all the edge deletions we obtain that $G - S$ is planar.
The claim follows.
\end{innerproof}

In the first case we do not change the size of the planar modulator and in the second case
we have $|S| \le |S'| + 2 \cdot |S' \cap D|$.
Therefore the lifting algorithm can turn any solution to $G'$ into a solution to $G$ within the guarantees of the lemma.

Finally, let us examine the radial distances in $H'$.
By \cref{lem:diameter:outerplanar-faces} we have
that whenever $x_1,x_2 \in V(T)$ are such that $x_1$ is parent to $x_2$ and $v \in \kappa(x_2)$ then there exists $u \in \kappa(x_1)$ so that $u,v$ share a face.
Let us fix a vertex $v_y \in \kappa(y)$.
By the observation above, there is a vertex $v_2 \in \kappa(q_2)$ so that $d_{H'}(v_y,v_2) = d_{H}(v_y,v_2) \le 2k + 5$.
As each path $P_i$, $i \in [d]$, has only one segment in $\kappa(q_1)$ and one in $\kappa(q_2)$, it also has only one segment in $V_Q$.
The region between $\kappa(q_1)$ and $\kappa(q_2)$ in  $H'$ comprises $d$ subdivided edges between $\kappa(q_1)$ and $\kappa(q_2)$.
Hence, there is a vertex $v_1 \in \kappa(q_1)$ which shares a face with $v_2$ in $H'$.
By the same argument as before, there is a vertex $v_x \in \kappa(x)$ at radial distance $2k+5$ from $v_1$ in $H'$.
This implies that $d_{H'}(v_y,v_x) \le (2k+5) + 1 + (2k+5) = 4k + 11$ and concludes the proof.
\end{proof}

Our strategy is to
mark the set of nodes in $T$ whose images cover the boundary $\partial(H)$ together with its LCA closure.
Suppose that $x, y \in V(T)$ were marked, $x$ is an ancestor of $y$, and $T_{x,y}$ contains no marked vertices.
If $\delta_T(x,y)$ is large and there are 6 internally vertex-disjoint $(\kappa(x),\kappa(y))$-paths, we can use \cref{lem:diameter:middle-layers} to compress the middle layers.
If the $(\kappa(x),\kappa(y))$-flow is less than 6, there is a restricted  $(\kappa(x),\kappa(y))$-separator of size at most 5.
If the number of such separators is small, we could just mark some more nodes in $T$.
We cannot do this however if there is a long sequence of small $(\kappa(x),\kappa(y))$-separators.
But in this case the subgraph of $H$ between the first and the last separator is disjoint from $\partial(H)$ and has boundary of size at most 10.
Recall that the \planardel problem admits a lossless protrusion replacer (\cref{lem:protrusion:planar-deletion})
which allows us to compress planar subgraphs of small boundary (\cref{lem:protrusion:planar-replacement}).
We use this tool in a preprocessing step to ensure that the distance between the first and last small separator cannot be too large.

\begin{lemma}\label{lem:diameter:diameter-replacement}
Let $\gamma$ be the function from \cref{def:protrusion:replacer} for the \planardel problem.
There is a polynomial-time algorithm that, given a graph $G$, integers $k,r$, a set $A \subseteq V(G)$, such that $G\brb A$ is planar, $G[A]$ is connected, and $r = |N_G(A)|$, returns an $r$-boundaried plane graph $H'$ so the following hold.
\begin{enumerate}
    \item  The radial diameter of $H'$ is at most $2\cdot (4k + \gamma(8019) + 12) \cdot (6r + 5)$ and the graph $H' \setminus \partial(H')$ is connected.
    \item $\mvp(G') \le \mvp(G)$, where $G'$ is the graph obtained from $G$ after replacing $N_G[A]$ with $H'$.
    \item Let $A' = V(H') \setminus \partial(H')$ denote the set of the new vertices inserted in place of $A$. There is a polynomial-time algorithm that, given a planar modulator $S'$ in $G'$ such that $|S' \cap A'| \le k$, outputs a planar modulator $S$ in $G$ such that $|S| \le |S'| + 2 \cdot |S' \cap A'|$ and $S \setminus A = S' \setminus A'$.
\end{enumerate}
\end{lemma}
\begin{proof}
Let $H$ denote a boundaried plane graph obtained by an arbitrary plane embedding of $G \brb A$.
We shall perform a series of modifications of $H$ which does not affect the boundary set.
First, we get rid of planar subgraphs with small boundary.
We iterate over all sets $X_0 \subseteq V(H) \setminus \partial(H)$ of size at most 10 and inspect whether there is a connected component $C$ of $H - X_0$ which (a) does not contain any vertex from $\partial(H)$, (b) has at least $\gamma(8019)$ vertices.
If we detect such a component $C$, we invoke \cref{lem:protrusion:planar-replacement} for $r = |N_H(C)| \le 10$ and replace $N_G[C]$ with a boundaried subgraph on less than $\gamma(27\cdot 27 \cdot 11) = \gamma(8019)$ vertices.
This does not affect the minimum vertex planarization number and allows to perform lossless solution lifting in polynomial time.
We perform such a replacement exhaustively until we cannot find a pair $(X_0,C)$ satisfying the requirements.
By \cref{lem:protrusion:planar-replacement}(\ref{lem:protrusion:planar-replacement:planar-subgraph}) for $X = V(G) \setminus A$ such a replacement does not affect planarity of the graph $G \brb A$.
By the same argument for $X = V(G) \setminus N_G[A]$ we get that the replacement cannot disconnect the graph $G[A]$.
In each step, the size of the graph decreases, so at some point there are no more such pairs $(X_0,C)$.
After this process terminates, we know that there is no separator of size at most 10, \bmp{which} would disconnect a subgraph of $H$ which does not touch the boundary and has at least $\gamma(8019)$ vertices.
We keep the variable $H$ to denote the graph after this modification.

Let us now consider the outerplanar decomposition $(T,\chi)$ of $H$.
For each $v \in \partial(H)$ we mark the node $x \in V(T)$ such that $v \in \kappa(x)$.
Let $Y_0$ denote the set of marked nodes plus the root node, $|Y_0| \le r + 1$, 
and $Y_1 = \overline{\mathsf{LCA}}(Y_0)$ (recall \cref{def:treewidth:lca}).
By \cref{lem:treewidth:lca} we know that $|Y_1| \le 2r+2$ and every
connected component of $T - Y_1$ has at most 2 neighbors.
Furthermore, \bmp{there} can be at most $2r+1$ such components with exactly 2 neighbors because contracting them forms a tree on $2r+2$ vertices.
Consider such a component $T_C \in \cc\cc(T-Y_1)$, $N_T(T_C) = \{x,y\}$.
Then one of the vertices $x,y$ is an ancestor of the other one.
Note that the path $T^P_{x,y}$ is contained in $N_T[T_C]$.
Let $x' \in T^P_{x,y}$ be the furthest (from $x$) vertex on the $(x,y)$-path for which there are 6 internally vertex-disjoint $(\kappa(x), \kappa(x'))$-paths in the graph $H \brb {\kappa(T_{x,x'})}$.
It might be the case that $x' = y$ \bmp{or $x' = x$}.
Next, let $y'  \in T^P_{x,y}$ be the furthest (from $y$) vertex on the $(x',y)$-path for which there are 6 internally vertex-disjoint $(\kappa(y), \kappa(y'))$-paths in the graph $H \brb {\kappa(T_{y,y'})}$.
If $x' = y$ then $y' = y$.

If $x',y'$ are distinct, we want to show that their distance $\delta_T(x', y')$ cannot be larger than some constant.
Suppose that $\delta_T(x', y') > 3$. 
Let $x'', y''$ be the neighbors of $x,y$, respectively, on the path $T^P_{x',y'}$.
By the definition of $x',y'$ there are no 6 internally vertex-disjoint $(\kappa(x), \kappa(x''))$-paths in the graph $H[\kappa(T_{x,x''})]$.
The sets $\kappa(x), \kappa(x'')$ are non-adjacent in $H$ so
by Menger's theorem there exists a restricted $(\kappa(x), \kappa(x''))$-separator $S_x$ in $H$ of size at most 5.
Analogous observations hold for $y,y''$ and there exists a restricted $(\kappa(y), \kappa(y''))$-separator $S_y$ in $H$ of size at most 5.
Consider the set $X_0 = S_x \cup S_y$ of size at most 10
and the connected component of $H-X_0$ containing the set $\kappa(T_{x'',y''})$.
By \cref{lem:diameter:outerplanar-separator},
it is disjoint from $\partial(H)$ (as $\partial(H)$ is contained in the $\kappa$-image of $Y_0$) so thanks to the preprocessing procedure we can conclude that it has less than $\gamma(8019)$ vertices.
This implies that $|T_{x'',y''}| < \gamma(8019)$ and so $\delta_T(x',y') = |T_{x',y'}| = |T_{x'',y''}| + 2 \le \gamma(8019) + 1$.

Let $Y_2$ be the union of $Y_1$ and all vertices marked as $x', y'$ in the previous step.
We have $|Y_2| \le |Y_1| + 2\cdot (2r+1) \le 6r+4$.
Each connected component $T_C \in \cc\cc(T - Y_2)$ 
with exactly two neighbors, $N_T(T_C) = \{x,y\}$,
satisfies one of the following: (a) $\delta_T(x,y) \le \gamma(8019) + 1$, or (b) there are 6 internally vertex-disjoint $(\kappa(x), \kappa(y))$-paths in the graph $H\brb {\kappa(T_{x,y})}$.
Furthermore, one of the vertices $x,y$ is an ancestor of the other one.
In order to reduce \bmp{the} depth of $T$ we will apply \cref{lem:diameter:middle-layers} to the components of type (b).
Let $T^1_C, T^2_C, \dots, T^\ell_C$, $N_T(T^i_C) = \{x^i,y^i\}$, be some ordering of the type (b) components of $T-Y_2$ for which $\delta_T(x^i,y^i) > 4k + 11$.
We have $T_{x^i,y^i} = T^i_C$ and these sets are pairwise disjoint.
By applying \cref{lem:diameter:middle-layers} to $x^i, y^i, T^i_C$ we obtain a set of vertices $V_D^i \subseteq \kappa(T^i_C)$ and a set of edges $E_D^i \subseteq E_G(\kappa(T^i_C), \kappa(T^i_C))$ to be removed.
Let $G_0,H_0$ be obtained from respectively $G,H$ by removing the edges from all $E_D^i$ and the vertices from all $V_D^i$.
Furthermore, let $G_i$, $i \in [\ell]$, be obtained from $G$ \bmp{by} performing only the removals from $V^j_D, E^j_D$ for $j > i$.
Clearly $\mvp(G_0) \le \mvp(G)$ as $G_0$ is a subgraph of $G$.
Let $A_0 = V(H_0) \setminus \partial(H_0)$ be the set of vertices that form a replacement for $A$.

First, we argue that a solution to $G_0$ can be lifted into a solution to $G$ with moderate additional cost.
Let $S_0 \subseteq V(G_0)$ be a planar modulator in $G_0$ such that $|S_0 \cap A_0| \le k$.
We construct a series of sets $S_1, \dots, S_\ell$ where $S_i$ is a refinement of $S_{i-1}$ with respect to $T^i_C$
so that $S_i$ is a planar modulator in $G_i$.
This holds for $i = 0$.
Let us consider $i \ge 1$ and suppose that the claim holds for $i-1$.
By \cref{lem:diameter:middle-layers}(\ref{lem:diameter:middle-layers:lifting}) we can compute a planar modulator $S_i$ in $G_i$ of size at most $|S_{i-1}| + 2\cdot |S_{i-1} \cap \kappa(T^i_C)|$ and $S_i \setminus \kappa(T^i_C) = S_{i-1} \kappa(T^i_C)$.
The latter condition guarantees that $S_{j} \cap \kappa(T^i_C) = S_{0} \cap \kappa(T^i_C)$ for $j < i$ because in the $j$-step we only add new vertices from $\kappa(T^j_C)$ and these sets are disjoint.
Therefore $|S_\ell| \le |S_0| + \sum_{i=1}^\ell 2 \cdot |S_0 \cap \kappa(T^i_C)| \le |S_0| + 2 \cdot |S_0 \cap A_0|$
and $S_\ell \setminus A = S_0 \setminus A_0$.
As $S_\ell$ is a planar modulator in $G_\ell = G$, the solution lifting algorithm works as intended.

Now, we analyze the radial diameter of $H_0$.
Let $v \in \partial(H_0) = \partial(H)$ and $y \in Y_2$ be such that $v \in \kappa(y)$.
Since the root node $x_1$ belongs to $Y_2$ and $|Y_2| \le 6r + 4$,
there a sequence $y = y_1, \dots, y_m = r$ of at most $m \le 6r + 4$ vertices from $Y_2$ 
so that $y_{i+1}$ is an ancestor of $y_i$ and $T_{y_i, y_{i+1}} \cap Y_2 = \emptyset$.
If  $T_{y_i, y_{i+1}}$ is a component of type (b) then by
\cref{lem:diameter:middle-layers}(\ref{lem:diameter:middle-layers:distance}) we know that \bmp{for} each $u \in \kappa(y_i)$ there is $u' \in \kappa(y_{i+1})$ such that their radial distance in $H_0$ is at most $4k+11$.
If $T_{y_i, y_{i+1}}$ is a component of type (a), then by \cref{lem:diameter:outerplanar-faces} we get the analogous property but with upper bound $\gamma(8019) + 1$.
Therefore $v$ is at radial distance at most $(4k + \gamma(8019) + 12) \cdot (6r + 4)$ from some vertex on the outer face.

To analyze the situation for vertices from $V(H_0) \setminus \partial(H_0)$ let us consider the outerplanar decomposition $(T_0, \kappa_0)$ of $H_0$.
By the estimation above, if $v \in \partial(H_0) = \partial(H)$ then $\opindex_{H_0}(v) \le (4k + \gamma(8019) + 12) \cdot (6r + 4)$.
Suppose there is a vertex $u \in V(H_0) \setminus \partial(H_0)$ for which $\opindex_{H_0}(u) > (4k + \gamma(8019) + 12) \cdot (6r + 4) + k + 3$.
Then there are at least $k+3$ outerplanarity layers around $u$ whose interiors \bmp{are} free of any boundary vertices.
Consider the sequence $u_1, u_2, \dots, u_{k+3} \in V(T_0)$ of the first ancestors of $u$ in $T_0$.
By \cref{lem:diamater:outerplanar-properties}(\ref{prop:connected}) each set $\kappa(u_i)$ induces a connected subgraph of $G_0$.
If $v \in \partial(H_0)$, then $\opindex_{H_0}(v) < \opindex_{H_0}(q_i)$ and
by \cref{lem:diamater:outerplanar-properties}(\ref{prop:cut}) $\kappa(u_i)$ is a restricted $(u,v)$-separator in $H_0$. 
Hence,
these sets form a nested sequence of $(u,\partial(H_0))$-separators in $H_0$ and thus 
these are $(u,N_{G_0}(A))$-separators in $G_0$.
We infer that these sets are $u$-planarizing and by \cref{lem:prelim:criterion:new} the vertex $u$ is $k$-irrelevant.
While such a vertex $u$ exists we can safely remove it;
let $H'$ be the graph obtained from $H_0$ after these removals.
The depth of the outerplanar decomposition tree of $H'$ is bounded by $(4k + \gamma(8019) + 12) \cdot (6r + 4) + k + 3$ and the claim follows from \cref{lem:diameter:outerplanar-diameter}.
\end{proof}

We will apply \cref{lem:diameter:diameter-replacement} to each connected component $C \in \cc\cc(G-X)$ where $X$ is a strong planar modulator supplied by \cref{lem:treewidth:final}.
We obtain some replacement $H_C$ for the boundaried graph $G\brb C$ equipped with a plane embedding of moderate diameter.
Because the lifting algorithm from  \cref{lem:diameter:diameter-replacement} increases the size of a solution only locally in $C$ and does not affect the solution outside $C$, we can perform the lifting iteratively, increasing the size of the solution only by a constant factor in the end.

\begin{proposition}\label{lem:diameter:final}
There is a polynomial-time algorithm that, given a graph $G$ with an $r$-strong planar modulator $X$ of fragmentation $s$, and integer $k$,
outputs a graph $G'$ together with a $r$-strong planar modulator $X'$ of fragmentation $s$ and a family of plane embeddings of $G'\brb {C'}$ for each connected component $C' \in \cc\cc(G'-X')$, 
so that
\begin{enumerate}
    \item $\mvp(G') \le \mvp(G)$,
    \item for each connected component $C' \in \cc\cc(G'-X')$, the computed plane embedding has radial diameter $\Oh(kr)$,
    \item given $G, G'$, and a planar modulator $S'$ in $G'$ of size at most $k$, one can find, in polynomial time, a planar modulator $S$ in $G$ of size at most $3 \cdot |S'|$.
\end{enumerate}
\end{proposition}
\begin{proof}
Let $C_1, C_2, \dots, C_\ell$ be some ordering of the connected \bmp{components} $C \in \cc\cc(G-X)$.
We process each $G\brb {C_i}$ with \cref{lem:diameter:diameter-replacement} iteratively.
Namely, let $G_0 = G$ and $G_{i}$ be obtained from $G_{i-1}$ by applying \cref{lem:diameter:diameter-replacement} to $G_{i-1} \brb {C_i}$, for $i \in [\ell]$.
We set $G' = G_\ell$.
Let $C'_i$ be the set of new vertices inserted in place of $C_i$ and $H'_i$ denote the plane graph given by the computed embedding of $G_i \brb {C'_i}$.
We have $\mvp(G_{i}) \le \mvp(G_{i-1})$ so $\mvp(G_\ell) \le \mvp(G)$.
Furthermore, the radial diameter of $H'_i$ is $\Oh(kr)$
and we are guaranteed that the graphs $G_i[C'_i] = H'_i \setminus \partial(H'_i)$ are connected so we do not increase the fragmentation of the modulator.
As $|N_{G_i}(C'_i)| = |N_{G_{i-1}}(C_i)|$ the modulator remains $r$-strong.

It remains to prove the solution lifting property.
Let $S_\ell \subseteq V(G_\ell)$ be a planar modulator in $G_\ell$ of size at most $k$.
We construct a series of sets $S_0, \dots, S_{\ell-1}$ where $S_i$ is a modification of $S_{i+1}$
so that $S_i$ is a planar modulator in $G_i$
of size at most $|S_\ell| + \bigcup_{j > i}^\ell 2 \cdot |S_\ell \cap C'_j|$.
Moreover, we show inductively that $S_i \cap C'_j = S_\ell \cap C'_j$ for $j \le i$, which holds for $i = \ell$.

We use the lifting algorithm from
\cref{lem:diameter:diameter-replacement}
to turn $S_{i+1}$ into a planar modulator $S_i$ in $G_i$ of size at most $|S_{i+1}| + 2 \cdot |S_{i+1} \cap C'_{i+1}|$ so that $S_i \setminus C_{i+1} = S_{i+1} \setminus C'_{i+1}$.
The latter condition implies that we do not alter the solution $j \ne i + 1$ and shows the invariant  $S_i \cap C'_j = S_\ell \cap C'_j$ for $j \le i$.
Therefore $|S_{i+1} \cap C'_{i+1}| = |S_\ell \cap C'_{i+1}| \le k$ as required by the lifting algorithm.
We estimate the final size of the planar modulator in $G_0$ as $|S_\ell| + \sum_{i=1}^\ell 2 \cdot |S_i \cap C'_i| = |S_\ell| + 
\sum_{i=1}^\ell 2 \cdot |S_\ell \cap C'_i| \le 3 \cdot |S_\ell|$.
\end{proof}

\section{Solution preservers}
\label{sec:preservers}

After the application of \cref{lem:diameter:final}, we can assume that the considered graph $G$ has a strong planar modulator $X \subseteq V(G)$ of fragmentation $k^{\Oh(1)}$
and for each $C \in \cc\cc(G-X)$ we are given a \bmp{plane} embedding of $G \brb C$ with radial diameter $k^{\Oh(1)}$.
In order to further simplify $G \brb C$ we take a two-step approach.
First, we mark a vertex set $S_C \subseteq C$ so that 
for any solution $S$ we can replace the set $S \cap C$ with a subset of $S_C \cup N_G(C)$ by paying a constant cost for each replacement.
In particular, we get that there is an $\Oh(1)$-approximately optimal planar modulator $S$ such that $S \cap C \subseteq S_C$.
\bmp{Afterwards, in Section~\ref{sec:compressing},} we compress the subgraphs within $C \setminus S_C$ for which we know that the approximately optimal solution does not remove any vertices inside them.
We begin with formalizing the concept of marking vertices to cover an approximate solution.

\begin{definition}\label{def:boundaried:preserver}
For a vertex set $A \subseteq V(G)$, we say that $S_A \subseteq A$ is an $\alpha$-preserver for $A$, if for any planar modulator $S$ in $G$ there exists a planar modulator $S' \subseteq (S \setminus A) \cup N_G(A) \cup S_A$ such that
$S' \setminus N_G[A] = S \setminus N_G[A]$ and
$|S'| \le |S| + \alpha \cdot |S \cap A|$.
\end{definition}

We are going to construct \bmp{these} solution preservers recursively \bmp{by} reducing the task to more structured boundaried graphs.
We analyze this process by decomposing a boundaried graph $H$ into a set $Y$ which can be regarded as an augmentation of the boundary $\partial(H)$ and a family of smaller boundaried graphs whose boundary lies on $\partial(H) \cup Y$.  

\begin{definition}\label{def:inseparable:boundaried-decomposition}
A boundaried decomposition of a boundaried graph $H$ is a pair $(Y, \mathcal{C})$, where $Y \subseteq V(H) \setminus \partial(H)$ and $\mathcal{C}$ is a partition of $V(H) \setminus (\partial(H) \cup Y)$, such that for each $C \in \mathcal{C}$ it holds that $N_H(C) \subseteq \partial(H) \cup Y$.
\end{definition}

For example, when $\cc = \cc\cc(H - \bmp{(\partial(H) \cup Y)})$ is given by the connected components of $H - \bmp{(\partial(H) \cup Y)}$ then $(Y,\cc)$ is a boundaried decomposition of $H$.
For our convenience, we will also consider cases when some elements of $\cc$ are disconnected.
Usually, we begin with finding a set $Y$ that imposes some properties on boundaried graphs $H \brb C$ for $C \in \cc$.
In order to bound the number of subcases to recurse on, we would like to upper bound the maximal size of $\cc$ in terms of $Y$.
This is not always possible, but at least we can bound the number of sets $C \in \cc$ with more than two neighbors.

\begin{lemma}[{\cite[Lem. 13.3]{fomin2019kernelization}}]
\label{lem:boundaried:bipartite}
Let $G$ be a planar graph, $X \subseteq V(G)$, and let $N_3$ be a set of vertices from $V(G) \setminus X$ such that every vertex from $N_3$ has at least three neighbors in \bmp{$X$}. Then, $|N_3| \le 2\cdot |X|$.
\end{lemma}

For a vertex set $X \subseteq V(G)$ in a planar graph $G$, 
consider \bmp{the} graph obtained from $G$ by
contracting each connected \bmp{component} of $G-X$ into a single vertex.
Since contractions preserve planarity,
we get the following corollary \bmp{using the definition of fragmentation (Definition~\ref{def:fragmentation}).}

\begin{lemma}\label{lem:boundaried:components-degree-3}
For a planar graph $G$ and a vertex set $X \subseteq V(G)$, the fragmentation of $X$ in $G$ is at most $2\cdot |X|$.
\end{lemma}

The following \mic{lemmas provide} us with a framework to compute solution preservers recursively using boundaried decompositions.
The property $S' \setminus N_G[A] = S \setminus N_G[A]$ in \cref{def:boundaried:preserver} plays a role \mic{in the second lemma} because by modifying a solution locally we do not affect its intersection with other parts of the decomposition.

\begin{lemma}\label{lem:inseparable:preserver:two-sets}
Consider a graph $G$ and vertex sets $A, Y \subseteq V(G)$.
If $S_A \subseteq A$ is an $\alpha$-preserver for $A$, for some  $\alpha \ge 1$, then $S_A \cup Y$ is an $\alpha$-preserver for $A \cup Y$.
\end{lemma}
\begin{proof}
Let $S \subseteq V(G)$ by any planar modulator in $G$.
By the definition of an $\alpha$-preserver, there exists a planar modulator $S' \subseteq V(G)$ such that $S' \subseteq (S \setminus A) \cup N_G(A) \cup S_A$, $S' \setminus N_G[A] = S \setminus N_G[A]$, and
$|S'| \le |S| + \alpha\cdot |S \cap A| \le  |S| + \alpha\cdot |S \cap (A \cup Y)|$.
The first property can be restated as $S' \subseteq (S \setminus (A \cup Y)) \cup N_G(A \cup Y) \cup (S_A \cup Y)$ because inserting $Y$ at the end covers all the modifications to the first two terms.
Since $N_G[A] \subseteq N_G[A \cup Y]$ we get $S' \setminus N_G[A \cup Y] = S \setminus N_G[A \cup Y]$.
This satisfies all the three conditions of an $\alpha$-preserver for $A \cup Y$.
\end{proof}

\begin{lemma}\label{lem:inseparable:preserver:non-adjacent}
Consider a graph $G$ and vertex sets $A, B \subseteq V(G)$, such that $N_G[A] \cap B = \emptyset$.
If $S_A \subseteq A, S_B \subseteq B$ are $\alpha$-preservers for respectively $A,B$, for some  $\alpha \ge 1$, then $S_A \cup S_B$ is an $\alpha$-preserver for $A \cup B$.
\end{lemma}
\begin{proof}
Let $S \subseteq V(G)$ by any planar modulator in $G$.
By the definition of an $\alpha$-preserver, there exists a planar modulator $S_1 \subseteq V(G)$ such that $S_1 \subseteq (S \setminus A) \cup N_G(A) \cup S_A$, $S_1 \setminus N_G[A] = S \setminus N_G[A]$, and
$|S_1| \le |S| + \alpha\cdot |S \cap A|$.
The second property implies that $S_1 \cap B = S \cap B$.
Next, we apply the definition for the set $B$ and solution $S_1$.
We obtain a planar modulator $S_2 \subseteq V(G)$ such that $S_2 \subseteq (S_1 \setminus B) \cup N_G(B) \cup S_B$, $S_2 \setminus N_G[B] = S_1 \setminus N_G[B]$, and
$|S_2| \le |S_1| + \alpha\cdot |S_1 \cap B| \le |S| + \alpha\cdot |S \cap A| + \alpha\cdot |S \cap B| = |S| + \alpha\cdot |S \cap (A \cup B)|$.

We check that $S_1 \setminus B \subseteq S \setminus (A \cup B) \cup N_G(A) \cup S_A$ and $N_G(B) \cap A = N_G(A) \cap B = \emptyset$ so $S_2 \subseteq S \setminus (A \cup B) \cup N_G(A \cup B) \cup (S_A \cup S_B)$.

Finally, both $N_G[A], N_G[B]$ are subsets of $N_G[A \cup B]$ so we trivially obtain
$S_2 \setminus N_G[A \cup B] = S_1 \setminus N_G[A \cup B] = S \setminus N_G[A \cup B]$.
\end{proof}

From these two observations, we can easily obtain the following useful lemma.

\begin{lemma}\label{lem:inseparable:preserver}
Consider a graph $G$, vertex set $A \subseteq V(G)$, and a boundaried decomposition $(Y,\mathcal{C})$ of $G\brb A$.
Suppose that for each set $C \in \mathcal{C}$ we are given an $\alpha$-preserver $S_C \subseteq C$ for $C$.
Then the set $Y \cup \bigcup_{C \in \mathcal{C}} S_C$ is an $\alpha$-preserver for $A$.
\end{lemma}
\begin{proof}
We fix any ordering $C_1, \dots, C_m$ of $\mathcal{C}$.
For any distinct $i,j \in [m]$ we have $N_G[C_i] \cap C_j = \emptyset$.
From \cref{lem:inseparable:preserver:non-adjacent} we obtain that $S_{C_1} \cup S_{C_2}$ is an $\alpha$-preserver for $C_1 \cup C_2$ and, by repeating this argument, that $\bigcup_{i=1}^\ell S_{C_i}$ is an $\alpha$-preserver for $\bigcup_{i=1}^\ell C_i$, for each $\ell \in [m]$.
\mic{
Next, we apply \cref{lem:inseparable:preserver:two-sets} to $\bigcup_{i=1}^\ell C_i$ and $Y$, to get that
$\bigcup_{C \in \mathcal{C}} S_C \cup Y$ is an $\alpha$-preserver for their union, that is $A$.}
\end{proof}

Suppose that $G$ is a graph, $X \subseteq V(G)$ is a planar modulator in $G$, and for each component $C \in \cc\cc(G-X)$ we have computed an $\alpha$-preserver $S_C \subseteq C$.
We can apply \cref{lem:inseparable:preserver} to the boundaried graph $(G, \emptyset)$ (with the empty boundary) and the boundaried decomposition $(X, \cc\cc(G-X))$ to
get that $X' = X \cup \bigcup_{C \in \cc\cc(G-X)} S_C$ is an $\alpha$-preserver for $V(G)$.
In other words,
for every planar modulator $S \subseteq V(G)$ in $G$, including the optimal one, there exists a planar modulator $S'$ in $G$ such that $S' \subseteq X'$ and $|S'| \le \alpha \cdot |S|$.
This will allow us to restrict our attention only to the approximate solutions contained in $X'$.

\subsection{Reduction to single-faced boundaried graphs}

In this section, we show
how to reduce the task of constructing a solution preserver for $G\brb A$ with a given embedding of bounded diameter
to the case where in the embedding of $G\brb A$ all the boundary vertices lie on the outer face.
As an intermediate step, we shall consider \circum boundaried plane graphs, which \bmp{are slightly more general than \nice plane graphs.} 
This reduction will also come in useful in \cref{sec:compressing}.

As we work with boundaried plane graphs of bounded radial diameter, we can mark a bounded-size set of faces covering the boundary and forming a connected area on the plane.
Then the remaining parts of the graph interact with their boundaries only through the outer faces. \bmp{Recall the definition of a radial graph from Section~\ref{sec:prelims:planargraphs}.}

\begin{definition}
Let $G$ be a plane graph and let $\mathsf{Rad}(G) = (V(G) \cup F, E)$ be the radial graph of~$G$.
We say that $S \subseteq V(G)$ is radially connected if there exists $F' \subseteq F$ such that $\mathsf{Rad}(G)[S \cup F']$ is connected. 
\end{definition}

Given a set of vertices in a plane graph of moderate radial diameter we want to find its superset that is radially connected.
A bound on the radial diameter translates to a bound on the diameter of the radial graph.
We will take advantage of the following simple observation about bounded-diameter graphs.

\begin{lemma}\label{lem:boundaried:dominating}
Let $G$ be a connected graph of diameter bounded by $d$.
Then for any \bmp{non-empty} vertex set $S \subseteq V(G)$ there exists $S' \supseteq S$ such that $G[S']$ is connected and \bmp{$|S'| \le d \cdot |S|$}.
Furthermore, \bmp{such a} set $S'$ can be found in polynomial time when $G$ and $S$ are given.
\end{lemma}

\begin{proof}
We prove by induction on~$\ell$ that if~$G[S]$ has~$\ell$ connected components, there is a connected superset~$S' \supseteq S$ of size at most~$|S| + \ell (d-1)$. This proves that any set~$S$ has a connected superset of size at most~$d \cdot |S|$ since~$G[S]$ has at most~$\ell \leq |S|$ connected components, and therefore gives the lemma.

If~$\ell=1$ then~$S' = S$ suffices. Otherwise, let~$C_1,C_2$ be distinct connected components of~$G[S]$ and let~$v_1 \in C_1, v_2 \in C_2$. Since the diameter of~$G$ is at most~$d$, there is a~$(v_1,v_2)$-path~$P$ in~$G$ of length at most~$d$, which has at most~$d-1$ internal vertices~$I$. Then~$S \cup I$ has at most~$\ell-1$ connected components. By induction, there is a connected superset~$S'$ of~$S \cup I$ on at most~$|S \cup I| + (\ell - 1) \cdot (d-1) \leq |S| + \ell (d-1)$ vertices, which concludes the proof. The inductive argument easily turns into a recursive polynomial-time algorithm.
\end{proof}

\begin{lemma}\label{lem:boundaried:radially-dominating}
Let $G$ be a plane graph of radial diameter bounded by $d$.
Then for any vertex set $S \subseteq V(G)$ there exists $S' \supseteq S$ such that $S'$ is radially connected in $G$ and $|S'| \le |S| \cdot \bmp{(2d+2)}$.
Furthermore, the set $S'$ can be found in polynomial time when $G$ and $S$ are given.
\end{lemma}
\begin{proof}
Let $\mathsf{Rad}(G) = (V(G) \cup F, E)$ be the radial graph of $G$.
Note that $\mathsf{Rad}(G)$ is always connected, even if $G$ is not.
Consider some pair $u,v \in V(G)$.
By assumption \bmp{$d_G(u,v) \leq d$} and so there exists a $(u,v)$-path in $\mathsf{Rad}(G)$ of length at most $2d$, \bmp{as any jump between vertices that share a face~$f$ can be emulated in~$\mathsf{Rad}(G)$ by visiting the vertex representing~$f$}. 
Furthermore, each vertex from $F$ is adjacent in $\mathsf{Rad}(G)$ to some vertex from $V(G)$, so
$\mathsf{Rad}(G)$ has diameter at most $2d+2$.
By Lemma~\ref{lem:boundaried:dominating}, there exists a set $S_R \supseteq S$, $S_R \subseteq V(\mathsf{Rad}(G))$, so that $\mathsf{Rad}(G)[S_R]$ is connected and $|S_R| \le |S| \cdot \bmp{(2d+2)}$, and this set can be efficiently computed.
The claim follows by
taking $S' = S_R \cap V(G)$.
\end{proof}

We proceed with decomposing a boundaried plane graph of moderate diameter into \circum boundaried plane graphs by augmenting the boundary to be radially connected.
Recall \cref{def:prelim:circum} of a \circum boundaried plane graph $H$: all vertices of $\partial(H)$ must be embedded in the outer face of $H-\partial(H)$.

\begin{lemma}\label{lem:boundaried:to-outerplanar}
There is a polynomial-time algorithm that, given a boundaried plane graph $H$ of radial diameter at most $d$ and a (possibly empty) vertex set $Y \subseteq V(H) \setminus \partial(H)$,
outputs a vertex set $Y' \subseteq V(H) \setminus \partial(H)$ of size at most $\bmp{(2d+2)}\cdot (|Y| + |\partial(H)|)$, such that for each connected component $C$ of $H - (\partial(H) \cup Y')$ the boundaried plane graph $H \brb C$ is \circum. 
\end{lemma}
\begin{proof}
We apply \cref{lem:boundaried:radially-dominating} for $S =  \partial(H) \cup Y$,
to find a set $S' \supseteq \partial(H)$ of size at most $\bmp{(2d+2)}\cdot (|Y| + |\partial(H)|)$, which is radially connected in $H$.
Consider a connected component $C$ of $H - S'$
\mic{and the plane embedding of $H$ truncated to the subgraph induced by $C$.
Suppose there are at least two faces of $H[C]$ containing vertices from $N_H(C)$ (which is a subset of $S'$).}
Then there exists a cycle in $H-S'$ with at least one vertex of $S'$ in the interior and at least one vertex of $S'$ in the exterior.
This contradicts the fact that $S'$ is radially connected in $H$.
We infer that there is only one face of $H[C]$ which contains vertices from $N_H(C)$, 
hence the graph $H\brb C$ is \circum.
\mic{Note that we can always modify the embedding by swapping the outer face with any other chosen face. } 
The claim follows by taking $Y' = S' \setminus \partial(H)$.
%
\end{proof}

We are already close to the case with the boundary lying on the outer face.
There is however a subtle difference between \circum and \nice boundaried plane graphs.
If $H$ is \circum, then the edges incident to the boundary of $H$ may separate some vertices of $\partial(H)$ from the outer face of $H$. 

We shall employ \cref{lem:undeletable:treewidth-modulator} 
to decompose $H$ into boundaried subgraphs adjacent to at most two vertices from $\partial(H)$.
We want to say that such a boundaried subgraph $H\brb C$ becomes \nice after discarding these two special vertices.
\mic{This is done by exploiting minimality of a particular separator which guarantees that there exists a subgraph disjoint from $N_H[C]$, that can be made connected without violating planarity, adjacent to
each vertex in $N_H(C)$.
Contracting this subgraph gives the closure graph of $H\brb C$ and allows us to use the criterion from \cref{obs:prelim:single-faced}.}

In order to use
\cref{lem:undeletable:treewidth-modulator} we need a tree decomposition of width $k^{\Oh(1)}$.
Such a decomposition is computed in 
\cref{sec:grid} but later we apply several graph modifications in \cref{sec:diameter}, including protrusion replacement.
This might invalidate the tree decomposition and we need to compute it anew.

\begin{lemma}\label{lem:boundaried:to-single-faced}
There is a polynomial-time algorithm that, given an $r$-boundaried \circum plane graph $H$ of radial diameter at most $d$, \mic{such that $r \ge 3$},
outputs a vertex set $Y \subseteq V(H) \setminus \partial(H)$ of size at most $18 \cdot dr$, such that for each connected component $C$ of $H - (\partial(H) \cup Y)$
it holds that
\begin{enumerate}
    \item $|N_H(C) \cap \partial(H)| \le 2$, and
    \label{lem:boundaried:to-single-faced:item:partial}
    \item the boundaried graph $H \brb C - (N_H(C) \cap \partial(H))$ is \nice.
    \label{lem:boundaried:to-single-faced:item:single-face}
\end{enumerate}
\end{lemma}
\begin{proof}
The graph $H$ has at most $d$ outerplanarity layers, so its treewidth is at most $3d-1$~\cite{bodlaender1998partial}.\bmpr{Actually there is a poly-time algorithm constructing a decomposition of width at most~$3d-1$ \url{http://www.cs.uu.nl/research/techreps/repo/CS-1988/1988-14.pdf}. It never made it into a full paper, though. Several sources just cite this algorithm as existing and running in time~$\Oh(dn)$, but we can keep the text like this.}
We invoke the polynomial-time $\frac{3}{2}$-approximation algorithm for treewidth on planar graphs~\cite{SeymourR94}\footnote{The work~\cite{SeymourR94} focuses on computing branchwidth exactly on planar graphs and the $\frac{3}{2}$-approximation for treewidth follows from the close relation between treewidth and branchwidth.} to compute a tree decomposition $(T,\chi)$ of $H$ of width at most $\eta = \frac{3}{2}(3d-1)$.

We supply the computed tree decomposition to \cref{lem:undeletable:treewidth-modulator} for $S = \partial(H)$ and obtain a vertex set $Y \subseteq V(H) \setminus \partial(H)$ such that $|Y| \le 4(\eta + 1)\cdot r \le 18\cdot dr$ and
each connected component of $H - (\partial(H) \cup Y)$ has at most 2 neighbors in $\partial(H)$ and at most $2\eta$ neighbors in $Y$.
The latter property is not used in this proof. 

In order to ensure condition (\ref{lem:boundaried:to-single-faced:item:single-face}) of the lemma, we need to trim the set $Y$.
Consider the following reduction rule: while there exists a vertex $v \in Y$ such that $Y \setminus \{v\}$ still satisfies condition (\ref{lem:boundaried:to-single-faced:item:partial}), remove $v$ from $Y$.
Let $Y' \subseteq Y$ be the set obtained by exhaustive application of the reduction rule.
\mic{Observe that $Y' \ne \emptyset$ because otherwise the only connected component of $H - \partial(H)$ would have $r \ge 3$ neighbors in $\partial(H)$ which does not satisfy condition (\ref{lem:boundaried:to-single-faced:item:partial}).}

\micr{The definition forbids a \nice graph to have empty boundary. I added assumption $r\ge 3$ and propagated it}


Consider a connected component $C$ of $H - (\partial(H) \cup Y')$ and $v \in N_H(C) \cap Y'$.
Since $v$ has not been removed in the application of the reduction rule, there is a different connected component of $H - (\partial(H) \cup Y')$ which is adjacent to $v$ and to some $t \in \partial(H)$.
Let $H^1$ be obtained from $H$ by
(1) removing $N_H(C) \cap \partial(H)$, and (2)
replacing $\partial(H) \setminus N_H(C)$ with a single vertex $u$ adjacent to entire \bmp{set} $N_H(\partial(H))$.
This operation can be seen as inserting a vertex adjacent to a subset of vertices lying on the outer face of $H - \partial(H)$, so $H^1$ is planar.
By the observation above, the set $R_{H^1}(u, Y')$ is adjacent to each $v \in N_{H^1}(C) \cap Y'$ and it is disjoint from $C$.
Let $H^2$ be obtained from $H^1$ by contracting the connected vertex set $R_{H^1}(u, Y')$ into $u$. 
This graph contains $H \brb C - (N_H(C) \cap \partial(H))$ as a  boundaried subgraph $H^2 \brb C$, together with $u \not\in C$ adjacent to each vertex in $N_{H^2}(C)$.
Therefore, $H^2$ is the closure graph of $H \brb C - (N_H(C) \cap \partial(H))$ which, as $H^2$ is planar and connected, implies that $H \brb C - (N_H(C) \cap \partial(H))$ is \nice (see \cref{obs:prelim:single-faced}).
\end{proof}

\paragraph*{Retrieving solution preservers from boundaried decompositions}
We have shown how to decompose the bounded-diameter case into the \circum case, and the \circum case into the \nice case.
Now we want to use \cref{lem:inseparable:preserver} to ``backpropagate'' these reductions and retrieve the solution preservers.
We first state a crucial proposition which allows us to construct a solution preserver for a single-faced boundaried graph.

\begin{restatable}{proposition}{restInseparablePreserverNice}
\label{lem:inseparable:preserver-nice}
Let $A \subseteq V(G)$ be such that $G\brb A$ is a \nice boundaried graph.
Then there exists a set $S_A \subseteq A$ of size at most $|N_G(A)|^2$ which is a 168-preserver for $A$.
Furthermore, $S_A$ can be computed in polynomial time.
\end{restatable}

In order to keep the focus of the current section, we postpone the 
proof of \cref{lem:inseparable:preserver-nice} to the end of \cref{sec:preserver:inseparability-to-preserver} and first discuss its implications.
As the case of a \circum boundaried graph can be reduced to the case with a \nice boundaried graph via \cref{lem:boundaried:to-single-faced}, our framework yields \bmp{a} construction of solution preservers for \circum boundaried graphs.

This construction may require removal of at most two vertices from the obtained boundaried subgraph $G \brb A$ to make it \nice.
This however does not impose much trouble: if $S$ is a solution and $S \cap A$ is non-empty we can ``charge'' this part of the solution to pay for the removal of the two vertices, turning an $\alpha$-preserver into an $(\alpha+2)$-preserver.
Similarly we argue that when $|N_G(A)| \le 2$ then if $S \cap A$ is non-empty then we may just remove $N_G(A)$ so in this case \bmp{the} empty set is a 2-preserver for $A$.

\begin{lemma}\label{lem:inseparable:preserver-circum}
Consider a graph $G$, a vertex set  $A \subseteq V(G)$ such that \mic{$|N_G(A)| = r \ge 3$}, and a \circum embedding of $G\brb A$ of radial diameter at most $d$.
Then there exists a 170-preserver $S_A \subseteq A$ for $A$ of size $\Oh(d^3r^3)$. 
Furthermore, $S_A$ can be computed in polynomial time when \bmp{given} $G, A$, and the embedding of $G\brb A$.
\end{lemma}
\begin{proof}
We apply \cref{lem:boundaried:to-single-faced}
to compute a set $Y \subseteq A$ of size at most $18 \cdot dr$, such that for each connected component $C$ of $G \brb A - (N_G(A) \cup Y)$ there exists a set $D_C = N_G(C) \cap N_G(A)$ of size at most 2 so that the boundaried graph $G \brb C - D_C$ is \nice.
Let $\mathcal{C}$ be the family of connected components of $G \brb A - (N_G(A) \cup Y)$.
Then $(Y,\cc)$ \bmp{is} a boundaried decomposition of $G \brb A$.
We partition $\mathcal{C}$ into families $\mathcal{C}_{\le 2}$, comprising \bmp{those} components that have at most 2 neighbors, and $\mathcal{C}_{\ge 3}$, which contains all the remaining components.

For $C \in \mathcal{C}_{\le 2}$ a 2-preserver is given by \bmp{the} empty set: for any planar modulator $S \subseteq V(G)$ satisfying $|S \cap C| \ge 1$ we can replace $S \cap C$ with $N_G(C)$.
Consider now a component $C \in \mathcal{C}_{\ge 3}$.
Let $G' = G - D_C$;
we have that $G' \brb C = G\brb C - D_C$ is a \nice boundaried graph.
We apply \cref{lem:inseparable:preserver-nice} to compute a 168-preserver $S_C \subseteq C$ for $C$ in $G'$ of size at most $|N_H(C)|^2 = \Oh(d^2r^2)$.
We claim that  $S_C$ is a 170 preserver for $C$ in $G$.
To see this, consider a planar modulator $S \subseteq V(G)$ in $G$ satisfying $|S \cap C| \ge 1$.
Then $S_1 = S \setminus D_C$ is a planar modulator in $G'$.
By the definition of an $\alpha$-preserver, there exists a planar modulator $S_2$ in $G'$, such that (1) $S_2 \subseteq (S_1 \setminus C) \cup N_{G'}(C) \cup S_C$, (2) $S_2 \setminus N_{G'}[C] = S_1 \setminus N_{G'}[C]$, and (3) $|S_2| \le |S_1| + 168 \cdot |S_1 \cap C|$.
We set $S_3 = S_2 \cup \bmp{D_C}$, so $|S_3| \le |S| + 170 \cdot |S \cap C|$.
It is a planar modulator in $G$ and, since $D_C \subseteq N_G(C)$, it is a correct replacement for $S$.

By \cref{lem:inseparable:preserver} we obtain that $S_A = Y \cup \bigcup_{C \in \mathcal{C}_{\ge 3}} S_C$ is a 170-preserver for $A$.
We estimate $|\mathcal{C}_{\ge 3}| \le 2\cdot |N_G(A) \cup Y| = \Oh(d r)$ with \cref{lem:boundaried:components-degree-3} so the total size of $S_A$ is $\Oh(d^3r^3)$.
\end{proof}

Finally, we employ the reduction from \cref{lem:boundaried:to-outerplanar} to extend the construction above to the case of boundaried plane graphs with bounded diameter.
The following proposition encapsulates all the results about solution preservers into a statement to be used in the main proof.

\begin{proposition}\label{lem:inseparable:preserver-diameter}
Consider a graph $G$, a vertex set $A \subseteq V(G)$ such that $|N_G(A)| = r$, and a plane embedding of $G\brb A$ of radial diameter at most $d$.
Then there exists a 170-preserver $S_A \subseteq A$ for $A$ of size  $\Oh(d^7r^4)$.  
Furthermore, $S_A$ can be computed in polynomial time when $G, A$, and the embedding of $G\brb A$ are given.
\end{proposition}
\begin{proof}
Let $H$ be a boundaried plane graph of  radial diameter at most $d$ given by the embedding of $G\brb A$.
We apply \cref{lem:boundaried:to-outerplanar} to $H$ with $Y = \emptyset$ to compute a set $Y'$ of size at most $\bmp{(2d+2)}\cdot r$ such that for each connected component $C$ of $H - (\partial(H) \cup Y')$ the  boundaried plane graph $H \brb C$ is \circum.
Let $\mathcal{C}$ be the family of connected components of $H - (\partial(H) \cup Y)$.
We partition $\mathcal{C}$ into families $\mathcal{C}_{\le 2}$, comprising \bmp{those} components that have at most 2 neighbors, and $\mathcal{C}_{\ge 3}$, which contains all the remaining components. \bmp{We proceed similarly as in \cref{lem:inseparable:preserver-circum}.}

For $C \in \mathcal{C}_{\le 2}$ a 2-preserver is given by \bmp{the} empty set: for any planar modulator $S \subseteq V(G)$ satisfying $|S \cap C| \ge 1$ we can replace $S \cap C$ with $N_G(C)$.
For each component $C \in \mathcal{C}_{\ge 3}$
we have that $|N_H(C)| \le |Y \cup \partial(H)| = \Oh(dr)$.
Furthermore, the radial diameter bound from $G \brb A$ transfers to $G\brb C$. 
We take advantage of \cref{lem:inseparable:preserver-circum} to compute a 170-preserver $S_C \subseteq C$ of size $\Oh(d^6r^3)$.

By \cref{lem:inseparable:preserver} we obtain that $S_A = Y \cup \bigcup_{C \in \mathcal{C}_{\ge 3}} S_C$ is a 170-preserver for $A$.
We estimate $|\mathcal{C}_{\ge 3}| \le 2\cdot |N_G(A) \cup Y| = \Oh(d r)$ with \cref{lem:boundaried:components-degree-3} so the total size of $S_A$ is $\Oh(d^7r^4)$.
\end{proof}

\subsection{Inseparability in single-faced boundaried graphs}

This section, together with the following one, are devoted to proving \cref{lem:inseparable:preserver-nice}.
To this end, we first develop a framework of inseparability in \nice boundaried graphs and reduce computing a solution preserver to the case where the given \nice boundaried graph is inseparable.
In the next section, we will focus on the inseparable case.

Given a \nice boundaried plane graph $H$ we can consider a natural cyclic ordering of $\partial(H)$.
Since \bmp{the vertices of} $\partial(H)$ lie on the outer face, we can augment the embedding by inserting a special vertex $u_H$ adjacent to all vertices in $\partial(H)$.
This operation turns $H$ into its 
closure graph $\widehat H$ (see \cref{obs:prelim:single-faced}).
After this operation we get a unique cyclic ordering of the neighbors of the special vertex $u_H$, which are the vertices from $\partial(H)$.
\mic{When considering a \nice boundaried plane graph $H$, we always implicitly assume a fixed cylic ordering of $\partial(H)$, given by some embedding of~$\widehat H$}. 
For such an ordering, we \mic{
define a \emph{canonical division} as a partition $\partial (H) = T_1 \uplus T_2$ such that $T_1,T_2$ are non-empty connected segments in the cyclic ordering of $\partial(H)$.}
We say that sets $T_1, T_2 \subseteq T$ are \emph{non-crossing} if they are non-empty and there exists a canonical division $(T'_1, T'_2)$
such that $T'_1 \supseteq T_1$, $T'_2 \supseteq T_2$.

\begin{observation}\label{obs:inseparable:canonical}
Consider a \nice boundaried plane graph $H$ and a canonical division $(T_1, T_2)$ of $\partial (H)$.
Let $H''$ be obtained by adding two new vertices $t_1, t_2$, adjacent to $T_1$ and $T_2$, respectively.
Then $H''$ is planar.
\end{observation}


\bmp{The following definition plays a key role in this section.}

\begin{definition} \label{def:cseparation}
A $c$-separation $(T_1,T_2,S)$ in a boundaried graph $H$ is given by an unrestricted $(T_1, T_2)$-separator $S$, such that $T_1, T_2 \subseteq \partial(H)$ are non-empty and disjoint, and $c\cdot |S| \le \min(|T_1|, |T_2|)$.
We say that a boundaried graph $H$ is \unbreak if it does not admit a $c$-separation.

Given a \nice  boundaried plane graph $H$, 
we say that a $c$-separation $(T_1,T_2,S)$ is canonical if $(T_1, T_2)$ is a canonical division of $\partial (H)$ and $S$ is an inclusion-wise minimal unrestricted $(T_1, T_2)$-separator.
\end{definition}

Our current goal is to reduce the task of finding solution preservers into the case when the given \nice boundaried graph is \unbreak for some constant $c$.
As long as some $c$-separation exists, ideally we would like to find it and recurse on the both sides of the separator.
It is however unclear whether a $c$-separation can be detected in polynomial time or whether the obtained boundaried subgraphs are always \nice. 
First we show that when we can separate some non-crossing subsets of the boundary with a small separator, we can do the same for some canonical division. 

\begin{lemma}\label{lem:inseparable:non-crossing}
Let $H$ be a \nice boundaried plane graph and $T_1, T_2 \subseteq \partial(H)$ be non-crossing.
Suppose there exists an unrestricted $(T_1,T_2)$-separator $S$,
such that  $c\cdot |S| \le \min(|T_1|, |T_2|)$.
Then $H$ admits a canonical $c$-separation.
\end{lemma}
\begin{proof}
{Let~$T'_1 \supseteq T_1, T'_2 \supseteq T_2$ be disjoint segments of $\partial(H)$ witnessing the non-crossing property.}
Let $H''$ be obtained from $H$ by adding a new vertex $t_1$ adjacent to all $v \in T_1$ and a new vertex $t_2$ adjacent to all $v \in T_2$. 
Since $T_1,T_2$ are non-crossing, $H''$ is planar by \cref{obs:inseparable:canonical}. \bmp{We can assume that each vertex~$v \in T'_1$ either belongs to~$T_1$ or to~$R = R_{H''}(t_2, S) \setminus \{t_2\}$, since if neither holds then we can add~$v$ to~$T_1$ while preserving the preconditions to the lemma.}

\bmp{We introduce some terminology for the rest of the proof. Suppose} $V_1, V_2 \subseteq T_1$ \bmp{are} two disjoint segments of terminals.
\bmp{Then the} set $\partial(H) \setminus (V_1 \cup V_2)$ comprises of two segments $U_1, U_2$ ($U_1$ possibly empty) such that $T_2 \subseteq U_2$.
We say that a vertex $v \in \partial(H)$ lies \emph{between} $V_1$ and $V_2$ when $v \in U_1$.


We divide $T_1$ into a sequence of \bmp{(maximal)} segments $Q_1, \dots, Q_m$, ordered with respect to the embedding, such that the boundary vertices \emph{between} $Q_i$ and $Q_{i+1}$ belong to $R$.
Let $C_i$ be the set of vertices reachable from any $v \in Q_i$ in $H''$  without going through $R$ or $t_1$, that is,
$C_i = R_{H''}(Q_i, R \cup \{t_1\})$.

We claim that the sets $(C_i)_{i=1}^m$ are pairwise disjoint. 
\mic{Assume otherwise that $C_i \cap C_j \ne \emptyset$ for some $i < j$.
Then there exist vertices $v_1 \in Q_i$ and $v_3 \in Q_j$ connected by a path~\bmp{$P_{1,3}$} in $V(H) - R$.
Let $v_2 \in T'_1 \cap R$ be some vertex lying between  $Q_i$ and $Q_j$.
}
Since $v_2 \in R$, there exists a path \bmp{$P_{2,4}$} connecting $v_2$ to some vertex $v_4 \in T'_2$ such that $V(P) \subseteq  R$.
Note that the numbering of vertices $v_1, v_2, v_3, v_4$ obeys the cyclic ordering of terminals given by the \nice embedding. 
Let $H_0$ be a graph obtained from $H$ by inserting edges $v_1v_2, v_2v_3, v_3v_4, v_4v_1$ and a new vertex $v_5$ adjacent $\{v_1, v_2, v_3, v_4\}$.
By the choice of the ordering, $H_0$ is planar.
\mic{The path~\bmp{$P_{1,3}$} within $V(H) \setminus R$ connecting $v_1$ and $v_3$ is disjoint from the path~\bmp{$P_{2,4}$} connecting $v_2$ and $v_4$ within $R$.}
By contracting the \mic{paths $P_{1,3}$, $P_{2,4}$ to edges}
we obtain a minor model of $K_5$ in $H_0$, which is \bmp{a}  contradiction with the assumption that \mic{$C_i \cap C_j \ne \emptyset$}.

\bmp{Using the disjointness property we finish the proof.} Let $S_i = S \cap C_i$ for $i \in [m]$.
The sets $(S_i)_{i=1}^m$ are pairwise disjoint and the sets $(Q_i)_{i=1}^m$ sum up to $T_1$, hence \bmp{the assumption~$c \cdot |S| \leq \min(|T_1|,|T_2|)$ implies that $c \cdot |S_i| \le |Q_i|$ for some $i \in [m]$.}
The set $S_i$ is an unrestricted $(Q_i, R)$-separator in $H$ and an unrestricted $(Q_i, Q_j)$-separator in $H$ for each $j \ne i$ because \bmp{$C_i \cap C_j = \emptyset$}.
Therefore, $S_i$ is an unrestricted $(Q_i, \bmp{\partial(H)} \setminus Q_i)$ separator.
Since $c \cdot |S_i| \le c \cdot |S| \le |T_2| \le |\partial(H) \setminus Q_i|$, we obtain a canonical $c$-separation.
\end{proof}

We now show that \bmp{the} existence of a $(2c)$-separation implies \bmp{the} existence of a $c$-separation of non-crossing subsets, which \bmp{in turn} implies \bmp{the} existence of a canonical $c$-separation.

\begin{lemma}\label{lem:inseparable:any-separator}
Let $H$ be a \nice boundaried graph, $T_1, T_2 \subseteq \partial(H)$ be disjoint and non-empty,
and $S$ be an unrestricted $(T_1,T_2)$-separator,
such that  $2c\cdot |S| \le \min(|T_1|, |T_2|)$.
Then $H$ admits a canonical $c$-separation
with respect to any given \nice embedding.
\end{lemma}
\begin{proof}
Let us fix a \nice embedding of $H$.
Due to Lemma~\ref{lem:inseparable:non-crossing}, it suffices to show that $S$ is an unrestricted $(Q_1,Q_2)$-separator for some non-crossing $Q_1, Q_2 \subseteq \partial(H)$, such that $c\cdot |S| \le \min(|Q_1|, |Q_2|)$.

{Since~$|T_2| \geq 2c\cdot |S|$, we can find two segments of $\partial(H)$ that each have at least~$c \cdot |S|$ vertices from~$T_2$.} Hence there exists a canonical division $(Q'_1, Q'_2)$ of $\partial(H)$, such that 
$c\cdot |S| \le \min(|Q'_1 \cap T_2|, |Q'_2 \cap T_2|)$.
{As the segments~$(Q'_1,Q'_2)$ partition~$T_1 \subseteq \partial(H)$ and~$|T_1| \geq 2c \cdot |S|$,} we have $c\cdot |S| \le |Q'_1 \cap T_1|$ or $c\cdot |S| \le |Q'_2 \cap T_1|$.
Suppose this inequality holds for \bmp{$Q'_i$}, $i \in \{1,2\}$.
Let $Q_1 = Q'_i \cap T_1$ and $Q_2 = T_2 \setminus Q'_i$.
Since $Q_1 \subseteq T_1$ and $Q_2 \subseteq T_2$, the set $S$ is an unrestricted $(Q_1,Q_2)$-separator.
\end{proof}

As the number of canonical \bmp{divisions} is quadratic in the size of the boundary, we can find a canonical $c$-separation in polynomial time.
Together with \cref{lem:inseparable:any-separator}, this yields a 2-approximation algorithm for finding separations.

\begin{lemma}\label{lem:inseparable:compute-separation}
There is a polynomial-time algorithm that, given a \nice boundaried graph $H$ and an integer $c>0$, either correctly concludes that $H$ is $(2c)$-inseparable or finds a canonical $c$-separation with respect to some \nice embedding of $H$.
\end{lemma}
\begin{proof}
Let us start with constructing a~\nice embedding of $H$: it can be done in polynomial time since it suffices to construct any plane embedding of the closure of $H$ (see \cref{obs:prelim:single-faced}).
By \cref{lem:inseparable:any-separator}, a~\nice boundaried graph either is $(2c)$-inseparable or admits a canonical $c$-separation with respect to any \nice embedding.
We enumerate all $|\partial(H)|^2$-many canonical divisions with respect to the constructed embedding.
For each canonical division $(T_1,T_2)$ we check if there exists $S$ such that $(T_1,T_2,S)$ is a canonical $c$-separation.
This task reduces to adding two new vertices $t_1,t_2$, adjacent respectively to $T_1$ and $T_2$, and finding a minimum $(t_1,t_2)$-separator.
If for some  canonical division $(T_1,T_2)$ we find $S$ such that $c\cdot |S| \le \min(|T_1|, |T_2|)$, we can return it: it must be an inclusion-wise minimal unrestricted $(T_1,T_2)$-separator.
Otherwise we report that $H$ is $(2c)$-inseparable.
\end{proof}

Whenever we find a canonical $c$-separation we want to augment the boundary by marking the vertices in the separator
to split the given \nice boundaried graphs into two smaller \nice boundaried graphs.
We will add the marked vertices to the solution preserver and we proceed recursively.
A priori this construction may mark a large number of vertices so we need some kind of a measure to ensure that the constructed solution preserver cannot be too large.
Let $t$ be the size of the original boundary, $s$ be the size of the separator, and $t_1,t_2$ be the boundary sizes of the two smaller \nice boundaried graphs on which we recurse.
Because we work with a canonical $c$-separation we have that $c\cdot s \le \min(t_1,t_2)$.
This implies that even \bmp{when} $t_1 + t_2 > t$ these numbers are in some sense balanced.
It turns out that it suffices to take $c=4$ to ensure that $t_1^2 + t_2^2 < t^2$ which provides us with a measure that decreases in each step and governs the number of possible recursive splits. 

\begin{lemma}\label{lem:inseparable:inequality}
Let $s \ge 0,\, t > 0,\, t_1>0,\, t_2>0$ be integers satisfying $t_1 \ge 4s, t_2 \ge 4s,\, t + 2s \ge t_1 + t_2$. Then $t^2 \ge t_1^2 + t_2^2 + s + 1$.
\end{lemma}
\begin{proof}
First consider the case $s = 0$.
We have $t \ge t_1 + t_2 > 0$ so $t^2 \ge (t_1 + t_2)^2 = t_1^2 + t_2^2 + 2t_1t_2 \ge  t_1^2 + t_2^2 + 1$ because $t_1, t_2$ are positive.

Now consider the case $s \ge 1$.
As $t_1 + t_2 - 2s \ge 6s \ge 0$, the inequality $t \ge t_1 + t_2 - 2s$ is maintained with both sides squared.
We enroll the squared expression, rearrange the terms, and bound $4s^2 \ge s + 1$.
\begin{eqnarray*}
t^2 &\ge& (t_1 + t_2 - 2s)^2 = t_1^2 + t_2^2 + 4s^2 + 2t_1t_2 - 4st_1 - 4st_2  \\
&=& t_1^2 + t_2^2 + 4s^2 + {t_1}(t_2 - 4s) + {t_2}(t_1 - 4s) \ge t_1^2 + t_2^2 + s + 1.
\end{eqnarray*}
\bmp{This concludes the proof.}
\end{proof}

Now we show that when splitting a \nice boundaried graph along a canonical separation, we indeed obtain two smaller \nice boundaried graphs.
We do not construct their embeddings directly but rather refer to the convenient criterion from \cref{obs:prelim:single-faced}.
Note that the definition of a \nice boundaried graph does not require it to be connected although connectivity will later follow from being \unbreak.

\begin{lemma}\label{lem:inseparable:decomposition-step}
Let $H$ be a \bmp{\nice} boundaried plane graph
and let $(T_1, T_2, S)$ be a canonical 4-separation.
Then there exist sets $V_1, V_2 \subseteq V(H)$, so that 
\begin{enumerate}
    \item $(V_1, V_2)$ is a partition of $V(H) \setminus (\partial(H) \cup S)$,
    \item $N_H(V_1) \subseteq T_1 \cup S$ and $N_H(V_2) \subseteq T_2 \cup S$,
    \item $H\brb {V_1}, H\brb {V_2}$ are \nice boundaried graphs,
    \item $|N_H(V_1)|^2 + |N_H(V_2)|^2 +  |S| + 1 \le |\partial(H)|^2$.
\end{enumerate}
Furthermore, the sets $V_1, V_2$ can be found in polynomial time, when $H$ and $(T_1, T_2, S)$ are given.
\end{lemma}
\begin{proof}
If $|S| = 0$, then $T_1, T_2$ are disconnected in $H$.
We can group the vertices of $H - \partial(H)$
into two sets $V_1, V_2$, where $N_H(V_i) \subseteq T_i$ for $i \in \{1,2\}$.
Since $T_1,T_2$ are non-empty \bmp{by Definition~\ref{def:cseparation}}, the numbers $0,|\partial(H)|, |T_1|, |T_2|$, satisfy the prerequisites of Lemma~\ref{lem:inseparable:inequality}, and so $|N_H(V_1)|^2 + |N_H(V_2)|^2 + |S| + 1 \le |\partial(H)|^2$. \bmp{Since~$H$ is \nice, the same is easily seen to hold for~$H \brb {V_1}$ and~$H \brb {V_2}$.}

Let us assume $|S| \ge 1$ for the rest of the proof.
Consider the connected components of $H - (\partial(H) \cup S)$.
None of them can have neighbors in both $T_1$ and $T_2$.
Therefore we can group \bmp{the} vertices of $V(H) \setminus (\partial(H) \cup S)$ into two sets $V_1, V_2$, where $V_i$ is the union of \bmp{the} components with neighbors in $T_i \cup S$.
If a component has neighbors only in $S$, we can choose its allocation arbitrarily.

We are going to show that $H\brb {V_1}$ is \nice---the proof for $H\brb {V_2}$ is symmetric.
Note that $N_H(V_1) \subseteq T_1 \cup S$.
Let $H_0$ be obtained from $H$ by adding two new vertices $t_1, t_2$ adjacent to respectively $T_1$ and $T_2$.
Then $H_0$ is planar and \bmp{by Definition~\ref{def:cseparation}} $S$ is an inclusion-wise minimal $(t_1,t_2)$-separator in $H_0$.
Therefore, each vertex from $S$ is adjacent to both $R_{H_0}(t_1, S)$ and $R_{H_0}(t_2, S)$.
Let $H_1$ be obtained from $H_0$ by
\begin{enumerate}
    \item adding an edge $t_1t_2$---this preserves  planarity, 
    \item contracting $R_{H_0}(t_2, S)$ into $t_2$,
    \item contracting the edge $t_1t_2$.
\end{enumerate}
The graph $H_1$ is planar and contains $H\brb {V_1}$ as a boundaried subgraph $H_1 \brb {V_1}$, plus one additional vertex adjacent to every vertex in $T_1 \cup S \supseteq N_{H_1}(V_1)$.
Also, $H_1$ is a contraction of the closure $\widehat{H}$ of $H$, hence it is planar and connected.
Finally, let us remove from $H_1$ all the vertices of $(T_1 \cup S) \setminus N_{H_1}(V_1)$, \bmp{which trivially} preserves planarity and \bmp{preserves} connectivity \bmp{through~$t_1$}.
\bmp{The resulting graph} 
is isomorphic to the closure of $H\brb {V_1}$.
We  have thus obtained that the closure of $H\brb {V_1}$ is planar and connected so by \cref{obs:prelim:single-faced} the boundaried graph $H$ is \nice.

Finally, let us prove the inequality.
Since  $(T_1, T_2, S)$ is a 4-separation, it holds that $\min(|T_1|,|T_2|) \ge 4 \cdot |S|$.
As $|T_1 \cup S| + |T_2 \cup S| \le |\partial(H)| + 2\cdot |S|$, the numbers $|S|,|\partial(H)|,|T_1 \cup S|,|T_2 \cup S|$ satisfy the prerequisites of Lemma~\ref{lem:inseparable:inequality}, and so $|T_1 \cup S|^2 + |T_2 \cup S|^2 + |S| + 1\le |\partial(H)|^2$.
Since $N_H(V_i) \subseteq T_i \cup S$ for $i \in \{1,2\}$, the claim follows.
\end{proof}

We are ready to decompose a \nice boundaried graph $H$ into boundaried subgraphs which are \nice and 8-inseparable.
We define a \nice decomposition of $H$ as a boundaried decomposition $(Y, \mathcal{C})$ of $H$ (recall \cref{def:inseparable:boundaried-decomposition}), such that for each $C \in \mathcal{C}$ it holds that $H \brb C$ is \nice.

\begin{lemma}\label{lem:inseparable:decomposition-full}
There is a polynomial-time algorithm that, given a \nice boundaried graph $H$,
returns its \nice decomposition $(Y, \mathcal{C})$, so that  $|Y| \le |\partial(H)|^2$ and for each $C \in \mathcal{C}$ it holds that $H \brb C$ is 8-inseparable.
\end{lemma}
\begin{proof}
We decompose $H$ recursively starting with a \nice decomposition $(Y_0, \mathcal{C}_0) = (\emptyset, \{V(H) \setminus \partial(H)\})$.
Let us define \bmp{the} measure $\mu(Y, \mathcal{C}) = |Y| + \sum_{C \in \mathcal{C}} |N_H(C)|^2$.
We have $\mu(Y_0, \mathcal{C}_0) = |\partial(H)|^2$.
We are going to refine the \nice decomposition until all the boundaried graphs $H \brb C$, for $C \in \mathcal{C}$, are 8-inseparable, in each step decreasing the measure $\mu$.

Let $(Y_i, \mathcal{C}_i)$ be a \nice decomposition of $H$.
We use \cref{lem:inseparable:compute-separation} with $c=4$ to either conclude that all $H \brb C$, for $C \in \mathcal{C}$, are 8-inseparable (then the procedure terminates)
or to find a canonical 4-separation $(T_1,T_2,S)$ with respect to some \nice embedding of $H\brb C$ for some $C \in \mathcal{C}$. 
Now we apply \cref{lem:inseparable:decomposition-step} to $H \brb C$ and $(T_1,T_2,S)$.
We obtain sets $V_1, V_2 \subseteq C \setminus S$,
which form a partition of $C \setminus S$ and satisfy $N_H(V_i) \subseteq N_H(C) \cup S$.
Furthermore, each $H\brb {V_i}$ is \nice.
We define a~new \nice decomposition $(Y_{i+1}, \mathcal{C}_{i+1})$, where $Y_{i+1} = Y_i \cup S$ and $\mathcal{C}_{i+1}$ is obtained from $\mathcal{C}_{i}$ by replacing $C$ with $V_1$ and $V_2$.
\cref{lem:inseparable:decomposition-step} guarantees that $|N_H(V_1)|^2 + |N_H(V_2)|^2 + |S|  + 1\le |N_H(C)|^2$, therefore $\mu(Y_{i+1}, \mathcal{C}_{i+1})  \le \mu(Y_{i}, \mathcal{C}_{i}) - 1$. 

In each step the measure $\mu(Y_i, \mathcal{C}_i)$ decreases, so the algorithm terminates after at most $ \mu(Y_0, \mathcal{C}_0) = |\partial(H)|^2$ steps.
Let $(Y_\ell, \mathcal{C}_\ell)$ be the final \nice decomposition.
We have $|Y_\ell| \le \mu(Y_\ell, \mathcal{C}_\ell) \le \mu(Y_0, \mathcal{C}_0) = |\partial(H)|^2$.
For each $C \in \mathcal{C}_\ell$, the boundaried graph $H \brb C$ is \nice (by the invariant of a \nice decomposition) and 8-inseparable (because the algorithm has terminated).
\end{proof}

\subsection{From inseparability to a solution preserver}
\label{sec:preserver:inseparability-to-preserver}

In this section we show when $G\brb A$ is \nice and $c$-inseparable then the empty set forms an $\Oh(1)$-preserver for $A$.
In other words, for any solution $S$ each vertex from $S\cap A$ can be replaced by $\Oh(1)$ vertices in $N_G(A)$.
First we need several observation about balls in the metric given by the weak radial distance.
Recall that for a plane graph $G$ and $u,v \in V(G)$ the weak radial distance $w_G(u,v)$ is one less than the minimum length of a sequence of vertices that
starts at $u$, ends in $v$, and in which every two consecutive vertices either share an edge or lie on a common \emph{interior} face. Recall also that $\ww^\ell_G(X)$ denotes the ball of weak radial distance at most~$\ell$ around vertex set~$X$ \mic{(see Figure \ref{fig:radial-ball} for an illustration)}. 
The reason why we consider the weak radial distance instead of just the radial distance (allowing to make paths also through the outer face) becomes evident in the following proof.

\begin{lemma}\label{lem:inseparable:halo-non-crossing}
Let $H$ be a single-faced boundaried plane graph and let $X \subseteq V(H)$.
Suppose that $T_1, T_2 \subseteq \partial(H)$ are non-crossing.
If $\ww^\ell_H(X)$ is an unrestricted $(T_1, T_2)$-separator, then there exists $X' \subseteq \ww^\ell_H(X)$ such that $X'$ is an unrestricted $(T_1, T_2)$-separator and $|X'| \le (2\ell + 1)\cdot |X|$.
\end{lemma}
\begin{proof}
Let $H''$ be obtained from $H$ by adding a new vertex $t_1$ adjacent to all $t$ in $T_1$ and a new vertex $t_2$ adjacent to all $t$ in $T_2$. 
Since $T_1,T_2$ are non-crossing, $t_1,t_2$ can be inserted without modifying the embedding of $H$, within the outer face.
Let us abbreviate $X^i = \ww^i_H(X)$.
By the assumption, $X^\ell$ is a \bmp{$(T_1,T_2)$}-separator in $H''$.

Consider a directed graph $R$ with vertex set $X^\ell$ defined as follows.
For each $x \in X^i \setminus X^{i-1}$, $i \in [\ell]$, we add a single edge from $x$ to any vertex $y \in X^{i-1}$ such that $x,y$ \bmp{share} an interior face or $xy \in E(H)$.
Such vertex $y$ always exists by the definition of the metric $w_H$.
Since any $x \in (X^\ell \setminus X)$ has out-degree 1 in $R$, we can define a function
$r \colon (X^\ell \setminus X) \to X^\ell$ given by the arcs of~$R$.

Let $X' \subseteq X^\ell$ be defined by the following procedure.
Initialize $X' = X^\ell$ and while there exists $x \in X' \setminus X$ of in-degree 0 in $R$, such that $X' \setminus x$ is still a \bmp{$(T_1,T_2)$}-separator in $H''$, remove $x$ from $X'$ and from $R$.
Let $C_1, C_2$ be the connected components in $H'' - X'$ containing \bmp{$T_1$} and \bmp{$T_2$}, respectively.

We shall prove that $|X'| \le (2\ell+1)\cdot|X|$.
Let $R_x \subseteq X'$ be the set of vertices $y \in X'$ admitting a path $y \to x$ in $R$, including $x$.
Fix $x \in X$ and suppose that there are 3 vertices $x_1, x_2, x_3 \in R_x$, each of in-degree 0 in $R$. 
Since none of these 3 vertices were removed during construction of $X'$, each of them is connected to both $C_1$ and $C_2$ in \mic{$H''$}.
Let $H_0$ be obtained from $H''$ be adding a new vertex $u_f$ for each interior face of $H$ and connecting $u_f$ to all vertices lying on this face.
Recall that $H''$ was obtained by adding two vertices inside the outer face, so they do not interfere with inserting vertices within interior faces.
Therefore, $H_0$ is still a planar graph.
Let $Y \subseteq V(H_0)$ be obtained from $R_x \setminus \{x_1, x_2, x_2\}$ by adding all the adjacent face-vertices.
Since $x$ and $r(x)$ either share an edge or an interior face, the graph $H_0[Y]$ is connected.
The vertices $x_1, x_2, x_3$ and the sets $Y, C_1, C_2$ form branch sets of a $K_{3,3}$ minor in $H_0$, which is a contradiction to the assumption that such vertices $x_1, x_2, x_3$ exist.

We infer that $R_x$ can contain at most 2 vertices of in-degree 0 in $R$, and because the out-degree in $R$ is at most 1, we obtain that
$|R_x \cap X^i| \le 2$ for $i \in [\ell]$ and so
$|R_x| \le 2\ell + 1$.
Since the sets $R_x$ sum up to $X'$, it follows that $|X'| \le (2\ell+1)\cdot|X|$.
%
\end{proof}

\cref{lem:inseparable:halo-non-crossing} implies that for a set $X$ and integer $\ell$, the set $\ww_H^\ell(X)$ behaves similarly as a set of size $(2\ell+1) \cdot |X|$ with respect to separating non-crossing subsets of the boundary.
The following lemma deepens this analogy: we show
that in a $c$-inseparable boundaried plane graph $H$, if the number of boundary vertices is large compared to the size of some set~$X$ and integer $\ell$, then there must be unique component of~$H - \ww_H^\ell(X)$ that contains almost all boundary vertices and all the other components have at most~$\Oh(c \cdot \ell \cdot d)$ of them. 

\begin{lemma}\label{lem:inseparable:large-component}
Let $H$ be a \unbreak single-faced $r$-boundaried plane graph, $\ell$ be a positive integer, and let $X \subseteq V(H)$ be of size $d$.
Suppose that $3c \cdot (2\ell + 1) \cdot d\le r$. 
Then there is a single connected component in $H - \ww_H^\ell(X)$ containing at least $r - c \cdot (2\ell+1)\cdot d$ vertices from $\partial(H)$.
\end{lemma}
\begin{proof}
Let $X' = \ww_H^\ell(X)$.
We will consider two cases depending on the maximal number of boundary vertices in a single component of $H - X'$.
First let us assume that each such component contains at most $c \cdot (2\ell + 1) \cdot d$ elements from $\partial(H)$ or $\partial(H) \subseteq X'$.
Consider the cyclic ordering of $\partial(H)$ given by the \nice embedding.
Let us choose an arbitrary element $t_1 \in \partial(H)$ and fix the ordering $t_1, \dots, t_r$ of $\partial(H)$ to start at this element.

Let $S_i = \{t_i\}$ if $t_i \in X'$ and otherwise let $S_i = R_H(t_i, X') \cap \partial(H)$
be the subset of $\partial(H)$ reachable from $t_i$ in $H - X'$.
Let $j$ be the smallest index for which $|\bigcup_{i=1}^j S_i| \ge c \cdot (2\ell + 1) \cdot d$.
By assumption it holds that $|S_j| \le c \cdot (2\ell + 1) \cdot d$, so we have
$|\bigcup_{i=1}^j S_i| \le 2c \cdot (2\ell + 1) \cdot d$.
\mic{We define $T_1 = \{t_i \mid i \in [j]\}$} 
and $T_2 = T \setminus \bigcup_{i=1}^j S_i$.
Then $T_1, T_2$ are disjoint, non-crossing, and $X'$ is an unrestricted $(T_1, T_2)$-separator.
\mic{Lemma~\ref{lem:inseparable:halo-non-crossing} guarantees that there exists an unrestricted $(T_1, T_2)$-separator $X'' \subseteq X'$ of size at most $(2\ell + 1) \cdot d$.
Note that for each $i \in [j]$ satisfying $t_i \not\in X'$, the set $S_i$ is contained in a single connected component of $H-X'$, hence \bmp{it} is also contained in a single connected component of $H-X''$ which is disjoint from $T_2$.
Therefore $X''$ is also an unrestricted $(\bigcup _{i=1}^j S_i, T_2)$-separator.
By the construction, $|\bigcup _{i=1}^j S_i| \in [c \cdot (2\ell + 1) \cdot d,\, 2c \cdot (2\ell + 1) \cdot d]$.
The lower bound $3c \cdot (2\ell + 1) \cdot d\le r$ implies that $|T_2| = r - \bigcup _{i=1}^j S_i \ge c \cdot (2\ell + 1) \cdot d$,
which contradicts the assumption that $H$ is \unbreak.}

Suppose now that there exists a connected component $C$ of $H - X'$ so that $|C \cap \partial(H)| > c\cdot (2\ell + 1) \cdot d$.
Let $T_1 = \{t\}$ for an arbitrary $t \in C \cap \partial(H)$ and $T_2 = \partial(H) \setminus C$.
Clearly $X'$ is an unrestricted $(T_1, T_2)$-separator and $T_1,T_2$ are non-crossing.
By Lemma~\ref{lem:inseparable:halo-non-crossing} we obtain an unrestricted $(T_1, T_2)$-separator $X''$ such that  $|X''| \le (2\ell + 1) \cdot d$ and $X'' \subseteq X'$.
The latter property implies that $X''$ is also  an unrestricted $(\partial(H) \cap C, T_2)$-separator.
Since $H$ is \unbreak, \bmp{we have} 
$|T_2| <  c\cdot (2\ell + 1) \cdot d$.
Therefore $|C \cap \partial(H)| = r - |T_2| > r - c\cdot (2\ell + 1) \cdot d$, and so $C$ is the desired component.
\end{proof}

When $G\brb A$ is \nice and $c$-inseparable, then for a solution $S$ in $G$ we consider the set $\ww_{G\brb A}^3(S \cap A)$ after fixing some \nice embedding of $G \brb A$.
This will allow us to construct 3 connected separators between any vertex $x \in S \cap A$ and the large component from \cref{lem:inseparable:large-component} in order to argue for their irrelevance in the solution.
To construct such separators we rely on the following proposition.
A plane triangulation is a plane graph with every face being incident to exactly three vertices. \bmp{Recall that when using the word \emph{separator} without any adjective, we refer to a \emph{restricted} separator that does not intersect the elements being separated.}

\begin{proposition}[{\cite[Prop. 8.2.3]{mohar2001graphs}}]
\label{lem:inseparable:traingulation}
Let $u$ and $v$ be distinct vertices in a plane triangulation $G$.
If $S$ is a minimal $(u,v)$-separator, then $G[S]$ is an induced, non-facial cycle.
\end{proposition}

We remark that this result has been originally formulated for sets whose removal disconnect the graph
but the same argument holds for minimal $(u,v)$-separators.

\begin{lemma}\label{lem:inseparable:halo-irrelevant}
Let $G$ be a \mic{connected} plane graph, $x \in X \subseteq V(G)$, and
$v \in V(G) \setminus \ww_G^\ell(X)$ {for some~$\ell \in \mathbb{N}$}.
Then there exists either \bmp{a} vertex $y \in \ww_G^\ell(X) \setminus X$ which separates $v$ from $x$ or
$\ell$ \bmp{vertex-disjoint subsets of $\ww_G^\ell(X) \setminus X$} forming a sequence of nested $(v,x)$-separators. 
In the latter case, each separator is inclusion-wise minimal and either induces a cycle or a path with both endpoints \bmp{on} the outer face. 
\end{lemma}
\begin{proof}
\mic{Consider the block-cut tree of~$G$ and the sequence of vertices $u_1, u_2, \dots, u_m$, where $u_1 = x$, $u_m = v$, and $u_2, \dots, u_{m-1}$ are the articulation points between $x$ and $v$ in $G$.}  \bmp{If} $x,v$ belong to the same biconnected component, then $m=2$.
\mic{More precisely, $\{u_i\}$ is a $(v,x)$-separator for $i \in [2,m-1]$ and there are no more singleton  $(v,x)$-separators in $G$ apart from those in the sequence.}
Let $j$ be the highest index such that $u_j \in X$.
Note that $j$ is well defined, as $u_1 = x \in X$,
and furthermore $j \le m - 1$.
If $u_{j+1} \in \ww_G^\ell(X)$, then in particular $u_{j+1} \ne v$ and we can set $y = u_{j+1}$ to satisfy the claim.
Otherwise, let $B \subseteq V(G)$ be the biconnected component of $G$ containing $u_j$ and  $u_{j+1}$.
If suffices to show that the claim holds for $x' = u_j \in X$ and $v' = u_{j+1} \not\in \ww_G^\ell(X)$, therefore we have reduced the problem to the case when $x, v$ lie in a single biconnected component $B$ of $G$.
Let us focus on this case from now \bmp{on}.

Let the graph $G_B'$ be obtained from $B$ by
\begin{enumerate}
    \item adding a new vertex $u_f$ in the center of each face $f$ in $B$, 
    \item for each new vertex $u_f$, connecting it to all the vertices from $B$ lying on the face $f$.
\end{enumerate}
Let us note several properties of $G_B'$.
Let the set $U_0$ consist of the new single vertex put in place of the outer face \mic{of $G$}, if that vertex belongs to $V(G_B')$, or be empty if the outer face is not adjacent to any vertex of $B$.
\begin{enumerate}
    \item If $u, v \in V(B)$ and $uv \in E(G_B')$, then $uv \in E(G)$.
    \item If $S \subseteq V(B)$ is a $(v,x)$-separator in $G_B' - U_0$, then $S$ is a $(v,x)$-separator in $G$.
    \item Since $B$ is a biconnected component of $G$, each face of $B$ forms a cycle and thus $G_B'$ is a plane triangulation.
    \item Let $S_j = \{u \in V(B) \mid w_G(X, u) = j\}$. 
    Then $S_j \cup U_0$ is a $(v,x)$-separator in $G_B'$ for $j \in [\ell]$.
\end{enumerate}

The last property requires a~short justification.
Suppose there exists a $(v,x)$-path $P'$ in $G_B' - (S_j \cup U_0$) and let $P$ be the sequence of vertices from $V(B)$ occurring \bmp{on} $P'$.
Since $P'$ does not use the vertex corresponding to the outer face, each pair of consecutive vertices $y_i, y_j$ from $P$ share an interior face, thus $|w_G(X,y_i) - w_G(X,y_j)| \le 1$.
Since $w_G(X,x) = 0$ and $w_G(X,v) > \ell \ge j$, $P$~must contain a vertex from $S_j$; a contradiction.

We shall now show that for each $j \in [\ell]$ there exists a set $S'_j \subseteq S_j$ such that $G[S'_j]$ is connected and $S'_j \cup U_0$ is a $(v,x)$-separator in $G_B'$ \bmp{of one of the two desired forms}.
Let $S^0_j$ be \bmp{an} inclusion-wise minimal subset of $S'_j \cup U_0$ which is still a $(v,x)$-separator in $G_B'$.
Since $G_B'$ is a plane triangulation, \cref{lem:inseparable:traingulation} applies and we get that $G_B'[S^0_j]$ is an induced cycle.
If $S^0_j \cap U_0 = \emptyset$, then $S^0_j$ induces a cycle in $G$.
Otherwise, recall that $U_0$ contains at most one element representing the outer face.
Removing a~single vertex from a cycle turns it into a path, hence $G_B'[S^0_j \setminus U_0] = G[S^0_j \setminus U_0]$ induces a path. \bmp{Since the vertex in~$U_0$ was placed in the outer face}, the path has both \bmp{endpoints} \mic{on the outer face of~$G$.}
Note that $S'_j = S^0_j \setminus U_0$ is a $(v,x)$ separator in $G_B' - U_0$.
Because $S'_j \subseteq S_j \subseteq V(B)$, we obtain that $S'_j$ is the desired  $(v,x)$-separator in $B$, and thus in $G$.

These separators are clearly \bmp{vertex-}disjoint and contained in $\ww_G^\ell(X) \setminus X$.
The fact that \bmp{they can be ordered into a nested sequence} follows from \bmp{their} connectivity and minimality.
For $i\ne j$, the set $S'_i$ must be fully contained in some component of $G-S'_j$, precisely either in $R_G(x,S'_j)$ or in $R_G(v,S'_j)$.
From minimality we get that each vertex of $S'_j$ is adjacent to both $R_G(x,S'_j)$ and $R_G(v,S'_j)$.
In the first case $S'_i$ must separate $x$ from $S'_j$ and the second case it must separate $v$ from $S'_j$.
Therefore, either
$R_G(x,S'_i) \subseteq R_G(x,S'_j)$
or $R_G(x,S'_j) \subseteq R_G(x,S'_i)$.
\end{proof}

Note that even though the proof gives a natural order of separators $(S'_i)_{i=1}^\ell$, they might be nested in a different order.
This is because we require $S'_i \subseteq \{u \in B \mid w_G(X, u) = i\}$, not $S'_i \subseteq \{u \in B \mid w_G(x, u) = i\}$. 
This is necessary to guarantee that $S'_i \cap X = \emptyset$, which will be an important property in the proof of \cref{lem:inseparable:replacement}.

When given a \nice embedding of \bmp{a} $c$-inseparable \bmp{graph} $G\brb A$ and a solution $S$ in $G$, we are going to remove $S \cap A$ from the solution and replace it with some subset of $N_G(A)$, of size comparable to $|S \cap A|$.
As the new solution removes some vertices from  $N_G(A)$ we need to ensure that this cannot disconnect the separators constructed in \cref{lem:inseparable:halo-irrelevant}.

\begin{lemma}\label{lem:inseparable:separator-truncated}
Let $G$ be a \mic{connected} plane graph, $u, v \in V(G)$, and $S \subseteq V(G)$ be a minimal $(u,v)$-separator such that $G[S]$ either induces a cycle or a path with both endpoints \bmp{on} the outer face.
If $Y \subseteq V(G) \setminus \{u,v\}$ is a subset of vertices lying on the outer face of $G$ and $S\not\subseteq Y$, then $S \setminus Y$ is a connected $(u,v)$-separator in $G-Y$.
\end{lemma}
\begin{proof}
It is clear that  $S \setminus Y$ is a $(u,v)$-separator in $G-Y$ and we only need to show that it is connected.
We will rely on the observation that, since $S$ is minimal,
each vertex $s \in S$ is adjacent to both $R_G(u,S)$ and $R_G(v,S)$.
Let $G'$ be obtained from $G$ by adding a new vertex $f$ adjacent to all vertices in $Y$.
Since they lie on the outer face of $G$, the graph $G'$ is planar.

Suppose that $S$ induces a cycle and that $S \setminus Y$ is not connected.
Then there exist vertices $y_1, y_2 \in \bmp{Y \cap S}$ and a partition of $S \setminus \{y_1, y_2\}$ into connected vertex sets $S_1, S_2$
so that each edge set $E_G(y_i, S_j)$ for $i,j \in\{1,2\}$ is non-empty.
The sets $R_G(u,S), S_1 \cup R_G(v,S), S_2, \{y_1, f\}, \{y_2\}$ form branch sets of a $K_5$-minor model in $G'$, which gives a contradiction.

Now suppose that $S$ induces a path with both endpoints \bmp{on} the outer face and that $S \setminus Y$ is not connected.
Then there exists vertex $y \in \bmp{Y \cap S}$ and a partition of $S \setminus \{y\}$ into connected vertex sets $S_1, S_2$
so that each of them contains a vertex lying on the outer face.
The sets $R_G(u,S), R_G(v,S), \{f\}, S_1, S_2, \{y\}$ form branch sets of a $K_{3,3}$-minor model in $G'$, which again leads to a contradiction.
\end{proof}

We are ready to show that when $G\brb A$ is \nice and 8-inseparable then $\emptyset$ is an $\alpha$-preserver for $A$, where $\alpha =  3 \cdot 8 \cdot 7 = 168$.
We first fix a \nice embedding of $G\brb A$.
When $S$ is a solution in $G$ and $X = S \cap A$,
we consider terminal pairs separated by $\ww_{G \brb A}^3(X)$ and
use \cref{lem:inseparable:large-component} to identify
a subset of $N_G(A)$ of size comparable to $|X|$, which intersects all the separated terminal pairs.
We can add it to $S$, creating a larger solution $S'$ (see Figure \ref{fig:radial-ball}).
We need to argue that vertices from $X$ can be now removed from $S'$ without invalidating the solution.
To this end, we use \cref{lem:inseparable:halo-irrelevant} to construct 3 nested connected separators between any $x \in X$ and $N_G(A) \setminus S'$ in $G-(S'\setminus \{x\})$ \bmp{which} allows us \bmp{to} use \bmp{the} planarity criterion from \cref{lem:prelim:criterion:new}.

\begin{figure}[h]
\centering
\includegraphics{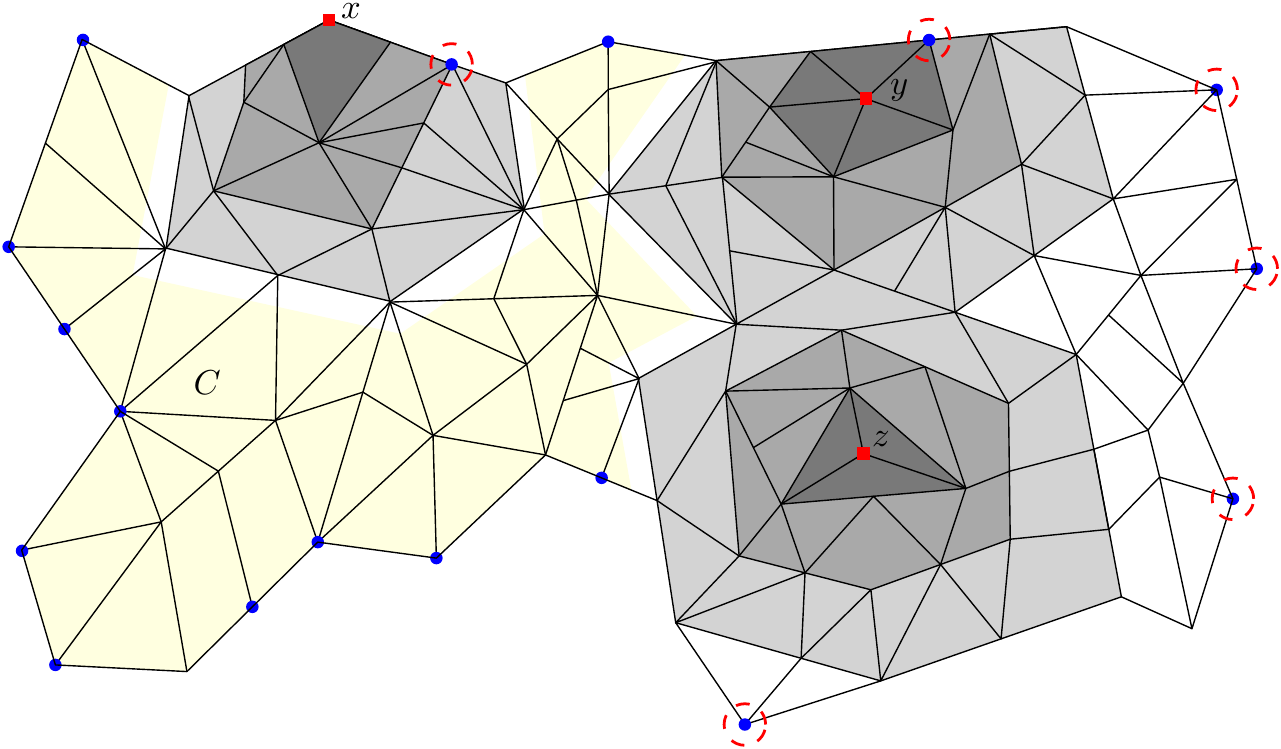}
\caption{An illustration \bmp{of} the proof of \cref{lem:inseparable:replacement}.
In the plane embedding of $G\brb A$ the boundary vertices from $N_G(A)$ are drawn in blue; they lie on the outer face.
The red squares represent vertices $X = \{x,y,z\}$ in the intersection of a solution $S$ and the set $A$.
The sets $\ww_{G \brb A}^1(X)$, $\ww_{G \brb A}^2(X)$, $\ww_{G \brb A}^3(X)$ are sketched with shades of gray so, e.g., the vertices incident to the dark gray faces belong to $\ww_{G \brb A}^1(X)$.
Due to \cref{lem:inseparable:large-component}, when $G\brb A$ is $c$-inseparable there exists a connected component $C$ of $G\brb A - \ww_{G \brb A}^3(X)$, containing all but $\Oh(|X|)$ boundary vertices.
The new solution $S'$ is obtained from $S$ \bmp{by} adding the vertices from $N_G(A) \setminus C$ (enclosed by red circles) and removing vertices from $X$.
For each vertex $v \in X$ we can construct 3 connected nested $(v,C)$-separators in the gray area, which certifies that inserting $v$ back into the graph does not affect planarity.
} \label{fig:radial-ball}
\end{figure}

\begin{lemma}\label{lem:inseparable:replacement}
Let $A \subseteq V(G)$ be such that $G\brb A$ is \bmp{an} 8-inseparable \nice boundaried planar graph.
Then for any planar modulator $S \subseteq V(G)$,
there exists a planar modulator $S' \subseteq (S \setminus A) \cup N_G(A)$,
such that $S' \setminus N_G[A] = S \setminus N_G[A]$ and $|S'| \le |S| + 168 \cdot |S \cap A|$.
\end{lemma}
\begin{proof}
Let $r = |N_G(A)|$, $X = S \cap A$, and $d = |X|$.
If $3 \cdot 8 \cdot 7 \cdot d > r$, then we set $S' = (S \setminus A) \cup N_G(A)$.
Since the graph $G[A]$ is planar, $G - (A \cup S')$ is a subgraph of $G-S$, and $G-S'$ is a disjoint union of these graphs, we get that $S'$ is also a planar modulator \bmp{of the desired size}. \bmp{(This simple case actually forms the worst-case scenario for the size of~$S'$.)}

Suppose now that $3 \cdot 8 \cdot 7 \cdot d \le r$.
\mic{Note that all boundary vertices are connected in $G\brb A$ due to 8-inseparability and, by the definition of being \nice \mic{ (\cref{def:prelim:single-faced})}, $G\brb A$ is connected.}
Let us fix any \nice embedding of $G\brb A$, so 
the set $\ww_{G \brb A}^3(X)$ is well-defined.
We apply Lemma~\ref{lem:inseparable:large-component} to $G \brb A$ and $X = S \cap A$, with $c = 8$ and $\ell = 3$, to infer that there exists a single connected component $C$ of $G \brb A - \ww_{G \brb A}^3(X)$ containing at least $r - 8 \cdot 7 \cdot d$ vertices from $N_G(A)$.
We set $S' = (S \setminus A) \cup (N_G(A) \setminus C)$, which satisfies $S' \setminus N_G[A] = S \setminus N_G[A]$.
Clearly, $|N_G(A) \setminus C| \le 8 \cdot 7 \cdot d$ and so
$|S'| \le |S| +  8 \cdot 7 \cdot |S \cap A|$.

We need to show that $S'$ is a planar modulator.
If $N_G(A) \cap C \subseteq S$, then $N_G(A) \subseteq S'$ and we again obtain that $G-S'$ is a disjoint union of planar graphs.
Suppose otherwise that there exists $v \in (N_G(A) \cap C) \setminus S$.
Let $X_0$ be an inclusion-wise minimal subset of $X$, such that $G - (S' \cup X_0)$ is planar.
Note that $X_0$ is well-defined because $S \subseteq S' \cup X$ and  thus $G - (S' \cup X)$ \bmp{is} planar.

We are going to prove that $X_0 = \emptyset$\bmp{, thereby establishing that~$S'$ is a solution}. 
Suppose otherwise and let $x \in X_0$.
We apply
Lemma~\ref{lem:inseparable:halo-irrelevant}
to $G \brb A$, and $v,x,X$ as above.
We obtain that in $G\brb A$ there exists either (a) a vertex $y \in \ww_{G \brb A}^3(X) \setminus X$ which separates $v$ from $x$ or (b) 3 nested \bmp{vertex-disjoint} connected $(v,x)$-separators $Y_1, Y_2, Y_3 \subseteq \ww_{G \brb A}^3(X) \setminus X$.
Since $\ww_{G \brb A}^3(X) \subseteq N_G[A] \setminus C$ and $C \ni v$ induces a connected subgraph of $G \brb A$, these must also be $(x,C)$-separators in $G\brb A$.

Let $X_1 = X_0 \setminus \{x\}$.
We want to show that $G_1 = G-(S'\cup X_1)$ is still planar.
Let us note an easy observation (*) that whenever $Y \subseteq N_G[A] \setminus (S' \cup X)$ is a restricted $(x,C \setminus S')$-separator in $G\brb A - S'$, then $R_{G_1}(x,Y) \subseteq A$.
To see this, consider any path in $G\brb A$ that starts at $x$
and reaches $N_G(A) \setminus S' \subseteq C$:
it clearly must go through $S' \cup Y$, so it is not present in $G_1-Y$.

In case (a) there exists a vertex $y \in \ww_{G \brb A}^3(X) \setminus X$ which separates $x$ from $C$ in $G \brb A$.
If $y \in S'$ then 
by observation (*) for $Y = \emptyset$ we get that
the connected component of $x$ in $G_1$ is contained in $A$ so it is planar
and the remaining components of $G_1$ are subgraphs of $G - (S' \cup X_0)$, so $G_1$ is planar.
If $y \not\in S'$, then by observation (*) for $Y = \{y\}$ we get that \bmp{$R_{G_1}(x,\{y\}) \subseteq A$, so that}~$R_{G_1}[x,\{y\}] \subseteq N_G[A]$. \bmp{Hence the latter} again induces a planar subgraph of $G_1$.
For any other component $U$ of $G_1 - y$
the graph $G_1\brb U$ is
a subgraph of $G - (S' \cup X_0)$, hence it is planar, and each biconnected component of $G_1$ is planar \bmp{and so} we infer that $G_1$ is planar.
This contradicts the assumption that $X_0$ is \bmp{minimal}.

In case (b) there are 3 nested connected $(x,C)$-separators $Y_1, Y_2, Y_3 \subseteq \ww_{G \brb A}^3(X) \setminus X$.
If for any $i \in [3]$ the set $Y_i$ is contained in $S'$, then by observation (*) for $Y = \emptyset$
we get that the connected component of $x$ in \bmp{$G_1$} is contained in $A$ so it is planar
and the remaining components of $G_1$ are subgraphs of $G' - X_0$, so $G_1$ is planar.
Otherwise let $Y'_i = Y_i \setminus S'$ for $i \in [3]$.
Note that \bmp{$S' \cap N_G[A] \subseteq N_G(A)$} lies on the outer face of $G \brb A$, so by \cref{lem:inseparable:separator-truncated} the sets  $Y'_1, Y'_2, Y'_3$ form nested connected $(x,C)$-separators in $G\brb A - S'$.
Note that $Y'_i \cap X_1 \subseteq \bmp{Y_i} \cap X = \emptyset$ so $Y'_i \subseteq N_G[A] \setminus (S' \cup X)$, so these are present in the graph $G_1$.
For each $i \in [3]$
we take advantage of observation (*) with $Y = Y'_i$ to
see that $R_{G_1}[x, Y'_i] \subseteq N_G[A]$ and so these \bmp{sets} induce planar subgraphs of $G_1$.
Furthermore, the graph $G_1 - x = G - (S' \cup X_0) $ is planar by assumption.
By
\cref{lem:prelim:criterion:new} the graph $G_1$ is planar.
This again contradicts the \bmp{minimality of~$X_0$}. 
We conclude that $G-S'$ is planar.
\end{proof}

We have shown that when $G\brb A$ is \nice and 8-inseparable then any solution to $G$ can be made disjoint from $A$ with a small cost.
This finishes the proof of \cref{lem:inseparable:preserver-nice} (restated below) and the \bmp{construction} of solution preservers for \nice boundaried graphs.

\restInseparablePreserverNice*
\begin{proof}
We apply \cref{lem:inseparable:decomposition-full} to compute, in polynomial time, a boundaried decomposition $(Y, \mathcal{C})$ of $G\brb A$, so that $|Y| \le |N_G(A)|^2$ and for each $C \in \mathcal{C}$ it holds that $G\brb C$ is \nice and 8-inseparable.
By \cref{lem:inseparable:replacement}, the empty set forms a 168-preserver for each $C \in \mathcal{C}$.
Hence, \cref{lem:inseparable:preserver} implies
that $Y$ is a 168-preserver for $A$.
\end{proof}

\section{Compressing the negligible components}
\label{sec:compressing}

After collecting the solution preservers, we obtain a planar modulator $X' \subseteq G $ with the following property:  there exists a 170-approximate solution $S'$ such that $S' \subseteq X'$.
In the final step of the algorithm we want to compress each connected component $C$ of $G-X'$; these component are called \emph{negligible}.
By \cref{lem:boundaried:to-outerplanar} we can augment the solution preservers computed for boundaried plane graphs with bounded radial diameter, so that the negligible components can be represented with boundaried plane graphs which are \circum.
Recall that in a \circum boundaried plane graph $H$
\mic{the image of each vertex in~$\partial(H)$ belongs to the outer face of~$H - \partial(H)$.}
We refer to the set of vertices of $H-\partial(H)$ lying on the outer face as $\Delta(H) \subseteq V(H) \setminus \partial(H)$; equivalently, this is the first outerplanarity layer of  $H-\partial(H)$.
We are going to decompose $H$ into cycles and bridges, making it easier to analyze~\bmp{$H$} through the criterion from
\cref{lem:planarity:characterization}. \bmp{Refer to Figure~\ref{fig:bridgedecomp} for an illustration of the following concept.}



\begin{definition}\label{def:undeletable:decomposition}
A bridge decomposition of a \circum $r$-boundaried plane graph $H$ is a pair $(Y, \bb)$ satisfying the following conditions:
\begin{enumerate}
    \item $Y \subseteq \Delta(H)$ and $|Y| \le 24 \cdot r$,
    \item $\bb$ is a subfamily of the connected components of $H - (\partial(H) \cup \Delta(H))$ \bmp{with} $|\bb| \le 24 \cdot r$,
    \item each connected component of $H - (\partial(H) \cup Y \cup \bigcup_{A \in \bb} V(A))$ has at most 2 neighbors in $\partial(H)$, at most $10$ neighbors in $Y$, and is adjacent to at most 10 components from $\mathcal{B}$.
\end{enumerate}
\end{definition}

\begin{figure}
    \centering
    \includegraphics{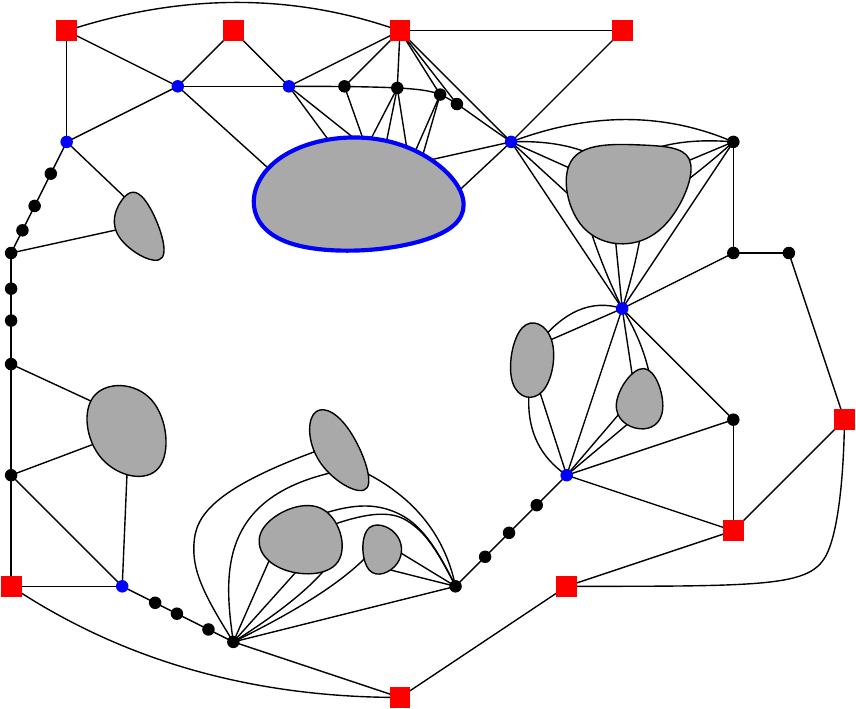}
    \caption{Schematic illustration of a bridge decomposition of a circumscribed 9-boundaried graph~$H$. The boundary~$\partial(H)$ is drawn as red squares. The vertices of~$\Delta(H)$ are shown by circles. Each connected component of~$H - (\partial(H) \cup \Delta(H))$ is a shaded region. The set~$Y$ consists of the blue circles and the set~$\mathcal{B}$ of the single component inside the blue curve. After removing the blue and red objects, the neighborhoods of the resulting components are bounded.}
    \label{fig:bridgedecomp}
\end{figure}

Observe that if we only could assume that the subgraphs from $\bb$ have constant size, then we could mark all the vertices in $Y \cup \bigcup_{B \in \bb} V(B)$, and the leftover components would have only $\Oh(1)$ neighbors each, making them amenable to protrusion replacement.
However these subgraphs may be large and we will need to compress them.

First we show that a bridge decomposition can be easily computed with known techniques for processing bounded-treewidth subgraphs.

\begin{lemma}\label{lem:undeletable:decomposition-find}
There is a polynomial-time algorithm that, given a \circum $r$-boundaried plane graph $H$, returns \bmp{a} bridge decomposition \bmp{of~$H$}.
\end{lemma}
\begin{proof}
Let $\bb_0$ denote the family of all connected components of $H - (\partial(H) \cup \Delta(H))$.
Consider a \circum $r$-boundaried plane graph $H'$ obtained from $H$ by contracting each component from $\bb_0$ into a single vertex.
It is clear that this operation does not affect the boundary and the outer face, so we can write $\partial(H') = \partial(H)$ and $\Delta(H') = \Delta(H)$.
Let us refer 
to the set of \mic{vertices obtained by contracting components from $\bb_0$} as $B' \subseteq V(H')$.
We have a natural mapping $\phi_B \colon B' \to \bb_0$ which \bmp{maps} vertices of $B'$ to their pre-image before contraction. 
The graph $H'-\partial(H')$ is 2-outerplanar (see Definition~\ref{def:diameter:outerplanarity-layers}) therefore its treewidth is at most 5~\cite{bodlaender1998partial}.
\bmp{A} corresponding tree decomposition of width 5 can be found in polynomial time~\cite{bodlaender1993linear}.

We apply \cref{lem:undeletable:treewidth-modulator} to $H'$ with $S=\partial(H')$ and $\eta=5$ to mark a set $Y' \subseteq V(H') \setminus \partial(H')$ such that $|Y'| \le 24\cdot r$ and
each connected component of $H' - (\partial(H') \cup Y')$ has at most 2 neighbors in $\partial(H')$ and at most $10$ neighbors in $Y'$.
Let \bmp{$Y  = Y' \cap \Delta(H')$ be the set of marked vertices on the outer face, $\bb = \{ \phi_B(b) \mid b \in Y' \cap B' \}$} be the set of connected components from $\bb_0$ that got marked after contraction.
Each connected component of $H - (\partial(H) \cup Y \cup \bigcup_{B \in \bb} V(B))$ corresponds to a connected component of $H' - (\partial(H') \cup Y')$ and its neighborhood is the union of the corresponding branch sets.
The claim follows.
\end{proof}


The next lemma explains the correspondence between a bridge decomposition and the concept of a $C$-bridge.
The assumption that $H$ is \circum plays an important role here. \bmp{Recall the notion of overlap graph~$O$ from Lemma~\ref{lem:planarity:characterization}.}

\begin{lemma}\label{lem:undeletable:c-bridge}
Let $(Y, \bb)$ be a bridge decomposition of a \circum boundaried plane graph $H$ and $B \in \bb$.
Then there exists a cycle $C$ such that $V(C) \subseteq \Delta(H)$ and $B$ is a connected component of~$H - V(C)$. Furthermore, the graph $O(H - \partial(H), C)$ has no edges, that is, no two $C$-bridges in $H - \partial(H)$ overlap.
\end{lemma}
\begin{proof}
\bmp{Consider a \circum boundaried plane graph $H'$ obtained from $H$ by contracting each component of $H - (\partial(H) \cup \Delta(H))$ into a single vertex.
Let $b$ be the vertex given by contracting $B$. Since~$B \cap \Delta(H) = \emptyset$, vertex~$b$ does not lie on the outer face of~$H'$. Hence the image of~$b$ lies in an interior face~$f$ of~$H[\Delta(H)]$ consisting of vertices which lie on the outer face of~$H$, which implies that the boundary cycle of~$f$ is a simple cycle~$C$ with~$V(C) \subseteq \Delta(H)$. We proceed to show that~$C$ has the desired properties.

Observe that~$N_{H'}(b) \subseteq V(C)$: by planarity of the embedding,~$b$ cannot have neighbors on the exterior of~$C$, while~$b$ cannot have neighbors embedded in the proper interior of~$C$ as they would have been part of the same connected component of~$H - (\partial(H) \cup \Delta(H))$ and therefore would have been contracted when forming~$b$. Hence~$N_{H'}(b) \subseteq V(C)$ which implies~$N_H(B) \subseteq V(C)$. Since~$H[B]$ is connected, this implies~$B$ is a connected component of~$H - V(C)$. It remains to prove that no two $C$-bridges in~$H - \partial(H)$ overlap.

Note that for each connected component~$A$ of~$H - V(C)$ that is embedded on the exterior of~$C$, we have~$|N_G(A) \cap V(C)| \leq 2$: if~$A$ had at least three neighbors on the cycle~$C$, then one of these neighbors does not lie on the outer face of~$H$ since it is enclosed by~$A$ and a path around the cycle. The same argument shows that if~$|N_G(A) \cap V(C)| = 2$, then the two neighbors of~$A$ on~$C$ are consecutive along~$C$.
Note that by definition, such a $C$-bridge does not overlap any other $C$-bridges.
\mic{This holds also when $|N_G(A) \cap V(C)| = 1$.}
Finally, no $C$-bridge that is a chord of~$C$ can be embedded on the exterior of~$C$, as it would again separate a vertex of~$C$ from the outer face.

When it comes to $C$-bridges corresponding to connected components of~$H - V(C)$ or chords of~$C$ embedded on the interior of~$C$, by \cref{lem:prelim:overlap} any two $C$-bridges which overlap are embedded on opposite sides of~$C$, so there are no overlaps between $C$-bridges embedded in the interior. The claim follows.}
\end{proof}

Let $G$ be a graph, $D \subseteq V(G)$, and \bmp{consider} the embedding of $G\brb D$ as a \circum boundaried plane graph $H$.
If $A \in \bb$ is a $C$-bridge for some cycle $C$ within $\Delta(H)$, then any planar modulator $S$ in $G$ which is disjoint from $D$, remains valid if we replace $A$ with an equivalent $C$-bridge, that is, having the same attachments at $C$ and not violating planarity of $A \cup C$ (see
\cref{lem:planarity:characterization}).
This holds in particular for a bridge obtained by contracting $A$ into a single vertex.
However such a modification is insufficient to provide the backward safeness: for any solution of size at most $k$ (possibly intersecting $D$) we need to be able to lift it back.
A possible ``dirty'' solution is to mark such vertices as undeletable.
This however would lead to producing an instance of a different problem \bmp{(which could be harder to approximate)} and would by insufficient to obtain \cref{thm:approximation}.
Instead, we will replace $A$ with a \emph{well} of depth $\Oh(k)$ and perimeter $k^{\Oh(1)}$.
We perform this operation in several steps.
First we introduce the concept of drawing cycles within $C$ surrounding the vertices from $A$.

\begin{definition}
Let $A$ be a vertex subset of a plane graph $G$, so that the induced embedding of $G \brb A$ is \circum.
We say that $A$ is $c$-well-nested in $G$ if the following conditions hold:
\begin{enumerate}

    \item the \mic{first}~$c$ \bmp{outerplanarity} layers of $G[A]$ are cycles of fixed length $m \ge 3$, the $i$-th outer cycle comprising 
    vertices $v^i_1, \dots, v^i_m$, \bmp{ordered clockwise},
    

    \item for each $j \in [m]$, there is at most one vertex $u \in N_G(A)$ such that $uv^1_j \in E(G)$,
   
   
    \item $v^i_j$ is adjacent to $v^{i+1}_j$ for \bmp{all} $i \in [c-1]$, $j \in [m]$.
\end{enumerate}
\end{definition}

\begin{figure}
\centering
\includegraphics{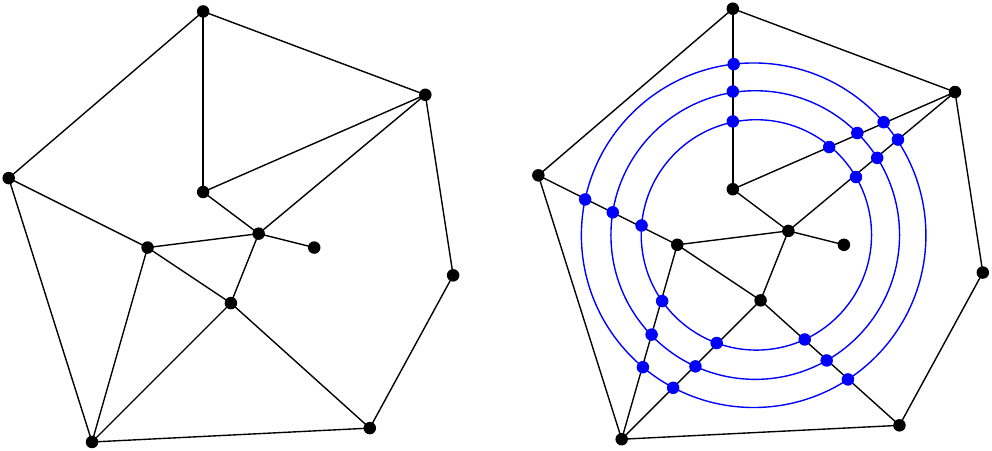}
\caption{A result of 3-nesting the vertex set given by the five vertices in the center.} \label{fig:nesting}
\end{figure}

\begin{definition}\label{def:undeletable:nesting}
Let $A$ be a vertex subset of a (boundaried) plane graph $G$, so that the induced embedding of $G \brb A$ is \circum
and $E_G(A, V(G) \setminus A) \ge 3$.
The act of $c$-nesting the set $A$ means
replacing $G$ by $G'$ given by
\begin{enumerate}
    \item enumerating the edges in $E_G(A,V(G) \setminus A)$ in clockwise order with respect to the embedding of $G\brb A$; let us refer to them as $e_1, e_2, \dots, e_m$, 
    \item subdividing each edge $e_i$ $c$ times, creating vertices $v^1_i, v^2_i, \dots, v^c_i$, counting towards the endpoint in $A$,
    \item connecting vertices $v^j_i$ and $v^j_{i+1}$ for $j \in [c], i \in [m-1]$, as well as  $v^j_m$ and $v^j_1$ for \bmp{each} $j \in [c]$,
\end{enumerate}
and replacing $A$ with $A'$ being the union of $A$ and all the new vertices.
\mic{See Figure \ref{fig:nesting}.}
We treat \bmp{the} remaining vertices and the corresponding points in the embedding as unaffected, that is, $V(G') \setminus A' = V(G) \setminus A$.
\end{definition}

The operation of $c$-nesting is clearly backward safe as the original graph is a minor of the graph after the modification.
We take note of several other useful properties of $c$-nesting.

\begin{observation}\label{obs:undeletable:nesting}
Let $H$ be a \circum boundaried plane graph with a bridge decomposition $(Y, \bb)$ and $A \in \bb$.
Suppose that $H'$ and $A'$ 
have been obtained by $c$-nesting the set $A$.
Then the following hold.
\begin{enumerate}
    \item $H'$ is a \circum boundaried plane graph, 
    \item $(Y, \bb')$ is a bridge decomposition of $H'$, where $\bb'$ is given by replacing $A$ with $A'$ in $\bb$,
    \item the set $A'$ is $c$-well-nested in $H'$,
    \item the size of each of the $c$ outermost layers in $H'[A']$ equals the number of edges in $E_{H'}(A', V(H') \setminus A')$, which is the same as the number of edges in $E_{H}(A, V(H) \setminus A)$,
    \label{obs:undeletable:nesting:item:size}
    \item $H$ can be obtained back from $H'$ by removing or contracting edges which are non-incident to the boundary \bmp{$\partial(H) = \partial(H')$}.
    \label{obs:undeletable:nesting:item:minor}
\end{enumerate}
\end{observation}

We now prove the forward safeness of $c$-nesting with respect to solutions disjoint from the subgraph being decomposed.

\begin{lemma}\label{lem:undeletable:nesting-safeness}
Consider a graph $G$, vertex set $D \subseteq V(G)$,
so that $G\brb D$ admits a \circum embedding, and \bmp{let} $H = G\brb D$ denote the corresponding \circum boundaried plane graph.
Let $(Y, \bb)$ be a bridge decomposition of $H$ and $A \in \bb$.
Suppose that $(H',A')$ has been obtained by $c$-nesting the set $A$ in $H$ and let $G' = G_\partial[\overline{D}] \oplus H'$.
If $S \subseteq V(G) \setminus D$ is a planar modulator in $G$, then
it is a planar modulator in~$G'$.
\end{lemma}
\begin{proof}
By Lemma~\ref{lem:undeletable:c-bridge} there exists a cycle $C$ such that $V(C) \subseteq \Delta(H) \subseteq D$ and $A$ is \bmp{a connected component of~$H - V(C)$.} 
Since $S \cap D = \emptyset$,
$A$ is also a \bmp{connected component of $(G-S)-V(C)$}.
We apply Lemma~\ref{lem:planarity:characterization} with respect to $G-S$ and $C$.
The overlap graph $O(G-S,C)$ is not affected by $c$-nesting the set $A$, and the only $C$-bridge that has been modified is the one generated \bmp{by component $A$ of~$(G-S)-V(C)$, which is replaced by $A'$.}
The graph induced by $A' \cup V(C)$ is a subgraph of
$H'$, hence it is planar.
Therefore $S$ is still a planar modulator in $G_\partial[\overline{D}] \oplus H'$.
\end{proof}

\mic{
After $c$-nesting some set of vertices $A$ for $c=\Omega(k)$, they become surrounded by $\Omega(k)$ nested cyclic $A$-planarizing separators.
This suffices to make each vertex in $A$ $k$-irrelevant.}
We would like now to remove the vertex set $A$ but the number of vertices introduced during $c$-nesting might still be large if the number of edges outgoing from $A$ was large.
In order to handle this, we first state a simple criterion for proving that some separator planarizes a vertex. 

\begin{lemma}\label{lem:undeletable:boundaried-separator}
Let $G$ be a graph, $A \subseteq V(G)$, $S \subseteq A$, and $v \in A \setminus S$.
If $R_{G \brb A}(v, S) \cap N_G(A) = \emptyset$ then 
 $R_{G \brb A}(v, S) =  R_{G}(v, S)$.
\end{lemma}
\begin{proof}
Assume the contrary: then there exists a path $P$ in $G-S$ that starts at $v$ and ends at $u \not\in R_{G \brb A}(v, S)$.
The path $P$ cannot be fully contained in $A$ because then we would have $u \in R_{G \brb A}(v, S)$.
Therefore the path $P$ intersects $N_G(A)$: let $P'$ be the subpath of $P$ starting at $v$ and ending \bmp{with} the first occurrence of a vertex $x$ from $V(P) \cap N_G(A)$.
Then $P'$ witnesses that $x \in  R_{G \brb A}(v, S)$ but this contradicts the assumption that $R_{G \brb A}(v, S) \cap N_G(A) = \emptyset$.
The claim follows.
\end{proof}

The next lemma states then when the set $A \in \bb$ is $(k+5)$-nested in $G\brb A$ and the number of edges outgoing from $A$ is large compared to $k$ and $|N_G(A)|$, then there exists an edge which is $k$-irrelevant.
Recall the notion of a pseudo-nested sequence from \cref{def:prelim:pseudo-nested}.
In the following proof, we slightly abuse the notation and write $R_G(v, P)$ to denote $R_G(v, V(P))$ when $P$ is a subgraph of $G$ (a path or a cycle).

\begin{lemma}\label{lem:undeletable:irrelevant-edge}
Let $G$ be a graph, $k$ be an integer, and
$D \subseteq V(G)$ \bmp{with} $|N_G(D)| = r$.
Suppose that $G\brb D$ has a
 \circum embedding and \bmp{let} 
$(Y, \bb)$ be a bridge decomposition of $G\brb D$ with respect to this embedding.
Suppose that $A \in \bb$ is $(k+5)$-well-nested and satisfies $|E_G(A,N_G(A))| > 2 \cdot r^2 \cdot (2k+7)^3$.
Then there exist vertices $v \in \partial_G(A)$, $u \in N_G(A)$, such that $uv \in E(G)$ and one of the following holds:
\begin{enumerate}
    \item there exists a sequence of $k+3$ nested $v$-planarizing connected vertex sets in $G$, or 
    \item there exists a sequence of $k+3$ pseudo-nested $v$-planarizing cycles in $G$. 
\end{enumerate}
Furthermore, such \bmp{a} pair $(u,v)$ can be found in polynomial time.
\end{lemma}
\begin{proof}
Let $A_i$ be the $i$-th outermost cycle in $G[A]$, where $i \in [k+5]$, and let $m$ denote the length of each cycle $A_i = (v^i_1, \dots, v^i_m)$.
Let $C$ be \bmp{a} cycle for which $N_G(A) \subseteq V(C) \subseteq D$ and $A$ is a \bmp{connected component of~$G \brb D - V(C)$ and therefore of~$G - V(C)$,}  guaranteed to exist by Lemma~\ref{lem:undeletable:c-bridge}.
We shall consider three scenarios. \bmp{We refer to Figure~\ref{fig:negligible} for an illustration.}

First, suppose there exists a vertex $u \in N_G(A)$ such that $|N_G(u) \cap A| \ge 2k + 7$.
{Since the vertices from $N_G(u) \cap A$ lie on the cycle $A_1$, we can choose an integer sequence $(j_1, j_2, \dots, j_{2k+7})$ so that (1) for $i \in [2k+7]$ we have $j_i \in [m]$ and $uv^1_{j_i} \in E(G)$, (2) the vertices $v^1_{j_1}, v^1_{j_2}, \dots, v^1_{j_{2k+7}}$ lie in this order on the cycle $A_1$,
and (3) the edges $uv^1_{j_1}, uv^1_{j_{2k+7}}$ separate the exterior of \bmp{cycle} $A_1$ into a bounded region containing all the remaining edges $uv^1_{j_2}, \dots uv^1_{j_{2k+6}}$ and an unbounded region containing all vertices from $N_G(D)$ (see Figure~\ref{fig:negligible} on the left, where the edges~$uv^1_{j_1}, \ldots, uv^1_{j_{2k+7}}$ are the edges between~$u$ and the blue cycle~$A_1$, from left to right).
}


Let $v = v^1_{j_{k+4}}$.
\bmp{Intuitively, a cycle~$C_i$ consists of two paths starting in~$u$: the two paths respectively go to the vertices~$i$ steps left and right of~$v$, move inwards over~$i+1$ cycles of the well, connecting to each other via a subpath of the cycle $A_{i+1}$.}
\mic{
More precisely,
for $i \in [k+3]$ we set $h(i,1) =  j_{k+4-i}$ and $h(i,2) =  j_{k+4+i}$ to denote the two indices that are in distance $i$ from the middle.
We define a path $P^1_i$ as $(u, v^1_{h(i,1)}, v^2_{h(i,1)}, \dots, v^{i+1}_{h(i,1)})$ and similarly  $P^2_i$ as $(u, v^1_{h(i,2)}, v^2_{h(i,2)}, \dots, v^{i+1}_{h(i,2)})$.
We define a cycle $C_i$ as a concatentaion of paths $P^1_i$, $P^2_i$ and a path 
that connects $v^{i+1}_{h(i,1)}$ and $v^{i+1}_{h(i,2)}$ within the cycle $A_{i+1}$;
there are two choices for the latter path and we choose the one that, together with $P^1_i$ and $P^2_i$, separates $v$ from $V(A_{i+2})$.} 

We show that $(C_1, \dots, C_{k+3})$ is a sequence of pseudo-nested $v$-planarizing cycles.
For each $i \in [k+3]$,
the cycle $C_i$ in the plane graph $G\brb D$ separates the plane into two regions: one contains $v$ and the other one $V(A_{i+2})$ and $N_G(D)$.
By \cref{lem:undeletable:boundaried-separator} we get that
$R_G(v, C_i) = R_{G \brb D}(v, C_i) \subseteq A$.
This implies that $R_G[v, C_i]$ induces a planar graph.
Furthermore, let
$1 \le i < j \le k+3$.
By the construction, $V(C_i) \cap V(C_j) = \{u\}$.
The cycle $C_i$ separates $v$ from $V(C_j) \setminus \{u\}$ in $G\brb D$.
By the same argument as above, we obtain that
$R_G(v, C_i) \cap V(C_j) = \emptyset$.

Before analyzing the second and third scenario, let us greedily construct a family of paths $\mathcal{P}$,
which is initialized as empty.
As long as we can find a path in $G\brb D$ connecting $A$ to $N_G(D)$ which is internally vertex-disjoint from the already constructed paths, we choose a shortest path with this property and insert it to $\mathcal{P}$.
Such a family can be easily constructed in polynomial time. \mic{Recall that $A_1$ is the set of vertices on the outer face of~$G[A]$, i.e., its first outerplanarity layer.} 
By the definition of a $c$-well-nested set, each vertex $x \in A$ can have at most one neighbor in $N_G(A)$, and it is possible only when $x \in A_1$.
\bmp{Because we always choose a shortest path, which does not contain vertices from~$A$ in its interior, the $A$-endpoints of the paths in $\mathcal{P}$ are therefore distinct and belong to $A_1$.} 
\bmp{The usage of shortest paths also implies} that (*) the interior of each path  $P \in \mathcal{P}$ does not contain any vertices from $N_G(D)$ and (**) the set $V(P)$ contains exactly one vertex from $N_G(A)$.

\begin{figure}
    \centering
    \includegraphics[width=\textwidth]{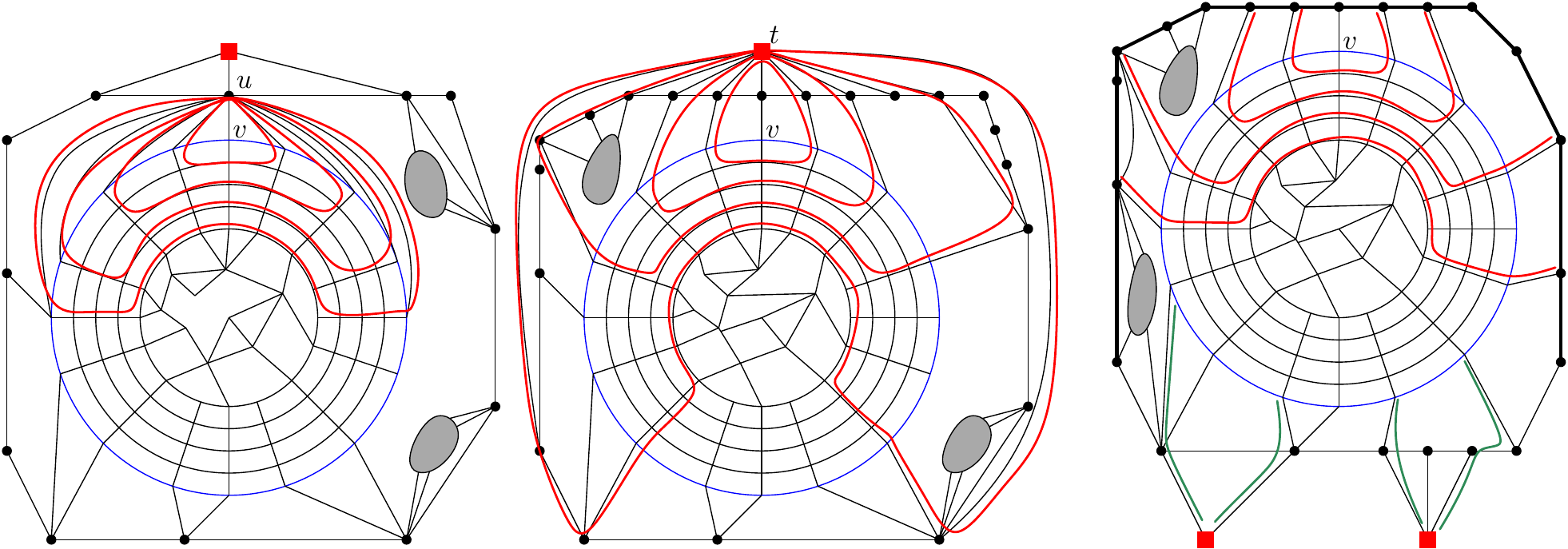}
    \caption{Illustration of three scenarios in the proof of \cref{lem:undeletable:irrelevant-edge}. Each figure shows a circumscribed boundaried plane graph~$G \brb D$ whose boundary vertices are shown as red squares, for which the vertices of~$\Delta(G \brb D)$ are visualized with black circles. A bridge decomposition of the graph consists of~$Y = \emptyset$, while~$\mathcal{B}$ consists of a single connected component~$A$ of~$(G \brb D) - (\partial(G \brb D) \cup \Delta(G \brb D))$, bounded by the blue cycle $A_1$. Left: a sequence of pseudo-nested $v$-planarizing cycles constructed in the first scenario when some vertex~$u \in N_G(A)$ has many neighbors in~$A$. Middle: a sequence of pseudo-nested $v$-planarizing cycles constructed in the second scenario. Right: A sequence of nested connected $v$-planarizing sets constructed in the third scenario. The cycle~$C$ is the cycle formed by the black round vertices. The segment~$S$ of this cycle is highlighted in bold. The paths from the family~$\mathcal{P}$ are drawn in green.}
    \label{fig:negligible}
\end{figure}

As the second scenario we consider the case $|\mathcal{P}| > r^2\cdot(2k+7)$.
\bmp{As~$|N_G(D)| = r$, this implies that} for some $t \in N_G(D)$ there exists a subfamily $\mathcal{P}_t \subseteq \mathcal{P}$
of $\ell > r\cdot(2k+7)$ internally vertex-disjoint $(A,t)$-paths in $G\brb D$.
{Since the $A$-endpoints of these paths are distinct and lie on the cycle $A_1$, we can choose an ordering $(P_1, \dots, P_\ell)$ of $\mathcal{P}_t$ and an integer sequence $(j_1, j_2, \dots, j_{\ell})$ so that (1) for $i \in [\ell]$ we have $j_i \in [m]$ and $P_i$ is a $(t,v^1_{j_i})$-path, (2) the vertices $v^1_{j_1}, v^1_{j_2}, \dots, v^1_{j_{\ell}}$ lie in this order on the cycle $A_1$,
and (3) the paths $P_1, P_\ell$ separate the exterior of \bmp{cycle} $A_1$ into a bounded region containing all the remaining paths $P_2, \dots P_{\ell-1}$ and an unbounded region.}
\mic{Since $A_1$ is a cycle, we
can construct regions $R_1, \dots, R_{\ell-1}$, so that $R_i$ is the bounded subset of the plane surrounded by $P_i$, $P_{i+1}$, and $A_1$ (exclusively), for each distinct $i,j \in [\ell-1]$ we have $R_i \cap R_j = \emptyset$ and for each $i \in [\ell-1]$ the region $R_i$ does not intersect any paths from  $(P_1, \dots, P_\ell)$.
}

\mic{
We say that index $i \in [\ell-1]$ is \emph{contaminated} if $R_i$ contains a vertex from $N_G(D)$. 
\bmp{Since the sets~$R_i$ are disjoint and~$t \in N_G(D)$ is not properly contained in any of them,} there can be at most $r-1$ contaminated indices and there exists a \bmp{contiguous} subsequence $[\ell_1, \ell_3] \subseteq [\ell-1]$ of length $2k+7$, which do not contain any contaminated indices. 
For $i \in [k+3]$ and $\ell \in [\ell_1, \ell_3]$
we define the extended path $Q^i_\ell$ as the concatenation of the path $P_{\ell}$ with the path $(v^1_{j_\ell}, v^2_{j_\ell}, \dots, v^{i+1}_{j_\ell})$ within $A$.
Let $\ell_2$ be the middle index in $[\ell_1, \ell_3]$.
\mic{For $i \in [k+3]$, 
we define cycle $C_i$ as the concatenation of the extended paths $Q^i_{\ell_2-i}$, $Q^i_{\ell_2+i}$ (they intersect only at $t$) with the path
connecting the endpoints of the extended paths within  $A_{i+1}$ (see Figure~\ref{fig:negligible} in the middle).
Similarly as before, there are two choices for the path in $A_{i+1}$ and we choose the one that makes $V(C_i)$ a $(v,V(A_{i+2}))$-separator in $G\brb D$.}
} 

Let $v = v^1_{j_{\ell_2}}$.
We show that $(C_1, \dots, C_{k+3})$ is a sequence of pseudo-nested $v$-planarizing cycles.
Let $1 \le i < j \le k+3$.
By the construction, $V(C_i) \cap V(C_j) = \{t\}$.
The cycle $C_i$ in the plane graph $G\brb D$ separates the plane into two regions: one containing $v$ and the other one containing $V(C_j) \setminus \{t\}$.
By the elimination of contaminated regions and property (*), the set $R_{G \brb D}(v, V(C_i))$ does not contain any  vertices from $N_G(D)$.
By \cref{lem:undeletable:boundaried-separator} we get that
$R_G(v, C_i) = R_{G \brb D}(v, C_i) \subseteq D$.
This implies that $R_G[v, C_i]$ induces a planar graph and $R_G(v, C_i) \cap V(C_j) = \emptyset$.

Finally, suppose
that none of the previous two scenarios occurs.
We have  $|\mathcal{P}| \le r^2\cdot(2k+7)$ and $|N_G(A)| \ge |E_G(A, N_G(A))| \, / \, (2k+6) \ge 2 \cdot r^2 \cdot (2k+7)^2$.
Let $U = N_G(A) \cap \bigcup_{P \in \mathcal{P}} V(P)$.
By property (**) each path $P \in \mathcal{P}$ goes through exactly one vertex from $N_G(A)$
and so $|U| \le r^2\cdot(2k+7)$.
\bmp{Recall the cycle~$C$ in~$G \brb D$ for which~$N_G(A) \subseteq V(C)$.} 
The set $U$ divides the cycle $C$ into $|U|$ segments;
these segments are considered without their endpoints in $U$ and some of them may be empty.
By the pigeonhole principle, one of these segments must contain
at least $(|N_G(A)| - |U|)\, /\, |U| \ge |N_G(A)|\, /\, (2\cdot |U|) \ge 2k+7$ vertices from $N_G(A)$. 
Let us refer to this segment as $S \subseteq V(C)$ and equip it with the order inherited from the clockwise orientation of the cycle $C$.
Then there exists  a sequence of distinct vertices $u_1, u_2, \dots, u_{2k+7}$ from $S \cap N_G(A)$, respecting the order above
\mic{and an integer sequence $(j_1, j_2, \dots, j_{2k+7})$ so that (1) for $i \in [\ell]$ we have $j_i \in [m]$ and $u_iv^1_{j_i} \in E(G)$, and (2) the vertices $v^1_{j_1}, v^1_{j_2}, \dots, v^1_{j_{2k+7}}$ lie in this order on the cycle $A_1$.}


We set $v = v^1_{j_{k+4}}$.
and $h(i,1) = j_{k+4-i}$ and $h(i,2) =  j_{k+4+i}$.
\mic{For $i \in [k+3]$ we define a path $P^1_i$ as
$(u_{k+4-i}, v^1_{{h(i,1)}}, v^2_{{h(i,1)}}, \dots, v^{i+1}_{{h(i,1)}})$ and $P^2_i$ as
$(u_{k+4+i}, v^1_{{h(i,2)}}, v^2_{{h(i,2)}}, \dots, v^{i+1}_{{h(i,2)}})$.
The path $P_i$ is obtained as a concatenation of $P^1_i$, the path connecting $v^{i+1}_{{h(i,1)}}$ and $v^{i+1}_{{h(i,2)}}$ within $A_{i+1}$, and finally $P^2_i$.
As before, we choose the path within $A_{i+1}$ to make $V(P_i)$ a $(v, V(A_{i+2}))$-separator in $G\brb A$
(see Figure~\ref{fig:negligible} on the right).}
Note that $V(P_i) \subseteq A \cup \bmp{(S \cap N_G(A))}$ is disjoint from all the paths from $\mathcal{P}$.
We shall prove that $(V(P_1), \dots, V(P_{k+3}))$ is a sequence of nested connected $v$-planarizing sets.

We want to show that for each $i \in [k+3]$ it holds that $R_{G \brb D}(v, P_i) \cap N_G(D) = \emptyset$.
Suppose otherwise, that there exists a path $P^v$ in $G \brb D - V(P_i)$ which connects $v$ to $N_G(D)$.
This path must intersect the interior of some path $P' \in \mathcal{P}$ because of the greedy construction. 
The set $V(P') \cap V(C)$ contains a vertex from outside $S$.
Therefore the set $V(P^v) \cup V(P')$ contains a path $P''$ connecting $v$ to $V(C) \setminus S$ in $G \brb D - V(P_i)$.
Let us consider the intersection of $V(P'')$ with $V(C)$.
First, $V(P'')$ must intersect the subsegment of $S$ between $u_{k+4-i}$ and $u_{k+4+i}$ (exclusively)
by the construction of $P_i$; let us refer to this subsegment as $S_i$.
Hence, there is a subpath of $P''$ which connects $S_i$ to $V(C) \setminus (P_i \cup S_i)$, internally disjoint from $V(C)$.
The interior of this subpath must be contained in some $C$-bridge $B$ other \bmp{than the one generated by the component $A$ of~$G - V(C)$}. 
However, as $u_{k+4-i}$ and $u_{k+4+i}$ are attachments of the bridge \bmp{generated by the connected component $A$ of~$G - V(C)$}, the bridges $A$ and $B$ overlap.
This is impossible due to \cref{lem:undeletable:c-bridge}.

We have shown that $R_{G \brb D}(v, P_i) \cap N_G(D) = \emptyset$, so by \cref{lem:undeletable:boundaried-separator} we get that $R_{G}(v, P_i) = R_{G \brb D}(v, P_i)$ and so $R_{G}[v, P_i]$ induces a planar graph.
Let $i < j \le k+3$.
By the construction, $V(P_i)$ and $V(P_j)$ are disjoint.
To see that $R_{G \brb D}(v, P_i) \cap V(P_j) = \emptyset$, observe that $V(P_j)$ belongs to the same component of $G \brb D - V(P_i)$ as $V(A_{i+2})$.
By the construction, this component does not contain $v$.

We have thus proven that there always exist a vertex $v$, with a neighbor in $N_G(A)$, that satisfies the conditions of the lemma.
The proof relies on degree counting and computing shortest paths, therefore it can be turned into a polynomial-time algorithm.
\end{proof}

We are ready to prove the main lemma of this section.
Given the boundaried graph $G\brb D$ with a \circum embedding, we compute \bmp{a} bridge decomposition $(Y,\bb)$ to identify $\Oh(r)$ subgraphs that \bmp{need} to compressed.
For $A \in \bb$, we first perform $c$-nesting for $c=k+5$, to make it amenable to \cref{lem:undeletable:irrelevant-edge}.
This allows us to reduce the number of edges outgoing from $A$.
Next, we perform $c$-nesting again (this time it suffices to take $c=k+3$) to surround all the vertices in $A$ with $k+3$ cycles of length $k^{\Oh(1)}$ and remove vertices from $A$ due to their $k$-irrelevance.
The properties of the bridge decomposition guarantees that marking all the vertices in $Y$ and \bmp{the} union of closed neighborhoods from $\bb$ (in total $k^{\Oh(1)}$ vertices now) \bmp{yields a decomposition into} subgraphs of constant neighborhood \bmp{size}, \bmp{which can be reduced by protrusion replacement}.

\begin{lemma}\label{lem:undeletable:reduction}
Let $G$ be a graph, $k$ be an integer, and
$D \subseteq V(G)$ with $|N_G(D)| = r$.
Suppose that $H = G\brb D$ admits a
 \circum embedding.
There is a polynomial-time algorithm that, given $G, D, k$, and a \circum embedding of $H$,
outputs a labeled $r$-boundaried planar graph $H'$, and a~vertex set $F' \subseteq V(H') \setminus \partial(H')$ such that the following conditions hold.
\begin{enumerate}
    \item each connected component of $H' - (\partial(H') \cup F')$ has at most 32 neighbors,
    \item $|F'| \le 50 \cdot r^3 \cdot (2k+7)^4$, 
    \item if $S \subseteq V(G) \setminus D$ is a planar modulator in $G$, then it is a planar modulator in $H' \oplus G_\partial[\overline{D}]$.\label{lem:undeletable:reduction:forward}
\item there is a polynomial-time algorithm that, given $G,D,G' = H' \oplus G_\partial[\overline{D}]$, and a planar modulator $S'$ in $G'$ of size at most $k$, returns a planar modulator in $G$ of size at most $|S'|$.
\end{enumerate}
\end{lemma}
\begin{proof}
We begin \bmp{by} 
computing a bridge decomposition $(Y,\mathcal{B})$ of $G\brb D$ (\cref{lem:undeletable:decomposition-find}).
For a set $A \in \mathcal{B}$ \mic{with at least 3 edges in $E_H(A, V(H) \setminus A)$}
we will perform a series of graph
modifications and analyze a series of tuples $(H_i, Y_i, \mathcal{B}_i, A_i)$, where \bmp{$H_i$} is a boundaried plane graph, $(Y_i,\mathcal{B}_i)$ is a bridge decomposition of $H_i$, and $A_i \in \mathcal{B}_i$.
\mic{The lower bound on the number of outgoing edges is necessary to perform $c$-nesting in a well-defined way.}
As a goal, we want to replace $A$ with a subgraph of moderate size and moderate neighborhood.
There are two kinds of safeness we need to ensure during the graph modification process.
\emph{Forward safeness} means that for every vertex set $S \subseteq V(G_\partial[\overline{D}])$ if $S$ is a planar modulator in $H_i \oplus G_\partial[\overline{D}]$, then it is also a planar modulator in $H_{i+1} \oplus G_\partial[\overline{D}]$.
\emph{Backward safeness} means that any planar modulator $S$ in $H_{i+1} \oplus G_\partial[\overline{D}]$ of size at most $k$ can be turned, in polynomial time, into a planar modulator in $H_i \oplus G_\partial[\overline{D}]$ of size at most $|S|$.

Let $\alpha = 2 \cdot r^2 \cdot (2k+7)^3$ and $A \in \bb$ be fixed.
We have $H_1 = G \brb D$, $(Y_1,\mathcal{B}_1)$ = $(Y,\mathcal{B})$, and $A_1 = A \in \mathcal{B}$.
We obtain $(H_2, Y_2, \mathcal{B}_2, A_2)$ by
$(k+5)$-nesting the set $A_1$.
The forward safeness of this modification follows directly from \cref{lem:undeletable:nesting-safeness}.
The backward safeness holds because $H_1 \oplus G_\partial[\overline{D}]$ is a minor of $H_2 \oplus G_\partial[\overline{D}]$ (\cref{obs:undeletable:nesting}(\ref{obs:undeletable:nesting:item:minor})).
In the next step we reduce the number of outgoing edges from $A_2$ in $H_2$. 
As long as it holds that
$|E_{H_2}(A_2,V(H_2) \setminus A_2)| > \alpha$,
we apply \cref{lem:undeletable:irrelevant-edge}
to find vertices $v \in \partial_{H_2}(A_2)$, $u \in N_{H_2}(A_2)$, such that $uv \in E(H_2)$ and one of the following holds:
\begin{enumerate}
    \item there exists a sequence of $k+3$ nested $v$-planarizing connected vertex sets in $H_2 \oplus G_\partial[\overline{D}]$, or
    \item there exists a sequence of $k+3$ pseudo-nested $v$-planarizing cycles in $H_2 \oplus G_\partial[\overline{D}]$. 
\end{enumerate}
In both cases the edge $uv$ is $k$-irrelevant due to Lemmas~\ref{lem:prelim:connected-separators} and~\ref{lem:prelim:pseudo-nested}, which implies backward safeness of removing $uv$.
The forward safeness of an edge removal is trivial.

At some point we reach a situation where
these reductions rules are no longer applicable.
Let $(H_3, Y_3, \mathcal{B}_3, A_3)$ denote the graph and the decomposition after performing the edge deletions.
We have
$|E_{H_3}(A_3,V(H_3) \setminus A_3)| \le \alpha$, which
implies that $|N_{H_3}(A_3)| \le \alpha$.
We obtain $(H_4, Y_4, \mathcal{B}_4, A_4)$ by $(k+3)$-nesting the set $A_3$.
The forward and backward safeness again follow from \cref{lem:undeletable:nesting-safeness} and \bmp{the minor} relation.
By \cref{obs:undeletable:nesting}(\ref{obs:undeletable:nesting:item:size}) we know that the size of each of the $(k+3)$ outermost layers in $H_4[A_4]$ is at most $\alpha$ and $|N_{H_4}(A_4)| \le \alpha$.
Consider now a vertex $v \in A_4$ which \bmp{does} not belong to any of the $(k+3)$ outermost layers in $H_4[A_4]$.
These layers form a sequence of $k+3$ nested $v$-planarizing connected sets in $H_4 \oplus G_\partial[\overline{D}]$, so $v$ is $k$-irrelevant in $H_4 \oplus G_\partial[\overline{D}]$.
We can thus remove all such vertices
and obtain the final tuple $(H_5, Y_5, \mathcal{B}_5, A_5)$.
The forward safeness is trivial and the backward safeness follows from the definition of $k$-irrelevance.
In the resulting graph we have $|A_5| \le (k+3)\cdot \alpha$ and $|N_{H_5}(A_5)| \le \alpha$.

With the recipe for compressing a single bridge component in $\mathcal{B}$ without affecting the other parts of the graph, we can perform this replacement for every $A \in \bb$ \mic{satisfying $|E_H(A, V(H) \setminus A)| \ge 3$}.
Let $(H', Y', \bb')$
be obtained from $(G \brb D, Y, \bb)$ by such a procedure
and $\bb'_3 \subseteq \bb'$ denote the subfamily of sets with at least 3 outgoing edges, for which we have performed the modification.
We assign labels to $\partial(H')$ to make them consistent with the (implicit) labeling of $\partial(H)$.
We set $G' = H' \oplus G_\partial[\overline{D}]$. 
The properties of forward and backward safeness \bmp{imply} conditions (\ref{lem:undeletable:reduction:forward})
and (4)---the existence of an algorithm that lifts solutions from $G'$ to $G$.
\mic{Let $F'$ be the union of $N_{H'}[A]$ over all $A \in \bb'_3$
plus the union of $N_{H'}(A)$ over all $A \in \bb' \setminus \bb'_3$ plus $Y'$.
For $A \in \bb' \setminus \bb'_3$ it holds $|N_{H'}(A)| \le 2$ and for  $A \in \bb'$ it holds $|N_{H'}[A]| \le (k+4) \cdot \alpha$.
}
We have $|\bb'| = |\bb|$, $|Y'| = |Y|$, and so $|F'| \le |\bb'| \cdot (k+4) \cdot \alpha + |Y'| \le 24r \cdot (2 \cdot r^2 \cdot (2k+7)^4 + 1) \le 50 \cdot r^3 \cdot (2k+7)^4$.

Finally, we want to estimate the neighborhood size of connected components $R'$ of $H' - (\partial(H') \cup F')$.
\mic{If $R' \in \bb' \setminus \bb'_3$ then clearly $|N_{H'}(R')| \le 2$.}
Otherwise, the vertex set $V(R')$ is a subset of $V(R)$, where $R$ is some component of $H - (\partial(H) \cup Y \cup \bigcup_{A \in \bb} V(A))$.
Hence, by the properties of a bridge decomposition, $R$ can have at most 2 neighbors in $\partial(H)$, at most 10 neighbors in $Y$, and can be adjacent to at most 10 components from $\bb$.
This means that $R'$ also can have at most 2 neighbors in $\partial(H')$ and at most 10 neighbors in $Y'$.
To estimate the number of neighbors of $R'$ in $N_{H'}[A]$, for $A \in \bb'$, first observe that these neighbors must belong to $N_{H'}(A)$.
It remains to show that $R'$ can be adjacent to at most 2 vertices from $N_{H'}(A)$  for a fixed $A \in \bb'$.
Suppose otherwise, that there are 3 such vertices $v_1, v_2, v_3$.
Because $(Y', \bb')$ is a bridge decomposition of $H'$, they must lie on the outer face of $H' - \partial(H')$.
We can thus add a vertex $u$ to $H' - \partial(H')$, adjacent to all $v_1, v_2, v_3$, while maintaining planarity.
However, $v_1,v_2,v_3$ together with $u, V(R'), V(A)$ form a minor model of $K_{3,3}$; a contradiction.
In summary, the set $N_{H'}(R')$ can comprise of at most 2 vertices from $\partial(H')$, 10 vertices from $Y'$, and $2\cdot 10$ vertices in total from $N_{G'}[A]$, where $A \in \bb'$, which gives at most 32 neighbors of $R'$.
\end{proof}

\section{Wrapping up} \label{sec:wrapup}

We are ready to combine all the intermediate results into a proof of the main theorem.
\cref{lem:diameter:final} provides us with a strong modulator $X$ of fragmentation $k^{\Oh(1)}$ and a family of embeddings for $G\brb C$ for $C \in \cc\cc(G-X)$ of radial diameter $k^{\Oh(1)}$.
We will now summarize Sections~\ref{sec:preservers} and~\ref{sec:compressing} with \cref{lem:undeletable:final}.
First we use \cref{lem:inseparable:preserver-diameter} to compute a 170-preserver $S_C$ for $C$,
and then \cref{lem:boundaried:to-outerplanar} to find a superset $Y_C$ of $S_C$ so that $N_G(C) \cup Y_C$ is radially connected in $G\brb C$.
The union of sets $Y_C$ together with $X$ gives a modulator $X'$ in $G$ containing some 170-approximate solution.
Furthermore, for $C' \in \cc\cc(G-X')$ we get a \circum embedding of $G \brb {C'}$.
\cref{lem:undeletable:reduction} allows us to mark $k^{\Oh(1)}$ vertices in $C'$ so that the remaining components have only $\Oh(1)$ neighbors.
Such components can be reduced to be of size $\Oh(1)$ using protrusion replacement.

\begin{proposition}\label{lem:undeletable:final}
There is a polynomial-time algorithm that, given integers $k, d, r, s$, a graph $G$ along with an $r$-strong planar modulator $X \subseteq V(G)$ of fragmentation $s$,
and for each $C \in \cc(G-X)$ a plane embedding of $G\brb C$ of radial diameter at most $d$,
returns a graph $G'$, such that
\begin{enumerate}
    \item $|V(G')| = \Oh(d^{64}r^{32}k^8s^2)$,
    \item $\mvp(G') \le 170 \cdot \mvp(G)$,
    \item there is a polynomial-time algorithm that, given $G, G'$, and a planar modulator $S'$ in $G'$ of size at most $k$, outputs a~planar modulator $S$ in $G$ such that
$|S| \le |S'|$.
\end{enumerate}
\end{proposition}
\begin{proof}
We partition $\cc\cc(G-X)$ into families $\mathcal{C}_{\le 2}$, comprising \bmp{those} components that have at most 2 neighbors, and $\mathcal{C}_{\ge 3}$, which contains all the remaining components.
As observed in the proof of \cref{lem:inseparable:preserver-circum}, for 
$C \in \cc_2$ the empty set forms a 2-preserver.
For each $C \in \cc_{\ge 3}$, 
we apply \cref{lem:inseparable:preserver-diameter} to compute a 170-preserver $S_C \subseteq C$ of size $\Oh(d^7r^4)$.
Note that $(X, \cc_{\le 2} \cup \cc_{\ge 3})$ is a boundaried decomposition of $(G,\emptyset)$ (a boundaried graph with an empty boundary).
By \cref{lem:inseparable:preserver} we obtain that $\widehat X = X \cup \bigcup_{C \in \cc_{\ge 3}} S_C$ forms a 170-preserver for $(G,\emptyset)$.
By definition, for any planar modulator $S$ in $G$ (in particular for the optimal one), there exists a planar modulator $\widehat{S} \subseteq \widehat X$ such that $|\widehat{S}| \le 170 \cdot |S|$.

We proceed with compressing the components $C \in \cc_{\ge 3}$.
Recall that \mic{$|N_G(C)| \le r$ and} we work with a plane embedding of $G \brb C$ of radial diameter at most $d$.
We apply \cref{lem:boundaried:to-outerplanar} to $Y = S_C$ to compute a set $Y_C \supseteq S_C$, $Y_C \subseteq C$ of size $\bmp{(2d+2)}\cdot (|N_G(C)| + |S_C|) = \Oh(d^8r^4)$ such that for each connected component $A$ of $G \brb C - (N_G(C) \cup Y_C)$ the boundaried plane graph $G \brb A$ is \circum.
Furthermore, by \cref{lem:boundaried:components-degree-3} there are at most $2\cdot (|N_G(C)| + |Y_C|)$ such components with at least 3 neighbors.
Let $X_0 = X \cup \bigcup_{C \in \cc_{\ge 3}} Y_C$; note that $\widehat{X} \subseteq X_0$.
\mic{Each connected component $C_0 \in \cc\cc(G-X_0)$ is contained in some connected component $C \in \cc\cc(G-X)$ so $N_G(C_0)  \subseteq N_G(C) \cup Y_C$ and thus  $|N_G(C_0)| = \Oh(d^8r^4)$.}
Again, we partition $\cc\cc(G-X_0)$ into families $\mathcal{C}^0_{\le 2}$, comprising \bmp{those} components that have at most 2 neighbors, and $\mathcal{C}^0_{\ge 3}$, which contains all the remaining components.
The number of elements of  $\mathcal{C}^0_{\ge 3}$ is at most  $\sum_{C \in \cc_{\ge 3}} (2 \cdot (|N_G(C)| + |Y_C|)) = \Oh(d^8r^4s)$.

Now for any $C_0 \in \cc^0_{\ge 3}$ we can take advantage of the \circum embedding of $G\brb {C_0}$ and apply
\cref{lem:undeletable:reduction} to $G\brb {C_0}$.
This allows us to perform a replacement of $N_G[{C_0}]$ with a compatible boundaried graph $H'_{C_0}$ such that any planar modulator disjoint from ${C_0}$ (in particular one contained in $X_0$) remains valid and we can perform a lossless lifting of any solution of size at most $k$.
Furthermore, we are given a vertex set $F'_{C_0} \subseteq V(H'_{C_0}) \setminus \partial(H'_{C_0})$ of size at most $50 \cdot |N_G({C_0})|^3 \cdot (2k+7)^4$ so that each connected component of $H'_{C_0} - (\partial(H'_{C_0}) \cup F'_{C_0})$ has at most 32 neighbors.
As before, we estimate the number of such components with at least 3 neighbors by $2\cdot |F'_{C_0}|$ (by \cref{lem:boundaried:components-degree-3}).
\mic{We have upper bounded $|N_G(C_0)|$ by $\Oh(d^8r^4)$, hence $|F'_{C_0}| = \Oh(d^{24}r^{12}k^4)$.}

After performing this replacement for each $C_0 \in \cc^0_{\ge 3}$ \bmp{we} obtain a new graph $G'$ such that \mic{the set $X_0$} stays intact and any solution $S \subseteq X_0$ remains valid.
Since $\widehat{X} \subseteq X_0$ we get that $\mvp(G') \le 170 \cdot \mvp(G)$.
Let $X' = X_0 \cup \bigcup_{{C_0} \in \cc^0_{\ge 3}} F'_{C_0}$.
The fragmentation of $X'$ in $G'$ can be \mic{again bounded by $\sum_{{C_0} \in \cc^0_{\ge 3}} (2 \cdot (|F'_{C_0}|)) = \Oh(d^8r^4s) \cdot 
\Oh(d^{24}r^{12}k^4) = \Oh(d^{32}r^{16}k^4s)$.}
Thanks to \cref{lem:undeletable:reduction} every component of $G'-X'$ has at most 32 neighbors.

By \cref{lem:protrusion:planar-deletion}
the \planardel problem admits a lossless protrusion replacer.
Let $\gamma$ be the function from \cref{def:protrusion:replacer} and $g(r) = 27(2r+7)(r+1)$.
Using \cref{lem:protrusion:planar-replacement} we can replace any connected component of $G'-X'$ which is larger than $\gamma(g(32)) = \Oh(1)$ with one smaller than $\gamma(g(32))$ without affecting the value of $\mvp(G')$ and so that any solution to the new instance can be losslessly lifted to a solution to the original instance.
We perform such a replacement for each component of $G'-X'$ with at least 3 neighbors.
This yields a total number of vertices proportional to the fragmentation of $X'$ in $G'$, that is $\Oh(d^{32}r^{16}k^4s)$.

The number of connected components from $\cc\cc(G'-X')$ with at most 2 neighbors may be large so we handle them differently.
We partition these components into groups sharing the same neighborhood.
There are ${\binom{|X'|}{2}} + |X'| \le |X'|^2$ possible neighborhoods.
Let $X_0 \subseteq X'$ be of size at most 2 and let $C(X_0)$ denote the union of these components $C_i$ from $\cc\cc(G'-X')$ whose neighborhood is exactly $N_{G'}(C_i) = X_0$.
Similarly as above we repeatedly perform protrusion replacement for each $C(X_0)$.
The graph \bmp{$G'\brb{C(X_0)}$} is connected so \cref{lem:protrusion:planar-replacement} applies and if $C(X_0)$ \bmp{is} larger than $\gamma(g(2))$ then we can replace it
with a subgraph on at most $\gamma(g(2))$ vertices.
This gives $\gamma(g(2)) \cdot |X'|^2 = \Oh(d^{64}r^{32}k^8s^2)$ vertices in total.
\end{proof}

{
Finally, we combine \cref{lem:undeletable:final} with the propositions from the previous sections to obtain our main technical contribution.
The final exponent in the approximate kernelization evaluates to 7082.
In the solution-lifting part we require a strict inequality $|S'| < \frac{k+1-3\cdot\ell}{3}$ to avoid off-by-one issues in the proof of \cref{thm:main}.}

\begin{thm}\label{thm:wrapping:summary}
There is a polynomial-time algorithm that, given a graph $G$ and an integer $k$,
either
\begin{enumerate}
    \item 
outputs a planar modulator in $G$ of size $\mvp(G)$, or 
\item correctly concludes that $\mvp(G) > k$, or
\item 
outputs a graph $G'$ \bmp{on} $\Oh(k^{7082})$ \bmp{vertices} and integer $\ell$, such that if $\mvp(G) \le k$, then $\mvp(G') \le 170 \cdot (\mvp(G) - \ell)$.
Furthermore, there is a polynomial-time algorithm that, given $G, k, G'$, and a planar modulator $S'$ in $G'$ of size $|S'| < \frac{k+1-3\cdot\ell}{3}$,
outputs a~planar modulator $S$ in $G$ such that
$|S| \le 3\cdot |S'| + 3 \cdot \ell$.
\end{enumerate}
\end{thm}
\begin{proof}
We will analyze a series of graphs $G, G_1, G_2, G_3, G_4 = G'$ obtained from $G$ by applying the intermediate reductions.
For each reduction we need to prove (a) the forward safeness: that the existence of a planar modulator in $G_i$ implies the  existence of a planar modulator in $G_{i+1}$ of comparable size, and (b) the backward safeness: that when given a planar modulator in $G_{i+1}$ within particular size bound we can lift it to a planar modulator in $G_i$.

We begin with \cref{lem:grid:final}.
We either find a planar modulator in $G$ of size $\mvp(G) \le k$, or 
correctly conclude that $\mvp(G) > k$, or
construct a graph $G_1$ together with a tree decomposition of width $\Oh(k^{36})$.
In the first two cases we can terminate the algorithm.
In the last case we also obtain
an integer $k_1 \le k$, such that if $\mvp(G) \le k$, then $\mvp(G_1) = \mvp(G) - (k - k_1)$.
Furthermore, there is a polynomial-time algorithm that, given $G, k, G_1$, and a planar modulator $S_1$ in $G_1$ of size at most $k_1$,
outputs a~planar modulator $S$ in $G$ such that
$|S| = |S_1| + (k - k_1)$.

Next, we apply \cref{lem:treewidth:final} to $G_1, k_1$, and the computed tree decomposition of width $t = \Oh(k^{36})$.
Again, if we conclude that $\mvp(G_1) > k_1$ then it implies that $\mvp(G) > k$ so we can terminate the algorithm.
Otherwise, we compute a graph $G_2$ with an $r$-strong planar modulator $X_2$, where $r = 2t+2$, and an integer $k_2 \le k_1$, such that
$\mvp(G_1) \le k_1$ implies $\mvp(G_2) \le \mvp(G_1) - (k_1 - k_2)$
and the fragmentation of $X_2$ in $G_2$ is $s = \Oh(k^{49}t^{48}) = \Oh(k^{1777})$.
Furthermore, given a planar modulator $S_2$ in $G_2$, one can turn it into a planar modulator in $G_1$ of size at most $|S_2| + 3\cdot (k_1 - k_2)$ in polynomial time.

We process the pair $(G_2, X_2)$ with \cref{lem:diameter:final} for the parameter value $k_2$.
We obtain a new graph $G_3$, an $r$-strong modulator $X_3$ in $G_3$ of fragmentation $s$,
and
a family of plane embeddings of $G_3\brb {C}$ for each connected component $C \in \cc\cc(G_3-X_3)$, each with a radial diameter 
$d = \Oh(kr) =  \Oh(k^{37})$.
We know that $\mvp(G_3) \le \mvp(G_2)$ and there is
a lifting algorithm that can turn a solution $S_3$ of size at most ${k_2}$ in $G_3$ into a solution of size at most $3 \cdot |S_3|$ in $G_2$.

The last modification applies \cref{lem:undeletable:final} to the decomposition computed so far and the parameter value $k_2$.
We obtain a graph $G_4$ with $\Oh(d^{64}r^{32}k^8s^2) = \Oh(k^{7082})$ vertices ($37\cdot 64+36\cdot 32+8+1777\cdot 2 = 7082$), such that $\mvp(G_4) \le 170 \cdot \mvp(G_3)$ and any planar modulator $S_4$ of size at most $k_2$ in $G_4$ can be turned, in polynomial time, into a planar modulator of size at most $|S_4|$ in $G_3$.

We set $\ell = k - k_2$. 
We check the forward safeness: $\mvp(G) \le k \Rightarrow \mvp(G_1) \le k_1 \Rightarrow \mvp(G_2) \le k_2 = k - \ell$
so $\mvp(G_2) \le \mvp(G) - \ell$.
Then $\mvp(G_4) \le 170 \cdot \mvp(G_3) \le 170 \cdot \mvp(G_2) \le 170 \cdot (\mvp(G) - \ell)$.
Now we check the backward safeness.
Suppose that we are given a planar modulator in $S_4$ in $G_4$ of size less than $\frac{k + 1 - 3 \cdot \ell}{3} \le k_2$.
We follow the chain of the lifting algorithms backwards from $G_4$ to $G$ to compute the following planar modulators:
\begin{itemize}
    \item $S_3$ in $G_3$ of size at most $|S_4|$,
    \item $S_2$ in $G_2$ of size at most $3 \cdot |S_3| < k + 1 - 3 \cdot \ell$, (so $|S_2| \le k - 3 \cdot \ell  \le k_2$),
    \item $S_1$ in $G_1$ of size at most $|S_2| + 3 \cdot (k_1 - k_2) \le k - 3 \cdot (k - k_1) \le k_1$, and finally
    \item $S$ in $G$ of size at most $|S_1| + (k - k_1)$.
\end{itemize}
We check that $|S| \le |S_2| + 3 \cdot \ell \le 3 \cdot |S_4| + 3 \cdot \ell$.
The claim follows.
\end{proof}

The main theorem is now a simple corollary from \cref{thm:wrapping:summary}.
\mic{It remains to unroll the definition of the approximate kernelization and reformulate the guarantees in this language.}
Recall that for a graph $G$, integer $k$,
and $S \subseteq V(G)$, we define $\Pi_\vp(G,k,S) = \min(|S|, k+1)$ if $S$ is a planar modulator in $G$, and $\Pi_\vp(G,k,S) = \infty$ otherwise.
We have that $\opt_{\Pi_\vp}(G,k) =  \min(\mvp(G), k+1)$.
Capping the solution cost at $k+1$ cannot increase the \bmp{approximation} factor of a solution, that is, $\frac{\Pi_\vp(G,k,S)}{\opt_{\Pi_\vp}(G,k)} \le \frac{|S|}{\mvp(G)}$.
When $|S| \le k$, then no capping occurs and we obtain equality.

\restThmMain*
\begin{proof}
Let $(G,k)$ be the input instance \bmp{and recall the formalization of the parameterized minimization problem~$\Pi_\vp$ from Section~\ref{sec:prelims}}.
As \bmp{the} reduction algorithm, we run the reduction algorithm from \cref{thm:wrapping:summary}.
If it terminates with option (1) or (2), we output \bmp{an} arbitrary instance, e.g., a single-vertex graph and $k' = 0$.
Otherwise, we output $G'$ and $k' = \lceil\frac{k+1-3\cdot\ell}{3}\rceil$.

The solution-lifting algorithm, given $G,k, G'$, and a planar modulator $S'$ in $G'$, first executes the reduction algorithm from \cref{thm:wrapping:summary} again.
If it terminates with option (1) we ignore $S'$ and return the optimal solution.
Also, if it terminates with option (2) we ignore $S'$ and return $S = V(G)$ with cost $\Pi_\vp(G,k,V(G)) = k+1$.
Note that when $\mvp(G) > k$ then $OPT_\vp(G,k) = k+1$ hence $V(G)$ is optimal with respect to the capped solution cost.
If the lifting algorithm terminates with option (3) but
$|S'| > \frac{k+1-3\cdot\ell}{3}$, we return $S = V(G)$: we will check that we are allowed to do this.
Otherwise we take advantage of the solution-lifting algorithm from \cref{thm:wrapping:summary}.

Let $\beta = \frac{|S'|}{\mvp(G')} \ge 1$ and $\beta^{\mathsf{cap}} = \frac{\Pi_\vp(G',k',S')}{\opt_\vp(G',k')} \ge 1$.
We have $\beta^{\mathsf{cap}} \le \beta$.
If $(510 \cdot \beta^{\mathsf{cap}}) \cdot \mvp(G) \ge k + 1$, then \bmp{by definition} any solution in $G$ evaluates at $\Pi_\vp$ to at most $k+1$ and \bmp{we get} the desired $(510 \cdot \beta^{\mathsf{cap}})$-approximation no matter whether we use the solution-lifting algorithm or just return $V(G)$.

{
Suppose now that $(510 \cdot \beta^{\mathsf{cap}}) \cdot \mvp(G) < k + 1$.
We are going to show that this implies $|S'| < \frac{k+1-3\cdot\ell}{3} \le k'$
and that the size of the solution $S$ returned by the solution-lifting algorithm is within the approximation guarantee.
From \cref{thm:wrapping:summary} we have $\mvp(G') \le 170 \cdot (\mvp(G) - \ell)$ and so
$\mvp(G') < \frac{k+1}{3\cdot \beta^{\mathsf{cap}}} - 170 \cdot \ell \le \frac{k+1-3\cdot\ell}{3} \le k'$.
Therefore the optimum in $G'$ is not being capped and $\opt_\vp(G',k') = \mvp(G')$.
\begin{align*}
\Pi_\vp(G',k',S') &= \beta^{\mathsf{cap}} \cdot \mvp(G') \le \beta^{\mathsf{cap}} \cdot 170 \cdot (\mvp(G) - \ell) \\ &<
\beta^{\mathsf{cap}} \cdot 170 \cdot \left(\frac{k+1}{510 \cdot \beta^{\mathsf{cap}}} - \ell\right) \le \frac{k+1}{3} - 170\cdot \beta^{\mathsf{cap}}\cdot \ell \le \frac{k+1-3\cdot\ell}{3}.
\end{align*}
In the last inequality we estimate simply $170\cdot \beta^{\mathsf{cap}}\cdot \ell \ge \ell$. 
This implies that $\Pi_\vp(G',k',S') \le k'$ so $\Pi_\vp(G',k',S') = |S'|$ and $\beta = \beta^{\mathsf{cap}}$.}
The inequality $|S'| < \frac{k+1-3\cdot\ell}{3}$
allows us to use the solution-lifting algorithm.
We estimate the size of the returned solution $S$. 
\[
|S| \le 3 \cdot |S'| + 3 \cdot \ell = 3 \cdot \beta \cdot \mvp(G') + 3\cdot \ell \le 3 \cdot \beta \cdot 170 \cdot (\mvp(G) - \ell) + 3 \ell \le 510  \cdot \beta \cdot \mvp(G).
\]
\bmp{This concludes the proof.}
\end{proof}

\subparagraph*{Polynomial-time approximation}
As a consequence of \cref{thm:main} we give improved approximation algorithms for \planardel.
This is obtained by pipelining the approximate kernelization with the following result by Kawarabayashi and Sidiropoulos.

\begin{thm}[\cite{kawarabayashi2017polylogarithmic}]
\label{thm:wrapping:kawarabayashi}
\planardel admits {the} following randomized approximation algorithms on $n$-vertex graphs.
\begin{enumerate}
    \item $\Oh(n^{\eps})$-approximation in time $n^{\Oh(1/\eps)}$ for any $\eps > 0$,
    \item $\Oh((\log n)^{32})$-approximation in time $n^{\Oh\left(\frac{\log n}{\log\log n}\right)}$.
\end{enumerate}
\end{thm}

We give a general observation on the result of applying $\rho(n)$-approximation to the outcome of approximate kernelization. 

\begin{lemma} \label{lem:wrapping:approximation}
Let $\alpha, c, d, r \in \mathbb{N}$ be constants and $\rho, f \colon \mathbb{N} \to  \mathbb{N}$ be functions satisfying:
\begin{enumerate}
\item \planardel admits an $\alpha$-approximate kernelization of size $f$, such that $f(k) \le c \cdot k^d$ for $k \ge r$, and
\item $\alpha \cdot \rho(k^{2d}) \le k$ for $k \ge r$. 
\end{enumerate}
Suppose that \planardel admits a (randomized) 
$\rho(n)$-approximation algorithm running in time $T(n)$ where $n$ is the number of vertices in the input.
Then there is a (randomized) algorithm that, given a graph $G$ and integer $\ell$, runs in time $\Oh(n^{r + \Oh(1)}) + T(c^{2d} \cdot \ell^{2d})$,
and either outputs a planar modulator in $G$ of size at most $\alpha \cdot \rho(c^{2d} \cdot \ell^{2d}) \cdot \mvp(G)$
or correctly concludes that $\mvp(G) > \ell$.
\end{lemma}
\begin{proof}
If $\ell < r$ we can find a solution of size at most $\ell$ in time $\Oh(n^{r+1})$ so suppose that $\ell \ge r$.
We execute the algorithm from \cref{thm:main} on the instance $(G, k)$ where $k = c \cdot \ell^2$ and obtain a new instance $(G',k')$,
where $G'$ has at most
$n' = c^{d+1} \cdot \ell^{2d} \le (c \cdot \ell)^{2d}$ vertices.
Next, we execute the $\rho(n)$-approximation algorithm on $G'$ and obtain a planar modulator $S'$ in $G'$ of size at most $\rho((c \cdot \ell)^{2d}) \cdot \mvp(G')$.
We use the solution lifting algorithm to turn $S'$ into a planar modulator $S$ in $G$ so that
\[
\frac{\Pi_\vp(G,k,S)}{\opt_\vp(G,k)} \le \alpha \cdot \frac{\Pi_\vp(G',k',S')}{\opt_\vp(G',k')} \le \alpha \cdot \frac{|S'|}{\mvp(G')} \le \alpha \cdot \rho((c \cdot \ell)^{2d}).
\]
We will show that when $\mvp(G) \le \ell$ then actually $|S| \le k$ and so $\frac{\Pi_\vp(G,k,S)}{\opt_\vp(G,k)} = \frac{|S|}{\mvp(G)}$.
As $\ell \le k$ this already implies that $\opt_\vp(G,k) = \mvp(G)$.
It remains to show that $\Pi_\vp(G,k,S) \le k$ which implies
$\Pi_\vp(G,k,S) = |S|$.
From the estimation above, we have $\Pi_\vp(G,k,S) \le \alpha \cdot \rho((c \cdot \ell)^{2d}) \cdot \opt_\vp(G,k)$ but by the assumption on $\rho$ and since $r \le \ell \le c \cdot \ell$ \bmp{and~$\opt_\vp(G,k) = \mvp(G) \leq \ell$} this is upper bounded by 
$(c \cdot \ell) \cdot \ell = k$.
\mic{Hence, if we get $|S| > k$, we can conclude that $\mvp(G) > \ell$.
Otherwise, $S$~is an $(\alpha \cdot \rho((c \cdot \ell)^{2d}))$-approximate solution.} 
\end{proof}

{Since approximation factors of the form~$(\log n)^{\Oh(1)}$ and~$\Oh(n^{\eps})$ for sufficiently small~$\eps$ satisfy the conditions of \cref{lem:wrapping:approximation}, we obtain the following consequence.}

\approximationThm*
\begin{proof}
We apply \cref{lem:wrapping:approximation} for polynomial function $f$ and functions $\rho$ from \cref{thm:wrapping:kawarabayashi}.
To see part (1), consider $\eps < \frac{1}{2d}$, where $d$ is the exponent in the kernel size from \cref{thm:main}, and $\rho(k) = \Oh(k^\eps)$.
Then $\rho(k^{2d}) = o(k)$ and the asymptotic conditions in \cref{lem:wrapping:approximation} are met.
We obtain an algorithm that, given a graph $G$ and integer $\ell$, runs in time $\Oh(n^{r + \Oh(1)}) + \ell^{\Oh(\frac{1}{2d\eps})}$, where $r$ is a constant from \cref{lem:wrapping:approximation} depending on $\eps$, and either returns a solution of size $\Oh(\ell^{2d\eps})\cdot \mvp(G)$ or concludes that $\mvp(G) > \ell$.
We execute this algorithm for $\ell = 1,2,3,\dots$ until we obtain a solution.
Then we know that $\opt = \mvp(G) \ge \ell$ and the returned solution gives an $\Oh(\opt^{2d\eps})$-approximation.
The claim follows by adjusting $\eps$.

The proof of part (2) is analogous. We consider $\rho(k) = \Oh((\log k)^{32})$ for which  $\rho(k^{2d}) = o(k)$.
The algorithm from \cref{lem:wrapping:approximation} runs in time $n^{\Oh(1)} + \ell^{\Oh\left(\frac{\log \ell}{\log\log \ell}\right)}$ and gives approximation factor $\alpha\cdot \rho(\ell^{2d}) = \Oh((\log \ell)^{32})$.
\end{proof}

\section{Planarization criteria}
\label{app:planar}


\bmp{In this section we provide the proofs of the planarization criteria stated in Section~\ref{sec:prelims:criteria}.}
{We first show how the criterion from \cref{lem:planarity:characterization} can be used to prove that a graph remains planar after some modification.
The following lemma generalizes the case where $S$ is a connected restricted $(v,V(C))$-separator which was an argument used in earlier work~\cite{JansenLS14}.
This stronger version will be later needed to prove \cref{lem:prelim:pseudo-nested}.}

\begin{lemma}\label{lem:planar:overlap}
Let $G$ be a connected graph, $C$ be a cycle in $G$, $S \subseteq V(G)$, and $v \in V(G) \setminus (V(C) \cup S)$.
Suppose that $|S \cap V(C)| \le 1$, $G[S]$ is connected, $G[S \setminus V(C)]$ is non-empty and connected, $R_G(v,S) \cap V(C) = \emptyset$, and the overlap graph $O(G-v,C)$ is bipartite.
Then the overlap graph $O(G,C)$ is bipartite as well.
\end{lemma}
\begin{proof}
For a $C$-bridge $B$ we denote by $A(B) \subseteq V(C)$ its set of attachments \bmp{on} $C$.
We say that $B_1,B_2$ are \emph{incomparable} \bmp{if} $A(B_1) \not\subseteq A(B_2)$ and $A(B_2) \not\subseteq A(B_1)$.
We \bmp{start by showing} that inserting $v$ \bmp{into} $G-v$ cannot merge two incomparable $C$-bridges $B_1,B_2$ into one $C$-bridge.
Suppose otherwise.
Let $b_1$ and $b_2$.
There exist a $(v,b_1)$-path $P_1$ in $G$ with internal vertices in $V(B_1)$ and a $(v,b_2)$-path \bmp{$P_2$} with internal vertices in $V(B_2)$.
Suppose first that $S \cap V(C) = \emptyset$.
Then $S$ is a connected restricted $(v,V(C))$-separator in $G$ and both paths $P_1,P_2$ must cross $S$.
Let $P'_1,P'_2$ be obtained from $P_1,P_2$ by removing their endpoints; note that they are non-empty.
Then $V(P'_1) \cup V(P'_2) \cup S$ induces a connected subgraph of $(G-v)-V(C)$ and contains vertices from both $B_1, B_2$; a contradiction with the assumption that $B_1, B_2$ are distinct $C$-bridges in $G-v$.

Suppose now that $S \cap V(C) = \{t\}$.
Assume \bmp{without loss of generality} that $b_1 \ne t$, \bmp{which can be achieved by swapping~$B_1$ and~$B_2$ if needed: since they are incomparable, at least one of them has an attachment point other than~$t$.}
Since $R_G(v,S) \cap V(C) = \emptyset$, the path $P_1$ must intersect $S \setminus V(C)$.
As $G[S \setminus V(C)]$ is connected, we obtain $S \setminus V(C) \subseteq V(B_1)$ and, since $G[S]$ is connected, $t \in A(B_1)$.
If \bmp{there exists} $b_2 \in \bmp{A(B_2) \setminus \{t\}}$, then by the same argument we get $S \setminus V(C) \subseteq V(B_2)$, contradicting that $B_1, B_2$ are distinct $C$-bridges.
Therefore $A(B_2) = \{t\}$, which contradicts that $B_1, B_2$ are incomparable.


\bmpr{This proof is a bit roundabout. It would be nicer to establish the re-used property (the consequence of having to cross~$S$ first.}

\bmp{Using the derived property we conclude the proof.} If $v$ is an isolated vertex of $G - V(C)$, then its insertion to $G-v$ creates a new $C$-bridge.
The attachments of the bridge $\{v\}$ must belong to $S \cap V(C)$ (since $R_G(v,S) \cap V(C) = \emptyset$) so the new bridge can have only one attachment.
Such a bridge cannot overlap with any other bridge, so this cannot create any odd cycles in the overlap graph.
Otherwise, inserting $v$ merges some bridges $B_1, B_2, \dots, B_\ell$ (possibly $\ell = 1$) into one and \bmp{augments} the set of attachments by $N_G(v) \cap V(C)$.
There cannot be an incomparable pair in $B_1, B_2, \dots, B_\ell$, so there is one $C$-bridge $B_i$ so that $A(B_j) \subseteq A(B_i)$ for $j \in [\ell]$.
If $B_j$ is overlapping with some bridge $B'$, then also $B_i$ \bmp{overlaps} with $B'$.
Similarly as in the previous paragraph, we have $S \cap V(C) \subseteq A(B_i)$ and also $N_G(v) \cap V(C) \subseteq S \cap V(C)$.
Therefore inserting $v$ into $G-v$ cannot augment the set of attachments of $B_i$.
The graph $O(G,C)$ is thus obtained from $O(G-v,C)$ by simply removing vertices representing $B_j$ for $j \ne i$.
This also cannot create any odd cycles.
\end{proof}

\mic{This is sufficient to prove the ``classic'' planarity criterion (\cref{lem:prelim:criterion:old}).
We present the following, slightly stronger, lemma, in which one of the separators does not need to be cyclic but only connected.
This directly implies \cref{lem:prelim:criterion:old} and later will be useful for proving \cref{lem:prelim:criterion:new}.}

\begin{lemma}\label{lem:prelim:criterion:old-stronger}
Let $v_0,v_3 \in V(G)$ and let $S_1,S_2 \subseteq V(G)$ be disjoint nested restricted $(v_0,v_3)$-separators \bmp{in a connected graph~$G$}, so that $G[S_2]$ has a Hamiltonian cycle and $G[S_1]$ is connected. If $R_G[v_0,S_2]$ and $R_G[v_3,S_1]$ induce planar graphs, then~$G$ is planar.
\end{lemma}
\begin{proof}
First we show that if $G-v_0$ is planar then $G$ is planar as well.
Let $C$ be \bmp{a} Hamiltonian cycle in $G[S_2]$.
By \cref{lem:planarity:characterization} the overlap graph $O(G-v_0,C)$ is bipartite and for every $C$-bridge $B$ in $G-v_0$ the graph $B\cup C$ is planar.
We apply \cref{lem:planar:overlap} for the set $S = S_1$ which has empty intersection with $V(C)$, $R_G(v_0,S) \cap V(C) = \emptyset$, and the graph $G[S]$ is connected.
We obtain that the overlap graph $O(G,C)$ is bipartite.
Let $B_0$ be the $C$-bridge in $G$ containing $v_0$.
Then $B_0 \cup C$ is a subgraph of $G[R_G[v_0,V(C)]]$ which is planar by assumption.
For any other $C$-bridge $B$ in $G$, the graph $B\cup C$ is a subgraph of $G-v_0$, hence it is planar as well.
Therefore, from \cref{lem:planarity:characterization} we infer that $G$ is planar.

Now we give a proof of the lemma by induction on the number of vertices in $V(G) \setminus R_G[v_3,S_1]$.
If $v_0$ is the only vertex in $V(G) \setminus R_G[v_3,S_1]$ then $G-v_0$ is planar and the argument above implies planarity of $G$.
When $v_0$ is an articulation point in $G$ then 
only one connected component $C_0$ of $G-v_0$ may be not fully contained in $R_G[v,S_1]$ and by \cref{obs:planar:cutvertex} it suffices to focus on the graph $G_0 = G[V(C_0) \cup \{v_0\}]$. 
The sets $S_1, S_2$ as well as the vertex $v_3$ must be contained in $V(G_0)$.
If $v_0$ is not an articulation point, we set $G_0 = G$.
Let $v' \in V(G) \setminus R_G[v_3,S_1]$ be distinct from $v_0$ and contained in $G_0-v_0$.
We have that $S_1,S_2$ are disjoint nested restricted $(v',v_3)$-separators in a connected graph $G_0-v_0$.
By induction we obtain that $G_0-v_0$ is planar.
Applying the argument from the first paragraph concludes the proof.
\end{proof}

\subsection{The new criterion}

{We move on to proving the new planarity criterion, where instead of 2 cyclic nested separators, we work with 3 connected nested separators.}
 We start with the following observation about the behavior of nested separators.

\begin{observation} \label{obs:nested:reachability}
If~$(S_i)_{i=1}^\ell$ is a sequence of nested vertex-disjoint restricted~$(v_0,v_{\ell+1})$-separators in a connected graph~$G$ such that~$G[S_i]$ is connected for all~$i \in [\ell]$, then the following holds.
\begin{enumerate}
    \item For each~$i \in [\ell-1]$ (resp.~$i \in \{2, \ldots, \ell\}$) we have~$N_G[S_i] \subseteq R_G[v_0,S_{i+1}]$ (resp.~$N_G[S_i] \subseteq R_G[v_{\ell+1},S_{i-1}]$).
    \item For each~$i \in [\ell-1]$ we have~$R_G[v_0,S_i] \subseteq R_G[v_0,S_{i+1}]$ and~$S_i \subseteq R_G(v_0, S_{i+1})$.
\end{enumerate}
\end{observation}

The following elementary lemma localizes paths connecting consecutive separators in a nested sequence.

\begin{lemma} \label{lem:paths:connecting:separators}
Let~$S_1, S_2$ be two nested vertex-disjoint restricted~$(v_0,v_3)$-separators in a connected graph~$G$, such that~$G[S_i]$ is connected for each~$i \in [2]$. If~$P$ is an~$(S_1, S_2)$-path in~$G$ whose internal vertices are disjoint from~$S_1 \cup S_2$, then~$V(P) \subseteq R_G[v_0,S_2] \cap R_G[v_3,S_1]$.
\end{lemma}
\begin{proof}
Let~$P'$ be the subpath formed by the interior vertices of~$P$. Since each~$G[S_i]$ is connected and~$G$ is connected, we have~$S_i \subseteq R_G[x, S_i]$ for each~$x \in \{v_0, v_3\}$. If~$V(P') = \emptyset$, then together with the fact that~$S_1 \subseteq R_G[v_0,S_2]$ and~$S_2 \subseteq R_G[v_3, S_1]$ by Observation~\ref{obs:nested:reachability}, the lemma follows.

So suppose~$V(P') \neq \emptyset$ and orient~$P'$ so that it starts at a neighbor~$w_1$ of~$S_1$ and ends in a neighbor~$w_2$ of~$S_2$. By the previous argument, it suffices to show that~$V(P') \subseteq R_G[v_0,S_2] \cap R_G[v_3,S_1]$. By the first point of Observation~\ref{obs:nested:reachability} we have~$w_1 \in N_G[S_1] \subseteq R_G[v_0,S_2]$ and since~$P'$ is by definition disjoint from~$S_1 \cup S_2$, reachability of the first vertex of~$P'$ implies reachability of the entire subpath:~$V(P') \subseteq R_G[v_0,S_2]$. Conversely, since~$P'$ ends in~$w_2 \in N_G[S_2] \subseteq R_G[v_3,S_1]$ and~$P'$ does not intersect~$S_1 \cup S_2$, path~$P'$ is contained entirely in~$R_G[v_3,S_1]$, which completes the proof.
\end{proof}

\restPrelimCriterionNew*
\begin{proof}
Since the separators~$S_1,S_2,S_3$ are nested,~$S_2$ is an~$(S_1,S_3)$-separator. We first argue that we can assume that~$S_2$ is in fact \emph{inclusion-wise minimal} with respect to simultaneously being a (restricted)~$(S_1,S_3)$ separator and inducing a connected graph, as follows. If~$S'_2 \subsetneq S_2$ is also a connected~$(S_1,S_3)$ separator, then the sequence~$S_1, S'_2, S_3$ is also a sequence of nested~$(v_0,v_4)$ separators satisfying the lemma statement, since each~$(S_1,S_3)$-separator is also an~$(v_0,v_4)$-separator since they are nested. As the graphs induced by~$R_G[v_0,S_3]$ and~$R_G[v_4,S_1]$ are the same with respect to this sequence, the triple~$S_1,S'_2,S_3$ also satisfies the preconditions to the lemma. 

In the remainder may therefore assume that~$S_2$ is a~$(v_0,v_4)$-separator which is inclusion-minimal with respect to being a connected~$(S_1,S_3)$-separator.

We first deal with an easy case. If~$|S_2| = 1$, then the single vertex~$x$ in~$S_2$ is an articulation point that separates~$v_0$ from~$v_4$. In this case, it is easy to verify that for each connected component~$H$ of~$G - x$, the set~$V(H) \cup \{x\}$ is a subset of~$R_G[v_0, S_3]$ or~$R_G[v_4, S_1]$ and therefore induces a planar graph by assumption. Hence by Observation~\ref{obs:planar:cutvertex}, in this case we are guaranteed that~$G$ is planar.

In the remainder we assume that~$S_2$ is not a singleton, and is inclusion-minimal with respect to being a connected~$(S_1,S_3)$-separator. For the following claim, recall that~$R_G(x,S)$ denotes the set of vertices reachable from~$x$ in~$G-S$, that is, the vertex set of the connected component of~$G-S$ containing~$x$.

\begin{claim} \label{claim:stwo:hampath}
$G[S_2]$ has a Hamiltonian path whose endpoint(s) are adjacent to both~$R_G(v_0, S_2)$ and~$R_G(v_4, S_2)$.
\end{claim}
\begin{innerproof}
Let~$T$ be a spanning tree of the connected subgraph~$G[S_2]$. Since~$S_2$ is minimal with respect to being a connected~$(S_1,S_3)$-separator, each leaf~$\ell$ of~$T$ is adjacent to both the connected component of~$G - S_2$ containing~$S_1$ (on vertex set~$R_G(v_0, S_2$)) and the connected component of~$G - S_2$ containing~$S_3$ (on vertex set~$R_G(v_4, S_2$)). If~$T$ has at most two leaves, this proves the claim. Suppose for a contradiction that~$T$ has at least three leaves~$\ell_1,\ell_2,\ell_3$, so that the set~$I \subseteq V(T) = S_2$ of internal vertices of~$T$ is non-empty. 

Let~$Q_1$ be a vertex set containing~$S_1$ and three paths, defined as follows. Each leaf~$\ell_i$ for~$i \in [3]$ has a neighbor~$w_i$ in~$R_G(v_0, S_2)$; add a shortest path from~$w_i$ to~$S_1$ through~$R_G(v_0, S_2)$ to~$Q_1$, which exists since~$S_1 \subseteq R_G[v_0, S_2]$. Since~$G[S_1]$ is connected, the graph~$G[Q_1]$ is connected. The paths added to~$Q_1$ are contained in~$R_G[v_0, S_2] \subseteq R_G[v_0,S_3]$ by Observation~\ref{obs:nested:reachability}, and since~$S_1 \subseteq R_G[v_0, S_3]$ by connectivity of~$G$ and~$G[S_1]$ we have~$Q_1 \subseteq R_G[v_0, S_3]$.

The vertex set~$Q_3$ is defined similarly: it contains~$S_3$ and for each leaf~$\ell_i$ it contains a shortest path~$P'_i$ through~$R_G(v_4, S_1)$ from a neighbor~$w'_i \in N_G(\ell_i) \cap R_G(v_4, S_1)$ to~$S_3$. By applying Lemma~\ref{lem:paths:connecting:separators} to the~$(S_2, S_3)$ path~$P_i$ obtained by appending~$\ell_i \in S_2$ to~$P'_i$, for each path~$P'_i$ added to~$Q_3$ we have~$V(P'_i) \subseteq R_G[v_0, S_3]$. Since~$S_3 \subseteq R_G[v_0, S_3]$ we have~$Q_3 \subseteq R_G[v_0, S_3]$.

The sets defined so far form a minor model of~$K_{3,3}$: the vertex sets~$Q_1,Q_3,I$ are disjoint and connected and are each adjacent to the three vertices~$\ell_1,\ell_2,\ell_3 \notin Q_1 \cup Q_3 \cup I$. As we have shown~$Q_1,Q_3 \subseteq R_G[u,S_3]$, together with~$S_2 \subseteq R_G[u,S_3]$ this implies that the entire minor model is contained in~$R_G[u,S_3]$, contradicting the assumption that~$R_G[u,S_3]$ induces a planar subgraph of~$G$.
\end{innerproof}

The following claim shows that there is a cycle in~$G$ containing the entire set~$S_2$ (which is therefore a~$(u,v)$-separator) which is contained in the region of the graph `in between' the separators~$S_1$ and~$S_2$, and can therefore be shown to be disjoint from~$S_3$.

\begin{claim} \label{claim:cycle}
There is a cycle~$C$ in~$G$ such that~$S_2 \subseteq V(C) \subseteq R_G[v_0,S_2] \cap R_G[v_4,S_1]$.
\end{claim}
\begin{innerproof}
Let~$H$ be a Hamiltonian path of~$G[S_2]$ whose two endpoints~$\ell_1, \ell_2$ are adjacent to both~$R_G(v_0, S_2)$ and~$R_G(v_4, S_2)$, as guaranteed by Claim~\ref{claim:stwo:hampath}.

For each~$i \in [2]$, let~$P_i$ be an~$(S_1, \ell_i)$-path in~$G$ whose internal vertices are disjoint from~$S_1 \cup S_2$. Such a path exists since~$\ell_i$ has a neighbor in the connected set~$R_G(v_0, S_2) \supseteq S_1$. By Lemma~\ref{lem:paths:connecting:separators}, we therefore have~$V(P_i) \subseteq (R_G[v_0, S_2] \cap R_G[v_3, S_1])$. Since~$G[S_1]$ is connected, the subgraph~$G[S_2 \cup V(P_1) \cup V(P_2)]$ is connected and contains a simple~$(\ell_1, \ell_2)$-path~$P$ with at least one internal vertex, such that the interior of~$P$ is disjoint from~$S_2$. Then~$V(P) \subseteq S_2 \cup V(P_1) \cup V(P_2) \subseteq R_G[v_0, S_3] \cap R_G[v_4, S_1]$. Hence~$P$ together with the Hamiltonian~$(\ell_1, \ell_2)$-path~$T$ spanning~$S_2$ forms a cycle~$C$ as desired, completing the proof. 
\end{innerproof}

\mic{
Fix a cycle~$C$ as guaranteed by the previous claim. Note that $V(C) \cap S_3 = \emptyset$ since $S_3 \cap R_G[v_0, S_2] = \emptyset$ by the nesting property of the separators.
Similarly we infer that $v_0 \not\in V(C)$.
Furthermore, $V(C)$ is a $(v_0,S_3)$-separator because $S_2 \subseteq V(C)$ and $S_3$ is a $(V(C),v_4)$-separators because $V(C)$ is contained in the same component of $G - S_3$ as $S_2$.
Hence, $(V(C), S_3)$ are disjoint nested $(v_0,v_4)$-separators which satisfy the prerequisites of
\cref{lem:prelim:criterion:old-stronger} (modulo ordering) and this concludes the proof.}
\end{proof}

\subsection{Irrelevant vertices}

Finally, we show how the planarity criteria give sufficient conditions for a vertex/edge to be $k$-irrelevant.

\restPrelimConnectedSeparators*
\begin{proof}
Let $X \subseteq V(G) \setminus \{v\}$ be any planar modulator in $G-v$ of size at most $k$; we have that $(G-v) - X$ is planar.
We are going to show that $G-X$ is planar as well.
There must exist indices $i_1 < i_2 < i_3$ so that $S_{i_j} \cap X = \emptyset$ for $j \in [3]$.
The sets $S_{i_1}, S_{i_2}, S_{i_3}$ form a nested sequence (with respect to $v$) in the graph $G - X$.
\bmp{They} are $v$-planarizing and induce connected subgraphs of $G$, so the same holds in $G-X$.
Consider the connected component of $v$ in $G-X$; let us refer to it as $G_v$.
If $V(G_v) \subseteq R_{G-X}[v,S_{i_3}]$ then, since $S_{i_3}$ planarizes $v$, this component is planar.
Other connected components of $G-X$ are subgraphs of $(G-v) - X$ so the entire graph $G-X$ is planar.
Otherwise there exists a~vertex $u \in V(G_v) \setminus R_{G-X}[v,S_{i_3}]$.
Now the sets $S_{i_1}, S_{i_2}, S_{i_3}$ are nested $(v,u)$-separators.
Note that $R_{G-X}[u, S_{i_1}] \subseteq V(G) \setminus \{v\}$ so it induces a planar subgraph of $G-X$.
We apply \cref{lem:prelim:criterion:new} to infer that $G_v$ is planar and thus $G-X$ is planar.
This concludes the proof that $v$ is $k$-irrelevant.
If $e$ is an edge incident to $v$, then any solution to $G\setminus e$ is also a solution to $G-v$, so the edge case reduces to the vertex case.
\end{proof}

\restPrelimPseudoNested*
\begin{proof}
Let $t \in V(G)$, $t\ne v$, be the common intersection of every pair of sets $S_i,S_j$, $1 \le i < j \le k+3$. Let $X \subseteq V(G) \setminus \{v\}$ be any planar modulator in $G-v$ of size at most $k$; we have that $(G-v) - X$ is planar.
We are going to show that $G-X$ is planar as well.
We consider two cases depending on whether $t \in S$ or not.

Suppose first that $t \in X$ and let $X_t = X \setminus \{t\}$.
Then $X_t$ is a planar modulator in $G-t$ of size at most $k-1$.
Observe that $\{S_i \setminus \{t\} \mid i \in [k+3]\}$ are disjoint and form a (standard) nested (with respect to $v$) sequence in $G-t$.
\mic{They also induce connected subgraphs of $G-t$ as removal of a single vertex from a cycle cannot disconnect it.}
Analogously as in the proof of \cref{lem:prelim:connected-separators}, we consider indices $i_1 < i_2 < i_3$ so that $S_{i_j} \cap X_t = \emptyset$ for $j \in [3]$.
The sequence $(S_{i_1} \setminus \{t\}, S_{i_2} \setminus \{t\}, S_{i_3} \setminus \{t\})$ is nested and $v$-planarizing in $(G-t)-X_t$.
From the assumption that $(G-v)-X = ((G-t)-v) - X_t$ is planar and by \cref{lem:prelim:criterion:new} we get that $(G-t)-X_t = G - X $ is planar.

Suppose now that $t \not\in S$.
Again by a counting argument, we can pick $i, j \in [k+3],\, i < j$ such that $S_i \cap X = S_j \cap X = \emptyset$.
We have that $R_{G-X}(v,S_i) \cap S_j = \emptyset$ and $|S_i \cap S_j| = 1$.
Furthermore $G[S_i]$ and $G[S_i \setminus S_j]$ are connected.
Let $C$ be the Hamiltonian cycle in $G[S_j]$.
We shall use the criterion from \cref{lem:planarity:characterization} directly.
The graph $(G-v)-X$ is planar so the overlap graph $O((G-v)-X,C)$ is bipartite and for each $C$-bridge $B$ we have that $B \cup C$ is planar.
We apply \cref{lem:planar:overlap} to the graph $G-X$, cycle $C$, and the set $S_i$ to infer that the overlap graph $O(G-X,C)$ is also bipartite.
Let $B_v$ be the $C$-bridge in $G-X$ which contains $v$.
Then $B_v \cup C$ is subgraph of $G[R_G[v,S_j]]$ which is planar by the assumption.
Since the other $C$-bridges are unaffected,
\cref{lem:planarity:characterization} implies that $G-X$ is planar.

This concludes the proof that $v$ is $k$-irrelevant.
Again, if $e$ is an edge incident to $v$, then any solution to $G\setminus e$ is also a solution to $G-v$, so the edge case reduces to the vertex case.
\end{proof}

\section{Conclusion} \label{sec:conclusion}

We presented the first constant-factor approximate kernelization for \textsc{Vertex planarization}. While several of our reduction steps are lossy and cannot directly be used in an exact kernelization, we believe the insights from our work will be important towards settling whether \textsc{Vertex planarization} admits a polynomial-size exact kernelization. In turn, this will shed further light on the algorithmic complexity of \textsc{$\mathcal{F}$-minor free deletion} for families~$\mathcal{F}$ that do not contain any planar graphs.

At the heart of our lossy kernelization is the subroutine to compute a bounded-size \emph{preserver} which contains a good approximate solution. Through several reduction steps, we showed that it suffices to perform this computation on planar subgraphs which can be embedded in the plane such that the boundary vertices by which they communicate with the rest of the graph all lie on the outer face and are bounded in number. Whether an exact preserver of polynomial size exists for such subgraphs, and whether it can be computed efficiently, is an important open question for future work. To emphasize the importance of this question, we restate it here explicitly:

\begin{problem} \label{openproblem}
Determine whether there exists a polynomial~$f \colon \mathbb{N} \to \mathbb{N}$ such that the following holds. Every planar $k$-boundaried graph~$H$ that has an embedding with all boundary vertices on the outer face, contains a set~$A$ of~$f(k)$ vertices such that for any $k$-boundaried graph~$G$, there is an optimal solution~$S$ to the \textsc{Vertex planarization} problem on~$H \oplus G$ satisfying~$S \cap V(H) \subseteq A$.
\end{problem}

\bmp{In hindsight, the insight that the task of (approximately) 
kernelizing \textsc{Vertex planarization} can be reduced to analyzing \emph{planar} subgraphs interacting with the rest of the graph through their outer face, and therefore that the problem is amenable to techniques typically applied only for problems whose input graph is planar, is one of the conceptual take-away messages of this work.} Driven by this realization, we believe that the answer to Problem~\ref{openproblem} will be the key to settling the kernelization complexity of \textsc{Vertex planarization}.

\mic{\bmp{A} second question is whether \bmp{every} \textsc{$\mathcal{F}$-minor free deletion} problem admits a constant-factor approximate kernelization. \bmp{Although a number} of steps in our algorithm \bmp{can} be generalized to graphs excluding any fixed minor, we heavily depend on the properties of plane graphs when decomposing the locally planar components.
It is plausible that these arguments may be generalized to graph families (near-)embeddable on more complex surfaces.
\bmp{A technically daunting but possibly viable challenge is therefore to} try using the 
\bmp{Robertson-Seymour decomposition}~\cite{KawarabayashiTW20,RobertsonS03a} of graphs excluding a fixed minor $H$ to develop an analog of the solution preservers framework on  $H$-\bmp{minor}-free subgraphs \bmp{for arbitrary~$H$}.
}






\bibliographystyle{alphaurl}
\bibliography{main}

\end{document}